\documentclass[10pt,a4paper]{amsart}
\usepackage[foot]{amsaddr}

\usepackage[margin=1.0in]{geometry}

\usepackage{amsthm, amsfonts,amssymb,euscript}

\usepackage{amsmath}
\usepackage{euscript,bbm,color,slashed,enumerate,bm}
\usepackage{amssymb}
\usepackage{amsthm}
\usepackage{graphicx}
\usepackage{caption}
\usepackage{subcaption}
\usepackage{amsfonts}
\usepackage{float}
\usepackage{slashed}

\usepackage[usenames,dvipsnames,svgnames,table]{xcolor}
\usepackage[colorlinks=true]{hyperref}
\hypersetup{linkcolor=BrickRed, urlcolor=green, citecolor=blue, linktoc=page}

\numberwithin{equation}{section}

\newtheorem*{proposition*}{Proposition}
\newtheorem*{theorem*}{Theorem}
\newtheorem*{conjecture*}{Conjecture}
\newtheorem*{claim*}{Claim}
\newtheorem*{lemma*}{Lemma}
\newtheorem*{corollary*}{Corollary}
\newtheorem{theorem}{Theorem}[section]
\newtheorem{proposition}[theorem]{Proposition}

\newtheorem{lemma}[theorem]{Lemma}
\newtheorem{corollary}[theorem]{Corollary}

\newtheorem*{definition*}{Definition}
\newtheorem{definition}{Definition}[section]
\newtheorem*{assumption*}{\mathcal{A}ssumption}

\newtheorem*{remark*}{Remark}
\newtheorem{remark}{Remark}[section]

\newtheorem{thmx}{Theorem}
\newtheorem{mainthm}{Theorem}

\newcommand{\la}{\langle}
\newcommand{\ra}{\rangle}
\newcommand{\R}{\mathbb{R}}
\newcommand{\s}{\mathbb{S}}
\newcommand{\C}{\mathbb{C}}

\newcommand{\Z}{\mathbb{Z}}
\newcommand{\N}{\mathbb{N}}

\DeclareMathOperator{\supp}{\textnormal{supp}}
\newcommand{\snabla}{\slashed{\nabla}}

\newcommand{\hphi}{\hat{\phi}}

\newcommand{\re}{\mathfrak{Re}\,} 
\newcommand{\im}{\mathfrak{Im}\,}
\allowdisplaybreaks

\begin{document}

\author{Dejan Gajic$^*$$^1$}\author{Claude Warnick$^{\dag}$$^1$}
\address{$^{1}$\small University of Cambridge, Department of Pure Mathematics and Mathematical Statistics, Wilberforce Road, Cambridge CB3 0WB, United Kingdom }
\email{$^*$d.gajic@dpmms.cam.ac.uk}
\email{$^{\dag}$c.m.warnick@maths.cam.ac.uk}

\title{Quasinormal modes in extremal Reissner--Nordstr\"om spacetimes}

\begin{abstract}
We present a new framework for characterizing quasinormal modes (QNMs) or resonant states for the wave equation on asymptotically flat spacetimes, applied to the setting of extremal Reissner--Nordstr\"om black holes. We show that QNMs can be interpreted as honest eigenfunctions of generators of time translations acting on Hilbert spaces of initial data, corresponding to a suitable time slicing. The main difficulty that is present in the asymptotically flat setting, but is absent in the previously studied asymptotically de Sitter or anti de Sitter sub-extremal black hole spacetimes, is that $L^2$-based Sobolev spaces are not suitable Hilbert space choices. Instead, we consider Hilbert spaces of functions that are additionally Gevrey regular at infinity and at the event horizon. We introduce $L^2$-based Gevrey estimates for the wave equation that are intimately connected to the existence of conserved quantities along null infinity and the event horizon. We relate this new framework to the traditional interpretation of quasinormal frequencies as poles of the meromorphic continuation of a resolvent operator and obtain new quantitative results in this setting. 
\end{abstract}

\maketitle

\tableofcontents

\section{Introduction}
\label{sec:intro}

An important problem in the theory of general relativity is to classify the behaviour of gravitational radiation emitted by dynamical solutions to the vacuum Einstein equations
\begin{equation}
\label{eq:eve}
\textnormal{Ric}[g]=0,
\end{equation}
where $g$ is a Lorentzian metric and $\textnormal{Ric}[g]$ is the corresponding Ricci tensor. It is expected that a significant proportion of the gravitational radiation emitted by dynamical black hole solutions to \eqref{eq:eve} may be dominated by \emph{quasinormal modes} (QNMs), also known as \emph{resonant states} as they settle down to a stationary Kerr black hole solution \cite{kerr63}; see for example the numerics in \cite{buon07, berti07}, the first experimental observations of gravitational radiation in \cite{ligo} and subsequent further analysis in \cite{isi19}. QNMs are exponentially damped, oscillating solutions to linear wave equations on fixed stationary spacetime backgrounds that are characterized by a discrete set of complex frequencies called \emph{quasinormal frequencies} (QNFs) or scattering resonances. They were first observed in numerics of Vishveshwara \cite{vish70} and have been a prevalent topic in the theoretical physics literature ever since, see the review articles \cite{Kok1999, berti09, konop11} and references therein.\footnote{See Section \ref{sec:tradapproach} for further discussion regarding the precise definitions of QNMs in the literature.} QNMs may be viewed as the dispersive analogues of the \emph{normal modes} that dictate the dynamics of an idealized vibrating string. 

If one assumes the dominance of QNMs in the gravitational radiation emitted at late time intervals in the dynamical evolution of perturbations of Kerr black holes, it is possible to entertain the notion of ``black hole spectroscopy'' \cite{dreyer04, isi19} which is the inference of properties of black hole end state from precise experimental measurements of the most dominant QNFs in the gravitational radiation.

Recently, there have been significant advances towards a mathematical proof of the Kerr Stability Conjecture in linearized settings, see \cite{part3, Dafermos2016, john18, ma2017, dhr-teukolsky-kerr} and also \cite{anderssonkerr, haf19, hung2017}.\footnote{We refer to \cite{lecturesMD} for a comprehensive list of additional earlier foundational results in the direction of linear stability.} This conjecture asserts that initially small and localized metric perturbations of sub-extremal Kerr black hole solutions to \eqref{eq:eve} should decay in time and the metric should asymptotically approach a nearby Kerr solution.  We note moreover a recent related \emph{nonlinear} stability result in the black hole setting: the stability of Schwarzschild black holes under axisymmetric, polarized perturbations \cite{klainerman17}. Stability is established by showing that perturbations decay \emph{at least} inverse polynomially in time.

In fact, early heuristic and numerical analysis of the linearized problem initiated by Price \cite{Price1972, leaver85, CGRPJP94b} suggested that one can do \emph{no better} than proving inverse polynomial decay estimates, because at sufficiently late times, the leading-order behaviour should be \emph{exactly} inverse polynomial (this is sometimes referred to as ``Price's law''). The presence of so-called \emph{polynomial tails} in the context of the linear wave equation on asymptotically flat, spherically symmetric black hole backgrounds has recently been proved in a mathematically rigorous setting \cite{paper2, paper4, extremal-prl}, where it has moreover been connected to the existence of conservation laws along null hypersurfaces, first discovered by Newman and Penrose at null infinity \cite{np2} and discovered in a different guise by Aretakis at the event horizons of extremal black holes\footnote{See \cite{are18} for a comprehensive overview on extremal black holes.} \cite{aretakis1, aretakis2, aretakis4, aretakisglue}.\footnote{The connection between the conserved quantities of Newman--Penrose and Aretakis was first observed in \cite{hm2012, bizon2012}.} In Fourier space, polynomial tails can alternatively be related to the precise behaviour of resolvent operators corresponding to the wave equation near the zero time frequency, see \cite{leaver85,tataru3, hintz20a}.

The presence of polynomial tails illustrates a clear contrast with the non-dispersive setting of the wave equation on a compact domain, where normal modes form an orthonormal basis for the space of solutions and therefore characterize fully the dynamics and in particular the late-time behaviour. Let us also note that polynomial tails are characteristic to \emph{asymptotically flat} spacetimes. Indeed, in the proof of the nonlinear stability of slowly rotating Kerr--de Sitter black holes \cite{HV2016}, it was shown that in the cosmological setting there are \emph{no} polynomial tails and the asymptotic behaviour of gravitational radiation is instead determined by QNMs.\footnote{We refer to Section \ref{intro:cosmo} for a further discussion on the de Sitter black hole setting.}

It remains an open problem in the asymptotically flat setting to reconcile the expected polynomial tails in gravitational radiation with the observed exponentially decaying and oscillating QNMs. In order to discern the dominant behaviour in a given time interval and to determine \emph{how} this behaviour depends on the type of initial data, it is necessary to establish in the first place a suitably robust and quantitative understanding of QNMs in the context of the linear wave equation
\begin{equation}
\label{eq:waveequationvac}
\square_g\psi=0,
\end{equation}
with $\square_g$ the Laplace--Beltrami operator corresponding to an asymptotically flat black hole spacetime background $(\mathcal{M},g)$.

\textbf{In this paper, we present a novel, mathematically rigorous construction of quasinormal modes for \eqref{eq:waveequationvac} on extremal Reissner--Nordstr\"om black hole spacetimes}. Using this construction, we obtain new quantitative statements regarding the distribution of QNMs near the zero time frequency and the properties of infinite sums over all angular frequencies. We provide in particular the first construction of QNMs on extremal Reissner--Nordstr\"om and the first construction of asymptotically flat black hole QNMs without the a priori assumption of fixed angular frequencies. 

We moreover demonstrate for each extremal Reissner--Nordstr\"om quasinormal frequency in a sector of the complex plane and for each sequence of Reissner--Nordstr\"om--de Sitter spacetimes with cosmological constants $\Lambda$ approaching zero and black hole charges $e$ approaching the extremal Reissner--Nordstr\"om value, the existence of a corresponding converging sequence of Reissner--Nordstr\"om--de Sitter quasinormal frequencies.\footnote{See Section \ref{intro:cosmo} for further details regarding the role of $\Lambda$.} We present rough statements of the main theorems of the paper in Section \ref{sec:roughmainthms}; see Section \ref{sec:mainthms} for precise statements of the main results.

We restrict to \emph{extremal} Reissner--Nordstr\"om only because of technical simplifications resulting from the existence of simple, closed-form expressions for conserved quantities along the event horizon and future null infinity.\footnote{In the sub-extremal case, there are no conserved quantities at the horizon but conserved quantities at null infinity are still present.} The exploitation of the aforementioned conservation laws along future null infinity and the event horizon plays a fundamental role in the new type of $L^2$-based ``Gevrey estimates'' that are introduced in this paper and involve functions that are \emph{Gevrey regular} at infinity and near the horizon, a notion that first appeared in \cite{gev18}.\footnote{The spaces of Gevrey functions can be viewed as lying ``between'' smooth and analytic functions. See Section \ref{def:gev} for a precise definition. Moreover, in the sub-extremal setting we expect Gevrey regularity at infinity (i.e.\ regularity with respect to the variable $\frac{1}{r}$ for large values of $r$, with $r$ a radial coordinate) to be sufficient, which allows us to include in particular all smooth and compactly supported data.} In future work, we hope to explore applications and extensions of the methods introduced in the present paper to a quantitative study of QNMs in \emph{sub-extremal} Reissner--Nordstr\"om and Kerr spacetimes.\footnote{An alternative strategy to deal with the sub-extremal case in the setting of fixed angular frequency solutions is to view the relevant operators on sub-extremal Reissner--Nordstr\"om as compact perturbations of the analogous operators in Minkowski or extremal Reissner--Nordstr\"om and apply more directly the methods of the present paper and the companion paper \cite{gajwar19b}. We refer to the companion paper \cite{gajwar19b} for related arguments involving compact operator perturbations.}

A key feature of the methods introduced in the present paper is the demonstration that, in contrast with previous approaches, QNMs may be interpreted as honest eigenfunctions of the infinitesimal generator of time translations acting on a suitable Hilbert space of initial data, thus placing them on a similar footing to normal modes. This interpretation was first introduced in the asymptotically anti de Sitter setting in \cite{warn15}. In contrast with the asymptotically (anti) de Sitter settings, where it is sufficient to consider (modified) $L^2$-based Sobolev spaces as Hilbert spaces, the Hilbert spaces in the asymptotically flat setting are shown to consist of functions that are Gevrey regular at future null infinity and at the future event horizon. The results of the present paper demonstrate that such spaces are natural candidate initial data spaces for investigating the role of quasinormal modes in the time evolution of solutions to \eqref{eq:waveequationvac}. Note that weighted $L^2$-based Sobolev spaces are \emph{not} well-suited for this purpose, see Theorem 4.5 and Remark 4.4 in \cite{aagscat} and the discussion in Section \ref{sec:smoothmode}.

\subsection{Traditional approaches to defining quasinormal modes}
\label{sec:tradapproach}
The wave equation \eqref{eq:waveequationvac} on a Reissner--Nordstr\"om spacetime with mass $M$ and charge $e$ (see Section \ref{sec:spacetimes} for more details) takes the following form in standard $(t,r,\theta,\varphi)$ coordinates when restricted to the coefficients $\psi_{\ell m}(t,r)$ in a spherical harmonic expansion $\psi(t,r,\theta,\varphi)= \sum_{\ell\in \N_0,m\in \Z, |m|\leq \ell}\psi_{\ell m}(t,r)Y_{\ell m}(\theta,\varphi)$:
\begin{equation}
\label{eq:asympflatode}
D^{-1}r^{-1}\left[ (D\partial_r)^2(r\psi_{\ell m})-\partial_t^2 (r\psi_{\ell m})-V_{\ell} \cdot r\psi_{\ell m}\right]=0,
\end{equation}
with 
\begin{align*}
D(r)=&\:1-2Mr^{-1}+e^2r^{-2},\\
V_{\ell}(r)=&\: \ell(\ell+1)r^{-2}D(r)+r^{-1}DD'(r),
\end{align*}
where $r\in (r_+,\infty)$ and $D(r_+)=0$.

Solutions to \eqref{eq:asympflatode} of the form $\psi_{\ell m}(t,r)=e^{st}\hat{\psi}_{\ell m}(r)$ with $s\in \C$ therefore satisfy the following (time-independent) Schr\"odinger equation:
\begin{align}
\label{eq:mainode}
\hat{\mathcal{L}}_{s,\ell}(\hat{\psi}_{\ell m})=&\:0,\quad \textnormal{where}\\  \nonumber
\hat{\mathcal{L}}_{s,\ell}(\cdot ):=&\:\frac{d^2}{dr_*^2} (r (\cdot))-(s^2+V_{\ell})(r(\cdot)),
\end{align}
and we changed from the coordinate $r$ to the radial coordinate $r_*(r)$, which satisfies $\frac{dr_*}{dr}=D^{-1}$ (see Section \ref{sec:spacetimes} for more details). 

Standard asymptotic analysis of second-order homogeneous ODE, see for example Chapter 7 Section 2 of \cite{olver74}, implies that solutions to \eqref{eq:mainode} must satisfy the asymptotic behaviour
\begin{equation*}
r\hat{\psi}_{\ell m}\sim e^{s r_*} \quad \textnormal{or} \quad r\hat{\psi}_{\ell m}\sim e^{-s r_*},
\end{equation*}
both as $r_*\to\infty$ and $r_*\to- \infty$, if $s\notin -\kappa_+ \N_0$, with $\kappa_+>0$ a constant known as the surface gravity of the event horizon, defined in Section \ref{sec:metric2}.\footnote{If $s\in -\kappa_+\N_0$, then there are still two linearly dependent solutions, but the asymptotic behaviour is slightly different, see for example \cite{olver74}.} Here, the notation ``$\sim$'' means that there exists smooth functions $P_{\textnormal{hor},\pm}(r)$ near $r=r_+$ and $P_{\textnormal{inf},\pm}(\rho)$ near $\rho=0$ and constants $A_+,A_->0$ and $B_+,B_-$, such that we can decompose
\begin{align}
\label{eq:asym1}
r\hat{\psi}_{\ell m}(r)=&\:A_+ e^{+s r_*}P_{\textnormal{hor},+}(r)+ A_-e^{- s r_*}P_{\textnormal{hor},-}(r)\quad \textnormal{near $r=r_+$ and},\\
\label{eq:asym2}
r\hat{\psi}_{\ell m}(r)=&\:B_+e^{+ s r_*}P_{\textnormal{inf},+}(1/r)+B_-e^{- s r_*}P_{\textnormal{inf},-}(1/r)\quad \textnormal{for large $r$}.
\end{align}
A key question that is relevant for the definition of quasinormal modes is whether the above decompositions are unique (up to renormalization). Approach 1 and Approach 2 below provide two different paths towards obtaining uniqueness, which allows one to define ``ingoing'' and ``outgoing'' solutions.

\subsubsection{Approach 1: quasinormal modes as solutions to a boundary value problem with convergent Taylor series}
\label{sec:approach1}

If $V_{\ell}$ was compactly supported in $r_*$, then the solution $r\psi_{\ell m}$ to \eqref{eq:mainode} at sufficiently large $|r_*|$ could have be written a sum of ingoing and outgoing travelling waves, so we would have made the canonical choice $P_{\textnormal{hor},\pm }\equiv P_{\textnormal{inf},\pm}\equiv 1$ for large $|r_*|$. In the seminal work \cite{chandra75}, \emph{quasinormal modes} were characterized as solutions to \eqref{eq:mainode} with $\re s<0$ (corresponding to compactly supported analogues of $V_{\ell}$) that are ``outgoing at infinity'' ($B_+=0$) and ``ingoing at the horizon'' ($A_-=0$).

The choice of $P_{\textnormal{hor},\pm }$ and $P_{\textnormal{inf},\pm}$ (which determines the notions of ``ingoing'' and ''outgoing'' by setting  $A_-=B_+=0$) is, however, not so straightforward since $V_{\ell}$ is \underline{not} compactly supported. Indeed, let $P_{\textnormal{inf},+}$ and $P_{\textnormal{inf},-}$ be a choice of smooth functions as in \eqref{eq:asym2} and consider the following \emph{ingoing} solution at infinity with respect to $P_{\textnormal{inf},\pm}$:
\begin{equation*}
r\hat{\psi}_{\ell m}(r)=B_+e^{+ s r_*}P_{\textnormal{inf},+}(1/r)=B_+e^{-s r_*}(e^{+2sr_*}P_{\textnormal{inf},+}(1/r)).
\end{equation*}
Define $\widetilde{P}_{\textnormal{inf},-}(1/r):=e^{+2sr_*}P_{\textnormal{inf},+}(1/r)$ and note that $\widetilde{P}_{\textnormal{inf},-}$ is also smooth in $\frac{1}{r}$ (because $\re s <0$!). Then the above solution is actually \emph{outgoing} with respect to $\widetilde{P}_{\textnormal{inf},-}$.

In the case of sub-extremal Reissner--Nordstr\"om, smoothness of $P_{\textnormal{hor},+}(r)$ in $r$ at $r=r_+$ does turn out to be a sufficient regularity condition to determine uniqueness of $P_{\textnormal{hor},+}(r)$ up to renormalization, since $e^{-2sr_*}P_{\textnormal{hor},-}$ (with $P_{\textnormal{hor},-}$ smooth) fails to be smooth in $r$ at $r=r_+$ (when $s\notin -\kappa_+\N_0$). In extremal Reissner--Nordstr\"om, however, $e^{-2sr_*}P_{\textnormal{hor},-}$ \underline{is} smooth in $r$; see also Section \ref{sec:smoothmode}.

One approach to deal with the non-uniqueness of the decomposition involving $P_{\textnormal{inf},\pm}$ in the sub-extremal case is to demand \underline{more than smoothness} by requiring sufficiently rapid decay of the coefficients of the Taylor series of $e^{+2 s r_*} r\hat{\psi}_{\ell m}(r)$ in terms of $z=1-r_+/r$ around $z=0$ so that the series converges at $z=1$.\footnote{Note that convergence is guaranteed at $z<1$ by analyticity. Indeed, analyticity of $P_{\textnormal{hor},+}(r)$ at $r=r_+$ follows from the fact that $r=r_+$ is a regular singular point of the ODE in the sub-extremal case, whereas $r=\infty$ corresponds to an irregular singular point and, \emph{a priori}, only smoothness of $P_{\textnormal{inf},-}(1/r)$ can be guaranteed at $\frac{1}{r}=0$.} This approach was first carried out by Leaver \cite{leaver85b} in the sub-extremal setting, following earlier ideas in the setting of the Schr\"odinger equation describing hydrogen molecule ions \cite{baber35}, and it forms the basis of the \emph{continued fraction method} for approximating quasinormal frequencies numerically. Indeed, in \cite{ans16} numerical evidence is given for the claim that the coefficients in the Taylor series of the function $e^{+s r_*}$ around $z=0$ indeed do \underline{not} decay sufficiently rapidly so as to guarantee convergence at $z=1$.  Furthermore, \cite{leaver85b} provides numerical evidence that the set of $s$ for which the Taylor series of $e^{+2 s r_*} r\hat{\psi}_{\ell m}(r)$ does converge is discrete.

The continued fraction method of Leaver was generalized to extremal Reissner--Nordstr\"om in \cite{extremal-rn-qnm}.

\subsubsection{Approach 2: quasinormal modes as poles of the meromorphic continuation of the resolvent}
\label{sec:approach2}
A different approach to defining QNMs is to interpret solutions to \eqref{eq:mainode} with the desired boundary behaviour as \emph{analytic continuations} in $s$ of analogous solutions with $\re s>0$ instead of $\re s<0$. Indeed, when $\re s>0$, one can unambiguously exclude the solutions with an undesired boundary behaviour $e^{+s |r_*|}$ (which now \emph{grow} exponentially in $|r_*|$) by demanding $r\hat{\psi}_{\ell m}$ to be suitably bounded at $r_*=\pm \infty$. Quasinormal frequencies then correspond to zeroes of the analytic continuation of the Wronskian corresponding to the two $\re s>0$ solutions which are outgoing at infinity and ingoing at the event horizon, respectively.\footnote{The analytic continuation of the Wronskian is only valid away from the non-positive real axis.} 

This approach was carried out by Bachelot--Motet-Bachelot in \cite{ba} in Schwarzschild (where $e=0$), and it constitutes the first mathematically rigorous result regarding the definition and distribution of QNMs.\footnote{See also earlier heuristic work \cite{det77, leaver85b} providing a similar interpretation of QNMs.}
\begin{thmx}[\cite{ba}]
\label{thm:bach}
Fix $\ell\in \N_0$. The resolvent operator $R_{\ell}(s)=\hat{\mathcal{L}}^{-1}_{s,\ell}: L^2(\R) \to L^2(\R)$ is a bounded linear operator that is holomorphic in $s$ when $\re s >0$ and for any pair $(\chi,\chi')\in C^{\infty}_c(\R)$, the operator
\begin{equation*}
\chi'\circ R_{\ell}(s) \circ \chi: L^2(\R) \to L^2(\R)
\end{equation*}
can be meromorphically continued to $\C\setminus \R_{\leq 0}$.
\end{thmx}
One can define ``quasinormal frequencies'' to be the poles of $R_{\ell}(s)$ and corresponding solutions to \eqref{eq:waveequationvac} to be ``quasinormal modes''. Theorem \ref{thm:bach} illustrates the discreteness of the set of quasinormal frequencies in the complex plane away from the non-positive real axis. We refer to \cite{sabaretto} for additional results regarding the distribution of quasinormal frequencies in a small conic sector around the imaginary axis when $|s|\to \infty$ and $\ell\to \infty$. 

A key feature in both \cite{ba} and \cite{sabaretto} is the complex scaling method which uses analyticity of the Schwarzschild metric components to analytically continue the coordinate $r_*$ away from the real axis into the complex plane. This method was introduced in \cite{agu71, bal71} and extended to a very general, ``black-box'' setting, in \cite{sjozwo91,sjo97}.

Note that it is not clear from the above discussions that the notions of ``quasinormal modes'' and ``quasinormal frequencies'' as prescribed in Approach 1 and Approach 2 actually agree (although there is ample numerical evidence, presented for example in \cite{leaver85b} and \cite{ba}). We refer to Remark \ref{rmk:relationapproach12} for a further discussion on this issue. In order to differentiate between the two approaches, we will refer in this paper to the (heuristic) QNMs/QNFs of Approach 1 as \emph{Leaver quasinormal modes/quasinormal frequencies} and the QNMs/QNFs of Approach 2 as \emph{resonant states}/\emph{scattering resonances}.

The resolvent operator with $\re s>0$ can be defined without the restriction to spherical harmonic modes of fixed $\ell$ in all Reissner--Nordstr\"om spacetimes, using basic uniform energy boundedness properties from \cite{aretakis1,Civin2014} and this definition can be extended to more general classes of spacetimes satisfying suitable energy boundedness properties. 

One would arrive at a natural extension of the notion of fixed $\ell$ scattering resonances by considering the poles of a meromorphic continuation of this more general resolvent operator to a suitable subset of the complex left half-plane. \textbf{It does not follow from the proof of Theorem \ref{thm:bach} that such a continuation exists as, in particular, one cannot rule out \emph{a priori} accumulation points of scattering resonances with fixed $\ell$ as $\ell\to \infty$!} We address this point in Theorem \ref{thm:polesresolvent} in the setting of extremal Reissner--Nordstr\"om.

\subsection{Quasinormal modes in time evolution}
\label{intro:laplacetransf}

The relevance of the scattering resonances described in Section \ref{sec:tradapproach} for the behaviour of general solutions to the initial value problem corresponding to \eqref{eq:waveequationvac} becomes clear when one considers the Laplace transform of solutions to \eqref{eq:waveequationvac}. Let us consider for simplicity initial data of the form:
\begin{equation*}
(\psi|_{t=0},\partial_t\psi|_{t=0})\in (C^{\infty}_c((r_+,\infty)_r\times \s^2))^2.
\end{equation*}
The corresponding solutions $\psi$ to \eqref{eq:waveequationvac} are uniformly bounded in $L^{\infty}$ in both sub-extremal and extremal Reissner--Nordstr\"om (see for example \cite{aretakis1,Civin2014}), so for $\re s>0$ the following Laplace transform is well-defined
\begin{equation*}
\hat{\psi}(s,r,\theta,\varphi)=\int_0^{\infty} e^{-s t} \psi(t,r,\theta,\varphi)\,dt,
\end{equation*}
and after restricting to fixed spherical harmonic modes, it follows immediately that $\hat{\psi}_{\ell m}$ is the unique solution to
\begin{equation*}
\hat{\mathcal{L}}_{s,\ell}(\hat{\psi}_{\ell m})=f_{\ell m},
\end{equation*} 
with 
\begin{equation*}
f_{\ell m}(r;s)=-s\cdot r\psi_{\ell m}|_{t=0}(r)-\partial_t(r\psi_{\ell m})|_{t=0}(r).
\end{equation*}

Note moreover that we can express
\begin{equation*}
\hat{\psi}_{\ell m}=R_{\ell}(s)(f_{\ell m}).
\end{equation*}

Let  $\chi,\chi'$ be cut-off functions in $r$, as in Theorem \ref{thm:bach} such that $\chi\equiv 1$ on $\supp f_{\ell m}$, then the following inverse Laplace transform can then be obtained via the \emph{Bromwich integral}:
\begin{equation*}
\chi' \cdot \psi_{\ell m}(t,r)=\lim_{S\to \infty} \frac{1}{2\pi i} \int_{s_0-i S}^{s_0+i S} (\chi'\circ R_{\ell}(s)\circ \chi)(f_{\ell m})\, ds,
\end{equation*}
for any $s_0>0$.

By the meromorphicity property of $\chi'\ \circ R_{\ell}(s)\circ \chi$, established in Theorem \ref{thm:bach}, we can deform the contour of integration in the Bromwich integral so that it enters the complex left half-plane (while avoiding $\R_{\leq 0}$), see Figure \ref{fig:contour} below. 

 \begin{figure}[H]
	\begin{center}
\includegraphics[scale=0.5]{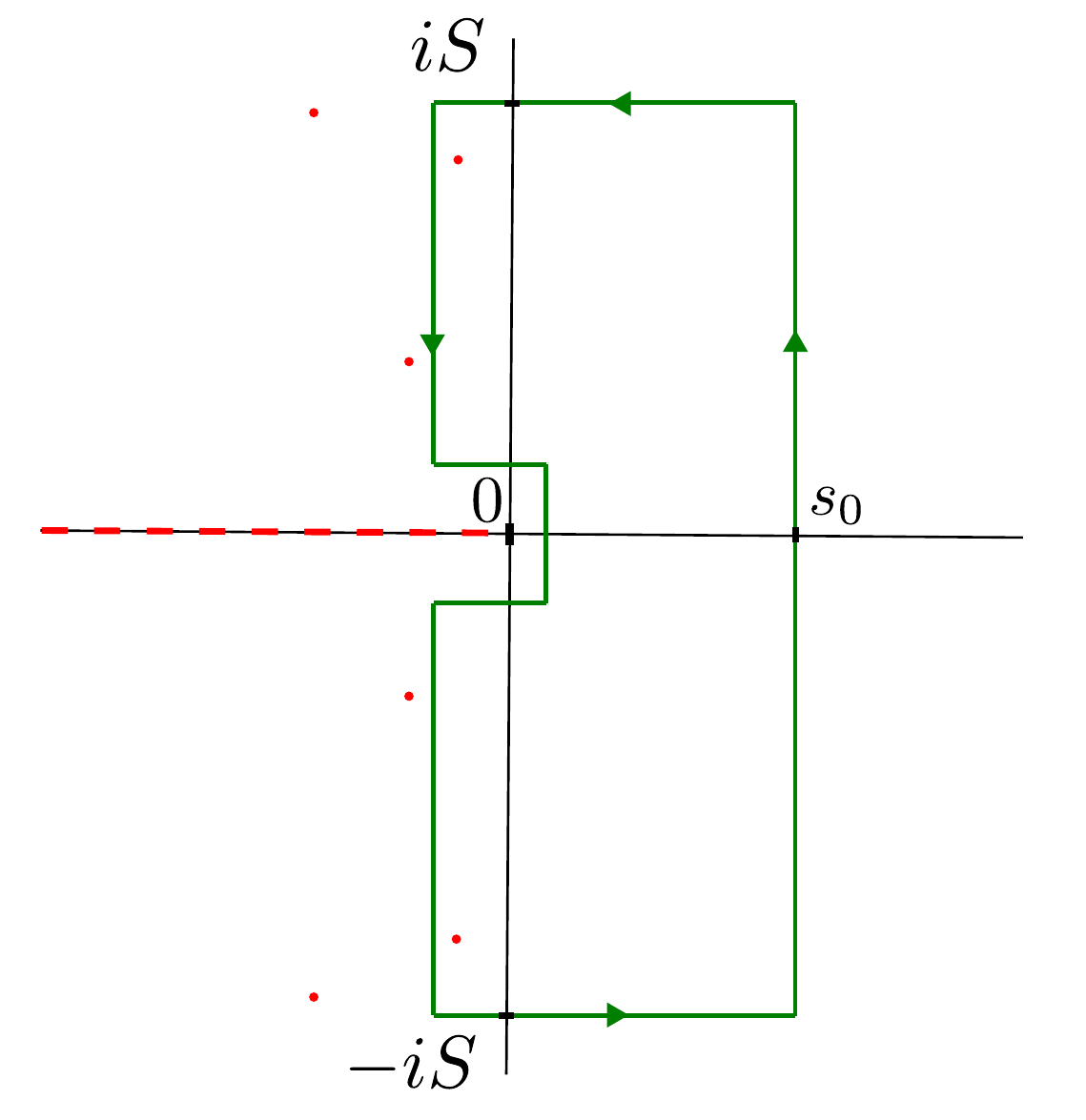}
\end{center}
\vspace{-0.2cm}
\caption{An example of a possible integration contour deformation. The dots indicate possible poles of the resolvent operator.}
	\label{fig:contour}
\end{figure}

Poles of $R_{\ell}(s)$ in the left half-plane, which correspond to scattering resonances, will contribute to the contour integral via the Residue Theorem as exponentially decaying and oscillating terms, proportional to resonant states, with coefficients determined by $f_{\ell m}$.

If one could determine the remaining contributions to the deformed contour integral, one could \emph{in theory} reconstruct the solution $\psi_{\ell m}$ to determine the dominant contributions of various parts of the contour integral at sufficiently late times. See for example the heuristics in Schwarzschild spacetimes in \cite{leaver85b} where the resonant state contribution is compared to a polynomial tail contribution, which dominates the very late-time behaviour \cite{paper2}.

The first observation of the appearance of exponentially damped and oscillating functions in the time evolution of solutions to \eqref{eq:waveequationvac} was made in a numerical setting in \cite{vish70} and actually preceded the time-independent definitions of quasinormal modes described in Section \ref{sec:tradapproach}.

\subsection{Open problems in the traditional approach}
\label{intro:points}
In order to make use of a contour deformation for the inverse Laplace transform, as in Section \ref{intro:laplacetransf}, for determining the role of quasinormal modes in the temporal behaviour of solutions to \eqref{eq:waveequationvac} arising from suitably regular, localized and generic initial data, additional information is needed with regards to the distribution of scattering resonances and further properties of the resolvent operator $R_{\ell}(s)$.

\textbf{In this paper, we will address the following points in the setting of extremal Reissner--Nordstr\"om spacetimes}:
\begin{enumerate}[(I)]
\item The resolvent operator ${R}_{\ell}(s)$ is only defined for fixed $\ell$. What happens when summing over \emph{all} spherical harmonic modes? Is it possible to meromophically continue $R(s)=\sum_{\ell \in \N_0} R_{\ell}(s)$?
\item What does the distribution of scattering resonances look like near $s=0$?
\item The $L^2$-norm of resonant states along $\{t=const.\}$ slices is infinite, due to their exponential growth in $|r_*|$ as $|r_*|\to \infty$ (see Section \ref{sec:tradapproach}).
\item It is tempting to view scattering resonances as eigenfunctions of the operator $\hat{\mathcal{L}}_{s,\ell}+s^2$, with eigenvalues $s^2$, in analogy with normal modes. However, in contrast with normal modes it is not clear what the corresponding Hilbert spaces should be (see (III)).
\item By Section \ref{intro:laplacetransf}, one has to restrict to initial data that are supported \emph{away} from the horizon and infinity in order to reconstruct $\psi_{\ell m}$ using the meromorphic continuation of $R_{\ell}(s)$. There is no physical motivation for solely considering initial data supported away from the horizon.
\item Since $R_{\ell}(s)$ has to be multiplied with cut-off functions in order for a meromorphic continuation to be well-defined, the contour deformation of the inverse Laplace transform can not be used directly to provide information about the behaviour in time at the horizon and null infinity. 
\end{enumerate}

In the companion paper \cite{gajwar19b}, we moreover address:
\begin{enumerate}[(I)]
\item[(VII)] What is the relation between Leaver QNFs and scattering resonances?
\end{enumerate}
In the setting of Schr\"odinger operators $\mathcal{L}_{s,\ell}$ as in \eqref{eq:mainode} with more general $V_{\ell}$, we address also in \cite{gajwar19b}:
\begin{enumerate}[(I)]
\item[(VIII)] It is possible to define quasinormal modes when the potential $V_{\ell}$ in \eqref{eq:mainode} behaves to leading-order like $\frac{1}{r_*^2}$, but is \underline{non-analytic} in the variable $\frac{1}{r_*}$?
\end{enumerate}
See also the recent work \cite{galzwo20} for progress towards addressing (VIII) using complex scaling methods.

\subsection{Comparison with the cosmological setting}
\label{intro:cosmo}
A key difficulty in determining the regularity class of homogeneous solutions to $\hat{\mathcal{L}}_{s,\ell}(\hat{\psi}_{\ell m})=0$, that is required to be able impose the desired outgoing boundary condition, is the slow $r_*^{-2}$ fall-off of $V_{\ell}$ when $r_*\to \infty$; see Section \ref{sec:tradapproach}. A similar difficulty is additionally present in the extremal case ($|e|=M$) when $r_*\to -\infty$, in contrast with the sub-extremal ($|e|<M$) case where $V_{\ell}\sim e^{\kappa_+ r_*}$ as $r_* \to -\infty$, with $\kappa_+$ the surface gravity of the event horizon (see Section \ref{sec:metric2}). This exponential fall-off is intimately related to the presence of a red-shift effect along the event horizon at $r=r_+$, see \cite{redshift}. 

In order to avoid dealing with the difficulty present at $r_*=\infty$, one can consider a different problem by replacing ``infinity'' with a \emph{cosmological horizon} along which there is an \emph{additional} red-shift effect (see for example \cite{lecturesMD}).

Indeed, we can consider modified Reissner--Nordstr\"om metrics by adding a cosmological constant term depending on an additional parameter $\Lambda>0$:
\begin{align*}
D(r)=&\:1-2Mr^{-1}+e^2r^{-2}-\frac{\Lambda}{3}r^2,\quad \textnormal{so that}\\
V_{\ell}\sim& e^{-\kappa_c r_*}\quad \textnormal{as $r_*\to \infty$},
\end{align*}
where $\kappa_c>0$ (the surface gravity of the cosmological horizon) and $r\in (r_+,r_c)$ for some $r_c>0$ with $D(r_+)=0$ and $D(r_c)=0$. The corresponding spacetimes are called Reissner--Nordstr\"om--de Sitter spacetimes and they are solutions to the electrovacuum Einstein--Maxwell equations with positive cosmological constant, see for example \cite{haw}. See Section \ref{sec:spacetimes} for more details.

Geometrically, the $\Lambda$-modification amounts to replacing null infinity with a Killing horizon at finite radius $r=r_c$ and as a result, when $\kappa_+,\kappa_c>0$, there is a red-shift effect along \emph{both} the event and the cosmological horizon. From an ODE perspective, the homogeneous ODE $\hat{\mathcal{L}}_{s,\ell}(\hat{\psi}_{\ell m})=0$ has an irregular singular point at $r_*=\infty$ and a singular point at $r_*=-\infty$ that is either regular ($\kappa_+>0$) or irregular ($\kappa_+=0$) in the $\Lambda=0$ case, whereas in the $\kappa_+,\kappa_c>0$ ($\Lambda>0$) case \emph{both} singular points are regular singular points.

In the $\kappa_+,\kappa_c>0$ setting it is possible to define scattering resonances as poles of the resolvent operator for the (massive) wave equation as in Theorem \ref{thm:bach}, without requiring a decomposition into spherical harmonics, addressing (I) of Section \ref{intro:points} in the $\Lambda>0$ setting; see \cite{sabaretto, bonyh, vasydesitter, semyon1}. See also the very general set-up developed by Vasy for defining scattering resonances in the $\Lambda>0$ setting  \cite{vasy1}, where the red-shift effect is indirectly also exploited. The role of scattering resonances in the $\Lambda>0$ setting on the time-evolution of solutions to the wave equation was moreover investigated in \cite{dyatlov2015asymptotics}.

The above works use microlocal methods that have previously been applied to study scattering resonances in a variety of settings (see \cite{zwordyabook} for a detailed overview of the corresponding literature) and they are closely related to the study of resonances on asymptotically hyperbolic Riemannian manifolds \cite{mazzeo1987}. The above results address moreover (II), (V), (VI) in the $\Lambda>0$ setting. Let us mention finally that quasinormal modes in Kerr--de Sitter spacetimes play a fundamental role in the proof of the nonlinear stability of Kerr--de Sitter in \cite{HV2016}.

A different perspective on QNMs was taken in \cite{warn15}, and when applied to the setting of Reissner--Nordstr\"om--de Sitter, it results in the following theorem that addresses all the points (I)---(VI) above in the $\Lambda>0$ setting:
\begin{thmx}[\cite{warn15}]
\label{thm:ads}
Let $k\in \N_0$ and $\Sigma$ be a hypersurface in sub-extremal Reissner--Nordstr\"om--de Sitter that intersects both horizons to the future of the bifurcation sphere. Then the \emph{infinitesimal generator of time translations}\footnote{If one denotes with $\mathcal{S}(\tau): {H}^{k+1}(\Sigma)\times H^{k}(\Sigma)\to {H}^{k+1}(\Sigma)\times H^{k}(\Sigma)$ the semigroup of time translations along a timelike Killing vector field that map initial data to the corresponding solution to \eqref{eq:waveequation} and its time derivative in Reissner--Nordstr\"om--de Sitter along a time slice at time $\tau$, then we can write $S(\tau)=e^{\tau \mathcal{A}}$, with $\mathcal{A}$ a densely defined operator. See Section \ref{sec:propgevreyphys} for more details.} $\mathcal{A}: {H}^{k+1}(\Sigma)\times H^{k}(\Sigma)\supseteq \mathcal{D}(\mathcal{A})\to {H}^{k+1}(\Sigma)\times H^{k}(\Sigma)$ is a well-defined (unbounded) closed, linear operator and:
\begin{enumerate}[\rm (i)]
\item The operator $\mathcal{A}$ has pure point spectrum $\Lambda_{QNF}^k$  in 
\begin{equation*}
\left\{ s\in \C\,|\, \re s>-\left(\frac{1}{2}+k\right)\min\{\kappa_+,\kappa_c\}\right\},
\end{equation*}
with $\bigcup_{k\in \N_0}\Lambda_{QNF}^k\subset \C$ a discrete subset, and the corresponding eigenfunctions are smooth.
\item For every scattering resonance there exists a sufficiently large $k$ such that the scattering resonance is an element of $\Lambda_{QNF}^k$ and there exists a subspace $\mathbf{H}^k_{\rm res}(\Sigma)< {H}^{k+1}(\Sigma)\times H^{k}(\Sigma)$ such that the union over $k$ of the set of eigenvalues of the restrictions $\mathcal{A}|_{\mathbf{H}^k_{\rm res}(\Sigma)}$ coincides \emph{precisely} with the total set of scattering resonances.
\end{enumerate}
\end{thmx}
We will refer to $\bigcup_{k\in \N_0}\Lambda_{QNF}^k$ as \textbf{regularity quasinormal frequencies} and the corresponding eigenfunctions as \textbf{regularity quasinormal modes}.\footnote{Note that in the specific case of Reissner--Nordstr\"om--de Sitter, by applying standard ODE theory, one can appeal to regularity of mode solutions to show that for $-s\notin \kappa_+\N_0\cup \kappa_c\N_0$, the regularity QNFs agree precisely with the scattering resonances. In the case of exact de Sitter spacetimes, the set of scattering resonances is empty but the set of regularity QNFs is non-empty and coincides with $-\kappa \N_0$, with $\kappa>0$ the surface gravity of the de Sitter horizon, see \cite{warn15} and the toy model problem discussed in Section \ref{intro:toy}.}

See also the related results in \cite{vasy1}, where, in addition, high frequency estimates are obtained. We note that in \cite{warn15}, a general asymptotically anti de Sitter $\Lambda<0$ setting is considered rather than Reissner--Nordstr\"om--de Sitter, but the methods can be straightforwardly applied to the $\Lambda>0$ Reissner--Nordstr\"om--de Sitter setting to arrive at Theorem \ref{thm:ads}; see also Appendix \ref{sec:red-shift}. See \cite{gan14} for an alternative consideration of the $\Lambda<0$ setting and \cite{gusmu2} for a related construction of ``quasimodes'' in the $\Lambda<0$ setting.
 
We emphasize that Theorem \ref{thm:ads} addresses (III) and (IV) by demonstrating that the restrictions of quasinormal modes to hypersurfaces intersecting the horizons to the future of the bifurcation spheres (rather than hypersurfaces of constant $t$, which intersect the bifurcation spheres) can be viewed as eigenfunctions in suitable $L^2$-based Sobolev spaces.

Furthermore, since normal frequencies can also be interpreted as eigenvalues of $\mathcal{A}$ (on a Sobolev space), corresponding to a non-degenerate wave equation on a spatially compact domain, Theorem \ref{thm:ads} moreover allows one to view normal modes and (regularity) quasinormal modes as the \emph{same kind of objects}.

We note that the proof of Theorem \ref{thm:ads} relies fundamentally on the red-shift estimates and the enhanced red-shift\footnote{Here, ``enhanced'' refers to an increase of the strength of the red shift effect when considering higher-derivative norms.} estimates near both horizons, developed in \cite{redshift}. \textbf{Since there is no red-shift at null infinity in the $\Lambda=0$ setting, and additionally no red-shift at the event horizon in the extremal ($|e|=M$) case, the methods of Theorem \ref{thm:ads} do \underline{not} apply in the present paper!}

The above results in the cosmological setting motivate another problem that will be addressed the present paper:
\begin{enumerate}[(IX)]
\item How do cosmological quasinormal modes and frequencies behave in the limit $\Lambda\downarrow 0$?
\end{enumerate}

\subsection{Asymptotically hyperboloidal and null foliations}
\label{intro:litasymflat}

Theorem \ref{thm:ads} demonstrates how, rather than considering hypersurfaces of constant $t$ in Reissner--Nordstr\"om--de Sitter (that intersect the bifurcation spheres), it is more natural to consider hypersurfaces interesting the event and cosmological horizons to the future of the bifurcation spheres when investigating candidate function spaces containing regularity QNMs. In this section we consider analogous foliations of Reissner--Nordstr\"om by hypersurfaces that intersect the event horizon and are moreover asymptotically hyperboloidal or null (they ``intersect'' future null infinity in the conformal picture; see Figure \ref{fig:foliations}).

 \begin{figure}[H]
	\begin{center}
\includegraphics[scale=0.5]{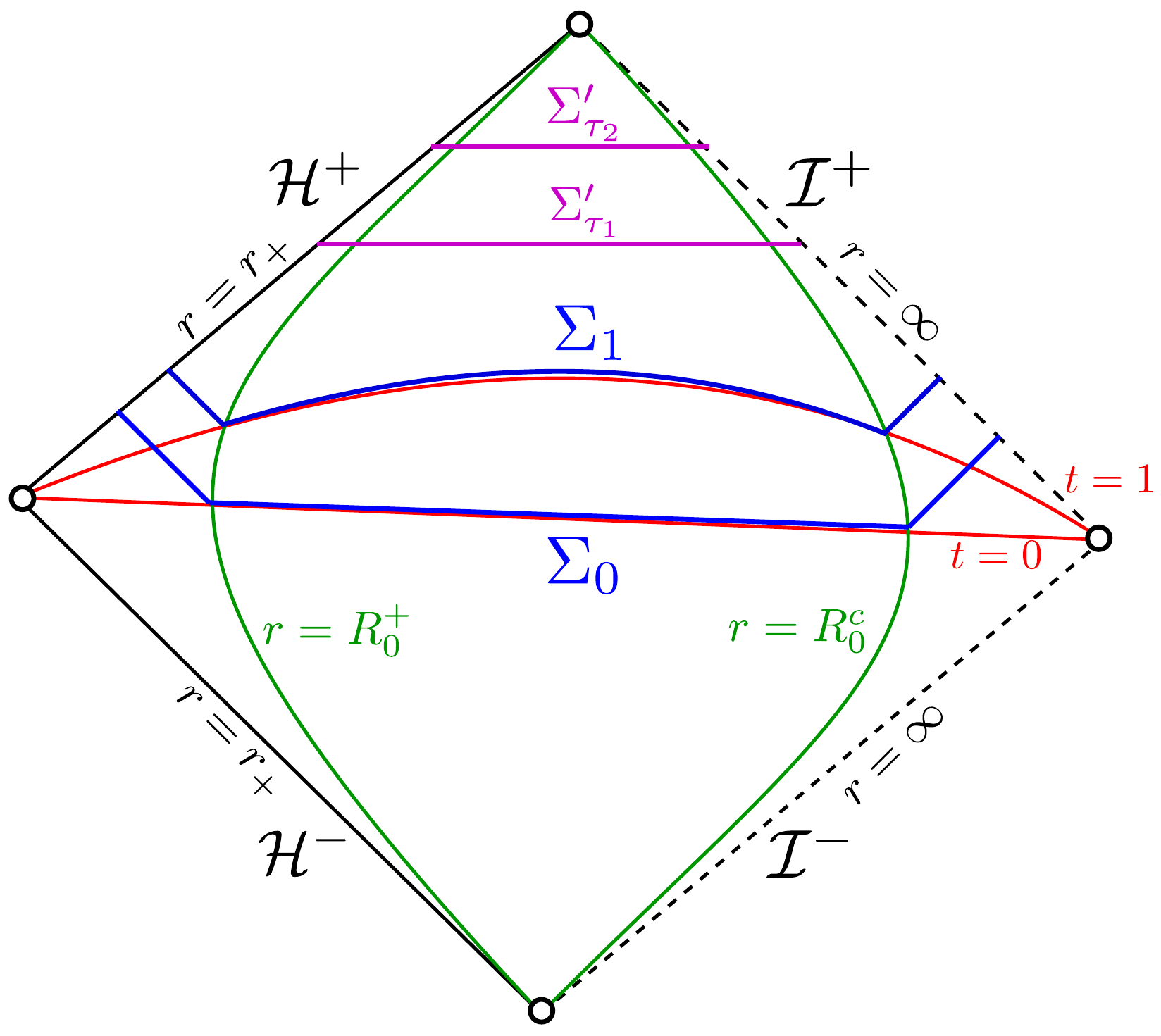}
\end{center}
\vspace{-0.2cm}
\caption{A Penrose diagrammatic depiction of examples of an asymptotically null foliation $\{\Sigma_{\tau}\}_{\tau\geq 0}$ and an asymptotically hyperboloidal foliation $\{\Sigma'_{\tau}\}_{\tau\geq 0}$, compared to a foliation of asymptotically flat $t$-level sets. }
	\label{fig:foliations}
\end{figure}

\subsubsection{Uniform decay estimates}
A foliation of stationary asymptotically flat spacetimes by asymptotically hyperboloidal or null hypersurfaces (intersecting the event horizon to the future of the bifurcation sphere), rather than the asymptotically flat foliation by $t$-level sets is well-suited for deriving uniform decay estimates for solutions to \eqref{eq:waveequationvac}. Indeed, while the energy associated to the vector field generating the time translation symmetry is \emph{conserved} along $t$-level sets, when considered along asymptotically hyperboloidal or null hypersurfaces, the energy is generically \emph{decreasing} in time due to energy radiation through the event horizon and future null infinity, see Figure \ref{fig:foliations}. We refer to \cite{newmethod, volker1, part3, moschidis1, aag18} and references therein for energy decay results along asymptotically hyperboloidal or null hypersurfaces in the asymptotically flat setting. Of particular relevance to the setting of the present paper is the \emph{Dafermos--Rodnianski $r^p$-weighted energy method} \cite{newmethod} and the extended methods in \cite{aag18} and \cite{paper4} , which relate the existence of hierarchies of $r$-weighted energy estimates along asymptotically hyperboloidal or null hypersurfaces to polynomial energy time-decay rates.

In the context of extremal Reissner--Nordstr\"om spacetimes, energy and pointwise decay estimates for solutions to \eqref{eq:waveequationvac} were first obtained by Aretakis in \cite{aretakis1, aretakis2}. Aretakis presented a novel instability phenomenon: transversal derivatives of waves along the event horizon generically do not decay and higher order transversal derivatives blow up asymptotically in time along the horizon. The mechanism for this instability is the presence of conserved quantities along the spheres foliating the event horizon.

The conserved quantities discovered by Aretakis are intimately connected to the conserved quantities along future null infinity that were discovered by Newman and Penrose \cite{np2} and are present in much more general, stationary, asymptotically flat spacetimes. Indeed, one can explicitly relate the conserved quantities at the horizon with the conserved quantities at null infinity via a conformal transformation that maps the horizon to null infinity \cite{couch}, see \cite{hm2012, bizon2012}.

In \cite{paper4} further results were obtained regarding \eqref{eq:waveequationvac} on extremal Reissner--Nordstr\"om, exploiting heavily the aforementioned connection between null infinity and the event horizon and the presence of conserved quantities along both. In particular, the precise leading-order behaviour in time was obtained for solutions to \eqref{eq:waveequationvac}, demonstrating the presence of polynomial tails, first predicted in heuristics and numerics \cite{hm2012,ori2013, sela}. See also \cite{paper2} for an illustration of the importance of conserved quantities along future null infinity in the sub-extremal setting.

\subsubsection{Smooth mode solutions}
\label{sec:smoothmode}

In \cite{aagscat}, a scattering theory is developed in extremal Reissner--Nordstr\"om involving non-degenerate energies. As an application of this theory, it is shown that for any smooth, superpolynomially decaying scattering data along the event horizon and future null infinity, there exists a corresponding unique solution to \eqref{eq:waveequationvac} that is smooth along an asymptotically null foliation and is moreover smooth with respect to the coordinate $\frac{1}{r}$ at $r=\infty$. In particular, one can construct smooth solutions in the above sense (with respect to $\frac{1}{r}$) and with the exact time dependence $e^{st}$ for \emph{any} $s\in \C$ with $\re s<0$. We refer to such solutions as \emph{mode solutions}.

If we were to consider the infinitesimal generator of time translations $\mathcal{A}$, as in Theorem \ref{thm:ads}, acting on a dense subset of a naturally $r$-weighted $L^2$-based Sobolev space of arbitrarily high order in extremal Reissner--Nordstr\"om, the point spectrum of $\mathcal{A}$ would therefore include the entire open left half-plane!

The above result illustrates that weighted Sobolev spaces are \emph{ill-suited} as choices for Hilbert spaces for which the point spectrum of $\mathcal{A}$ coincides with the set of quasinormal frequencies. Indeed, if we want to interpret \emph{all} eigenfunctions of $\mathcal{A}$ as regularity QNMs or we demand at least discreteness of the set of eigenvalues, we need to \underline{exclude} the above smooth mode solutions, but we cannot achieve this by simply restricting to a Sobolev space of suitably high regularity.

An alternative motivation for the failure of weighted Sobolev spaces as candidate Hilbert spaces, which applies also to sub-extremal Reissner--Nordstr\"om, can be found at the ODE level by considering the functions $P_{\textnormal{hor}, \pm}$ and $P_{\textnormal{inf}, \pm}$ appearing in \eqref{eq:asym1} and \eqref{eq:asym2}. Recall that \emph{both} $P_{\textnormal{inf},-}$ and $e^{2sr_*}P_{\textnormal{inf},+}$ are smooth in $1/r$ for large $r$, whereas $P_{\textnormal{hor}, +}$ is smooth in $r$ near $r=r_+$, but $e^{-2sr_*} P_{\textnormal{hor}, -}$ is \emph{not} when $|e|<M$. Hence, for \emph{any} frequency $s$ with $\re s<0$, we can restrict to homogeneous solutions with boundary behaviour $P_{\textnormal{hor}, +}$ near $r=r_+$ to obtain a smooth solution in $1/r$ near $r=\infty$. As a result, the corresponding radiation field $r \psi_{\ell m}$, with $\psi_{\ell m}$ a solution to \eqref{eq:waveequationvac}, is smooth along suitable hyperboloidal/asymptotically null hypersurfaces intersecting the horizon to the future of the bifurcation sphere.\footnote{When $|e|=M$, the mode solutions corresponding to \emph{any} homogeneous solution to the ODE are smooth at the horizon and at infinity with respect to $1/r$, which is consistent with the scattering theory viewpoint in \cite{aagscat} where exponentially decaying data along future null infinity and the event horizon that lead to smooth solutions can be chosen independently.} This implies that, in contrast with the $\Lambda\neq 0$ settings, a restriction to an arbitrarily regular Sobolev space will also not ensure discreteness of the set of eigenvalues of $\mathcal{A}$ in the sub-extremal ($\Lambda=0$) setting.

The above observations may instil doubts on the \emph{relevance} of QNMs when considering the time evolution of solutions to \eqref{eq:waveequationvac} in the $\Lambda=0$ setting, as in Section \ref{intro:laplacetransf}, arising from \emph{generic} smooth initial data along an asymptotically hyperboloidal or null slice, with $r\psi$ moreover smooth at infinity with respect to the coordinate $\frac{1}{r}$. Indeed, it is not immediately clear in this setting what singles out quasinormal mode solutions over other smooth mode solutions with arbitrary frequencies $s$. See, however, Remark \ref{rmk:QNMvssmoothmode} for a discussion on why QNMs \emph{are} relevant when considering smooth \underline{and suitably localized} initial data.

\subsection{A toy model}
\label{intro:toy}
In this section, we briefly discuss a toy model ODE which illustrates the key differences between the $\Lambda=0$ and $\Lambda>0$ settings and motivates the necessity of Gevrey regularity. We discuss this toy model at length in a companion paper \cite{gajwar19b}.

Consider the following Dirichlet problem:
\begin{align}
\label{eq:introtoy1} 
L_{s,\kappa}(u):=\frac{d}{dx}\left( (\kappa x+x^2) \frac{du}{dx}\right)+s \frac{du}{dx}=&\:f\quad \textnormal{for $x\in [0,1]$},\\\nonumber
u(1)=&\:0,
\end{align}
with $f$ a suitably regular function on $[0,1]$, $s\in \C$ and $\kappa\geq 0$. The parameter $\kappa$ plays the role of either $\kappa_c$, the surface gravity of the cosmological horizon, or $\kappa_+$, the surface gravity of the event horizon. When $\kappa=0$, the above equation models the equation for the spherical mean of $r\psi$ with $\psi$ a mode solution to \eqref{eq:waveequationvac} near infinity, with $x$ taking on the role of $\frac{1}{r}$, or near the event horizon of extremal Reissner--Nordstr\"om with $x=r-M$, where $M$ is the radius of the event horizon. The $\kappa>0$ case models the $\Lambda>0$ setting with $\psi$ a solution to the \emph{conformally invariant} Klein--Gordon equation \eqref{eq:waveequation}.\footnote{In fact, if we replace $s$ by $2s$ and $\kappa$ by $2\kappa_c$, the equation is \emph{precisely} the equation satisfied by the spherical mean of the radiation field on an exact de Sitter spacetime or the Minkowski spacetime, but with a reflecting boundary condition at $r=(1+r_c^{-1})^{-1}$ and $r=1$, respectively; see \eqref{eq:eqradfieldcos2}.}

One can easily verify that in the $\kappa>0$ case, any solution to $L_{s,\kappa}(u)=0$ can be expressed as a linear combination of the following solutions:
\begin{align*}
u_1(x)=&\: 1,\\
u_2(x)=&\: \left( \frac{x}{x+\kappa}\right)^{-\frac{s}{\kappa}}.
\end{align*}
Note that for $\re s<0$ and $s\notin -\kappa \N$,\footnote{When $s\in  -\kappa \N$, both $u_1$ and $u_2$ are smooth and $-\kappa \N$ may be thought of as regularity quasinormal frequencies which are not scattering resonances. See also the discussion in Section 6 of \cite{warn15}.} we have that $u_2\in C^k([0,1])\setminus C^{k+1}([0,1])$ for $k= \lfloor \kappa^{-1} |\re s|\rfloor$, whereas $u_1\in C^{\infty}([0,1])$. Hence, by restricting to $C^n([0,1])$, with $n\in \N$ suitably high depending on $\re s$, we can rule out the existence of homogeneous solutions to the Dirichlet problem and guarantee uniqueness of solutions to $L_{s,\kappa}(u)=f$.

In the $\kappa=0$ case, any solution to $L_{s,\kappa}(u)=0$ can be expressed as a linear combination of the following solutions:
\begin{align*}
u_1(x)=&\: 1,\\
u_2(x)=&\:e^{\frac{s}{x}}.
\end{align*}
If $\re s<0$, both solutions are smooth, so restricting to $C^n$ or $C^{\infty}$ spaces will \underline{not} guarantee uniqueness of $L_{s,\kappa}(u)=f$ in this case. Since $u_2$ fails to be analytic at $x=0$, one way of obtaining uniqueness of solutions to $L_{s,\kappa}(u)=f$ is to restrict to the space of analytic solutions. However, assuming a solution $u$ exists to \eqref{eq:introtoy1} with $f\equiv 1$, we can apply the equation \eqref{eq:introtoy1}, commuted with $\frac{d^k}{dx^k}$, $k=0,\ldots,n$ and $n\in \N_0$, to obtain
\begin{equation*}
u^{(n+1)}(0)=(-1)^n s^{-(n+1)}n!(n+1)!,
\end{equation*}
which implies that $u$ \emph{cannot} be analytic at $x=0$, so \emph{existence} fails in the analytic category. The behaviour of $u^{(n)}(0)$ in fact suggests instead the larger space of $(\sigma,2)$-Gevrey functions, with $\sigma \in \R_{>0}$, which are smooth functions such that moreover
\begin{equation*}
\sup_{x\in [0,1]} \left|\frac{d^n}{dx^n}u\right|(x)\leq C \sigma^{-n} (n!)^2
\end{equation*}
for all $n\in \N_0$ and for some constant $C>0$.\footnote{Note that a similar estimate with $(n!)^2$ replaced by $n!$ (1-Gevrey) would guarantee analyticity of $u$.}

Indeed, $u_2$ is not $(\sigma,2)$-Gevrey if $\sigma>|\re s|$ (see Lemma A.1.\ of \cite{gajwar19b}). Hence, loosely speaking, the role of the exponent $k$ in $C^k$ spaces, that plays an important role in the $\kappa>0$ setting, is taken on by $\sigma$ in the $\kappa=0$ setting. Note that in the present paper it is more convenient to work with an alternative notion of Gevrey functions, using $L^2$-norms rather than $L^{\infty}$-norms, see Section \ref{sec:enspaces}. 

\subsection{Main theorems}
\label{sec:roughmainthms}
We state below rough versions of the main theorems of the paper:

\begin{mainthm}[Rough statement]
\label{thm:rough1}
The set of eigenvalues of the infinitesimal generator of time-translations corresponding to the wave equation on extremal Reissner--Nordstr\"om,
\begin{equation*}
\mathcal{A}: \mathbf{H}\to \mathbf{H},
\end{equation*}
with $\mathbf{H}$ Hilbert spaces of initial data that are 2-Gevrey regular at infinity and at the horizon and are supported on suitable asymptotically null hypersurfaces, is \underline{discrete} when restricted to the sector\\ $\{|\textnormal{arg}(z)|<\frac{2}{3}\pi\}\subset \C$.
\end{mainthm}

\begin{remark}
We do not expect the sector $\left\{|\textnormal{arg}(z)|<\frac{2}{3}\pi\right\}$ to be optimal. In the context of a toy model equation, we in fact obtain a slightly larger sector in \cite{gajwar19b}.
\end{remark}

\begin{mainthm}[Rough statement]
\label{thm:rough2}
There exists subspaces $\mathbf{H}_{\rm res}< \mathbf{H}$ so that the union of eigenvalues of operators of the form
\begin{equation*}
\mathcal{A}: \mathbf{H}_{\rm res}\to \mathbf{H}_{\rm res},
\end{equation*}
restricted to $\left\{|\textnormal{arg}(z)|<\frac{2}{3}\pi\right\}$ agree precisely with poles of the meromorphic continuation of the resolvent operator corresponding to the constant $t$ foliation in this sector.
\end{mainthm}

We refer to the eigenvalues of $\mathcal{A}: \mathbf{H}\to \mathbf{H}$ in the sector $\left\{|\textnormal{arg}(z)|<\frac{2}{3}\pi\right\}$ as \textbf{regularity quasinormal frequencies} and the eigenvalues of $\mathcal{A}: \mathbf{H}_{\rm res}\to \mathbf{H}_{\rm res}$ as \textbf{scattering resonances}, in analogy with the nomenclature introduced in the cosmological setting in Section \ref{intro:cosmo}, and the corresponding eigenfunctions as regularity quasinormal modes and resonant states, respectively. See also the more precise statements in Definitions \ref{def:regqnf} and \ref{def:outqnf}.

\begin{mainthm}[Rough statement]
\label{thm:rough3}
For each regularity QNF $s$ and corresponding QNM $\hat{\psi}_s$ in extremal Reissner--Nordstr\"om, there exists a sequence of sub-extremal Reissner--Nordstr\"om--de Sitter spacetimes approaching extremal Reissner--Nordstr\"om with a corresponding sequence of regularity QNFs approaching $s$ and regularity QNMs approaching $\hat{\psi}_s$ with respect to the norm on the Hilbert space $\mathbf{H}$.
\end{mainthm}

See Theorem \ref{thm:discreigen}, Theorem \ref{thm:polesresolvent} and Theorem \ref{thm:convqnms} for more precise versions of  Theorem \ref{thm:rough1}, Theorem \ref{thm:rough2} and Theorem \ref{thm:rough3}, respectively. Theorems \ref{thm:rough1} and \ref{thm:rough2} address the points (I) and (III)--(VI) in Section \ref{intro:points} and Theorem \ref{thm:rough3} addresses (IX). In \cite{gajwar19b}, we address additionally (VII) and (VIII).

\begin{remark}
Theorems \ref{thm:rough1} and \ref{thm:rough2} illustrate that we can maintain the interpretation of quasinormal modes as eigenfunctions of $\mathcal{A}$ when $\Lambda=0$, with the key difference with the $\Lambda>0$ setting being that we need to adapt out choice of function spaces by restricting to suitably \emph{Gevrey regular} functions at infinity and at the horizon; see Section \ref{sec:enspaces}. In fact, these function spaces will allow us to simultaneously also consider the case of small $\Lambda> 0$ and derive estimates that are \emph{uniform in (small) $\Lambda$}.
\end{remark}

\begin{remark}
\label{rmk:relationapproach12}
Theorems \ref{thm:rough1} and \ref{thm:rough2} address the relation between Leaver QNFs and scattering resonances (see Section \ref{sec:approach1} and \ref{sec:approach2})  by presenting a ``third'' notion of regularity QNFs as eigenvalues of the infinitesimal generator of time translations for the wave equation on suitable Hilbert spaces (Theorem \ref{thm:rough1}). By Theorem \ref{thm:rough2}, scattering resonances form a subset of the set of regularity QNFs. In \cite{gajwar19b}, we show that, in the setting of a toy model problem, Leaver QNFs also form a subset of the set of regularity QNFs. \emph{The present paper and \cite{gajwar19b} therefore suggest that Leaver QNFs and scattering resonances can be related \underline{via} regularity QNFs}, which would be a resolution to problem (VII).
\end{remark}

\begin{remark}
\label{rmk:QNMvssmoothmode}
Another consequence of Theorems \ref{thm:rough1} and \ref{thm:rough2} is that the smooth mode solutions with \emph{arbitrary} frequencies $s$, which are mentioned in Section \ref{sec:smoothmode}, do \underline{not} play a role when considering initial data with suitable Gevrey regularity at infinity and at the horizon in the extremal setting. Moreover, we do not expect the Gevrey regularity condition at the horizon to be necessary in the sub-extremal setting; instead, we expect it can be replaced by finite Sobolev regularity at the horizon. Note that Gevrey regular data at infinity include in particular data with finite Sobolev norms that are \emph{compactly supported}. 
\end{remark}

If we restrict to solutions with a fixed angular frequency, we obtain an additional statement which addresses (II) in Section \ref{intro:points}.
\begin{mainthm}[Rough statement]
\label{thm:rough4}
If we restrict to a fixed angular frequency $\ell$, there exists a radius $\delta_{\ell}>0$ such that all the corresponding the regularity QNFs are supported outside of the ball $B_{\delta_{\ell}}$.
\end{mainthm}

See Theorem \ref{thm:smallslargel} for more a precise statement.

\begin{remark}
While Theorem \ref{thm:rough4} guarantees that QNFs supported on a bounded set of angular frequencies are supported away from the origin in the complex plane (when restricted to an appropriate sector of the complex plane), it remains an open problem to rule out an possible accumulation of QNFs at the origin as $\ell\to \infty$.
\end{remark}

\subsection{Overview of paper}
We give in this section an overview of the remainder of the paper.

\begin{itemize}
\item In Section \ref{sec:geom}, we introduce the necessary geometric notions and we set the notation that is used in the rest of the paper.
\item In Section \ref{sec:enspaces}, we introduce the main Hilbert spaces that play a role in the paper and we define the precise notions of Gevrey regularity that we will use.
\item Equipped with the notions and notation from Sections \ref{sec:geom} and \ref{sec:enspaces}, we then state precisely the main theorems of the paper in Section \ref{sec:mainthms}.
\item Before proving the theorems of Section \ref{sec:mainthms}, we present the structure of the main proofs in Section \ref{sec:mainideas}, highlighting the main new ideas and techniques that are introduced in this paper.
\item In Section \ref{sec:eq} we derive the main equations that we will use to prove estimates.
\item In Section \ref{sec:propgevreyphys} we derive the necessary estimates in physical space that allow us to make sense of the infinitesimal generator of time translations $\mathcal{A}$ on the desired Hilbert spaces.
\item In Sections \ref{sec:degelliptic} -- \ref{sec:lowfreqgev}, we derive the main frequency space estimates of the paper.
\item We apply the estimates from Sections \ref{sec:degelliptic} -- \ref{sec:lowfreqgev} in Section \ref{sec:constrresolvent} together with functional analytic arguments to arrive at the desired properties of $\mathcal{A}$.
\item Finally, we relate the eigenvalues of $\mathcal{A}$ to scattering resonances in Section \ref{sec:tradqnms}.
\end{itemize}

\subsection{Acknowledgements}
The authors would like to thank Mihalis Dafermos and Georgios Moschidis for helpful discussions. They would also like to thank an anonymous referee for a very thorough and helpful reading of the article.

\section{Geometric preliminaries}
\label{sec:geom}
In this section, we review elementary properties of the three-parameter family of Reissner--Nordstr\"om--de Sitter spacetimes. We moreover introduce coordinate charts that will allow us to treat \emph{uniformly} spacetime backgrounds with non-negative cosmological constants.
\subsection{Reissner--Nordstr\"om(--de Sitter) spacetimes}
\label{sec:spacetimes}
We treat both the cases $\Lambda>0$ and $\Lambda=0$. We start by considering the Lorentzian manifolds $(\widetilde{\mathcal{M}}_{\Lambda}, g_{\Lambda})$, with $\widetilde{\mathcal{M}}_{\Lambda}=\R_v\times (r_+,r_c)_r\times \s^2$ and:
\begin{equation*}
g_{\Lambda}=-D(r)dv^2+2dvdr+r^2(d\theta^2+\sin^2\theta d\varphi^2),
\end{equation*}
where
\begin{equation}
\label{eq:D}
D(r)=1-\frac{2M}{r}+\frac{e^2}{r^2}-\frac{\Lambda}{3}r^2,
\end{equation}
$M>0$, $e\in \R$ with $|e|\leq M$, $\Lambda> 0$, and $r_c=r_c(e,M,\Lambda)$ and $r_+=r_+(e,M,\Lambda)$ the largest and second-largest roots of the quartic polynomial $r^2D(r)$, respectively, which we will assume to be distinct; see Section \ref{sec:metric2} for further properties of the roots.

Given our choice of coordinates, it is immediate that we can embed $\widetilde{\mathcal{M}}_{\Lambda}$ into the manifold-with-boundary $\mathcal{M}_{\Lambda,+}=\R_v\times [r_+,r_c]_r\times \s^2$, such that
\begin{equation*}
\partial \mathcal{M}_{\Lambda,+}=\mathcal{H}^+\cup \mathcal{C}^-,
\end{equation*}
where $\mathcal{H}^+=\{r=r_+\}\subset\mathcal{M}_{\Lambda,+}$ is the \emph{future event horizon} and $\mathcal{C}^-=\{r=r_c\}\subset \mathcal{M}_{\Lambda,+}$ is the \emph{past cosmological horizon}; see Figure \ref{fig:penrose}. Both $\mathcal{H}^+$ and $\mathcal{C}^-$ are null hypersurfaces and Killing horizons with respect to the Killing vector field $T=\partial_v$ on $\mathcal{M}_{\Lambda,+}$. We moreover take $T$ to fix the time-orientation on $\mathcal{M}_{\Lambda,+}$.

Alternatively, we can introduce a coordinate $u=v-2r_*$ on $\widetilde{\mathcal{M}}_{\Lambda}$, where $\frac{dr}{dr_*}=D(r)$, and consider the coordinate chart $(u,r,\theta,\varphi)$ on $\widetilde{\mathcal{M}}_{\Lambda}=\R_u\times (r_+,r_c)_r\times \s^2$:
\begin{equation*}
g_{\Lambda}=-D(r)du^2-2dudr+r^2(d\theta^2+\sin^2\theta d\varphi^2).
\end{equation*}

Note that we can also introduce $t=\frac{v+u}{2}$ to cover $\widetilde{\mathcal{M}}_{\Lambda}$ with \emph{standard static coordinates} $(t,r,\theta,\varphi)$.

We can now embed $\widetilde{\mathcal{M}}_{\Lambda}$ into the manifold-with-boundary $\mathcal{M}_{\Lambda,-}=\R_u\times [r_+,r_c]_r\times \s^2$, such that the boundary can be decomposed as follows:
\begin{equation*}
\partial \mathcal{M}_{\Lambda,-}=\mathcal{H}^-\cup \mathcal{C}^+,
\end{equation*}
where $\mathcal{H}^-=\{r=r_+\}\subset \mathcal{M}_{\Lambda,-}$ is the \emph{past event horizon} and $\mathcal{C}^+=\{r=r_c\}\subset \mathcal{M}_{\Lambda,-}$ is the \emph{future cosmological horizon}. Both $\mathcal{H}^+$ and $\mathcal{C}^-$ are null hypersurfaces and Killing horizons with respect to the Killing vector field $\partial_u$. Furthermore, $\partial_u=T$ on $\widetilde{\mathcal{M}}_{\Lambda}$.

The main manifold under consideration in this paper is 
\begin{equation*}
\mathcal{M}_{\Lambda}=\widetilde{\mathcal{M}}_{\Lambda}\cup \mathcal{H}^+ \cup \mathcal{C}^+, 
\end{equation*}
which cannot be covered by either a single $(v,r)$ or $(u,r)$ chart. See Figure \ref{fig:penrose} for an illustration.
\begin{figure}[h!]
\centering
\begin{subfigure}[b]{0.45\textwidth}
\centering
\includegraphics[width=2.1in]{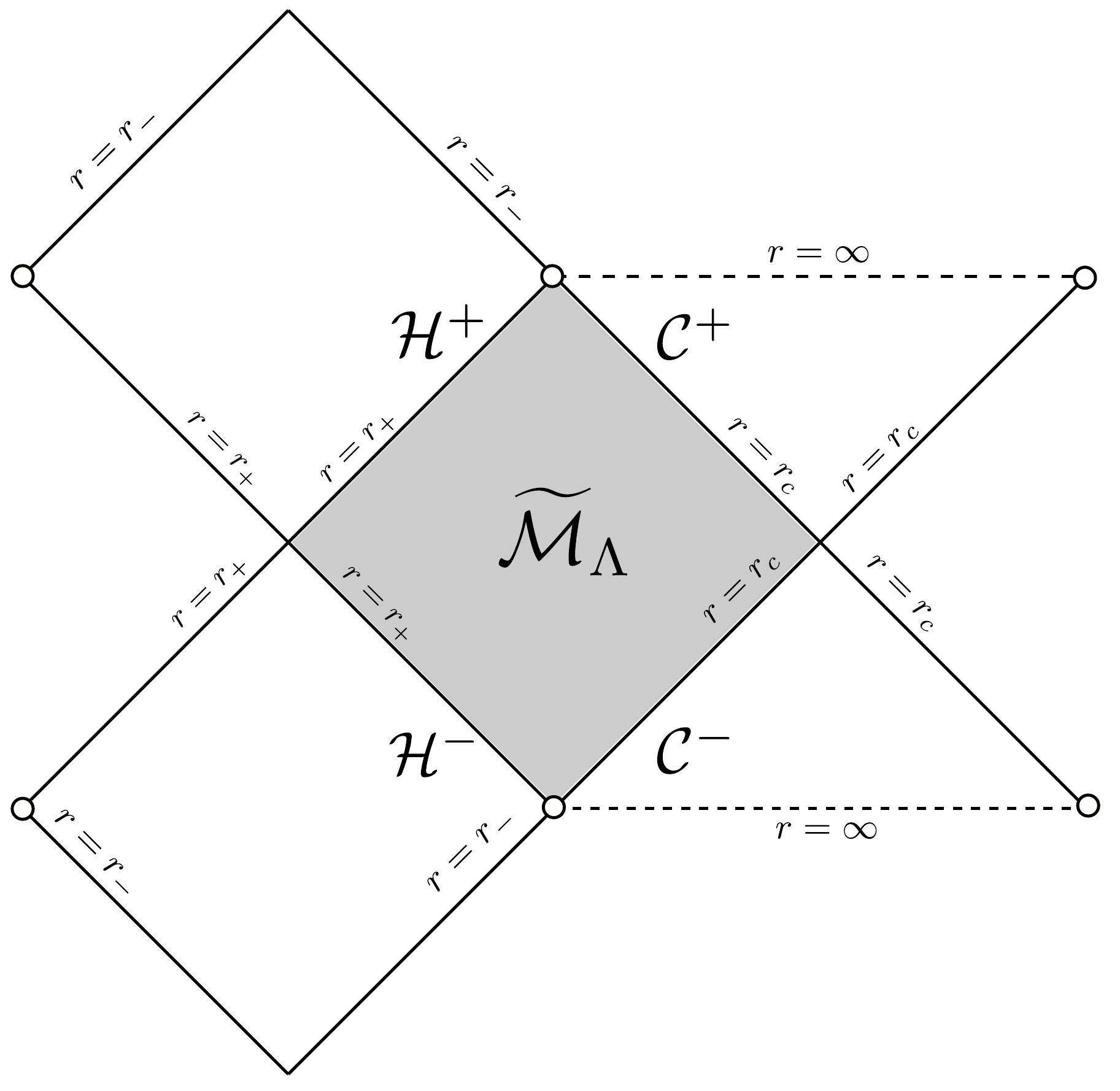}
\caption{Sub-extremal \\Reissner--Nordstr\"om--de Sitter.}
\end{subfigure}
\begin{subfigure}[b]{0.23\textwidth}
\includegraphics[width=1.5in]{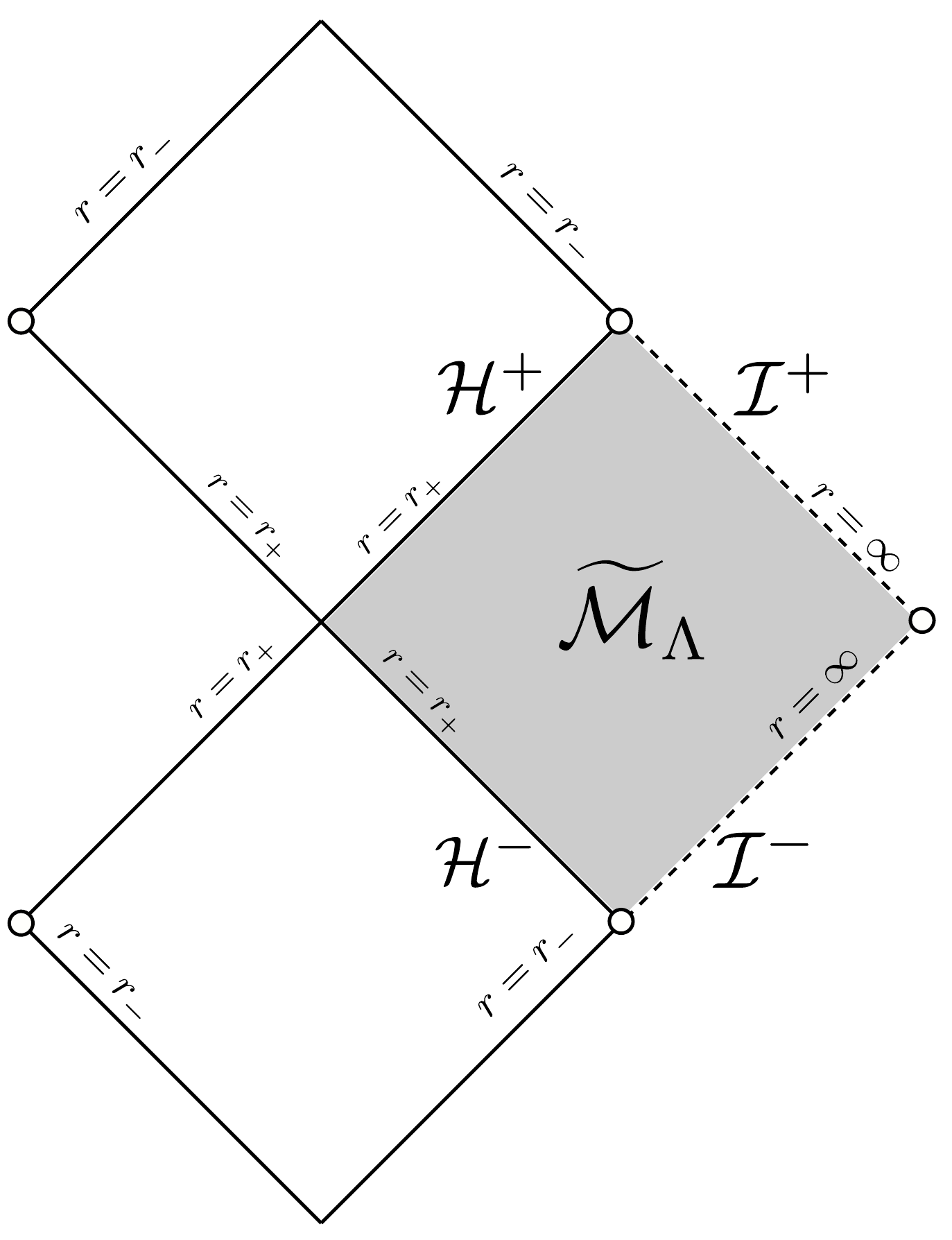}
\caption{Sub-extremal \\Reissner--Nordstr\"om}
\end{subfigure}
\begin{subfigure}[b]{0.3\textwidth}
\centering
\includegraphics[width=1.15in]{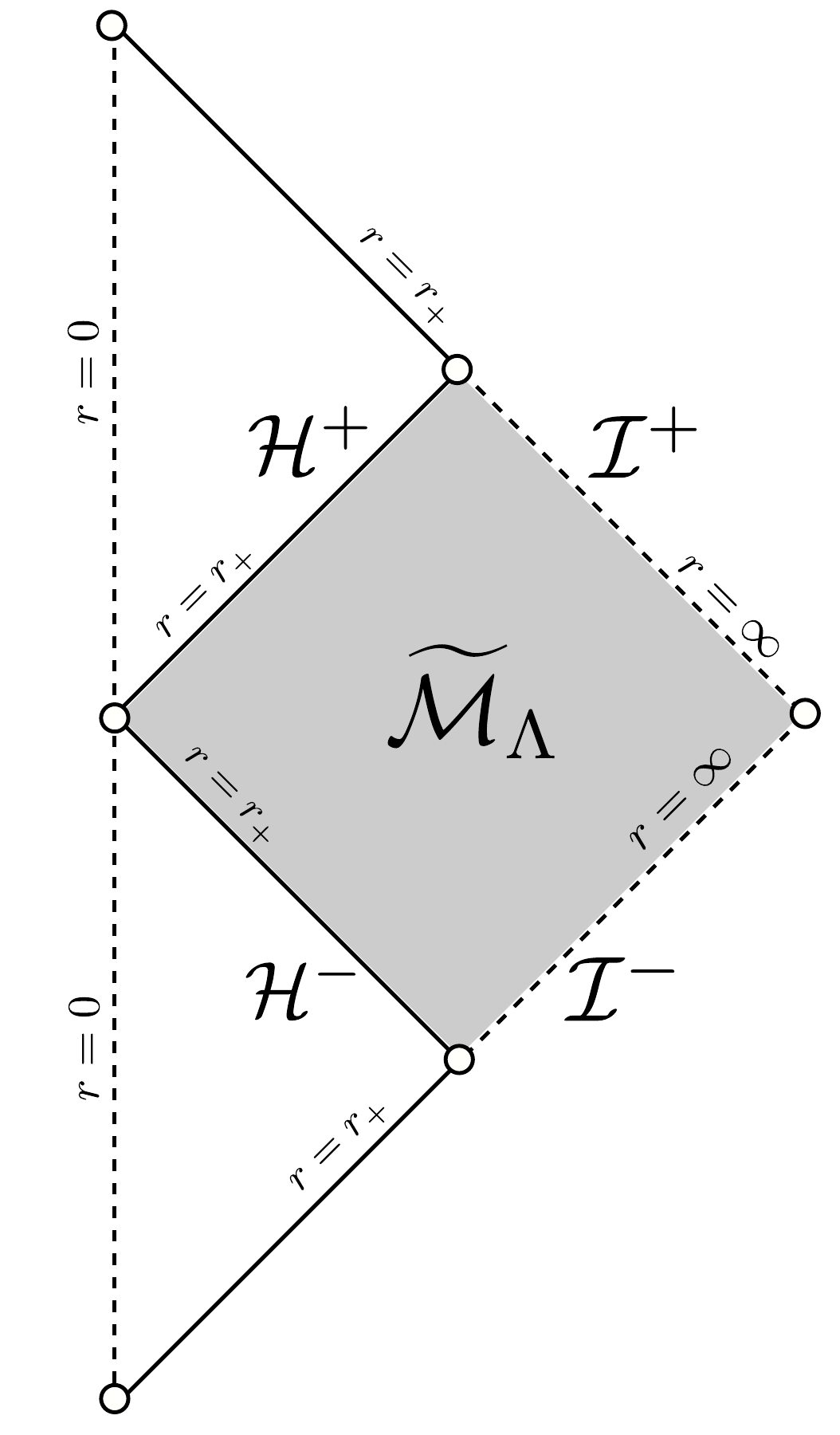}
\caption{Extremal\\ Reissner--Nordstr\"om}
\end{subfigure}
\caption{A Penrose diagrammatic representation of the embedding of $\widetilde{\mathcal{M}}_{\Lambda}$ in the Reissner--Nordstr\"om--de Sitter and Reissner--Nordstr\"om spacetimes in consideration.}
\label{fig:penrose}
\end{figure}

Now take $\Lambda=0$. The corresponding two-parameter family of Lorentzian manifolds, the Reissner--Nordstr\"om family $(\widetilde{\mathcal{M}}_{0}, g_{0})$ is defined as follows: $\widetilde{\mathcal{M}}_{0}=\R_v\times (r_+,\infty)_r\times \s^2$ and:
\begin{equation*}
g_{0}=-D(r)dv^2+2dvdr+r^2(d\theta^2+\sin^2\theta d\varphi^2),
\end{equation*}
where
\begin{equation*}
D(r)=1-\frac{2M}{r}+\frac{e^2}{r^2}.
\end{equation*} 
As above, we can extend $\widetilde{\mathcal{M}}_{0}$ by embedding it into a manifold with boundary, working in either $(v,r)$ or $(u,r)$ coordinates (defined as above, replacing $r_c$ with $\infty$.) to obtain
\begin{equation*}
\mathcal{M}_0=\widetilde{\mathcal{M}}_{0}\cup \mathcal{H}^+.
\end{equation*}
In contrast with the $\Lambda>0$, the boundary of $\mathcal{M}_0$ has only one connected component; there is no cosmological horizon present. Formally, one may think of ``$\{r=\infty\}$'' as replacing $\mathcal{C}^+$ above. This can be actually made precise by introducing a further extension of $\mathcal{M}_0$; see Section \ref{sec:confcoord} for more details. See Figure \ref{fig:penrose} for an illustration.

In fact, it is possible to extend $\mathcal{M}_{\Lambda,+}$ and $\mathcal{M}_{0,+}$ to $r<r_+$. The additional region in the extension is called the \emph{black hole region} and it will not play a role in this paper. Similarly, in the $\Lambda>0$ case it is possible to extend $\mathcal{M}_{\Lambda,-}$ to $r>r_c$. The additional region in this case is called the \emph{cosmological} or \emph{expanding region}.

\subsection{Further properties of the metric}
\label{sec:metric2}

Let $\Lambda>0$. We will denote $\mathfrak{l}^2=\frac{3}{\Lambda}$. Then it can be shown that the polynomial $r^2D(r)$ has four roots $r_0<0\leq r_-\leq r_+\leq r_c$, with $r_-=0$ if $e=0$ and $r_->0$ if $e\neq 0$. The roots $r_c$ and $r_+$ are the area radii at the cosmological horizons $\mathcal{C}^+$ and $\mathcal{C}^-$ and the event horizons $\mathcal{H}^+$ and $\mathcal{H}^-$, respectively. If $e\neq 0$, the root $r_-$ corresponds to the area radius at the \emph{inner} horizon. We may write
\begin{equation}
\label{eq:Dr2}
\begin{split}
r^2D(r)=&-\mathfrak{l}^{-2}(r-r_c)(r-r_+)(r-r_-)(r-r_0).
%=&-\mathfrak{l}^{-2}\Big[r^4-(r_c+r_++r_-+r_0)r^3+(r_cr_++r_cr_-+r_cr_0+r_+r_-+r_+r_0+r_-r_0)r^2\\
%&-(r_cr_+r_-+r_cr_+r_0+r_+r_-r_0)r+r_cr_+r_-r_0\Big]
\end{split}
\end{equation}
By comparing \eqref{eq:Dr2} with \eqref{eq:D}, we obtain the following identities for $r_c,r_+,r_-,r_0$:
\begin{align*}
r_0=&-(r_c+r_++r_-),\\
\mathfrak{l}^2=&\: r_c^2+r_+^2+r_-^2+r_-r_++r_-r_c+r_+r_c,\\
2M\mathfrak{l}^2=&\: r_c^2(r_++r_-)+r_+^2(r_c+r_-)+r_-^2(r_++r_c)+2r_cr_+r_-,\\
e^2\mathfrak{l}^2=&\:(r_c+r_++r_-)r_cr_+r_-.
\end{align*}
 Using that the surface gravities corresponding to the cosmological horizon and event horizon are given by $\kappa_c=-\frac{1}{2}D'(r_c)$ and $\kappa_+=\frac{1}{2}D'(r_+)$, respectively, we obtain:
 \begin{align}
 \label{eq:kappac}
 \kappa_c=&\:\frac{1}{2\mathfrak{l}^2} \frac{1}{r_c^2}(r_c-r_+)(r_c-r_-)(2r_c+r_++r_-),\\
  \label{eq:kappaplus}
  \kappa_+=&\:\frac{1}{2\mathfrak{l}^2} \frac{1}{r_+^2}(r_c-r_+)(r_+-r_-)(2r_++r_c+r_-).
 \end{align} 
 
 Let $\Lambda=0$. Then the polynomial $r^2D(r)$ has three roots $r_0<0\leq r_-\leq r_+$ and
\begin{align*}
r_0=&-(r_++r_-),\\
2M=&\: r_++r_-,\\
e^2=&\:r_+r_-.
\end{align*}
Suppose we fix $e$ and $M$, then $r_+(e,M,\Lambda)\to r_+(e,M,0)$ and $r_-(e,M,\Lambda)\to r_-(e,M,0)$ as $\mathfrak{l} \to \infty$. In particular, $r_+$ and $r_-$ stay bounded as $\mathfrak{l}\to \infty$. Therefore
\begin{equation*}
\frac{r_c}{\mathfrak{l}}\to 1
\end{equation*}
 as $\Lambda\downarrow 0$. It then follows that $\kappa_c(e,M,\Lambda)\to 0$ as $\Lambda\downarrow 0$ and $\kappa_+(e,M,\Lambda)\to \kappa_+(e,M)$ as $\Lambda\downarrow 0$.

\textbf{The main spacetimes of interest in this paper are the \emph{extremal Reissner--Nordstr\"om spacetimes}, which correspond to the limits $\Lambda\downarrow 0$ and $e^2\uparrow M^2$.} However, we will consider the bigger Reissner--Nordstr\"om--de Sitter family (with suitably small $\kappa_+$ and $\kappa_c$) in order to arrive at the desired estimates in extremal Reissner--Nordstr\"om.

\subsection{Conformal radial coordinates}
\label{sec:confcoord}
Let us introduce the \emph{conformal radial coordinate} $\rho_c=\frac{r_c-r}{r_cr}$, or equivalently, $\rho_c=\frac{1}{r}-\frac{1}{r_c}$. Then, $\frac{d\rho_c}{dr}=-r^{-2}$ and the metric $g_{\Lambda}$ takes the following form in $(u,\rho_c,\theta,\varphi)$ coordinates:
\begin{equation*}
g_{\Lambda}=r^2\left(-D(r)r^{-2}du^2+2du d\rho_c+d\theta^2+\sin^2\theta d\varphi^2\right)
\end{equation*}
with $u\in \R$, $\rho_c\in (0,\frac{r_c-r_+}{r_cr_+})$ and $r=\frac{1}{\rho_c+r_c^{-1}}$.

Using the expression for $\kappa_c$ in \eqref{eq:kappac}, we can further write
 \begin{equation}
 \label{eq:Dcosmo}
 Dr^{-2} =2\kappa_c\rho_c +(1-A_{r_c}^{(2)})\rho_c^2+A_{r_c}^{(3)}\rho_c^3+A_{r_c}^{(4)}\rho_c^4.
 \end{equation}
 where
 \begin{align*}
 A_{r_c}^{(2)}=&\:3\mathfrak{l}^{-2}\left[r_c(r_++r_-)+r_+^2+r_-^2-r_c^{-1}r_+r_-(r_++r_-)\right],\\
 A_{r_c}^{(3)}=&\:-(r_++r_-)+\mathfrak{l}^{-2}(4r_cr_+r_-+r_+^3+r_-^3+5r_+^2r_-+5r_-^2r_+),\\
 A_{r_c}^{(4)}=&\:e^2.
 \end{align*}
 Indeed, this follows from a Taylor expansion of $Dr^{-2}\mathfrak{l}^2$ in the variable $\rho_c$ around $\rho_c=0$, using that $Dr^{-2}\mathfrak{l}^2$ is a polynomial in $r^{-1}$ and therefore also in $\rho_c$:
  \begin{equation*}
 \begin{split}
 Dr^{-2}\mathfrak{l}^2=&\:\frac{r_c-r}{r}(1-r_+r^{-1})(1-r_-r^{-1})(1-r_0r^{-1})\\
 =&\: r_c\rho_c (1-r_+r^{-1})(1-r_-r^{-1})(1+(r_c+r_++r_-)r^{-1})\\
 =&\:r_c^{-2}(r_c-r_-)(r_c-r_+)(2r_c+r_-+r_+)\rho_c\\\
 &+\left(r_c^2-2r_-^2-2r_+^2+r_-r_+-2r_c(r_-+r_+)+3 r_c^{-1}r_-r_+(r_-+r_+)\right)\rho_c^2\\
 &+\left(-r_c(r_-+r_+)^2-r_c^2(r_-+r_+)+3r_-r_+(r_-+r_+)\right)\rho_c^3+r_-r_+r_c(r_c+r_-+r_+)\rho_c^4\\
 =&\: 2\kappa_c\mathfrak{l}^2\rho_c +(1-A_{r_c}^{(2)})\mathfrak{l}^2\rho_c^2+A_{r_c}^{(3)}\mathfrak{l}^2\rho_c^3+A_{r_c}^{(4)}\mathfrak{l}^2\rho_c^4.
\end{split}
 \end{equation*}
 From the above expression for $(Dr^{-2})(\rho_c)$ it follows that by considering the conformal metric $\hat{g}_{\Lambda}=r^{-2}g_{\Lambda}$ for $\Lambda\geq 0$, we can embed the Lorentzian manifold $(\widetilde{\mathcal{M}}_{\Lambda}, \hat{g}_{\Lambda})$ into the manifold-with-boundary $(\widehat{\mathcal{M}}_{\Lambda}, \hat{g}_{\Lambda})$, with $\widehat{\mathcal{M}}_{\Lambda}=\R_u \times [0,\frac{r_c-r_+}{r_cr_+})_{\rho_c}\times \s^2$. If $\Lambda>0$, we have that $\partial \widehat{\mathcal{M}}_{\Lambda}=\{\rho_c=0\}=\mathcal{C}^+$. In the $\Lambda=0$ case, we define
 \begin{equation*}
 \mathcal{I}^+:=\partial \widehat{\mathcal{M}}_{0}=\{\rho_c=0\},
 \end{equation*}
 and we refer to $\mathcal{I}^+$ as \emph{future null infinity}. Note that $\mathcal{I}^+$ is a \emph{degenerate Killing horizon}, i.e.\  it has vanishing surface gravity with respect to the conformal metric $\hat{g}_{\Lambda}$.
 
In order to be able to treat the spacetime regions near $\mathcal{H}^+$ and $\mathcal{C}+$ simultaneously, we introduce another radial coordinate that plays a similar role to $\rho_c$. Let $\rho_+=\frac{r-r_+}{r_+r}$, or equivalently, $\rho_+=\frac{1}{r_+}-\frac{1}{r}$. Then, $\frac{d\rho_+}{dr}=r^{-2}$. We can change from $(v,r,\theta,\varphi)$ coordinates to $(v,\rho_+,\theta,\varphi)$ coordinates on $\mathcal{M}_{\Lambda,+}$ and then it is immediate that
\begin{equation*}
\partial \mathcal{M}_{\Lambda,+}=\mathcal{H}^+=\{\rho_+=0\}
\end{equation*}
for both the $\Lambda>0$ and $\Lambda=0$ cases. Note that in this case there is no need to pass to a conformal metric and extend the spacetime as the original metric $g_{\Lambda}$ is well-defined at $\mathcal{H}^+$ when $\Lambda\geq 0$.

Using the expression for $\kappa_+$ in \eqref{eq:kappaplus}, we can write
 \begin{equation}
  \label{eq:Devent}
 Dr^{-2}= 2\kappa_+ \rho_+ +(1-A_{r_+}^{(2)})\rho_+^2+A_{r_+}^{(3)}\rho_+^3+A_{r_+}^{(4)}\rho_+^4,
 \end{equation}
 where
 \begin{align*}
  A_{r_+}^{(2)}=&\:3\mathfrak{l}^{-2}\left[r_+(r_c+r_-)+r_c^2+r_-^2-r_+^{-1}r_c r_-(r_c+r_-)\right],\\
 A_{r_+}^{(3)}=&\:(r_++r_-)-\mathfrak{l}^{-2}\left[4r_cr_+r_-+r_c^3+r_-^3+5r_c^2r_-+5r_-^2r_c\right],\\
 A_{r_+}^{(4)}=&\:e^2.
 \end{align*}
This follows from a Taylor expansion of $Dr^{-2}\mathfrak{l}^2$ in the variable $\rho_+$ around $\rho_+=0$, using that $Dr^{-2}\mathfrak{l}^2$ is a polynomial in $r^{-1}$ and therefore also in $\rho_+$:
   \begin{equation*}
 \begin{split}
 Dr^{-2}\mathfrak{l}^2=&\:-\frac{r-r_+}{r}(1-r_cr^{-1})(1-r_-r^{-1})(1-r_0r^{-1})\\
 =&\: -r_+\rho_+ (1-r_cr^{-1})(1-r_-r^{-1})(1+(r_++r_-+r_c)r^{-1})\\
 =&\: r_+^{-2}(r_c-r_+)(r_+-r_-)(2r_++r_c+r_-)\rho_+\\
 &+\left(-2r_c^2-2r_-^2+r_c r_--2(r_c+r_-)r_++r_+^2\right)\rho_+^2\\
 &+\left(-3r_cr_-(r_c+r_-)+r_+(r_c-r_-)^2+r_+^2(r_c+r_-)\right)\rho_+^3+r_cr_-r_+(r_c+r_-+r_+)\rho_+^4\\
 =&\: 2\kappa_+ \mathfrak{l}^2\rho_+ +(1-A_{r_+}^{(2)})\mathfrak{l}^2\rho_+^2+A_{r_+}^{(3)}\mathfrak{l}^2\rho_+^3+A_{r_+}^{(4)}\mathfrak{l}^2\rho_+^4.
 \end{split}
 \end{equation*}

It is convenient to introduce the following modification of $\rho_+$ and $\rho_c$. We will employ the notation $\rho$ when we do not discriminate between $\rho_+$ and $\rho_c$. Let
\begin{equation*}
\tilde{\rho}=\frac{\rho}{1-M\rho}.
\end{equation*}
Then
\begin{equation*}
\partial_{\rho}=\frac{d\tilde{\rho}}{d\rho} \partial_{\tilde{\rho}}=\frac{1}{(1-M\rho)^2} \partial_{\tilde{\rho}}.
\end{equation*}
Furthermore,
\begin{equation*}
\rho= \frac{\tilde{\rho}}{1+M\tilde{\rho}}.
\end{equation*}
Note in particular that $\tilde{\rho}(0)=0$.

 \subsection{Foliations}
 \label{sec:foliations}
We will construct a suitable spacelike hypersurface $\Sigma_0$ in $\mathcal{M}_{\Lambda}$.

In order to avoid ambiguity when passing from $(v,r)$ to $(u,r)$ coordinates, we introduce the following null vector fields:
\begin{align*}
L=&\:T+\frac{D}{2}\partial_r,\\
\underline{L}=&\:-\frac{D}{2}\partial_r,\\
\hat{\underline{L}}=&\:r^2 \partial_r,
\end{align*}
with respect to $(v,r)$ coordinates. The vector fields $L$ and $\underline{L}$ take the following form when expressed in $(u,r)$ coordinates:
\begin{align*}
L=&\:\frac{D}{2}\partial_r,\\
\underline{L}=&\:T-\frac{D}{2}\partial_r.
\end{align*}
Let us also define the rescaled vector field
\begin{equation*}
\hat{L}=-r^2\partial_r.
\end{equation*}
Consider a piecewise smooth function $h_{r_+}:(r_+,r_c)\to \R$, where $r_c\leq \infty$, satisfying the following properties: 
\begin{enumerate}[(i)]
\item $0\leq h_{r_+}(r)\leq \frac{2}{D(r)}$,
\item $|h_{+}(r)|\leq C_0$ in $\{r\leq r_++\min\{r_+, \frac{r_c-r_+}{2}\}\}$,
\item $|r^2(2D^{-1}-h_{r_+})(r)|\leq C_0M^2$ in $\{r\geq r_++\min\{r_+, \frac{r_c-r_+}{2}\}\}$,
\end{enumerate}
with $C_0>0$ a numerical constant. 

It will be useful to consider the following special choice of $h_{r_+}$: take $h_{r_+}$ to be a piecewise smooth function, with $h_{r_+}(r)=0$ when $r< r((\rho_+)_0)$, $h_{r_+}(r)=2D^{-1}(r)$ when $r> r((\rho_c)_0)$ and $0<h_{r_+}(r)<2D^{-1}$ smooth when $r((\rho_+)_0)\leq r \leq r((\rho_+)_0)$. Note that $h_{r_+}$ is therefore discontinuous at $r=r((\rho_+)_0)$ and $r((\rho_c)_0)$, and it is smooth everywhere else.

We define $h_{r_c}(r):=2D^{-1}(r)-h_{r_+}(r)$ and introduce the vector field 
\begin{equation*}
Y=-\frac{2}{D}\underline{L}+h_{r_+} T=\frac{2}{D}L-h_{r_c} T.
\end{equation*}
Note that the properties (i)--(iii) above imply that $Y$ must be spacelike or null everywhere. 

Let $\gamma: (r_+,r_c) \to \mathcal{R}_v \times (r_+,r_c)_r$ be the unique function that satisfies:
\begin{enumerate}[(a)]
\item $\lim_{r\downarrow r_+} \gamma(r)=(v_0,r_+)$, with $v_0>0$,
\item $\frac{d}{dr}\gamma(r)=Y$.
\end{enumerate}
Then $(\gamma,\theta_0,\varphi_0)$ defines an integral curve of $Y$ in $\mathcal{M}_{\Lambda}$ for each $\theta_0\in (0,\pi), \varphi_0\in (0,2\pi)$ that has a limit point $(v_0,r_+,\theta_0,\varphi_0)\in \mathcal{H}^+$ in $(v,r)$ coordinates.

The hypersurface $\Sigma_0=\gamma\times \s^2\subset \mathcal{M}_{\Lambda}$ is null or spacelike everywhere and has a boundary that intersects $\mathcal{H}^+$ at $\mathcal{H}^+\cap \{v=v_0\}$.

Let us denote with $v_{\Sigma_0}(r)$ the value of the $v$ coordinate along $\Sigma_0$ and $u_{\Sigma_0}(r)$ the value of the $u$ coordinate along $\Sigma_0$. By construction $\frac{dv_{\Sigma_0}}{dr}=h_{r_+}$.

Since $u=v-2r_*$, we have that 
\begin{equation*}
\frac{du_{\Sigma_0}}{dr}=h_{r_+}-2D^{-1}=-h_c(r).
\end{equation*}
Let $r_0<r_+$ be arbitrary. By property (i)---(iii) above, we can bound for all $r\in (r_0,r_c)$: $|h_c(r)|\leq C(r_0)M^2r^{-2}$, for some constant depending on $r_0$. After integrating in $r$ from $r=r_0$ to $r=r_c$ or $r=\infty$, it therefore follows that $u|_{\Sigma_0}$ attains a finite value as $r\to r_c$ or $r\to \infty$. In other words, the curve $(\gamma,\theta_0,\varphi_0)$ has a limit point on $\mathcal{C}^+$ or $\mathcal{I}^+$.

We define a foliation $\Sigma_{\tau}$ by flowing $\Sigma_0$ along the integral curves of $T$, with $T(\tau)=1$. We denote
\begin{equation*}
\mathcal{R}=\bigcup_{\tau \in [0,\infty)}\Sigma_{\tau}.
\end{equation*}
See Figure \ref{fig:foliations} for an illustration in the $\Lambda=0$ setting. 

By considering the extensions $\widehat{\Sigma}_{\tau}$ of $\Sigma_{\tau}$ in the larger manifold $\widehat{\mathcal{M}}_{\Lambda}\cup \mathcal{H}^+$, we can moreover consider the region
\begin{equation*}
\widehat{\mathcal{R}}=\bigcup_{\tau \in [0,\infty)}\widehat{\Sigma}_{\tau}.
\end{equation*}

We introduce the coordinate chart $(\tau,r, \theta,\varphi)$ associated to the foliation $\Sigma_{\tau}$. In these coordinates, we have that $\partial_{\tau}=T$ and $\partial_{r}=Y$.

We will also consider $(\tau,\tilde{\rho}_+,\theta,\varphi)$ coordinates. In these coordinates,
\begin{equation}
\label{eq:eqrhohat1}
\partial_{\tilde{\rho}_+}=(1+M\tilde{\rho}_+)^{-2}\partial_{\rho_+}=(1+M\tilde{\rho}_+)^{-2}r^2Y=(1+M\tilde{\rho}_+)^{-2}\hat{\underline{L}}+r^2h_{r_+}(1+M\tilde{\rho}_+)^{-2} T.
\end{equation}
If instead we consider $(\tau,\tilde{\rho}_c,\theta,\varphi)$ coordinates, we can express
\begin{equation}
\label{eq:eqrhohat2}
\partial_{\tilde{\rho}_c}=(1+M\tilde{\rho_c})^{-2}\partial_{\rho_c}=-(1+M\tilde{\rho}_c)^{-2}r^2Y=(1+M\tilde{\rho}_c)^{-2}\hat{L}+r^2h_{r_c}(1+M\tilde{\rho}_c)^{-2} T.
\end{equation}
Let us define $\hat{h}= r^2h(1+M\tilde{\rho})^{-2}$, where $h=h_{r_+}$ or $h=h_{r_c}$ and $\rho=\rho_+$ or $\rho=\rho_c$, respectively.

Let $(\rho_+)_0,(\rho_c)_0\in (0,\frac{r_c-r_+}{r_cr_+})$. We denote $r((\rho_+)_0)=r|_{\rho_+=(\rho_+)_0}$, $r((\rho_c)_0)=r|_{\rho_c=(\rho_c)_0}$ and assume $(\rho_+)_0,(\rho_c)_0$ are chosen such that $r((\rho_+)_0)<r((\rho_c)_0)$. 

Consider again the special choice of $h_{r_+}$ defined above: $h_{r_+}(r)=0$ when $r< r((\rho_+)_0)$, $h_{r_c}(r)=0$ when $r> r((\rho_c)_0)$ and $0<h_{r_+}(r)<2D^{-1}$ smooth when $r((\rho_+)_0)\leq r \leq r((\rho_+)_0)$.

With the above choice of $h_{r_+}$, we can split
\begin{align*}
\Sigma_{\tau}=&\:\underline{N}_{\tau}\cup S_{\tau}\cup N_{\tau},
\end{align*}
where $\underline{N}_{\tau}=\Sigma_{\tau}\cap \{r\leq r((\rho_+)_0)\}$ are ingoing null hypersurfaces, $N_{\tau}=\Sigma_{\tau}\cap \{ r\geq r((\rho_c)_0)\}$ are outgoing null hypersurfaces and $S_{\tau}=\Sigma_{\tau}\cap \{ r((\rho_+)_0)\leq r \leq r((\rho_c)_0)\}$ are spacelike.

\subsection{Additional notation}

For convenience, let us introduce the following notation: let  $0<(\rho_+)_0,(\rho_c)_0<\frac{r_c-r_+}{r_cr_+}$, then
\begin{align*}
R_0^+:=&\:r((\rho_+)_0),\\
R_0^c:=&\:r((\rho_c)_0).
\end{align*}

Suppose $\psi$ is a solution to the conformally invariant Klein--Gordon equation
 \begin{equation}
 \label{eq:waveequation}
 \square_{g_{\Lambda}}\psi=\frac{1}{6}R[g_{\Lambda}]\psi=\frac{2}{3}\Lambda \psi=2\mathfrak{l}^{-2}\psi,
 \end{equation}
 where $R[g_{\Lambda}]$ is the Ricci scalar corresponding to the metric $g_{\Lambda}$. Note that in the case $\Lambda=0$, \eqref{eq:waveequation} is simply the wave equation with respect to the Reissner--Nordstr\"om metric $g_0$.

We denote the components of the stress-energy tensor corresponding to \eqref{eq:waveequation} as follows:
\begin{equation*}
\mathbb{T}_{\mu \nu}[\psi]=\partial_{\mu}\psi \partial_{\nu }\psi-\frac{1}{2}g_{\mu \nu}[(g^{-1})^{\alpha \beta} \partial_{\alpha}\psi\partial_{\beta}\psi+2\mathfrak{l}^{-2}\psi^2 ]
\end{equation*}
and we denote with $\mathbf{n}_{\tau}$ the normal vector field to $\Sigma_{\tau}$ and $d\mu_{\tau}$ the induced volume form on $\Sigma_{\tau}$. We also use the notation $\mathbf{n}_{\Sigma_0}:=\mathbf{n}_0$ and $d \mu_{\Sigma_0}:=d\mu_0$. Note that if $\Sigma_{\tau}$ has an null piece, we let $d\mu_{\tau}=r^2\sin \theta d\theta d\varphi$ along the null piece and $\mathbf{n}_{\tau}=\underline{L}$ if the null piece is ingoing and $\mathbf{n}_{\tau}={L}$ if it is outgoing.
 \section{Gevrey regularity and Hilbert spaces}
  \label{sec:enspaces}
  
  In this section, we introduce a notion of Gevrey regularity and we define the main Hilbert spaces that will be relevant in the paper.
 \begin{definition}
 \label{def:gev}
We define the \textbf{$(\sigma,2)$-Gevrey inner products} $\la \cdot , \cdot \ra_{G^2_{\sigma,j,\rho_0}}$ with $j=0,1,2$ on the space $C^{\infty}((0,\rho_0);\C)$ with $\sigma \in {\R_{>0}}$ and $\rho_0>0$ as follows: let $f,g\in  C^{\infty}((0,\rho_0);\C)$, then
\begin{align*}
\la f, g \ra_{G^2_{\sigma,0,\rho_0}}=&\:\sum_{n=0}^{\infty} \frac{\sigma^{2n}}{n!^2 (n+1)!^2}\int_0^{\rho_0} \partial_{\rho}^n f \partial_{\rho}^n \overline{g}\,d\rho,\\
\la f, g \ra_{G^2_{\sigma,1,\rho_0}}=&\:\sum_{n=0}^{\infty} \frac{\sigma^{2n}}{n!^2 (n+1)!^2}\int_0^{\rho_0} \left[\partial_{\rho}^n f \partial_{\rho}^n \overline{g}+\partial_{\rho}^{n+1} f \partial_{\rho}^{n+1} \overline{g}\right]\,d\rho,\\
\la f, g \ra_{G^2_{\sigma,2,\rho_0}}=&\:\sum_{n=0}^{\infty} \frac{\sigma^{2n}}{n!^2 (n+1)!^2}\int_0^{\rho_0} \left[\partial_{\rho}^n f \partial_{\rho}^n \overline{g}+\partial_{\rho}^{n+1} f \partial_{\rho}^{n+1} \overline{g}+\rho^4\partial_{\rho}^{n+2} f \partial_{\rho}^{n+2} \overline{g}\right]\,d\rho
\end{align*}
We denote the corresponding norms by $||\cdot||_{G^2_{\sigma,0,\rho_0}}$, $||\cdot||_{G^2_{\sigma,1,\rho_0}}$,  $||\cdot||_{G^2_{\sigma,2,\rho_0}}$, respectively.

We refer to functions $f\in C^{\infty}((0,\rho_0);\C)$ with $||f||_{G^2_{\sigma,0,\rho_0}}<\infty$ as \textbf{$(\sigma,2)$-Gevrey functions} on $[0,\rho_0]$.
\end{definition}

\begin{lemma}
The spaces $(G^2_{\sigma,j,\rho_0}, \la \cdot , \cdot \ra_{G^2_{\sigma,j,\rho_0}})$, with $j=0,1,2,$ and
\begin{equation*}
G^2_{\sigma,j,\rho_0}:=\left\{f\in  C^{\infty}((0,\rho_0);\C)\:|\: ||f||_{G^2_{\sigma,j,\rho_0}}<\infty\right\},
\end{equation*}
are Hilbert spaces.
\end{lemma}
\begin{proof}
It is straightforward to see that $(G^2_{\sigma,j,\rho_0}, \la \cdot , \cdot \ra_{G^2_{\sigma,j,\rho_0}})$ are well-defined inner product spaces. We will show that every Cauchy sequence in $G^2_{\sigma,j,\rho_0}$ converges with respect to $||\cdot||_{G^2_{\sigma,j,\rho_0}}$. We will consider the case $j=0$. The cases $j=1,2$ can be treated similarly. Let $\{f_n\}$ be a Cauchy sequence in $G^2_{\sigma,0,\rho_0}$. By standard Sobolev embeddings on $(0,\rho_0)$ and completeness of the Sobolev spaces $H^N$, we must have that $f_n$ converges in $H^N$ to a smooth function $f$, for any $N\in \N$. 

By the Cauchy property of $\{f_n\}$, we have that for all $\epsilon>0$, there exists $L>0$ such that for all $k>l>L$ and for any $N\in \N_0$:
\begin{equation*}
\sum_{n=0}^{N} \frac{\sigma^{2n}}{n!^2(n+1)!^2}||\partial_{\rho}^n(f_k-f_{l})||^2_{L^2(0,\rho_0)}<\epsilon
\end{equation*}
Hence, by taking the limit $k \to \infty$ and using the convergence of $\{f_n\}$ in $H^N$, we obtain that for any $N\in \N_0$,
\begin{equation*}
\sum_{n=0}^{N} \frac{\sigma^{2n}}{n!^2(n+1)!^2}||\partial_{\rho}^n(f-f_{l})||^2_{L^2(0,\rho_0)}<\epsilon
\end{equation*}
We can now take the limit $N \to \infty$ to obtain $f-f_{\ell}\in G^2_{\sigma,0,\rho_0}$, and we can conclude that $f\in G^2_{\sigma,0,\rho_0}$.

\end{proof}

\begin{definition}
Let $0<\rho_0=(\rho_+)_0=(\rho_c)_0<\frac{r_c-r_+}{r_cr_+}$, with $r_c<\infty$ or $r_c=\infty$. For a function $f\in  C^{\infty}((r_+,r_c);\C)$, we denote $f_+=f|_{r\in (r_+,R^+_0)}$ and $f_c=f|_{r\in  (R_0^c,r_c)}$. We introduce the following inner product on $C^{\infty}((r_+,r_c);\C)$: let $f,g \in C^{\infty}((r_+,r_c);\C)$, then
\begin{align*}
\la f,g \ra_{\sigma,\rho_0}:=  &\:\la rf_+, r g_+ \ra_{G^2_{\sigma,0,\rho_0}}+ \la rf_c, r g_c \ra_{G^2_{\sigma,0,\rho_0}}+ \la f, g\ra_{L^2[R^+_0,R^c_0]},\\
\la f,g \ra_{\sigma,1,\rho_0}:= &\: \la rf_+, r g_+ \ra_{G^2_{\sigma,1,\rho_0} }+ \la rf_c, r g_c \ra_{G^2_{\sigma,1,\rho_0}}+ \la f, g\ra_{H^1[R^+_0,R^c_0]},\\
\la f,g \ra_{\sigma, 2, \rho_0}:= &\: \la rf_+, r g_+ \ra_{G^2_{\sigma,2,\rho_0} }+ \la rf_c, r g_c \ra_{G^2_{\sigma,2,\rho_0}}+ \la f, g\ra_{H^2[R^+_0,R^c_0]}.
\end{align*}
We denote the norms corresponding to $\la \cdot , \cdot \ra_{\sigma,\rho_0}$, $\la \cdot , \cdot \ra_{\sigma,1,\rho_0}$ and $\la \cdot , \cdot \ra_{\sigma,2,\rho_0}$ by $||\cdot||_{\sigma,\rho_0}$ , $||\cdot||_{\sigma,1,\rho_0}$ and  $||\cdot||_{\sigma,2,\rho_0}$, respectively. We then define the Hilbert spaces $H_{\sigma,\rho_0}$, $H_{\sigma,1,\rho_0}$, $H_{\sigma,2,\rho_0}$ as the completions of the spaces 
\begin{align*}
&\{f\in  C^{\infty}((r_+,r_c);\C)\:|\: ||rf_+||_{G^2_{\sigma,0,\rho_0}}+||rf_c||_{G^2_{\sigma,0,\rho_0}}<\infty\},\\
& \{f\in  C^{\infty}((r_+,r_c);\C)\:|\: ||rf_+||_{G^2_{\sigma,1,\rho_0}}+||rf_c||_{G^2_{\sigma,1,\rho_0}}<\infty\},\\
&\{f\in  C^{\infty}((r_+,r_c);\C)\:|\: ||rf_+||_{G^2_{\sigma,2,\rho_0}}+||rf_c||_{G^2_{\sigma,2,\rho_0}}<\infty\}.
\end{align*}
 with respect to the norms $||\cdot ||_{\sigma,\rho_0}$\,, $||\cdot ||_{\sigma,1,\rho_0}$ and $||\cdot ||_{\sigma,2,\rho_0}$, respectively.
\end{definition}

Let $\ell\in \N_0$ and consider the projection operators
\begin{align*}
\pi_{\ell}:& L^2(\s^2)\to L^2(\s^2),\\
\pi_{\ell}f=&f_{\ell}:=\sum_{m=-\ell}^{\ell} f_{\ell m} Y_{\ell m}(\theta,\varphi),
\end{align*}
with $f_{\ell m} \in \C$ and $Y_{\ell m}$, $m=-\ell,\ldots,\ell$ the spherical harmonics with angular momentum $\ell$.

The operator $\pi_{\ell}$ is well-defined on the  domains $C^{\infty}(\Sigma_0)$ and $C^{\infty}(S_0)$, and since it is a bounded linear operator with respect to $||\cdot||_{L^2(\Sigma_0)}$ and $||\cdot||_{L^2(S_0)}$, the following extensions are also well-defined:
\begin{align*}
\pi_{\ell}: L^2(\Sigma_0)\to L^2(\Sigma_0),\\
\pi_{\ell}: L^2(S_0)\to L^2(S_0).
\end{align*}
We use the notation $\pi_{\ell}$ both for the operator acting on $L^2(\s^2)$ and the extensions to $L^2(\Sigma_0)$ and $L^2(S_0)$, for the sake of notational convenience.

Let us now introduce the sets of fixed angular momentum $\ell$: 
\begin{equation*}
V_{\ell}=\ker (\pi_{\ell}-id)\times \ker (\pi_{\ell}-id)\subset L^2(\Sigma_0)\times L^2(S_0).
\end{equation*}

Note that $\partial_{\varphi} Y_{\ell m}=im Y_{\ell m}$. We denote moreover
\begin{equation*}
V_{\ell m}=(\ker (\pi_{\ell}-id)\times \ker (\pi_{\ell}-id))\cap (\ker(\partial_{\varphi}-im)\times  \ker(\partial_{\varphi}-im)).
\end{equation*}

\begin{definition}
Let $\Sigma_0=\underline{N}_0\cup S_0\cup N_0$. We introduce the following norm on $C^{\infty}(\Sigma_0;\C)\times C^{\infty}(S_0;\C)$: let $\Psi \in C^{\infty}(\Sigma_0;\C)$ and $\Psi'\in C^{\infty}(S_0;\C)$. We denote with $\psi$ the unique solution to the wave equation satisfying $\psi|_{\Sigma_0}=\Psi$ and $T\psi|_{S_0}=\Psi'$. Then:
\begin{align*}
||(\Psi,\Psi')||_{\mathbf{H}_{\sigma,\rho_0}}^2:=&\:\int_{\Sigma_0} \mathbb{T}(T,\mathbf{n}_{\Sigma_0})[\psi]\,d\mu_{\Sigma_0}+ \sum_{\ell=0}^{\infty}\Bigg[\sum_{n=1}^{\ell-1}\frac{\sigma^{2n}}{\ell^{2n}(\ell+1)^{2n}} (\ell+1)^4 \int_{\underline{N}_0} |\hat{\underline{L}}^n(r\Psi_{\ell})|^2\,d\omega d\rho_+ \\
&+\frac{\ell^{2\ell}(\ell+1)^{2\ell}}{\ell!^{2} (\ell+1)!^2} \sum_{n=\ell}^{\infty} \frac{\sigma^{2n}}{n!^2(n+1)!^2}  (n+1)^4\int_{\underline{N}_0} |\hat{\underline{L}}^n(r\Psi_{\ell})|^2\,d\omega d\rho_+\Bigg]\\
&+ \sum_{\ell=0}^{\infty}\Bigg[\sum_{n=1}^{\ell-1}\frac{\sigma^{2n}}{\ell^{2n}(\ell+1)^{2n}} (\ell+1)^4\int_{N_0} |\hat{{L}}^n(r\Psi_{\ell})|^2\,d\omega d\rho_c \\
&+\frac{\ell^{2\ell}(\ell+1)^{2\ell}}{\ell!^{2} (\ell+1)!^2} \sum_{n=\ell}^{\infty} \frac{\sigma^{2n}}{n!^2(n+1)!^2}  (n+1)^4 \int_{N_0} |\hat{{L}}^n(r\Psi_{\ell})|^2\,d\omega d\rho_c \Bigg].
\end{align*}
We denote with $\mathbf{H}_{\sigma,\rho_0}$ the completion of $\{ (f,g)\in C^{\infty}(\Sigma_0 ;\C)\times C^{\infty}(S_0;\C)\,|\, ||(f,g)||_{\mathbf{H}_{\sigma,\rho_0}}<\infty\}$ with respect to the norm $||(\cdot,\cdot)||_{\mathbf{H}_{\sigma,\rho_0}}$. Furthermore, $\mathbf{H}_{\sigma,\rho_0}$  is a Hilbert space with respect to the natural choice of inner product.
\end{definition}

In the proposition below, we state useful relations between the Hilbert spaces $H_{\sigma,i, \rho_0}$, with $i=0,1$, and $\mathbf{H}_{\sigma,\rho_0}$.

\begin{proposition}
Let $\ell\in \N_0$.
\begin{itemize}
\item[(i)]
Let $\Psi_{\ell m}\in H_{\sigma,1,\rho_0}$ and $\Psi_{\ell m}'\in L^2[ R_0^+,R_0^c]$ for all $|m|\leq \ell$. Denote $\Psi_{\ell}(r,\theta,\varphi)=\sum_{m=-\ell}^{\ell}\Psi_{\ell m}(r)Y_{\ell m}(\theta,\varphi)$ and $\Psi'_{\ell}(r,\theta,\varphi)=\sum_{m=-\ell}^{\ell}\Psi'_{\ell m}(r)Y_{\ell m}(\theta,\varphi)$. Then 
\begin{equation*}
(\Psi_{\ell},\Psi_{\ell}')\in \mathbf{H}_{\sigma,\rho_0}\cap V_{\ell}.
\end{equation*}
\item[(ii)]
Let $(\Psi_{\ell},\Psi_{\ell}')\in \mathbf{H}_{\sigma,\rho_0}\cap V_{\ell}$. Then we can write
\begin{equation*}
(\Psi_{\ell},\Psi_{\ell}')(r,\theta,\varphi)=\sum_{m=-\ell}^{m=\ell} (\Psi_{\ell m}(r) Y_{\ell m}(\theta,\varphi),\Psi'_{\ell m}(r) Y_{\ell m}(\theta,\varphi))
\end{equation*}
and we have that
\begin{equation*}
\Psi_{\ell m}\in H_{\sigma,1,\rho_0}\quad\textnormal{and}\quad  \Psi_{\ell m}'\in L^2[R_0^+,R_0^c].
\end{equation*}
\end{itemize}
\end{proposition}
\begin{proof}
The proof is a straightforward application of the definitions of $\mathbf{H}_{\sigma,\rho_0}$, $H_{\sigma,\rho_0}$, and $H_{\sigma,1,\rho_0}$. We use that the factor $(n+1)^4$ appearing in the infinite sums in the definition of $||\cdot||_{\mathbf{H}_{\sigma,\rho_0}}$ implies control over similar infinite sums with one additional $\rho_+$ or $\rho_c$ derivative, but no factor $(n+1)^4$. This allows us to conclude that not only $\Psi_{\ell m}\in H_{\sigma,\rho_0}$, but in fact $\Psi_{\ell m}\in H_{\sigma,1,\rho_0}$.
\end{proof}

We will moreover need the following Hilbert spaces in the case where $\kappa_+$ and $\kappa_c$ are strictly positive. Let $k\in \N$ and define
\begin{align*}
\widetilde{H}^k:=&\:L^2((R^+_0,R^c_0)\times \s^2)\cap H^k((r_+,R^+_0) \times \s^2) \cap H^k((R^c_0,r_c)\times \s^2),\\
\widetilde{H}^{k}_2:=&\:H^2((R^+_0,R^c_0)\times \s^2)\cap H^{k}((r_+,R^+_0) \times \s^2) \cap H^{k}((R^c_0,r_c)\times \s^2).
\end{align*}

\section{Precise statements of the main theorems}
\label{sec:mainthms}
In this section, we state the main theorems of the paper. We make use of the notation and concepts introduced in Sections \ref{sec:geom} and \ref{sec:enspaces}.
\subsection{Construction of regularity quasinormal modes}
We first provide a construction of regularity quasinormal modes in the extremal Reissner--Nordstr\"om setting (as referred to in Section \ref{sec:roughmainthms}) and describe their relation with the scattering resonances (cf. Approach 2 of Section \ref{sec:tradapproach}).
\begin{theorem}
\label{thm:discreigen}
Assume that either $\kappa_+,\kappa_c> 0$ or $\kappa_+=\kappa_c=0$. Let $\sigma \in {\R_{>0}}$ and let $\rho_0>0$ be suitably small. Then the family of solution operators
\begin{align*}
\mathcal{S}(\tau): \mathbf{H}_{\sigma,\rho_0}&\to\mathbf{H}_{\sigma,\rho_0},\\
  (\psi|_{\Sigma_{0}}, T\psi|_{S_{0}})\:& \mapsto  (\psi|_{\Sigma_{\tau}}, T\psi|_{S_{\tau}}),
\end{align*}
with $\psi$ the solution to \eqref{eq:waveequation} corresponding to initial data $(\psi|_{\Sigma_{0}}, T\psi|_{S_{0}})$, define a $C_0$-semigroup and the corresponding infinitesimal generator $\mathcal{A}: \mathbf{H}_{\sigma,\rho_0} \supseteq \mathcal{D}_{\sigma}(\mathcal{A})\to \mathbf{H}_{\sigma,\rho_0}$ satisfies the following properties: let 
\begin{equation*}
\Omega_{\sigma}=\left\{ z\in \C\:|\: \re z<0,\: |z|<\sigma,\: 3(\im z)^2-5(\re z)^2>{\sigma}^2\right\}\cup \{ z\in \C\:|\: \re z \geq 0, z\neq 0\}\subset \C,
\end{equation*}
then
\begin{enumerate}[\rm (i) ]
\item
\begin{equation*}
\textnormal{Spect}_{\rm point} (\mathcal{A})\cap \Omega_{\sigma}=\Lambda^{\sigma}_{QNF}
\end{equation*}
is independent of $\rho_0$, with $\Lambda^{\sigma}_{QNF}\subset \Omega_{\sigma}\cap \{\re z <0\}$ a discrete set of eigenvalues which moreover have finite multiplicity.
\item
The set
\begin{equation*}
\Lambda_{QNF}:=\bigcup_{\sigma \in {\R_{>0}}}\Lambda^{\sigma}_{QNF}\subset \left\{-\re z<\frac{1}{2}|z| \right\}=\left\{|\textnormal{arg}(z)|<\frac{2}{3}\pi\right\},
\end{equation*}
is a discrete subset of $\left\{|\textnormal{arg}(z)|<\frac{2}{3}\pi\right\}$ (i.e.\ with accumulation points only possible on the boundary of $\left\{|\textnormal{arg}(z)|<\frac{2}{3}\pi\right\}$ in $\C$).
\end{enumerate}
\end{theorem}
Theorem \ref{thm:discreigen} follows from Corollary \ref{cor:mainresult}. See also Figure \ref{fig:res} for a pictorial representation of $\Omega_{\sigma}$.

\begin{definition}
\label{def:regqnf}
We refer to the elements of $\Lambda_{QNF}$ as \textbf{regularity quasinormal frequencies} and the corresponding eigenvectors as \textbf{regularity quasinormal modes}. 
\end{definition}

\begin{theorem}
\label{thm:polesresolvent}
Assume that either $\kappa_+,\kappa_c> 0$ or $\kappa_+=\kappa_c=0$, and consider the operator
\begin{equation*}
\mathcal{L}_s(\hat{\psi})=\frac{d^2}{dr_*^2}(r\hat{\psi})-s^2 r\hat{\psi}+Dr^{-2}\mathring{\slashed{\Delta}}(r\hat{\psi})-r^{-1}DD' \cdot r\hat{\psi}-2\mathfrak{l}^{-2}D\cdot r\hat{\psi},
\end{equation*}
with $\mathring{\slashed{\Delta}}(\cdot)= \frac{1}{\sin \theta} \partial_{\theta}(\sin \theta \partial_{\theta}(\cdot))+\frac{1}{\sin^2\theta} \partial_{\varphi}^2(\cdot)$ the standard Laplacian on the unit round sphere.
\begin{enumerate}[\rm (i) ]
\item Then for all $\re s>0$ and any two smooth, compactly supported cut-off functions $\chi,\chi': (r_+,r_c)_r\to \R$ (with $r_c\leq \infty$), the inverse map
\begin{equation*}
\chi' \circ R(s)\circ \chi:=\chi' \circ \mathcal{L}_s^{-1}\circ \chi: L^2(\{t=0\})\to H^2(\{t=0\})
\end{equation*}
defines a holomorphic family of bounded linear operators that admits a meromorphic continuation to $\left\{|\textnormal{arg}(z)|<\frac{2}{3}\pi\right\} \subset \C$ with the set of poles $\Lambda_{QNF}^{\rm res}$ satisfying
\begin{equation*}
\Lambda_{QNF}^{\rm res} \subseteq \Lambda_{QNF}.
\end{equation*}
\item For all $\sigma \in {\R_{>0}}$ and for $\rho_0>0$ suitably small, there exist $\mathcal{S}(\tau)$-invariant subspaces $\mathbf{H}_{\sigma,\rho_0}^{\rm res}< \mathbf{H}_{\sigma,\rho_0}$ such that the infinitesimal generators $\mathcal{A}^{\rm res}$ corresponding to the restricted operators $\mathcal{S}(\tau): \mathbf{H}_{\sigma,\rho_0}^{\rm res}\to\mathbf{H}_{\sigma,\rho_0}^{\rm res}$ satisfy:
\begin{equation*}
\bigcup_{\sigma \in {\R_{>0}}} \textnormal{Spect}_{\rm point} (\mathcal{A}^{\rm res})\cap \Omega_{\sigma}=\Lambda_{QNF}^{\rm res}.
\end{equation*}
\end{enumerate}
\end{theorem}
Theorem \ref{thm:polesresolvent} follows from Proposition \ref{prop:merocont} and Proposition \ref{prop:relwithtradition}.

\begin{definition}
\label{def:outqnf}
We refer to the regularity quasinormal modes that are elements of the subspace $\mathbf{H}_{\sigma,\rho_0}^{\rm res}$ as \textbf{resonant states} and the corresponding eigenvalues as \textbf{scattering resonances}. 
\end{definition}

\subsection{Further properties regarding the distribution of quasinormal frequencies}
We state here additional results regarding the distribution of regularity quasinormal modes and their relation to regularity quasinormal modes in the setting of positive $\Lambda$.

\begin{theorem}
\label{thm:smallslargel}
The following additional properties hold for $\Lambda_{QNF}$:
\begin{enumerate}[\rm (i) ]
\item
Let $\ell\in \N$ and denote with $\Lambda_{QNF}^{\ell}$ the subset of $\Lambda_{QNF}$ corresponding to eigenvectors in $V_{\ell}$. Then
\begin{equation*}
\Lambda_{QNF}^{\ell}=\textnormal{Spect}(\mathcal{A}|_{V_{\ell}}),
\end{equation*}
i.e. $\mathcal{A}|_{V_{\ell}}$ has a pure point spectrum. Furthermore, for all $L>0$, there exists $\delta_L>0$ such that
\begin{equation*}
\bigcup_{\ell=0}^L\Lambda_{QNF}^{\ell}\cap \{z\in \C \,|\, |z|\leq \delta_L\}=\emptyset.
\end{equation*}
\item
Let $K \subset  \left\{|\textnormal{arg}(z)|<\frac{2}{3}\pi\right\}$ be a compact set. Then there exists $L=L(K)\in \N$ such that
\begin{equation*}
\ker (\mathcal{A}|_{\mathbf{H}_{\sigma,\rho_0}}-s)=\sum_{\ell=0}^L \pi_{\ell}\left(\ker (\mathcal{A}|_{\mathbf{H}_{\sigma,\rho_0}}-s)\right)
\end{equation*}
for all $\sigma \in {\R_{>0}}$, for $\rho_0>0$ suitably small and $s\in \Omega_{\sigma}\cap K$.
\item
The map $s\mapsto (\mathcal{A}|_{\mathbf{H}_{\sigma,\rho_0}\cap V_{\ell}}-s)^{-1}$ is meromorphic on $\Omega_{\sigma}$ with the poles coinciding with the elements of $\Lambda^{\ell}_{QNF}\cap \Omega_{\sigma}$.
\end{enumerate}
\end{theorem}
Theorem \ref{thm:smallslargel} is included in a combination of the results of Proposition \ref{prop:resolventzerokappa}, Proposition \ref{prop:fredholm} and Proposition \ref{prop:relationAL}.

\begin{remark}
Theorem \ref{thm:smallslargel} illustrates that when restricting to fixed angular frequencies, there are \underline{no} regularity quasinormal frequencies near the zero frequency (when restricted to $\left\{|\textnormal{arg}(z)|<\frac{2}{3}\pi\right\}$). Furthermore, regularity quasinormal modes are supported on a bounded set of angular frequencies.
\end{remark}

\begin{theorem}
\label{thm:convqnms}
Let $s_*\in \Lambda_{QNF}$ and let $\hat{\psi}_{s_*}$ a corresponding regularity quasinormal mode. Then for any sequence of sufficiently small positive cosmological constants $\Lambda_n$ approaching zero and  charges $e_n$ approaching $M$, there exists a corresponding sequence of regularity quasinormal frequencies $s_n$ for the equation \eqref{eq:waveequation} and a corresponding sequence of regularity quasinormal modes $\hat{\psi}_{n}$, such that for all suitably small $\rho_0>0$:
\begin{align*}
s_n\to&\: s_*,\\
||(\hat{\psi}_{s_*})_{\ell m}-(\hat{\psi}_{n})_{\ell m}||_{\sigma,\rho_0}\to&\: 0\quad \textnormal{for all $\ell\in \N$ and $m\in \Z$, with $|m|\leq \ell$}.
\end{align*}
\end{theorem}
Theorem \ref{thm:convqnms} is a reformulation of Proposition \ref{prop:conveigenf}.

\begin{remark}
Theorem \ref{thm:convqnms} shows that for each extremal Reissner--Nordstr\"om regularity quasinormal frequency, there is a corresponding converging sequence of sub-extremal  Reissner--Nordstr\"om--de Sitter regularity quasinormal frequencies.
\end{remark}

\section{Structure of proofs and main ideas and techniques}
\label{sec:mainideas}
In this section we sketch the logic and structure of the proofs of the theorems stated in Section \ref{sec:mainthms} and we highlight the main new ideas and techniques that are introduced in this paper. 

\subsection{Infinitesimal generators of time translations and resolvent operators}
The theorems in Section \ref{sec:mainthms} are concerned with the operator $\mathcal{A}$, which is a densely defined, closed, unbounded operator that generates the time translation semigroup corresponding to a mixed spacelike-null foliation of the spacetime (see Section \ref{sec:foliations}). Rather than inferring properties about the spectrum of $\mathcal{A}$ by proving estimates for $\mathcal{A}$ directly, we first consider the restrictions $\mathcal{A}_{\ell}=\mathcal{A}|_{V_{\ell}}$ to fixed spherical harmonic modes and use that the invertibility of $\mathcal{A}_{\ell}-s$ is equivalent to the existence of the operator $\hat{L}_{s,\ell}^{-1}$, which is the resolvent operator corresponding to fixed spherical harmonic modes on a mixed spacelike-null foliation. The precise definitions of $\mathcal{A}_{\ell}$ and $\hat{L}_{s,\ell}$ and their relation are described in detail in Section \ref{sec:propgevreyphys}. 

The equivalence of the invertibility of $\mathcal{A}_{\ell}-s$ with the existence of $\hat{L}_{s,\ell}^{-1}$ can also be easily seen when one considers \eqref{eq:waveequation} with respect to the Minkowski metric, i.e.\ the standard wave equation, and investigates the existence of the standard resolvent operator $(\Delta_{\R^3}-s^2)^{-1}$. In that case, the infinitesimal generator of time translation $\mathcal{A}$ corresponding to a foliation by hypersurfaces of constant $t$ (the standard time coordinate) is given by:
\begin{equation*}
(\mathcal{A}-s)\begin{pmatrix}
\Psi\\
\Psi'
\end{pmatrix}
=\begin{pmatrix}
0 & 1\\
1 & -s
\end{pmatrix}
\begin{pmatrix}
-\Delta_{\R^3}+s^2& 0\\
0 & 1
\end{pmatrix}
\begin{pmatrix}
-1 & 0\\
-s & 1
\end{pmatrix}
\begin{pmatrix}
{\Psi}\\
{\Psi}'
\end{pmatrix}.
\end{equation*}
Since the first and third matrix on the right-hand side above are clearly invertible, it follows immediately that invertibility of $\mathcal{A}-s$ is related to existence of $(\Delta_{\R^3}-s^2)^{-1}$, which is the resolvent operator with respect to a foliation by $t$-level sets. This kind of relation between $\mathcal{A}_{\ell}$ and $\hat{L}_{s,\ell}$ also plays a central role in the proof of Theorem \ref{thm:ads}. 

Note that \emph{in passing from $\mathcal{A}_{\ell}-s$ to $\hat{L}_{s,\ell}$, we lose the simple dependence of the operator on $s$ but we gain the ability to apply Fredholm theory}, see Section \ref{intro:fred}.

\subsection{Fredholm theory for resolvent operators}
\label{intro:fred}
We will show that the operator
\begin{equation*}
\hat{L}^{-1}_{s,\ell}:H_{\sigma, \rho_0}\to H_{\sigma,\rho_0}
\end{equation*}
is compact and holomorphic in $s$, provided $s$ lies in suitable sector of the complex plane (depending on the choice of $\sigma$, but independent of $\rho_0$), for sufficiently large $\ell$ (or for bounded $\ell$ and sufficiently small $|s|$), which allows us to apply the Analytic Fredholm Theorem to infer meromorphicity of $\hat{L}^{-1}_{s,\ell}$ for \underline{all} $\ell$.\footnote{One should think of the $\ell$-dependent terms in $\hat{L}^{-1}_{s,\ell}$ as zeroth order terms with a favourable sign, cf. the existence theory for solutions to uniformly elliptic PDE on compact domains, see for example Theorem 3 of \S 6.2 in \cite{evans}.}

 As a corollary, we conclude that $\mathcal{A}_{\ell}$ has a pure point spectrum in a suitable sector including the imaginary axis and moreover, if we restrict to eigenvalues in any compact subset, then all corresponding eigenfunctions of $\mathcal{A}$ must be supported on a bounded set of angular frequencies, so the spectra of $\mathcal{A}_{\ell}$ in fact determine fully the point spectrum of $\mathcal{A}$.

\subsection{Key estimates}
\label{intro:keyest}
As described in the previous paragraphs, we reduce the problem of characterizing the point spectrum of $\mathcal{A}$ to establishing compactness of the resolvent operator $\hat{L}^{-1}_{s,\ell}$ for suitably large $\ell$. We establish compactness by showing that in fact
\begin{equation*}
\hat{L}^{-1}_{s,\ell}(H_{\sigma, \rho_0})\subseteq H_{\sigma,\rho_0,2}
\end{equation*}
and using that $H_{\sigma,\rho_0,2}$ can be compactly embedded in $H_{\sigma,\rho_0}$. The compactness of this embedding follows from an analogue of the Rellich--Kondrachov theorem to the setting of $L^2$-based Gevrey norms.

The arguments discussed so far can be mostly considered ``soft'' (as they are primarily variations of well-stablished results in functional analysis), and the ``hard'' part of the proofs consists of proving the estimate:
\begin{equation}
\label{eq:introkeyest}
||\hat{\psi}||_{\sigma,\rho_0,2}\leq C_{\ell,s}||\hat{L}_{s,\ell}(\hat{\psi})||_{\sigma,\rho_0}
\end{equation}
with a constant $C_{\ell,s}>0$. The estimate \eqref{eq:introkeyest} is central to establishing existence and compactness of $\hat{L}^{-1}_{s,\ell}$. In fact, in order to prove Theorem \ref{thm:polesresolvent}, we need a refined version of \eqref{eq:introkeyest}, where we keep more precise track of the $\ell$-dependence in the constant, but we will ignore that point in the discussion in the present section.

We give an outline below of the main steps involved in obtaining \eqref{eq:introkeyest}. We will carry out the discussion for the shifted operator $L_{s,\kappa,\ell}:=L_{s,\kappa}-\ell(\ell +1)$ where $L_{s,\kappa}$ is the operator appearing in the toy model equation \eqref{eq:introtoy1}. The logic of the proof in the context of the toy model is very similar to the real problem. The key difference that appears when considering the true resolvent operator are discussed in Section \ref{intro:consvlaw} below.

\begin{enumerate}
\item The resemblance of the operator $\hat{L}_{s,\ell}$ to its toy model version $L_{s,\kappa=0, \ell}$ applies only in regions near the event horizon and infinity (where $r$ is large or $r$ is close to the horizon radius $M$). These regions are both modelled by the interval $ [0,1]$ in the toy model problem. We consider $\frac{d^n}{dx^n}(L_{s,\kappa,\ell}(u))$ in $[0,1]$ with $\kappa\geq 0$ and use \eqref{eq:introtoy1} to write:
\begin{equation}
\label{eq:keytoyeq}
\frac{d}{dx}\left((\kappa x+x^2)u^{(n+1)}\right)+(2nx+\kappa n)u^{(n+1)}+su^{(n+1)}=(\ell(\ell+1)-n(n+1))u^{(n)}+ \frac{d^n}{dx^n}(L_{s,\kappa,\ell}(u)),
\end{equation}
with $u^{(k)}=\frac{d^ku}{dx^k}$.

We take the square norm on both sides of \eqref{eq:keytoyeq} and integrate in $x$. In order to derive the desired estimate, we need to absorb the term $(\ell(\ell+1)-n(n+1))^2|u^{(n)}|^2$ that appears with a bad sign in the estimate. Observe that this term vanishes when $n=\ell$. Furthermore, when $n>\ell$, we can absorb it into the analogous estimate with $n$ replaced by $n-1$, provided we multiply with appropriate $n$-dependent weights. 

This observation motivates the following integration and summation over $n$:
\begin{equation*}
\sum_{n=\ell}^{N_{\infty}} \frac{\tilde{\sigma}^{2n}|s|^{2n}}{n!^2(n+1)!^2}\int_0^1\left[\cdot\right]\,dx.
\end{equation*}
Indeed, such a summation allows us to absorb non-coercive terms, which arise by taking the square norm of both sides of \eqref{eq:keytoyeq}, into the estimates of either order $n-1$ or $n+1$. The above summation moreover elucidates the relevance of $L^2$-based Gevrey regularity. In order for this procedure to work, it is in particular necessary that $\re s>-\frac{1}{2}|s|$ and we set $\sigma=\tilde{\sigma}|s|$.

A caveat of the above summation procedure is that at the top order $n=N_{\infty}$ we cannot absorb into order $n=N_{\infty}+1$. If we had assumed \emph{a priori} the finiteness of the Gevrey norm for $u$ that we want to estimate, then these top order terms would vanish in the limit $N_{\infty}\to \infty$. We do not make such an a priori assumption, but instead, we restrict to the case of positive $\kappa$ and make use of the presence of a \emph{red-shift effect} in the problem when $\kappa>0$ (using that $\kappa$ plays the role of a surface gravity), which is moreover \emph{stronger} for larger values of $n$ (this is known as the ``enhanced red-shift effect'' and it also plays an important role in the proof of Theorem \ref{thm:ads}), to absorb the top order term, assuming $N_{\infty}>N_{\kappa}$ with $N_{\kappa}\to \infty$ as $\kappa\downarrow 0$. After closing the estimate for $\kappa>0$, we can then take the limit $N_{\infty}\to \infty$ on both sides of the estimate to arrive at an $L^2$-based Gevrey estimate for $\kappa>0$ that is uniform in $\kappa$, \emph{without} control of the lower order derivatives with $0\leq n\leq \ell-1$.

\item We control derivatives with $0\leq n\leq \ell-1$ by modifying the estimates above with suitable exponential weights so that we can control lower-order terms by higher-order terms and boundary terms at $x=1$ via ``Carleman-type'' estimates. It is important in this step that we can close the Gevrey estimates with $n\geq \ell$ \underline{\emph{independently}} of the $n\leq \ell-1$ estimates (modulo boundary terms at $x=1$). We can control the boundary terms at $x=1$ in terms of the lowest order derivatives: $|u|^2(1)$ (which vanishes in the toy problem, by assumption) and $|\frac{du}{dx}|^2(1)$, together with derivatives of $L_{s,\kappa,\ell}(u)$, by applying \eqref{eq:keytoyeq} repeatedly. As a result, the constant multiplying $|\frac{du}{dx}|^2(1)$ grows exponentially in $\ell$.

\item In order to pass from $L^2$-based Gevrey estimates to the toy model analogue of \eqref{eq:introkeyest}, we need to estimate the $|\frac{du}{dx}|^2(1)$ boundary term that is multiplied by a constant that grows exponentially in $\ell$. In the real problem, this amounts to estimating boundary terms along a large fixed constant $r$ hypersurface and a constant $r$ hypersurface with $r$ close to the horizon area radius $M$, which arise from restricting the Gevrey estimates to the near-infinity and near-horizon regions, respectively. 

We make use of a new type of degenerate elliptic estimate, with a degeneracy that depends on $\ell$, to achieve the desired bounds. In the context of the toy model problem, we can write
\begin{equation*}
x^k\left|\frac{d}{dx}\left((\kappa x+x^2)\frac{du}{dx}\right)-\ell(\ell+1)u\right|^2=x^k\left|s\frac{du}{dx}+L_{s,\kappa,\ell}(u)\right|^2
\end{equation*}
for $k\geq 0$. We arrive at a suitable degenerate elliptic estimate by integrating the above equation over an interval $[\eta,1]$, where $\eta>0$ is suitably small, integrating by parts, and averaging out the $x=\eta$ boundary terms with a suitable cut-off function.  We can then absorb the term proportional to $x^k|s|^2|\frac{du}{dx}|^2$ that appears on the right-hand side into a term proportional to $x^{k+2}\ell(\ell+1)|\frac{du}{dx}|^2$ that appears with a good sign on the left-hand side, provided $\ell$ is chosen suitably large depending on $\eta$ and $|s|$.\footnote{Alternatively, for any $\ell$, we can take $|s|$ suitably small, depending on $\ell$, in order to close the Gevrey estimates for small $|s|$. This approach is taken in Section \ref{sec:lowfreqgev}.}

Furthermore, the degenerate elliptic estimate is only valid  if $0\leq k\lesssim \ell$. In other words, for larger $\ell$, we can introduce a \emph{stronger degeneracy} in the elliptic estimate at $x=0$. As a result, we can multiply the terms arising from the cut-off near $x=\eta$ with a factor that \emph{decays exponentially in $\ell$}.

Let us moreover note that physical space versions of the above degenerate elliptic estimates also play a key role in deriving the leading-order behaviour (polynomial tails) of $\psi_{\ell}$; see the upcoming work \cite{aagprice}.

\item We couple the Gevrey estimates with the above degenerate elliptic estimates, absorbing the coupling terms by using that the exponentially decaying factor in $\ell$ in the elliptic estimates \emph{cancels out} the exponentially growing factor in the Gevrey estimates that appears in front of the boundary term. We arrive at the toy model analogue of \eqref{eq:introkeyest} with $\kappa>0$, but $C_{\ell,s}$ still independent of $\kappa$.

\item Using the \emph{uniformity} in $\kappa$ of the $\kappa>0$ estimates, we can construct $L^{-1}_{s,\kappa=0,\ell}$ and establish its compactness as a limit of $\kappa>0$ operators as $\kappa \downarrow 0$.
\end{enumerate}

\subsection{Conservation laws for spherical harmonic modes}
\label{intro:consvlaw}
A crucial step in the above sketch of the $L^2$-based Gevrey estimates above is the vanishing of the term  $(\ell(\ell+1)-n(n+1))^2|u^{(n)}|^2$ when $n=\ell$. This is in fact a manifestation of a conservation law along null infinity in Minkowski that is present in physical space. Indeed, 
\begin{equation*}
\partial_u \left((r^2\partial_v)^{\ell+1}(r\psi_{\ell})\right)=0
\end{equation*}
along $\mathcal{I}^+$ for solutions $\psi_{\ell}$ to \eqref{eq:waveequation} on Minkowski arising from suitable Cauchy initial data that are supported on spherical harmonic modes with angular frequency $\ell$. 

In Reissner--Nordstr\"om spacetimes, however, we generically have that
\begin{equation*}
\partial_u \left((r^2\partial_v)^{\ell+1}(r\psi_{\ell})\right)\neq 0
\end{equation*}
along $\mathcal{I}^+$. Nevertheless, there are still conserved quantities present which are \emph{linear combinations} of $(r^2\partial_v)^{n}(r\psi_{\ell})$ with $n\leq \ell$. Hence, rather than considering derivatives of the form $(r^2\partial_v)^{n}(r\psi_{\ell})$, we need to modify the $n$-th order quantities in order to see the conservation law, so that we can close the Gevrey estimates starting at $n=\ell$.

In the extremal Reissner--Nordstr\"om case, the required modification can be done explicitly by an appropriate change of coordinates, see Section \ref{sec:hoconsvlaws}. Furthermore, there are similar higher-order conserved quantities present along the event horizon. The price we pay by carrying out the modification is that we have to deal with additional terms in the modified higher-order equations that have to be absorbed appropriately. Nevertheless, we arrive at $L^2$-based Gevrey estimates that are similar to those that appear in the toy model discussed above and the steps described in Section \ref{intro:keyest} still hold.

\subsection{Eigenvalues for small $\Lambda>0$ and $M-e$}
An advantage of deriving estimates that are uniform in the surface gravities $\kappa_c,\kappa_+$, and are therefore valid not only for extremal Reissner--Nordstr\"om, but also for near-extremal Reissner--Nordstr\"om--de Sitter spacetimes with small $\Lambda>0$, where $\kappa_+,\kappa_c>0$, is that we obtain as a \emph{corollary} the convergence of regularity quasinormal frequencies (eigenvalues of $\mathcal{A}$) for the conformally invariant Klein--Gordon equation \eqref{eq:waveequation} as $\kappa_+ \downarrow 0$ and $\kappa_c\downarrow 0$, which is the content of Theorem \ref{thm:convqnms}.
 
 \section{Main equations}
 \label{sec:eq}
 In this section we write down the equation \eqref{eq:waveequation} in coordinates adapted to either the future event horizon $\mathcal{H}^+$ or the future cosmological horizon $\mathcal{C}^+$, allowing us to easily take the limit $\Lambda \downarrow 0$ or $M \downarrow  |e|$.
  
In $(\tau,r, \theta,\varphi)$ coordinates, \eqref{eq:waveequation} can be expressed as follows:
\begin{equation}
\label{eq:waveeqhyperboloidal}
0=2(1-h_{r_+}D)\partial_{r}T\psi+r^{-2}\partial_{r}(Dr^2\partial_{r}\psi)-h_{r_+}(2-h_{r_+}D)T^2\psi+(2r^{-1}-r^{-2}(Dr^2h_{r_+})')T\psi+r^{-2} \mathring{\slashed{\Delta}}\psi-2\mathfrak{l}^{-2}\psi.
\end{equation}

 \subsection{Equation for the radiation field near the event horizon}
 \label{sec:waveeqevent}
 
It can easily be shown that the \emph{Friedlander radiation field} $\phi:=r\cdot \psi$ corresponding to a solution $\psi$ of \eqref{eq:waveequation} satisfies the following equation in $(v,r,\theta,\varphi)$ coordinates:
 \begin{equation}
 \label{eq:eqradfieldhor1}
 2\partial_v\partial_r\phi+\partial_r(D\partial_r\phi)-(r^{-1}D'+2\mathfrak{l}^{-2})\phi+r^{-2}\mathring{\slashed{\Delta}} \phi=0,
 \end{equation}
 where $\mathring{\slashed{\Delta}}$ is the Laplacian with respect to the unit round sphere. Let us denote for convenience in this section $\rho=\rho_+$. In $(v,\rho,\theta,\varphi)$ coordinates, \eqref{eq:eqradfieldhor1} reduces to
  \begin{equation*}
 2\partial_v\partial_{\rho}\phi+\partial_{\rho}(Dr^{-2}\partial_{\rho}\phi)-(rD'+2\mathfrak{l}^{-2}r^2)\phi+\mathring{\slashed{\Delta}} \phi=0,
 \end{equation*}
and hence,
   \begin{equation*}
   2\partial_v\partial_{\rho}\phi+\partial_{\rho}(Dr^{-2}\partial_{\rho}\phi)-(2Mr^{-1}-2e^2r^{-2})\phi+\mathring{\slashed{\Delta}} \phi=0,
    \end{equation*}
 where we used that
\begin{equation*}
rD'=2Mr^{-1}-2e^2r^{-2}-2\mathfrak{l}^{-2}r^2.
\end{equation*}

We can further write
\begin{equation*}
\begin{split}
2Mr^{-1}-2e^2r^{-2}=&\:2Mr_+^{-1}-2e^2r_+^{-2}+(4e^2r_+^{-2}-2Mr_+^{-1})r_+\rho-2e^2\rho^2\\
=&\:2r_+\kappa_++2l^{-2}r_+^2+(4e^2r_+^{-2}-2Mr_+^{-1})r_+\rho-2e^2\rho^2\\
=&\:B_{r_+}^{(0)}+B_{r_+}^{(1)}\rho+B_{r_+}^{(2)}\rho^2,
\end{split}
\end{equation*}
with coefficients
\begin{align*}
B_{r_+}^{(0)}=&\:2\kappa_+r_++2\mathfrak{l}^{-2}r_+^2,\\
B_{r_+}^{(1)}=&\:(4e^2r_+^{-2}-2Mr_+^{-1})r_+,\\
B_{r_+}^{(2)}=&\:-2e^2.
\end{align*}
Recall from Section \ref{sec:confcoord} that we can moreover write
   \begin{equation*}
 \begin{split}
 Dr^{-2}=&\: 2\kappa_+ \rho +(1-A_{r_+}^{(2)})\rho^2+A_{r_+}^{(3)}\rho^3+A_{r_+}^{(4)}\rho^4,
 \end{split}
 \end{equation*}
 where
 \begin{align*}
  A_{r_+}^{(2)}=&\:3\mathfrak{l}^{-2}\left[r_+^{-1}(r_+-r_-)(r_c^2+r_-r_c)+r_+(r_-+r_c)\right],\\
 A_{r_+}^{(3)}=&\:-2r_-+\mathfrak{l}^{-2}\left[(r_+-r_-)r_c^2+(r_+^2-r_-^2)r_c+2r_-^3 +3r_+^2r_-+3r_-^2r_+\right],\\
 A_{r_+}^{(4)}=&\:-e^2.
 \end{align*}  
    
Combining the above equations, we arrive at the following equation for the radiation field:
\begin{equation}
 \begin{split}
 \label{eq:eqradfieldhor2}
 0=&\:2\partial_v\partial_{\rho}\phi+\partial_{\rho}\left(\left((1-A_{r_+}^{(2)})\rho^2+2\kappa_+ \rho+A_{r_+}^{(3)}\rho^3+A_{r_+}^{(4)} \rho^4\right)\partial_{\rho}\phi\right)\\
 &-(B_{r_+}^{(0)}+B_{r_+}^{(1)}\rho +B_{r_+}^{(2)}\rho^2)\phi+\mathring{\slashed{\Delta}} \phi.
  \end{split}
\end{equation}

 \subsection{Equation for the radiation field near the cosmological horizon/future null infinity}
 \label{sec:waveeqevent}
 In $(u,r,\theta,\varphi)$ coordinates, the radiation field $\phi=r\cdot \psi$ satisfies the equation
 \begin{equation}
 \label{eq:eqradfieldcos1}
 2\partial_u\partial_r\phi-\partial_r(D\partial_r\phi)+(r^{-1}D'+2\mathfrak{l}^{-2})\phi-r^{-2}\mathring{\slashed{\Delta}} \phi=0.
 \end{equation}
Let us denote for convenience in this section $\rho=\rho_c$. In $(u,\rho,\theta,\varphi)$ coordinates, \eqref{eq:eqradfieldcos1} reduces to
   \begin{equation*}
   2\partial_u\partial_{\rho}\phi+\partial_{\rho}(Dr^{-2}\partial_{\rho}\phi)-(2Mr^{-1}-2e^2r^{-2})\phi+\mathring{\slashed{\Delta}} \phi=0.
    \end{equation*}

    We can further write
\begin{equation*}
\begin{split}
2Mr^{-1}-2e^2r^{-2}=&\:2Mr_c^{-1}-2e^2r_c^{-2}-(4e^2r_c^{-2}-2Mr_c^{-1})r_c\rho-2e^2\rho^2\\
=&\:-2r_c\kappa_c+2\mathfrak{l}^{-2}r_c^2-(4e^2r_c^{-2}-2Mr_c^{-1})r_c\rho-2e^2\rho^2.
\end{split}
\end{equation*}
Combining the above estimates, we arrive at
\begin{equation}
 \begin{split}
 \label{eq:eqradfieldcos2}
 0=&\:2\partial_u\partial_{\rho}\phi+\partial_{\rho}\left(\left((1-A_{r_c}^{(2)})\rho^2+2\kappa_c \rho+A_{r_c}^{(3)}\rho^3+A_{r_c}^{(4)} \rho^4\right)\partial_{\rho}\phi\right)\\
 &-(B_{r_c}^{(0)}+B_{r_c}^{(1)}\rho +B_{r_c}^{(2)}\rho^2)\phi+\mathring{\slashed{\Delta}} \phi,
  \end{split}
\end{equation}
with coefficients
\begin{align*}
B_{r_c}^{(0)}=&\:2\mathfrak{l}^{-2}r_c^2-2\kappa_cr_c,\\
B_{r_c}^{(1)}=&\:-(4e^2r_c^{-2}-2Mr_c^{-1})r_c,\\
B_{r_c}^{(2)}=&\:-2e^2,
\end{align*}
where we used the following expression, derived in Section \ref{sec:confcoord}:
 \begin{equation*}
 \begin{split}
 Dr^{-2}=&\: 2\kappa_c\rho +(1-A_{r_c}^{(2)})\rho^2+A_{r_c}^{(3)}\rho^3+A_{r_c}^{(4)}\rho^4,
 \end{split}
 \end{equation*}
 where
 \begin{align*}
 A_{r_c}^{(2)}=&\:3\mathfrak{l}^{-2}\left[r_c(r_++r_-)+r_+^2+r_-^2-r_c^{-1}r_+r_-(r_++r_-)\right],\\
 A_{r_c}^{(3)}=&\:-(r_++r_-)+\mathfrak{l}^{-2}(4r_cr_+r_-+r_+^3+r_-^3+5r_+^2r_-+5r_-^2r_+),\\
 A_{r_c}^{(4)}=&\:e^2.
 \end{align*}
 
 \subsection{Higher-order equations}
In this section we consider the higher-order derivatives $\partial_{\rho_c}^n\phi$ and $\partial_{\rho_+}^n\phi$, with $n\in \N$, and derive the corresponding higher-order equations. Observe first that we can write both \eqref{eq:eqradfieldhor2} and \eqref{eq:eqradfieldcos2} as follows
\begin{equation}
\label{eq:mainphysicalspaceeq}
\begin{split}
 0=&\:2T\partial_{\rho}\phi+\partial_{\rho}\left(\left((1-A_2)\rho^2+2\kappa \rho+A_3\rho^3+A_4 \rho^4\right)\partial_{\rho}\phi\right)-(B_0+B_1\rho +B_2\rho^2)\phi+\mathring{\slashed{\Delta}} \phi,
 \end{split}
\end{equation}
for appropriate coefficients $A_i\in \R$, with $i=2,3,4$, and $B_i\in \R$, with $i=0,1,2$, depending on $M$, $e$ and $\Lambda$, with $\kappa=\kappa_c$ or $\kappa=\kappa_+$, depending on the choice of coordinate $\rho=\rho_c$ or $\rho=\rho_+$.

\begin{proposition}
\label{prop:naivehoeq}
Denote $\phi_{(n)}=\hat{L}^{n}\phi$ or $\phi_{(n)}=\hat{\underline{L}}^{n}\phi$ . Then for each $n\in \N_0$, $\phi_{(n)}$ satisfies the following equation:
\begin{equation}
\label{eq:mainhophysicalspaceeq}
\begin{split}
 0=&\:2T\phi_{(n+1)}+\left((1-A_2)\rho^2+2\kappa \rho+A_3\rho^3+A_4 \rho^4\right)\phi_{(n+2)}\\
 &+2(n+1)\left((1-A_2)\rho+\kappa+\frac{3}{2}A_3\rho^2+2A_4 \rho^3\right)\phi_{(n+1)}\\
 &+\left[n(n+1)\left(1-A_2+3A_3 \rho+6A_4 \rho^2\right)-(B_0+B_1\rho +B_2\rho^2)\right]\phi_{(n)}+\mathring{\slashed{\Delta}} \phi_{(n)}\\
 &+\left[(n+1)n(n-1)\left(A_3+4A_4\rho\right) -n(B_1+2B_2\rho)\right]\phi_{(n-1)}\\
 &+\left[(n+1)n(n-1)(n-2)A_4-n(n-1)B_2\right]\phi_{(n-2)}
 \end{split}
 \end{equation}
\end{proposition}
\begin{proof}
It is convenient to work in $(u,\rho)$ or $(v,\rho)$ coordinates in which
\begin{equation*}
\phi_{(n)}=\partial_{\rho}^n\phi.
\end{equation*}
The equation follows by a simple induction argument, using \eqref{eq:mainphysicalspaceeq}, we have that for $n=0$:
\begin{equation*}
\begin{split}
 0=&\:2T\partial_{\rho}\phi+\left((1-A_2)\rho^2+2\kappa \rho+A_3\rho^3+A_4 \rho^4\right)\partial_{\rho}^2\phi+2\left((1-A_2)\rho+\kappa+\frac{3}{2}A_3\rho^2+2A_4 \rho^3\right)\partial_{\rho}\phi\\
 &-(B_0+B_1\rho +B_2\rho^2)\phi+\mathring{\slashed{\Delta}} \phi.
 \end{split}
\end{equation*}
In the inductive step, we apply the following identities:
\begin{align*}
(n+1)n+2(n+1)=&\:(n+2)(n+1),\\
(n+1)n(n-1)+3(n+1)n=&\: (n+2)(n+1)n,\\
(n+1)n(n-1)(n-2)+4(n+1)n(n-1)=&\: (n+2)(n+1)n(n-1).\qedhere
\end{align*}
\end{proof}

\subsection{Higher-order conservation laws}
\label{sec:hoconsvlaws}
The higher-order equations derived in Proposition \ref{prop:naivehoeq} have to be modified in order to use them in the Gevrey estimates of Section \ref{sec:gevrest}. In this section we carry out the required modification and we show that \emph{conservation laws} arise for the modified higher-order quantities along $\mathcal{H}^+$ and $\mathcal{I}^+$ when $\Lambda=0$ and $|e|=M$. In order to see the conservation laws, we need to restrict to the projections $\psi_{\ell}=\pi_{\ell}(\psi)$ of solutions $\psi$ to \eqref{eq:waveequation}, i.e. solutions satisfying ${\mathring{\slashed{\Delta}}}\psi_{\ell}=-\ell(\ell+1)\psi_{\ell}$ (see Section \ref{sec:enspaces}).

\textbf{The key goal of this section is to derive equations for modified higher-order quantities, which give rise to conservation laws, \emph{whilst keeping track of the precise dependence of constants on the \underline{order of differentiation} and the \underline{surface gravities} $\kappa_+$, $\kappa_c$.}} This is achieved in Lemma \ref{lm:eqmodquant}. It is only at this part of the paper where \emph{extremality} of the spacetime in the $\Lambda=0$ limit is fundamentally used via smallness of \emph{both} $\kappa_c$ and $\kappa_+$.\footnote{Although one could in principle try to derive an analogue of Lemma \ref{lm:eqmodquant} near $\mathcal{I}^+$ or $\mathcal{C}^+$, which is useful even when $\kappa_+$ is not small.}

In the discussion below, we will simultaneously investigate spacetimes regions near $\mathcal{H}^+$ or $\mathcal{C}^+$ ($\mathcal{I}^+$ when $\Lambda=0$), by working with the coordinate $\rho$, where $\rho=\rho_+$ or $\rho=\rho_c$, respectively.

In order to obtain the conservation law in the $\ell=0$ case, no modification is required of $\phi_{(n)}$. We consider \eqref{eq:mainhophysicalspaceeq} with $n=0$ and immediately see that the quantity
\begin{equation*}
\lim_{r\to \infty}\partial_{\rho_c}\phi_{0}(\tau,r,\theta,\varphi)
\end{equation*}
is independent of $\tau$ in the limit $\Lambda \downarrow 0$ and moreover,
\begin{equation*}
\lim_{r\downarrow r_+}\partial_{\rho_+}\phi_{0}(\tau,r,\theta,\varphi)
\end{equation*}
is independent of $\tau$ when both $\Lambda \downarrow 0$ and $|e|\to M$.

Now consider $\ell>0$ and \eqref{eq:mainhophysicalspaceeq} with $n=\ell$. When $n=\ell$, we see that
\begin{equation*}
\lim_{r\to \infty}\partial_{\rho}^{\ell+1}\phi_{\ell}(\tau,r,\theta,\varphi)
\end{equation*}
is independent of $\tau$ in the limit $\Lambda \downarrow 0$ if $M=0$. However, this conservation law fails to hold when $M>0$.

In order to derive conservation laws in the $M>0$ case when $\Lambda \downarrow 0$ \emph{and} $|e|\to M$, we introduce the quantity ${\Phi}:=w(\rho)^{-1}\phi$, where the weight function $w$ will be chosen appropriately. Let us consider \eqref{eq:mainhophysicalspaceeq} with $n=0$ and an additional inhomogeneous term $f$ on the left-hand side. We express this equation in terms of ${\Phi}$ as follows:
\begin{equation*}
\begin{split}
f=&\:2T\partial_{\rho}\phi+\partial_{\rho}\left(Dr^{-2} \partial_{\rho}(w{\Phi})\right)-(B_0+B_1\rho +B_2\rho^2)w{\Phi}+w\mathring{\slashed{\Delta}} {\Phi}\\
=&\:2T\partial_{\rho}\phi+w^{-1}\partial_{\rho}\left(Dr^{-2} w^2\partial_{\rho}{\Phi}\right)+\left[(Dr^{-2} w')' w^{-1}-(B_0+B_1\rho +B_2\rho^2)\right]w{\Phi}+w\mathring{\slashed{\Delta}} {\Phi}
 \end{split}
\end{equation*}
and hence, 

\begin{equation}
\label{eqwaveqeqtildepsi}
\begin{split}
w^{-1}f=&\:2Tw^{-1}\partial_{\rho}\phi+w^{-2}\partial_{\rho}\left(Dr^{-2} w^2\partial_{\rho}{\Phi}\right)+\left[(Dr^{-2} w')' w^{-1}-(B_0+B_1\rho +B_2\rho^2)\right]{\Phi}+\mathring{\slashed{\Delta}} {\Phi}.
 \end{split}
\end{equation}

We now choose $w$ so that the factor in front of the zeroth-order term ${\Phi}$ vanishes when $\kappa_c=\kappa_+=0$. We first define constants $\widetilde{A}_3$, $\widetilde{A}_4$, $\widetilde{B}_1$ and $\widetilde{B}_2$ that satisfy:
\begin{align*}
Dr^{-2}=&\:\rho^2(1-M\rho)^2+2\kappa \rho +A_2\rho^2+ \widetilde{A}_3\rho^3+\widetilde{A}_4 \rho^4,\\
B_0+B_1\rho+B_2\rho^2=&\:2M\rho(1-M\rho)+ B_0 +\widetilde{B}_1\rho+\widetilde{B}_2\rho^2.
\end{align*}
Note that $\widetilde{A}_3$, $\widetilde{A}_4$, $\widetilde{B}_1$ and $\widetilde{B}_2$ vanish when $\kappa_c=\kappa_+=0$ and $A_2$ and $B_0$ also vanish when $\kappa_c=\kappa_+=0$.

Let $w$ be a solution to the following differential equation:
\begin{equation*}
(\rho^2(1-M\rho)^2w')'=2M\rho(1-M\rho)w.
\end{equation*}
Then $w(\rho)$ must be of the form
\begin{equation*}
w(\rho)= C_1 (\rho-M\rho^2)^{-1}+C_2(1-M\rho)^{-1},
\end{equation*}
for some constants $C_1\in \R$ and $C_2\in \R$.

We set $C_1=0$ and $C_2=1$ to obtain:
\begin{equation*}
w(\rho)= (1-M\rho)^{-1}.
\end{equation*}
Note that $w(0)=1$.

We now recall the coordinate $\tilde{\rho}$ introduced in Section \ref{sec:confcoord}:
\begin{equation*}
\tilde{\rho}=\frac{\rho}{1-M\rho}.
\end{equation*}
Then
\begin{equation*}
\partial_{\rho}=\frac{d\tilde{\rho}}{d\rho} \partial_{\tilde{\rho}}=\frac{1}{(1-M\rho)^2} \partial_{\tilde{\rho}}=w^2 \partial_{\tilde{\rho}}
\end{equation*}
Furthermore,
\begin{equation*}
\rho= \frac{\tilde{\rho}}{1+M\tilde{\rho}}
\end{equation*}
and we can express
\begin{equation*}
w(\rho)=\frac{\tilde{\rho}}{\rho}=1+M\tilde{\rho}.
\end{equation*}

By the above observations, we can rewrite \eqref{eqwaveqeqtildepsi} as follows:
\begin{equation*}
\begin{split}
w^{-1}f=&\:2Tw\partial_{\tilde{\rho}}\phi+\partial_{\tilde{\rho}}\left(Dr^{-2} (1-M\rho)^{-4}\partial_{\tilde{\rho}}{\Phi}\right)+\mathring{\slashed{\Delta}} {\Phi}\\
&+\left[w^{-1}\partial_{\rho}((\rho^2(1-M\rho)^2+2\kappa \rho+A_2\rho^2+\widetilde{A}_3\rho^3+\widetilde{A}_4 \rho^4)\partial_{\rho}w)-(2M\rho(1-M\rho)+B_0+\widetilde{B}_1\rho +\widetilde{B}_2\rho^2)\right]{\Phi}\\
=&\:2T(1+M\tilde{\rho})\partial_{\tilde{\rho}}\phi+\partial_{\tilde{\rho}}\left(\tilde{\rho}^2\partial_{\tilde{\rho}}{\Phi}\right)+\partial_{\tilde{\rho}}\left((2\kappa \rho+A_2\rho^2+ \widetilde{A}_3\rho^3+\widetilde{A}_4 \rho^4)(1-M\rho)^{-4}\partial_{\tilde{\rho}}{\Phi}\right)+\mathring{\slashed{\Delta}} {\Phi}\\
&+\left[w^{-1}\partial_{\rho}(M(1-M\rho)^{-2}(2\kappa \rho+A_2\rho^2+\widetilde{A}_3\rho^3+\widetilde{A}_4 \rho^4))-(B_0+\widetilde{B}_1\rho +\widetilde{B}_2\rho^2)\right]{\Phi}.
 \end{split}
\end{equation*}

We can express:
\begin{equation*}
\begin{split}
(2\kappa \rho+A_2\rho^2+ \widetilde{A}_3\rho^3+\widetilde{A}_4 \rho^4)(1-M\rho)^{-4}=&\:2\kappa \frac{\tilde{\rho}^4}{\rho^3}+A_2 \frac{\tilde{\rho}^{4}}{\rho^2}+\widetilde{A}_3\frac{\tilde{\rho}^4}{\rho}+\widetilde{A}_4 \tilde{\rho}^4\\
=&\:2\kappa \tilde{\rho}(1+M\tilde{\rho})^3+A_2 \tilde{\rho}^{2}(1+M\tilde{\rho})^2+\widetilde{A}_3\tilde{\rho}^3(1+M\tilde{\rho})+\widetilde{A}_4 \tilde{\rho}^4.
\end{split}
\end{equation*}
Furthermore, we can express
\begin{equation*}
\begin{split}
w^{-1}&\partial_{\rho}(M(1-M\rho)^{-2}(2\kappa \rho+A_2\rho^2+\widetilde{A}_3\rho^3+\widetilde{A}_4 \rho^4))-(B_0+\widetilde{B}_1\rho +\widetilde{B}_2\rho^2)\\
=&\:M(1+M\tilde{\rho})\partial_{\tilde{\rho}}\left(2\kappa \tilde{\rho}(1+M\tilde{\rho})+A_2\tilde{\rho}^2+\widetilde{A}_3 \tilde{\rho}^{3}(1+M\tilde{\rho})^{-1}+\widetilde{A}_4\tilde{\rho}^4 (1+M\tilde{\rho})^{-2}\right)\\
&-\left(B_0+\widetilde{B}_1\tilde{\rho} (1+M\tilde{\rho})^{-1}+\widetilde{B}_2\tilde{\rho}^2 (1+M\tilde{\rho})^{-2}\right)\\
=&\:2M\kappa (1+M\tilde{\rho})(1+2M\tilde{\rho})+2MA_2 \tilde{\rho}(1+M\tilde{\rho})+M\widetilde{A}_3(3\tilde{\rho}^2-M\tilde{\rho}^{3}(1+M\tilde{\rho})^{-1})\\
&+M\widetilde{A}_4(4\tilde{\rho}^3(1+M\tilde{\rho})^{-1}-2M\tilde{\rho}^{4}(1+M\tilde{\rho})^{-2})-\left(B_0+\widetilde{B}_1\tilde{\rho} (1+M\tilde{\rho})^{-1}+\widetilde{B}_2\tilde{\rho}^2 (1+M\tilde{\rho})^{-2}\right)
\end{split}
\end{equation*}
to obtain:
\begin{equation}
\label{eq:waveqeqtildepsi2}
\begin{split}
(1+M\tilde{\rho})^{-1}f=&\:2T(1+M\tilde{\rho})\partial_{\tilde{\rho}}\phi+\partial_{\tilde{\rho}}\left(\tilde{\rho}^2\partial_{\tilde{\rho}}{\Phi}\right)\\
&+\partial_{\tilde{\rho}}\left((2\kappa \tilde{\rho}(1+M\tilde{\rho})^3+A_2 \tilde{\rho}^{2}(1+M\tilde{\rho})^2+\widetilde{A}_3\tilde{\rho}^3(1+M\tilde{\rho})+\widetilde{A}_4 \tilde{\rho}^4)\partial_{\tilde{\rho}}{\Phi}\right)+\mathring{\slashed{\Delta}} {\Phi}\\
&+\Big[ 2M\kappa (1+M\tilde{\rho})(1+2M\tilde{\rho})+2MA_2 \tilde{\rho}(1+M\tilde{\rho})+M\widetilde{A}_3(3\tilde{\rho}^2-M\tilde{\rho}^{3}(1+M\tilde{\rho})^{-1})\\
&+M\widetilde{A}_4(4\tilde{\rho}^3(1+M\tilde{\rho})^{-1}-2M\tilde{\rho}^{4}(1+M\tilde{\rho})^{-2})-\left(B_0+\widetilde{B}_1\tilde{\rho} (1+M\tilde{\rho})^{-1}+\widetilde{B}_2\tilde{\rho}^2 (1+M\tilde{\rho})^{-2}\right)\Big]{\Phi}.
 \end{split}
\end{equation}
Recall now that
\begin{equation*}
\phi=(1+M\tilde{\rho}){\Phi}.
\end{equation*}
Define moreover $f_{(0)}=(1+M\tilde{\rho})^{-1} f$, then we can write

\begin{equation}
\label{eq:waveqeqtildepsifin}
\begin{split}
f_{(0)}=&\:2T(1+M\tilde{\rho})^2\partial_{\tilde{\rho}}\Phi+2M(1+M\tilde{\rho})T{\Phi}+\partial_{\tilde{\rho}}\left(\tilde{\rho}^2\partial_{\tilde{\rho}}{\Phi}\right)\\
&+\partial_{\tilde{\rho}}\left(\left[2\kappa \tilde{\rho}(1+M\tilde{\rho})^3+A_2 \tilde{\rho}^{2}(1+M\tilde{\rho})^2+\widetilde{A}_3\tilde{\rho}^3(1+M\tilde{\rho})+\widetilde{A}_4 \tilde{\rho}^4\right]\partial_{\tilde{\rho}}{\Phi}\right)+\mathring{\slashed{\Delta}} {\Phi}\\
&+\Big[ 2M\kappa (1+M\tilde{\rho})(1+2M\tilde{\rho})+2MA_2 \tilde{\rho}(1+M\tilde{\rho})+M\widetilde{A}_3(3\tilde{\rho}^2-M\tilde{\rho}^{3}(1+M\tilde{\rho})^{-1})\\
&+M\widetilde{A}_4(4\tilde{\rho}^3(1+M\tilde{\rho})^{-1}-2M\tilde{\rho}^{4}(1+M\tilde{\rho})^{-2})-\left(B_0+\widetilde{B}_1\tilde{\rho} (1+M\tilde{\rho})^{-1}+\widetilde{B}_2\tilde{\rho}^2 (1+M\tilde{\rho})^{-2}\right)\Big]{\Phi}.
 \end{split}
\end{equation}
\begin{lemma}
\label{lm:eqmodquant} 
Define
\begin{align*}
{\Phi}_{(0)}:=&\: {\Phi},\\
{\Phi}_{(n)}:=&\:\partial_{\tilde{\rho}}^{n} {\Phi},\\
{f}_{(n)}:=&\: \partial_{\tilde{\rho}}^n((1+M\tilde{\rho})^{-1} f),
\end{align*}
where $\partial_{\tilde{\rho}}$ is defined with respect to the coordinates $(v,\tilde{\rho}_+,\theta,\varphi)$ or $(u,\tilde{\rho}_c,\theta,\varphi)$.

Let $0<\rho_0<\frac{r_c-r_+}{r_c r_+}$ and restrict to the region $\{\rho<\rho_0\}$. Then there exist functions ${F}_{n,k},\widetilde{F}_{n,k}:[0,\tilde{\rho}_0]\to \R$ for all $n\in \N$ and $k\in \N$ with $k\leq n+2$ such that 
\begin{equation}
\label{eq:howeighteqv1}
\begin{split}
{f}_{(n)}= &\:2(1+M\tilde{\rho})^2T{\Phi}_{(n+1)}+4\left(n+\frac{1}{2}\right)M(1+M\tilde{\rho})T{\Phi}_{(n)}+2n^2M^2T{\Phi}_{(n-1)}\\
&+[2\kappa \tilde{\rho} +\tilde{\rho}^2(1+ {F}_{n,n+2})]{\Phi}_{(n+2)}+ 2(n+1)[\tilde{\rho}+\kappa+\tilde{\rho}{F}_{n,n+1}]{\Phi}_{(n+1)}\\
&+\left[n(n+1)+\mathring{\slashed{\Delta}} \right]{\Phi}_{(n)}+\sum_{k=1}^{n} \frac{(n+1)!}{(k-1)!} M^{n-k}{F}_{n,k}{\Phi}_{(k)}+\sum_{k=0}^n \frac{n! (n-k+1)}{k!}M^{n-k}\widetilde{F}_{n,k}{\Phi}_{(k)},
\end{split}
\end{equation}
with $ {F}_{n,k},\widetilde{F}_{n,k}$ satisfying:
\begin{equation}
\label{eq:esttildeF}
||{F}_{n,k}||_{L^{\infty}([0,\rho_0])}+||\widetilde{F}_{n,k}||_{L^{\infty}([0,\rho_0])}\leq C_0(\kappa_++\kappa_c),
\end{equation}
with $C_0>0$ a constant that is independent of $n,k,\kappa_+,\kappa_c$.
\end{lemma}
\begin{proof}
We will obtain \eqref{eq:esttildeF} by acting with $\partial_{\tilde{\rho}}^n$ on both sides of \eqref{eq:waveqeqtildepsifin}. First of all, recall that the constants $A_2,\widetilde{A}_3,\widetilde{A}_4,B_0,B_1,B_2$ can all be bounded by $\kappa_++\kappa_c$

Then, observe that the factors in the terms on the first two lines on the right-hand side of \eqref{eq:waveqeqtildepsifin} are polynomials in $\tilde{\rho}$. Acting with $\partial_{\tilde{\rho}}^n$ therefore results straightforwardly in the terms on the first two lines on the right-hand side of \eqref{eq:esttildeF} and $\left[n(n+1)+\mathring{\slashed{\Delta}} \right]{\Phi}_{(n)}$.  This determines in particular $F_{n,n+2}$ and $F_{n,n+1}$, which are simply polynomials in $\tilde{\rho}$.

Furthermore, after acting with $\partial_{\tilde{\rho}}^n$ on the terms with polynomial factors in $\tilde{\rho}$ and applying the Leibniz rule, the terms involving $\kappa, A_2,\widetilde{A}_3,\widetilde{A}_4$ generate also terms involving $\Phi_{(k)}$, $n-2\leq k\leq n$ that can be grouped as follows:
\begin{equation*}
\sum_{k=n-2}^n  \frac{(n+1)!}{(k-1)!} M^{n-k}{F}_{n,k}{\Phi}_{(k)}.
\end{equation*}

We are left with applying $\partial_{\tilde{\rho}}^n$ to the term of the form $[\ldots]\Phi$ on the right-hand side of \eqref{eq:waveqeqtildepsifin}. This term also involves factors of the form $(1+M\tilde{\rho})^{-m}$ multiplying $\Phi$, so when we act with $\partial_{\tilde{\rho}}^n$ we will see derivatives of the form $\Phi_{(k)}$ for \underline{all} $0\leq k\leq n$. It therefore remains to show that the factors in front of $\Phi_{(k)}$ do not grow too fast in $n$, so as to respect the behaviour stated in \eqref{eq:esttildeF}.

We will need to use that, after applying the Leibniz rule, we have:
\begin{equation*}
\partial_{\tilde{\rho}}^j ((1+M\tilde{\rho})^{-m})=\frac{(-1)^{j} (j+m-1)! M^j}{(1+M \tilde{\rho})^{j+1}(m-1)!}.
\end{equation*}
The terms that will lead to the largest growth in $n$ are therefore terms involving factors $(1+M\tilde{\rho})^{-m}$ with the largest values of $m$. In the case under consideration, these will terms of the form:
\begin{equation*}
\begin{split}
\partial_{\tilde{\rho}}^n\left( (1+M\tilde{\rho})^{-2} {\Phi}\right)=&\:\sum_{k=0}^n \frac{n!}{k! (n-k)!} \partial_{\tilde{\rho}}^{n-k} \left( (1+M\tilde{\rho})^{-2}\right) {\Phi}_{(k)}\\
=&\:\sum_{k=0}^n\frac{n!}{k! }  \left[ \frac{(-1)^{n-k}  M^{n-k}}{(1+M \tilde{\rho})^{n-k+1}}\cdot \frac{(n-k+1)!}{(n-k)!}\right] {\Phi}_{(k)}\\
=&\:\sum_{k=0}^n(n-k+1)\frac{n!}{k! }  \frac{(-1)^{n-k}  M^{n-k}}{(1+M \tilde{\rho})^{n-k+1}}{\Phi}_{(k)}.
\end{split}
\end{equation*}
We can similarly decompose $n^l \tilde{\rho}^l\partial_{\tilde{\rho}}^{n-l}( (1+M\tilde{\rho})^{-2}\Phi)$ as sum over ${\Phi}_{(k)}$ to group all the terms coming from $\partial_{\tilde{\rho}}$ acting on $\tilde{\rho}^4(1+M\tilde{\rho})^{-2} {\Phi}$ as:
\begin{equation*}
\sum_{k=0}^n \frac{n! (n-k+1)}{k!}M^{n-k}\widetilde{F}_{n,k}{\Phi}_{(k)}.
\end{equation*}
The remaining terms in $[\ldots]\Phi$ on the right-hand side of \eqref{eq:waveqeqtildepsifin} can be dealt with \emph{mutatis mutandis} to conclude that \eqref{eq:howeighteqv1} holds.
\end{proof}

Note that by Lemma \ref{lm:eqmodquant}, it follows immediately that for $\Phi$ corresponding to a solution $\psi_{\ell}$ to \eqref{eq:waveequation} on extremal Reissner--Nordstr\"om ($\kappa_+=\kappa_c=0$), the quantity
\begin{equation*}
{\Phi}_{(\ell+1)}+2\left(\ell+\frac{1}{2}\right)M{\Phi}_{(\ell)}+\ell^2M^2{\Phi}_{(\ell-1)}
\end{equation*}
is conserved in $\tau$ along $\mathcal{I}^+$ (when $\rho=\rho_c$) or $\mathcal{H}^+$ (when $\rho=\rho_+$).

\subsection{Fixed-frequency operators}
\label{sec:fixedfreqop}
We consider solutions $\psi$ to \eqref{eq:waveequation} that take on the following form:
\begin{equation*}
\psi(\tau, r,\theta,\varphi)=e^{s \tau}\sum_{m=-\ell}^{\ell} \hat{\psi}_{\ell m}(r) Y_{\ell m}(\theta,\varphi).
\end{equation*}
For the sake of convenience, we will suppress the $m$ and $\ell$ in the subscript of $ \hat{\psi}_{\ell m}$. Furthermore, we denote $\hat{\Phi}:= (1+M\tilde{\rho})^{-1}r{\hat{\psi}}$ and $\hat{\phi}=r{\hat{\psi}}$ .

Then we use \eqref{eq:waveeqhyperboloidal} to obtain the following ODE for ${\hat{\psi}}$:
\begin{equation}
\label{eq:maineqhatpsi}
0=2(1-h_{r_+}D)s\partial_{r}{\hat{\psi}}+r^{-2}\partial_{r}(Dr^2\partial_{r}{\hat{\psi}})+[(2r^{-1}-r^{-2}(Dr^2h_{r_+})')s-h_{r_+}(2-h_{r_+}D)s^2]{\hat{\psi}}-r^{-2}\ell(\ell+1){\hat{\psi}}-2\mathfrak{l}^{-2}{\hat{\psi}}.
\end{equation}

Let us furthermore define the higher-order quantities $\hat{\Phi}_{(n)}$ as follows:
\begin{equation*}
\hat{\Phi}_{(n)}:= (\partial_{\tilde{\rho}}-s\hat{h})^n \hat{\Phi},
\end{equation*}
where we recall that $\hat{h}= r^2h(1+M\tilde{\rho})^{-2}$, with either $h=h_{r_+}$ or $h=h_{r_c}=2D^{-1}-h_{r_+}$, and where $\partial_{\tilde{\rho}}$ is defined with respect to the coordinates $(\tau,\tilde{\rho},\theta,\varphi)$.

Recall from Lemma \ref{lm:eqmodquant} the two different definitions of $\Phi_{(n)}$ corresponding to the choices $\rho=\rho_+$ and $\rho=\rho_c$, expressed in a coordinate invariant way:
\begin{align*}
\Phi_{(n)}=((1+M\tilde{\rho}_+)^{-2} \hat{\underline{L}})^n\Phi,\\
\Phi_{(n)}=((1+M\tilde{\rho}_c)^{-2} \hat{{L}})^n\Phi.
\end{align*}
By \eqref{eq:eqrhohat1} and \eqref{eq:eqrhohat2}, we can write both of these as:
\begin{equation*}
\Phi_{(n)}=(\partial_{\tilde{\rho}}-\hat{h}T)^n\Phi.
\end{equation*}
We can therefore conclude that
\begin{equation*}
({\Phi}_{(n)})_{\ell m}(\tau, r)=e^{s \tau}\hat{\Phi}_{(n)}(r),
\end{equation*}
with $\partial_{\tilde{\rho}}$ defined with respect to the coordinates $(\tau,\tilde{\rho},\theta,\varphi)$.

Then, \eqref{eq:mainhophysicalspaceeq} reduces to
\begin{equation}
\label{eq:maineqhatphin}
\begin{split}
 0=&\:2s\hat{\phi}^{(n+1)}+\left((1-A_2)\rho^2+2\kappa \rho+A_3\rho^3+A_4 \rho^4\right)\hat{\phi}^{(n+2)}\\
 &+2(n+1)\left((1-A_2)\rho+\kappa+\frac{3}{2}A_3\rho^2+2A_4 \rho^3\right)\hat{\phi}^{(n+1)}\\
 &\left[n(n+1)\left(1-A_2+3A_3 \rho+6A_4 \rho^2\right)-(B_0+B_1\rho +B_2\rho^2)\right]\hat{\phi}^{(n)}-\ell(\ell+1)\hat{\phi}^{(n)}\\
 &+\left[(n+1)n(n-1)\left(A_3+4A_4\rho\right) -n(B_1+2B_2\rho)\right]\hat{\phi}^{(n-1)}\\
 &+\left[(n+1)n(n-1)(n-2)A_4-n(n-1)B_2\right]\hat{\phi}^{(n-2)}
 \end{split}
 \end{equation}
and \eqref{eq:howeighteqv1} reduces to
\begin{equation}
\label{eq:homodifiedeq}
\begin{split}
0=&\: 2(1+M\tilde{\rho})^2s\hat{\Phi}_{(n+1)}+4\left(n+\frac{1}{2}\right)M(1+M\tilde{\rho})s\hat{\Phi}_{(n)}+2n^2M^2s\hat{\Phi}_{(n-1)}\\ \nonumber
&+[2\kappa \tilde{\rho}+\tilde{\rho}^2(1+ {F}_{n,n+2})]\hat{\Phi}_{(n+2)}+ 2(n+1)(\tilde{\rho}+\kappa+\tilde{\rho}{F}_{n,n+1})\hat{\Phi}_{(n+1)}\\ \nonumber
&+\left[n(n+1)-\ell(\ell+1) \right]\hat{\Phi}_{(n)}+\sum_{k=1}^{n} \frac{(n+1)!}{(k-1)!} M^{n-k}{F}_{n,k}\hat{\Phi}_{(k)}+\sum_{k=0}^n \frac{n! (n-k+1 )}{k!}M^{n-k}\widetilde{F}_{n,k}\hat{\Phi}_{(k)}.
\end{split}
\end{equation}
Equation \eqref{eq:maineqhatpsi} motivates the study of the differential operator $\hat{{L}}_{s,\ell,\kappa}$, which is defined as follows:
\begin{equation}
\label{eq:fixedfreqpsis}
\begin{split}
r^{-2}\hat{{L}}_{s,\ell,\kappa}({\hat{\psi}}):=&\:2(1-h_{r_+}D)s\partial_{r}{\hat{\psi}}+r^{-2}\partial_{r}(Dr^2\partial_{r}{\hat{\psi}})+[-h_{r_+}(2-h_{r_+}D)s^2+(2r^{-1}-r^{-2}(Dr^2h_{r_+})')s]{\hat{\psi}}\\ 
&-(r^{-2} \ell(\ell+1)+2\mathfrak{l}^{-2}) {\hat{\psi}},\\
=&\:2(1-h_{r_+}D)sr^{-1}\partial_{r}(r{\hat{\psi}})+r^{-2}\partial_{r}(Dr^2\partial_{r}{\hat{\psi}})+[-h_{r_+}(2-h_{r_+}D)s^2-(h_+D)'s]{\hat{\psi}}\\
&-(r^{-2}\ell(\ell+1)+2\mathfrak{l}^{-2}) {\hat{\psi}}.
\end{split}
\end{equation}
We moreover have the following relation
\begin{equation}
\label{eq:fixedfreqpsihat}
\begin{split}
&(\partial_{\tilde{\rho}}-s\hat{h})^n(r^3(1+M\tilde{\rho})^{-1}r^{-2}\hat{{L}}_{s,\ell,\kappa}({\hat{\psi}}))= 2(1+M\tilde{\rho})^2s\hat{\Phi}_{(n+1)}+4\left(n+\frac{1}{2}\right)M(1+M\tilde{\rho})s\hat{\Phi}_{(n)}\\
&+2n^2M^2s\hat{\Phi}_{(n-1)}\\ 
&+[2\kappa \tilde{\rho}+\tilde{\rho}^2(1+ {F}_{n,n+2})]\hat{\Phi}_{(n+2)}+ 2(n+1)(\tilde{\rho}+\kappa+\tilde{\rho}{F}_{n,n+1})\hat{\Phi}_{(n+1)}\\
&+\left[n(n+1)-\ell(\ell+1) \right]\hat{\Phi}_{(n)}+\sum_{k=1}^{n} \frac{(n+1)!}{(k-1)!} M^{n-k}{F}_{n,k}\hat{\Phi}_{(k)}+\sum_{k=0}^n \frac{n! (n-k+1)}{k!}M^{n-k}\widetilde{F}_{n,k}\hat{\Phi}_{(k)}.
\end{split}
\end{equation}
\textbf{Note that the operator $\hat{{L}}_{s,\ell,\kappa}$ depends on the choice of hypersurface $\Sigma$, i.e. the choice of $h_{r_+}$!}

Suppose
\begin{equation*}
r(1+M\tilde{\rho})^{-1}\hat{{L}}_{s,\ell,\kappa}({\hat{\psi}})=f,
\end{equation*}
then we denote
\begin{equation}
\label{def:fn}
f_n:=(\partial_{\tilde{\rho}}-s\hat{h})^n(f).
\end{equation}

\section{The solution operator semigroup}
\label{sec:propgevreyphys}

In this section, we will derive physical space estimates involving Gevrey norms, which will allow us to define the time-translation operator $\mathcal{S}(\tau)$ on the Hilbert space $\mathbf{H}_{\sigma,\rho_0}$ (see Section \ref{sec:enspaces}). In particular, we establish propagation of Gevrey regularity near $\mathcal{H}^+$ and $\mathcal{C}^+$ or $\mathcal{I}^+$.

Here, we will fix $\Sigma_0$ to be a mixed spacelike-null hypersurface of the form $\Sigma_0=\underline{N}_0\cup S_0\cup N_0$ by choosing $h_{r_+}=0$ for $r\leq R_0^+$ and $h_{r_c}=0$ for $r \geq R_0^c$, see Section \ref{sec:foliations}.
\begin{proposition}
\label{prop:physspacegevrey}
Let $\psi\in C^{\infty}( \hat{\mathcal{R}};\C)$ be a solution to \eqref{eq:waveequation}. Let $\sigma>0$. Then, for $\rho=\rho_+$ or $\rho=\rho_c$ and $\rho_0>0$ suitably small, there exist constants $B,C>0$, independent of $\sigma$, such that for all $\tau_1>0$ and $N_{\infty}\in \N$
\begin{equation}
\label{eq:physicalspacegevreylown2} 
\begin{split}
\int_{\Sigma_{\tau=\tau_1}}& \mathbb{T}(T,\mathbf{n}_{\Sigma_{\tau=\tau_1}})[\psi] \,d\mu_{\tau=\tau_1}+\sum_{\ell=0}^{\infty}\Bigg[\sum_{n=0}^{\ell-1} \frac{\sigma^{2n}}{\ell^{2n}(\ell+1)^{2n}}  (\ell+1)^4\int_{0}^{\rho_0}\int_{\s^2} |(\phi_{\ell})_{(n)}|^2(\tau_1,\rho,\theta,\varphi)\,d\omega d\rho\\
&+ \frac{\ell!^2 (\ell+1)!^2}{\ell^{2\ell}(\ell+1)^{2\ell}}\sum_{n=\ell}^{N_{\infty}} \frac{\sigma^{2n}}{n!^2(n+1)!^2}  (n+1)^4\int_{0}^{\rho_0}\int_{\s^2} |(\phi_{\ell})_{(n)}|^2(\tau_1,\rho,\theta,\varphi)\,d\omega d\rho\Bigg]\\
  \leq&\:  Ce^{B \sigma \tau_1}  \sum_{\ell=0}^{\infty}\Bigg[\sum_{n=0}^{\ell-1} \frac{\sigma^{2n}}{\ell^{2n}(\ell+1)^{2n}}  (\ell+1)^4 \int_{0}^{\rho_0}\int_{\s^2} (\phi_{\ell})_{(n)}|^2(0,\rho,\theta,\varphi)\,d\omega d\rho\\
  &+ \frac{\ell!^2 (\ell+1)!^2}{\ell^{2\ell}(\ell+1)^{2\ell}}\sum_{n=\ell}^{N_{\infty}} \frac{\sigma^{2n}}{n!^2(n+1)!^2}  (n+1)^4 \int_{0}^{\rho_0}\int_{\s^2} |(\phi_{\ell})_{(n)}|^2(0,\rho,\theta,\varphi)\,d\omega d\rho\Bigg]\\
  &+Ce^{B \sigma \tau_1}\int_{\Sigma_{\tau=0}} \mathbb{T}(T,\mathbf{n}_{\Sigma_{\tau=0}})[\psi] \,d\mu_{\tau=0}.
\end{split}
\end{equation}
\end{proposition}
\begin{proof}
In the notation below, we will suppress the subscript $\ell$ in $\psi_{\ell}$ and $\phi_{\ell}$. The higher-order quantities $\phi_{(n)}$ satisfy \eqref{eq:mainhophysicalspaceeq}. Multiplying both sides of \eqref{eq:mainhophysicalspaceeq} with $ \overline{\phi}_{(n+1)}$ and taking the real part of the resulting equation, gives:
\begin{equation}
\label{eq:mainidhophysicalspaceeq}
\begin{split}
 0=&\: T( |\phi_{(n+1)}|^2)+\frac{1}{2}\left((1-A_2)\rho^2+2\kappa \rho+A_3\rho^3+A_4 \rho^4\right)\partial_{\rho}(|\phi_{(n+1))}|^2)\\
 &+\frac{1}{2}(r^2h) \left((1-A_2)\rho^2+2\kappa \rho+A_3\rho^3+A_4 \rho^4\right) T(|\phi_{(n+1)}|^2) \\
 &+2(n+1)\left((1-A_2)\rho+\kappa+\frac{3}{2}A_3\rho^2+2A_4 \rho^3\right)|\phi_{(n+1)}|^2\\
 &+\left[n(n+1)\left(1-A_2+3A_3 \rho+6A_4 \rho^2\right)-(B_0+B_1\rho +B_2\rho^2)\right]\re (\phi_{(n)} \overline{\phi}_{(n+1)})+ \re (\mathring{\slashed{\Delta}} \phi_{(n)}\overline{\phi}_{(n+1)})\\
 &+\left[(n+1)n(n-1)\left(A_3+4A_4\rho\right) -n(B_1+2B_2\rho)\right]\re (\phi_{(n-1)}\overline{\phi}_{(n+1)})\\
 &+\left[(n+1)n(n-1)(n-2)A_4-n(n-1)B_2\right]\re (\phi_{(n-2)}\overline{\phi}_{(n+1)})
 \end{split}
 \end{equation}
 We apply the Leibniz rule to rewrite:
 \begin{equation*}
 \begin{split}
 \frac{1}{2}&\left((1-A_2)\rho^2+2\kappa \rho+A_3\rho^3+A_4 \rho^4\right)\partial_{\rho}(|\phi_{(n+1))}|^2)=\partial_{\rho}\left(\frac{1}{2}\left((1-A_2)\rho^2+2\kappa \rho+A_3\rho^3+A_4 \rho^4\right)|\phi_{(n+1))}|^2\right)\\
 &-\left[(1-A_2)\rho+\kappa +\frac{3}{2} A_3\rho^2+2A_4\rho^3\right]|\phi_{(n+1))}|^2.
 \end{split}
 \end{equation*}
 We now invoke the assumption $\mathring{\slashed{\Delta}} \psi=-\ell(\ell+1)\psi$ and integrate the above equation in $\bigcup_{\tau\in [0,\tau_1]} \Sigma_{\tau}\cap \{\rho\leq \rho_0\}$, with $\tau_1>0$ arbitrarily large, and we apply a weighted Young's inequality to the terms involving $\phi^{(k)}$ with $k\leq n$ to obtain the following estimate for $n\geq 1$: let $\epsilon,\sigma>0$, then
 \begin{equation}
 \label{eq:physicalspacegevrey}
 \begin{split}
 \int_{0}^{\rho_0}\int_{\s^2}& \left[1+\frac{1}{2}r^2h\left((1-A_2)\rho^2+2\kappa \rho+A_3\rho^3+A_4 \rho^4\right)  \right] |\phi_{(n+1)}|^2(\tau_1,\rho,\theta,\varphi)\ \,d\omega d\rho\\
 &+\frac{1}{2}\left((1-A_2)\rho_0^2+2\kappa \rho_0+A_3\rho_0^3+A_4 \rho_0^4\right) \int_0^{\tau_1}\int_{\s^2}|\phi_{(n+1))}|^2(\tau,\rho_0,\theta,\varphi)\,d\omega d\tau\\
 &+2(n+1) \int_0^{\tau_1}\int_{0}^{\rho_0}\int_{\s^2}\Bigg((1-A_2)\rho+\kappa+\frac{3}{2}A_3\rho^2+2A_4 \rho^3\\
 &-\frac{1}{2}(n+1)^{-1}\left[(1-A_2)\rho+\kappa +\frac{3}{2} A_3\rho^2+2A_4\rho^3\right]\Bigg)|\phi_{(n+1)}|^2\,d\omega d\rho d\tau\\
 \leq &\: \int_{0}^{\rho_0}\int_{\s^2} \left[1+\frac{1}{2}r^2h\left((1-A_2)\rho^2+2\kappa \rho+A_3\rho^3+A_4 \rho^4\right)  \right] |\phi_{(n+1)}|^2(0,\rho,\theta,\varphi)\ \,d\omega d\rho\\
 &+ C \epsilon^{-1} \sigma \int_0^{\tau_1} \int_{0}^{\rho_0}\int_{\s^2}|\phi_{(n+1)}|^2\,d\omega d\rho d\tau\\
 &+ \epsilon \sigma^{-1} \int_0^{\tau_1} \int_{0}^{\rho_0}\int_{\s^2}  \max\{\ell^2(\ell+1)^2,n^2(n+1)^2\}|\phi_{(n)}|^2+ M^2 [(n-1)^2n^2(n+1)^2+n^2]|\phi_{(n-1)}|^2\\
 &+M^4[(n-2)^2 (n-1)^2n^2(n+1)^2+n^2(n-1)^2]|\phi_{(n-2)}|^2 \,d\omega d\rho d\tau,
 \end{split}
 \end{equation}
  where we introduced appropriate factors of $M$ on the right-hand side to ensure that the constants $C,\epsilon>0$ can be taken to be dimensionless and $\sigma>0$ has the same dimension as $M^{-1}$. Note that all the terms appearing on the left-hand side are non-negative definite when $\rho_0$ is chosen sufficiently small, using moreover that $\kappa\geq 0$ has a good sign and that $A_2\to 0$ when $\kappa\to 0$.\footnote{One can also first set $M=1$, carry out the computation, and then place appropriate powers of $M$ in front of the terms on the right-hand side to ensure that the dimensions of all the terms in the inequality are the same.}
  
 For $n=0$ we instead multiply both sides of \eqref{eq:mainhophysicalspaceeq} with $\chi \overline{\phi}_{(1)}$, where $\chi$ is a non-negative smooth cut-off function, such that $\chi(\rho)=1$ in $\rho\leq \rho_0$ and $\chi(\rho)=0$ for $2\rho_0<\rho<\frac{r_c-r_+}{r_c r_+}$, and we integrate in the larger region $\rho\leq 2 \rho_0$. We assume $\rho_0$ is sufficiently small so that $2\rho_0<\frac{1}{2}\rho_0+\frac{1}{2}\frac{r_c-r_+}{r_c r_+}$. This generates additional terms with a factor $\partial_{\rho}\chi$, which are supported in compact regions of $r$, away from $r_c$ and $r_+$.
 
 Furthermore, we do not apply Young's inequality to estimate the integral of the terms $\re (\mathring{\slashed{\Delta}} \phi \chi \overline{\phi}_{(1)})$ and $\re (-B_0\phi \chi \overline{\phi}_{(1)})$ on the right-hand side of \eqref{eq:mainidhophysicalspaceeq}, and we instead estimate
 \begin{equation*}
 \int_{\s^2}  \re ((\mathring{\slashed{\Delta}}-B_0) \phi_{(0)}\chi \overline{\phi}_{(1)})\,d\omega=\int_{\s^2 } -\frac{1}{2}\partial_{\rho} ( \chi|\mathring{\snabla} \phi|^2+B_0\chi |\phi|^2)+hr^2 \chi T(|\mathring{\snabla} \phi|^2+B_0|\phi|^2)+\frac{1}{2}\frac{d\chi}{dr}(|\mathring{\snabla} \phi|^2+B_0|\phi|^2)\,d\omega,
 \end{equation*}
 where the terms involving $T$ on the right-hand side vanish for $\rho\leq \rho_0$ since $h=0$ for $\rho\leq \rho_0$.
We do apply Young's inequality and then a Hardy inequality \eqref{eq:hardy1} to estimate
 \begin{equation*}
 \begin{split}
 \int_{0}^{\tau_1}\int_0^{2\rho_0}\int_{\s^2}\chi|B_1\rho+B_2\rho^2||\Re(\phi_{(1)} \overline{\phi})|\,d\omega d\rho d\tau\leq &\: \int_{0}^{\tau_1}\int_0^{2\rho_0}\int_{\s^2} \epsilon \rho \chi |\phi_{(1)}|^2\,d\omega d\rho d\tau\\
 &+C\epsilon^{-1}\int_{0}^{\tau_1}\int_0^{2\rho_0}\int_{\s^2} (B_1^2 \rho_0+B_2^2\rho_0^2)\chi |\phi|^2 \,d\omega d\rho d\tau\\
 \leq &\:  \int_{0}^{\tau_1}\int_0^{2\rho_0}\int_{\s^2} (\epsilon +C\epsilon^{-1} \rho_0^2)\rho\chi |\phi_{(1)}|^2\,d\omega d\rho d\tau\\
 &+C\epsilon^{-1}  \int_0^{\tau_1}\int_{\rho_0}^{2\rho_0}\int_{\s^2} |\phi|^2\,d\omega d\rho d\tau,
  \end{split}
 \end{equation*}
 where $\epsilon>0$ can be taken arbitrarily small, and $\rho_0$ is chosen suitably small such that $C\epsilon^{-1}\rho_0<\epsilon$.

We then obtain
  \begin{equation}
  \label{eq:physicalspacegevrey1}
 \begin{split}
 \int_{0}^{2\rho_0}\int_{\s^2}& \chi \left[1+\frac{1}{2}r^2h\left((1-A_2)\rho^2+2\kappa \rho+A_3\rho^3+A_4 \rho^4\right)  \right] |\phi_{(1)}|^2(\tau_1,\rho,\theta,\varphi)\ \,d\omega d\rho \\
 &+2 \int_0^{\tau_1}\int_{0}^{2\rho_0}\int_{\s^2}\chi \Bigg((1-A_2-\epsilon)\rho+\kappa+\frac{3}{2}A_3\rho^2+2A_4 \rho^3\\
 &-\frac{1}{2}\left[(1-A_2)\rho+\kappa +\frac{3}{2} A_3\rho^2+2A_4\rho^3\right]\Bigg)|\phi_{(1)}|^2\,d\omega d\rho d\tau\\
 \leq &\: \int_{0}^{2\rho_0}\int_{\s^2} \chi \left[1+\frac{1}{2}r^2h\left((1-A_2)\rho^2+2\kappa \rho+A_3\rho^3+A_4 \rho^4\right)  \right] |\phi_{(1)}|^2(0,\rho,\theta,\varphi)\ \,d\omega d\rho\\
 &+C\epsilon^{-1}  \int_0^{\tau_1}\int_{\rho_0}^{2\rho_0}\int_{\s^2} |\phi|^2+|{\mathring{\snabla}}\phi|^2\,d\omega d\rho d\tau+C\sum_{\tau'=0,\tau_1}\int_{\rho_0}^{2\rho_0}\int_{\s^2} (|\phi|^2+|{\mathring{\snabla}}\phi|)|_{\tau=\tau'}\,d\omega d\rho.
 \end{split}
 \end{equation}
 
 We will now consider the higher-order quantities $\phi_{(n)}$ and treat the cases $n\leq \ell$ and $n\geq \ell+1$ separately.
 
 \emph{Case 1: $n\leq \ell-1$}
 Take $\rho_0$ suitably small and consider the following weighted summation over $n$:
 \begin{equation*}
 \sum_{n=0}^{\ell-1} \frac{\sigma^{2n}}{\ell^{2n}(\ell+1)^{2n}}\left[\cdot\right].
 \end{equation*}
 We apply \eqref{eq:physicalspacegevrey} and \eqref{eq:physicalspacegevrey1} to obtain:
\begin{equation}
\label{eq:pregevphys1}
\begin{split}
  \int_{0}^{\rho_0}&\int_{\s^2} \sum_{n=0}^{\ell-1} \frac{\sigma^{2n}}{\ell^{2n}(\ell+1)^{2n}} |\phi_{(n+1)}|^2(\tau_1,\rho,\theta,\varphi)\,d\omega d\rho\\
  \leq&\:  C\epsilon^{-1} \sigma \int_0^{\tau_1}\int_{0}^{\rho_0}\int_{\s^2}\sum_{n=0}^{\ell-1} \frac{\sigma^{2n}}{\ell^{2n}(\ell+1)^{2n}} |\phi_{(n+1)}|^2(\tau,\rho,\theta,\varphi)\,d\omega d\rho d\tau\\
  &+C \int_{0}^{\rho_0}\int_{\s^2} \sum_{n=0}^{\ell-1} \frac{\sigma^{2n}}{\ell^{2n}(\ell+1)^{2n}} |\phi_{(n+1)}|^2(0,\rho,\theta,\varphi)\,d\omega d\rho+ C\epsilon^{-1}  \int_0^{\tau_1}\int_{\rho_0}^{2\rho_0}\int_{\s^2} |\phi|^2+|{\mathring{\snabla}}\phi|^2\,d\omega d\rho d\tau\\
  &+C\sum_{\tau'=0,\tau_1}\int_{\rho_0}^{2\rho_0}\int_{\s^2} (|\phi|^2+|{\mathring{\snabla}}\phi|)|_{\tau=\tau'}\,d\omega d\rho.
\end{split}
\end{equation}

 \emph{Case 2: $n\geq \ell$}
 Take $\rho_0$ suitably small and consider the following weighted summation over $n$:
 \begin{equation*}
 \frac{\ell!^2 (\ell+1)!^2}{\ell^{2\ell}(\ell+1)^{2\ell}}\sum_{n=\ell}^{N_{\infty}} \frac{\sigma^{2n}}{n!^2(n+1)!^2}\left[\cdot\right],
 \end{equation*}
 where $N_{\infty}$ is arbitrarily large.
 
 We apply \eqref{eq:physicalspacegevrey} (when $\ell\geq 1$) and \eqref{eq:physicalspacegevrey1} (when $\ell=0$) to obtain
 \begin{equation*}
\begin{split}
\int_{0}^{\rho_0}&\int_{\s^2} \sum_{n=\ell}^{N_{\infty}}  \frac{\sigma^{2n}}{n!^2(n+1)!^2} |\phi_{(n+1)}|^2(\tau_1,\rho,\theta,\varphi)\,d\omega d\rho\\
   \leq&\:  C\epsilon^{-1} \sigma \int_0^{\tau_1}\int_{0}^{\rho_0}\int_{\s^2}\sum_{n=\ell}^{N_{\infty}} \frac{\sigma^{2n}}{n!^2(n+1)!^2} |\phi_{(n+1)}|^2(\tau,\rho,\theta,\varphi)\,d\omega d\rho d\tau\\
  &+C \int_{0}^{\rho_0}\int_{\s^2} \sum_{n=\ell}^{N_{\infty}}  \frac{\sigma^{2n}}{n!^2(n+1)!^2} |\phi_{(n+1)}|^2(0,\rho,\theta,\varphi)\,d\omega d\rho\\
&+ \epsilon \sigma^{-1}\int_0^{\tau_1}\int_{0}^{\rho_0}\int_{\s^2} \sum_{n=\ell}^{N_{\infty}}  \frac{\sigma^{2n}}{n!^2(n+1)!^2}\Big[(n+1)^2n^2|\phi_{(n)}|^2+M^2[(n+1)^2n^2(n-1)^2+n^2]|\phi_{(n-1)}|^2\\
&+M^4[(n+1)^2n^2(n-1)^2(n-2)^2+n^2(n-1)^2]|\phi_{(n-2)}|^2\Big]\,d\omega d\rho d\tau\\
&+C\epsilon^{-1}  \int_0^{\tau_1}\int_{\rho_0}^{2\rho_0}\int_{\s^2} |\phi|^2+|{\mathring{\snabla}}\phi|^2\,d\omega d\rho d\tau+C\sum_{\tau'=0,\tau_1}\int_{\rho_0}^{2\rho_0}\int_{\s^2} (|\phi|^2+|{\mathring{\snabla}}\phi|)|_{\tau=\tau'}\,d\omega d\rho.
    \end{split}
\end{equation*}
Given the choice of summation over $n$, we can group the terms on the right-hand side with a factor $\epsilon>0$ with terms lower in the summation to obtain:
\begin{equation*}
\begin{split}
\int_{0}^{\rho_0}&\int_{\s^2} \sum_{n=\ell}^{N_{\infty}}  \frac{\sigma^{2n}}{n!^2(n+1)!^2} |\phi_{(n+1)}|^2(\tau_1,\rho,\theta,\varphi)\,d\omega d\rho\\
   \leq&\:  C\epsilon^{-1} \sigma \int_0^{\tau_1}\int_{0}^{\rho_0}\int_{\s^2}\sum_{n=\ell}^{N_{\infty}} \frac{\sigma^{2n}}{n!^2(n+1)!^2} |\phi_{(n+1)}|^2(\tau,\rho,\theta,\varphi)\,d\omega d\rho d\tau\\
  &+C \int_{0}^{\rho_0}\int_{\s^2} \sum_{n=\ell}^{N_{\infty}}  \frac{\sigma^{2n}}{n!^2(n+1)!^2} |\phi_{(n+1)}|^2(0,\rho,\theta,\varphi)\,d\omega d\rho\\
&+ C\epsilon \sigma^{-1}(\sigma^2+\sigma^4+\sigma^6)\int_0^{\tau_1}\int_{0}^{\rho_0}\int_{\s^2} \sum_{n=\ell-3}^{N_{\infty}-1}  \frac{\sigma^{2n}}{n!^2(n+1)!^2}|\phi_{(n+1)}|^2\,d\omega d\rho d\tau\\
&+C\epsilon^{-1}  \int_0^{\tau_1}\int_{\rho_0}^{2\rho_0}\int_{\s^2} |\phi|^2+|{\mathring{\snabla}}\phi|^2\,d\omega d\rho d\tau+C\sum_{\tau'=0,\tau_1}\int_{\rho_0}^{2\rho_0}\int_{\s^2} (|\phi|^2+|{\mathring{\snabla}}\phi|)|_{\tau=\tau'}\,d\omega d\rho.
    \end{split}
\end{equation*}
We can rewrite the above expression further to obtain:
\begin{equation}
\label{eq:pregevphys2}
\begin{split}
   \frac{\ell!^2 (\ell+1)!^2}{\ell^{2\ell}(\ell+1)^{2\ell}}\int_{0}^{\rho_0}&\int_{\s^2} \sum_{n=\ell}^{N_{\infty}}  \frac{\sigma^{2n}}{n!^2(n+1)!^2} |\phi_{(n+1)}|^2(\tau_1,\rho,\theta,\varphi)\,d\omega d\rho\\
   \leq&\: C(\epsilon^{-1}\sigma+\epsilon \sigma^5)  \frac{\ell!^2 (\ell+1)!^2}{\ell^{2\ell}(\ell+1)^{2\ell}}\int_0^{\tau_1}\int_{0}^{\rho_0}\int_{\s^2}\sum_{n=\ell}^{N_{\infty}} \frac{\sigma^{2n}}{n!^2(n+1)!^2} |\phi_{(n+1)}|^2(\tau,\rho,\theta,\varphi)\,d\omega d\rho d\tau\\
  &+C  \frac{\ell!^2 (\ell+1)!^2}{\ell^{2\ell}(\ell+1)^{2\ell}}\int_{0}^{\rho_0}\int_{\s^2} \sum_{n=\ell}^{N_{\infty}}  \frac{\sigma^{2n}}{n!^2(n+1)!^2} |\phi_{(n+1)}|^2(0,\rho,\theta,\varphi)\,d\omega d\rho\\
  &+C\epsilon (\sigma+\sigma^5)  \frac{\sigma^{2(\ell-1)}}{\ell^{2(\ell-1)} (\ell+1)^{2(\ell-1)} } \int_0^{\tau_1}\int_{0}^{\rho_0}\int_{\s^2} |\phi_{(\ell)}|^2(\tau,\rho,\theta,\varphi)\,d\omega d\rho d\tau\\
  &+ C\epsilon (\sigma+\sigma^5)   \frac{\sigma^{2(\ell-2)}}{\ell^{2(\ell-2)} (\ell+1)^{2(\ell-2)} } \int_0^{\tau_1}\int_{0}^{\rho_0}\int_{\s^2} |\phi_{(\ell-1)}|^2(\tau,\rho,\theta,\varphi)\,d\omega d\rho d\tau\\
    &+ C\epsilon (\sigma+\sigma^5) \frac{\sigma^{2(\ell-3)}}{\ell^{2(\ell-3)} (\ell+1)^{2(\ell-3)} } \int_0^{\tau_1}\int_{0}^{\rho_0}\int_{\s^2} |\phi_{(\ell-2)}|^2(\tau,\rho,\theta,\varphi)\,d\omega d\rho d\tau\\
    &+C\epsilon^{-1}  \int_0^{\tau_1}\int_{\rho_0}^{2\rho_0}\int_{\s^2} |\phi|^2+|{\mathring{\snabla}}\phi|^2\,d\omega d\rho d\tau+C\sum_{\tau'=0,\tau_1}\int_{\rho_0}^{2\rho_0}\int_{\s^2} (|\phi|^2+|{\mathring{\snabla}}\phi|)|_{\tau=\tau'}\,d\omega d\rho.
\end{split}
\end{equation}
We can combine \eqref{eq:pregevphys2} with \eqref{eq:pregevphys1} to group the last three spacetime integrals on  on the right-hand side of \eqref{eq:pregevphys2} with the terms on the right-hand side of \eqref{eq:pregevphys1} and obtain:
\begin{equation*}
\begin{split}
  \sum_{n=0}^{\ell-1}& \frac{\sigma^{2n}}{\ell^{2n}(\ell+1)^{2n}}  \int_{0}^{\rho_0}\int_{\s^2} |\phi_{(n+1)}|^2(\tau_1,\rho,\theta,\varphi)\,d\omega d\rho\\
 &+ \frac{\ell!^2 (\ell+1)!^2}{\ell^{2\ell}(\ell+1)^{2\ell}}\sum_{n=\ell}^{N_{\infty}} \frac{\sigma^{2n}}{n!^2(n+1)!^2}  \int_{0}^{\rho_0}\int_{\s^2} |\phi_{(n+1)}|^2(\tau_1,\rho,\theta,\varphi)\,d\omega d\rho\\
   \leq&\: C\sum_{n=0}^{\ell-1}\frac{\sigma^{2n}}{\ell^{2n}(\ell+1)^{2n}}  \int_{0}^{\rho_0}\int_{\s^2} |\phi_{(n+1)}|^2(0,\rho,\theta,\varphi)\,d\omega d\rho\\
 &+ C\frac{\ell!^2 (\ell+1)!^2}{\ell^{2\ell}(\ell+1)^{2\ell}}\sum_{n=\ell}^{N_{\infty}} \frac{\sigma^{2n}}{n!^2(n+1)!^2}  \int_{0}^{\rho_0}\int_{\s^2} |\phi_{(n+1)}|^2(0,\rho,\theta,\varphi)\,d\omega d\rho\\
   &+C\epsilon^{-1}\sigma \sum_{n=0}^{\ell-1} \frac{\sigma^{2n}}{\ell^{2n}(\ell+1)^{2n}}  \int_0^{\tau_1}\int_{0}^{\rho_0}\int_{\s^2} |\phi_{(n+1)}|^2(\tau,\rho,\theta,\varphi)\,d\omega d\rho\\
   &+C(\epsilon^{-1}\sigma+\epsilon \sigma^5)  \frac{\ell!^2 (\ell+1)!^2}{\ell^{2\ell}(\ell+1)^{2\ell}}\int_0^{\tau_1}\int_{0}^{\rho_0}\int_{\s^2}\sum_{n=\ell}^{N_{\infty}} \frac{\sigma^{2n}}{n!^2(n+1)!^2} |\phi_{(n+1)}|^2(\tau,\rho,\theta,\varphi)\,d\omega d\rho d\tau\\
    &+C\epsilon^{-1}  \int_0^{\tau_1}\int_{\rho_0}^{2\rho_0}\int_{\s^2} |\phi|^2+|{\mathring{\snabla}}\phi|^2\,d\omega d\rho d\tau+C\sum_{\tau'=0,\tau_1}\int_{\rho_0}^{2\rho_0}\int_{\s^2} (|\phi|^2+|{\mathring{\snabla}}\phi|)|_{\tau=\tau'}\,d\omega d\rho.
\end{split}
\end{equation*}
We then apply a standard Gr\"onwall estimate to obtain
\begin{equation}
\label{eq:physicalspacegevreylown} 
\begin{split}
 \sum_{n=0}^{\ell-1}& \frac{\sigma^{2n}}{\ell^{2n}(\ell+1)^{2n}}  \int_{0}^{\rho_0}\int_{\s^2} |\phi_{(n+1)}|^2(\tau_1,\rho,\theta,\varphi)\,d\omega d\rho\\
 &+ \frac{\ell!^2 (\ell+1)!^2}{\ell^{2\ell}(\ell+1)^{2\ell}}\sum_{n=\ell}^{N_{\infty}} \frac{\sigma^{2n}}{n!^2(n+1)!^2}  \int_{0}^{\rho_0}\int_{\s^2} |\phi_{(n+1)}|^2(\tau_1,\rho,\theta,\varphi)\,d\omega d\rho\\
  \leq&\:  Ce^{B (\sigma+\sigma^5) \tau_1} \Bigg[ \sum_{n=0}^{\ell-1} \frac{\sigma^{2n}}{\ell^{2n} (\ell+1)^{2n}}  \int_{0}^{\rho_0}\int_{\s^2} |\phi_{(n+1)}|^2(0,\rho,\theta,\varphi)\,d\omega d\rho\\
  &+ \frac{\ell!^2 (\ell+1)!^2}{\ell^{2\ell}(\ell+1)^{2\ell}}\sum_{n=\ell}^{N_{\infty}} \frac{\sigma^{2n}}{n!^2(n+1)!^2}  \int_{0}^{\rho_0}\int_{\s^2} |\phi_{(n+1)}|^2(0,\rho,\theta,\varphi)\,d\omega d\rho\\
  &+  \int_0^{\tau_1}\int_{\rho_0}^{2\rho_0}\int_{\s^2} |\phi|^2+|{\mathring{\snabla}}\phi|^2\,d\omega d\rho d\tau+\sum_{\tau'=0,\tau_1}\int_{\rho_0}^{2\rho_0}\int_{\s^2} (|\phi|^2+|{\mathring{\snabla}}\phi|)|_{\tau=\tau'}\,d\omega d\rho\Bigg],
\end{split}
\end{equation}
with $B>0$ a suitably large constant.

It remains to bound the integrals supported in $\rho_0\leq \rho\leq 2\rho_0$. This can be done by applying a standard $T$-energy boundedness estimate (see for example \cite{aretakis1} in the extremal Reissner--Nordstr\"om setting) and a Gr\"onwall inequality, together with the fundamental theorem of calculus to estimate integrals of $|\phi|^2$ in terms of the $T$-energy.  

By relabelling the summation variable, we then conclude that:
\begin{equation*}
\begin{split}
\int_{\Sigma_{\tau=\tau_1}}& \mathbb{T}(T,\mathbf{n}_{\Sigma_{\tau=\tau_1}})[\psi] \,d\mu_{\tau=\tau_1}+\sigma^2\sum_{n=1}^{\ell-1} \frac{\sigma^{2n}}{\ell^{2n}(\ell+1)^{2n}}  (\ell+1)^4\int_{0}^{\rho_0}\int_{\s^2} |\phi_{(n)}|^2(\tau_1,\rho,\theta,\varphi)\,d\omega d\rho\\
&+ \sigma^2\frac{\ell!^2 (\ell+1)!^2}{\ell^{2\ell}(\ell+1)^{2\ell}}\sum_{n=\ell}^{N_{\infty}} \frac{\sigma^{2n}}{n!^2(n+1)!^2} (n+1)^4\int_{0}^{\rho_0}\int_{\s^2} |\phi_{(n)}|^2(\tau_1,\rho,\theta,\varphi)\,d\omega d\rho\\
  \leq&\:  Ce^{B (\sigma+\sigma^5) \tau_1} \Bigg[ \sigma^2 \sum_{n=1}^{\ell-1} \frac{\sigma^{2n}}{\ell^{2n}(\ell+1)^{2n}}  (\ell+1)^4 \int_{0}^{\rho_0}\int_{\s^2} |\phi_{(n)}|^2(0,\rho,\theta,\varphi)\,d\omega d\rho\\
  &+ \sigma^2\frac{\ell!^2 (\ell+1)!^2}{\ell^{2\ell}(\ell+1)^{2\ell}}\sum_{n=\ell}^{N_{\infty}} \frac{\sigma^{2n}}{n!^2(n+1)!^2}  (n+1)^4 \int_{0}^{\rho_0}\int_{\s^2} |\phi_{(n)}|^2(0,\rho,\theta,\varphi)\,d\omega d\rho\\
  &+\int_{\Sigma_{\tau=0}} \mathbb{T}(T,\mathbf{n}_{\Sigma_{\tau=0}})[\psi] \,d\mu_{\tau=0}\Bigg]. \qedhere
\end{split}
\end{equation*}
\end{proof}

\begin{remark}
The estimate in Proposition \ref{prop:physspacegevrey} cannot be made uniform in time $\tau_1$ \emph{without} losing the uniformity in $\kappa$. This is because, in the case for extremal Reissner--Nordstr\"om, for each $\ell$ the terms with $n\geq \ell+2$ on the left-hand side. will grow polynomially in $\tau$ when $\kappa_+=0$ or $\kappa_c=0$. This is a manifestion of the \emph{Aretakis instability}, discovered in \cite{aretakis1, aretakis2}, together with its analogue along null infinity \cite{aag18}.
\end{remark}

\begin{definition}
\label{def:solop}
For $\rho_0>0$ suitably small, the solution operator corresponding to \eqref{eq:waveequation} is the bounded linear operator
\begin{align*}
\mathcal{S}(\tau): \mathbf{H}_{\sigma,\rho_0}&\to \mathbf{H}_{\sigma,\rho_0}\\
(\Psi,\Psi')& \mapsto  (\psi|_{\Sigma_{\tau}}, T\psi|_{S_{\tau}}),
\end{align*}
where $\psi$ is the unique solution to \eqref{eq:waveequation} corresponding to initial data $(\Psi,\Psi')$.
By Proposition \ref{prop:physspacegevrey}, the above definition makes sense.
\end{definition}

By Proposition \ref{prop:physspacegevrey}, there exist constants $B,C>0$ independent of $\kappa_+$, $\kappa_c$ and $\tau$ such that:
\begin{equation}
\label{eq:mainboundS}
||\mathcal{S}(\tau)||^2\leq C e^{B\sigma \tau}.
\end{equation}
\begin{theorem}
The family of solution operators $S(\tau)$ from Definition \ref{def:solop} define a $C_0$-semigroup on $\mathbf{H}_{\sigma,\rho_0}$.
\end{theorem}
\begin{proof}
We omit the proof as it follows directly from the proof of Theorem 3.16 of \cite{warn15} combined with \eqref{eq:mainboundS}. We can apply  \eqref{eq:mainboundS} because initial data $(\Psi,\Psi')$ with $r\Psi\in C^{\infty}(\widehat{\Sigma},\C)$ and $\Psi'\in C^{\infty}(S,\C)$ lead to unique solutions $\psi\in  C^{\infty}(\widehat{\mathcal{R}},\C)$ of \eqref{eq:waveequation}, since the equivalent equations \eqref{eq:eqradfieldhor2} and \eqref{eq:eqradfieldcos2} have coefficients in $C^{\infty}(\widehat{\mathcal{R}},\C)$, so standard local-in-time energy estimates apply. 

To obtain strong continuity, we first make the slight abuse of notation: $||\psi|_{\tau}||_{\mathbf{H}_{\sigma,\rho_0}}=||\mathcal{S}(\tau)(\Psi,\Psi')||_{\mathbf{H}_{\sigma,\rho_0}}$, with $\psi$ the solution to \eqref{eq:waveequation} arising from initial data $(\Psi,\Psi')$. Then it follows from a Minkowski inequality that
\begin{equation*}
	||\mathcal{S}(\tau_1)(\Psi,\Psi')-(\Psi,\Psi')||_{\mathbf{H}_{\sigma,\rho_0}}\leq \int_{0}^{\tau_1} ||T\psi|_{\tau}||_{\mathbf{H}_{\sigma,\rho_0}}\,d\tau\leq C e^{B\sigma \tau_1}\tau_1 ||T\psi|_{0}||_{\mathbf{H}_{\sigma,\rho_0}}.
\end{equation*}
For general $(\Psi,\Psi')\in \mathbf{H}_{\sigma,\rho_0}$, $||T\psi|_{0}||_{\mathbf{H}_{\sigma}}$ need not be finite, since it involves derivatives of $(\Psi,\Psi')$. But if we assume that $(\Psi,\Psi')\in \mathbf{H}_{\sigma',\rho_0}\cap C^{\infty}(\Sigma_0)$, with $\sigma'>\sigma$, then it follows that $||T\psi|_{0}||_{\mathbf{H}_{\sigma,\rho_0}}$ is finite, so $||\mathcal{S}(\tau_1)(\Psi,\Psi')-(\Psi,\Psi')||_{\mathbf{H}_{\sigma,\rho_0}}\to 0$ as $\tau_1\downarrow 0$. Convergence in the general case then follows from a standard approximation argument, using that $\mathbf{H}_{\sigma',\rho_0}$ is dense in $\mathbf{H}_{\sigma,\rho_0}$ if $\sigma'>\sigma$, which in turn follows from considering convolutions with appropriate Gaussian mollifiers, which have sufficiently large Gevrey regularity.
\end{proof}

\begin{definition}
Let
\begin{equation*}
\mathcal{D}_{\sigma}(\mathcal{A}):=\left\{(\Psi,\Psi')\in \mathbf{H}_{\sigma,\rho_0}\,\Big|\, \lim_{\tau \downarrow 0} \frac{\mathcal{S}(\tau)(\Psi,\Psi')-(\Psi,\Psi')}{\tau}\:\: \textnormal{exists in $\mathbf{H}_{\sigma,\rho_0}$ } \right\}
\end{equation*}
and
\begin{align*}
\mathcal{A}: \mathbf{H}_{\sigma,\rho_0}\supseteq \mathcal{D}_{\sigma}(\mathcal{A}) \to & \mathbf{H}_{\sigma,\rho_0}\\
(\Psi,\Psi')\mapsto & \lim_{\tau \downarrow 0} \frac{\mathcal{S}(\tau)(\Psi,\Psi')-(\Psi,\Psi')}{\tau}.
\end{align*}
We refer to $(\mathcal{D}_{\sigma}(\mathcal{A}),\mathcal{A})$ as the infinitesimal generator of the semigroup $\mathcal{S}(\tau)$ and to $\mathcal{D}_{\sigma}(\mathcal{A})$ as the domain of $\mathcal{A}$.
\end{definition}

\begin{theorem}
\label{thm:semigroup}
The infinitesimal generator $(\mathcal{D}_{\sigma}(\mathcal{A}),\mathcal{A})$ satisfies the following properties:
\begin{itemize}
\item[(i)] $\mathcal{D}_{\sigma}(\mathcal{A})$ is dense in $\mathbf{H}_{\sigma,\rho_0}$.
\item[(ii)]$(\mathcal{D}_{\sigma}(\mathcal{A}),\mathcal{A})$ is a closed operator
\item[(iii)] There exists $B>0$ such that the resolvent $(\mathcal{A}-s)^{-1}: \mathbf{H}_{\sigma,\rho_0}\to \mathcal{D}_{\sigma}(\mathcal{A})$ exists and is a bounded linear operator for $\mathfrak{Re}(s)>B$.
\end{itemize}
\end{theorem}
\begin{proof}
The statements (i)--(iii) are standard properties of strongly continuous semigroups, see for example Theorem 2 in \S 7.4 of \cite{evans} and Theorem 11.6.1 of \cite{hille1957}. The statement (iii) follows by invoking additionally \eqref{eq:mainboundS}.
\end{proof}
Consider the differential operator $\hat{{L}}_{s,\ell,\kappa}$. We can take as the domain of $\hat{{L}}_{s,\ell,\kappa}$ to be the function space $\mathcal{D}_{\sigma}(\hat{{L}}_{s,\ell,\kappa})$ defined as the closure of
\begin{equation*}
\{f\in C^{\infty}((r_+,r_c);\C)\,|\, ||f||_{H_{\sigma,\rho_0}}+||\hat{{L}}_{s,\ell,\kappa}(f)||_{H_{\sigma,\rho_0}}<\infty\}
\end{equation*}
under the norm $ ||\cdot||_{H_{\sigma,\rho_0}}+||\hat{{L}}_{s,\ell,\kappa}(\cdot)||_{H_{\sigma,\rho_0}}$.

We define the following restricted operators:
\begin{equation*}
\mathcal{A}_{\ell}=\mathcal{A}|_{\mathcal{D}_{\sigma}(\mathcal{A})\cap V_{\ell}}:\mathcal{D}_{\sigma}(\mathcal{A})\cap V_{\ell}\to \mathbf{H}_{\sigma,\rho_0}\cap V_{\ell}.
\end{equation*}

\begin{proposition}
\label{prop:relationAL}
Let $\kappa_+\geq 0$ and $\kappa_c\geq 0$. 
\begin{itemize}
\item[(i)]
The map
\begin{equation*}
(\mathcal{A}_{\ell}-s)^{-1}: \mathbf{H}_{\sigma,\rho_0}\cap V_{\ell}\to \mathcal{D}_{\sigma}(\mathcal{A})\cap V_{\ell}
\end{equation*}
is a well-defined bounded linear operator and is the inverse of $\mathcal{A}_{\ell}-s$ if and only if
\begin{equation*}
\hat{{L}}^{-1}_{s,\ell,\kappa}:  H_{\sigma,\rho_0}\to \mathcal{D}_{\sigma}(\hat{{L}}_{\kappa,\ell,s})
\end{equation*}
exists as a bounded operator with $\mathcal{D}_{\sigma}(\hat{{L}}_{s,\ell,\kappa})\subset H_{\sigma,1,\rho_0}$.
\item[(ii)]
Let $(\Psi,\Psi')\in \ker (\mathcal{A}-s)$, with $ (\mathcal{A}-s): \mathcal{D}_{\sigma}(\mathcal{A})\to \mathbf{H}_{\sigma,\rho_0}$ and assume that $(\mathcal{A}_{\ell}-s)^{-1}$ exists as well-defined bounded operator for all $\ell\geq L$. Then 
\begin{equation*}
\pi_{\geq L}(\Psi)=\pi_{\geq L}(\Psi')=0.
\end{equation*}
\end{itemize}
\end{proposition}
\begin{proof}

Observe that we can rearrange \eqref{eq:waveequation} to obtain:
\begin{equation*}
\begin{split}
h_{r_+}(2-h_{r_+}D)T^2\psi=&\:2(1-h_{r_+}D)r^{-1}\partial_{r}T(r\psi)+r^{-2}\partial_{r}(Dr^2\partial_{r}{\psi})-(h_+D)'T{\psi}-(r^{-2}\mathring{\slashed{\Delta}}+2\mathfrak{l}^{-2}) {\psi}\\
=&\:-h_{r_+}(2-h_{r_+}D)P_1 (T\psi)+r^{-2}\partial_{r}(Dr^2\partial_{r}{\psi})+(r^{-2}\mathring{\slashed{\Delta}}-2\mathfrak{l}^{-2}) {\psi},
\end{split}
\end{equation*}
where $P_1$ is the following operator:
\begin{equation*}
P_1=-\frac{1}{h_{r_+}(2-h_{r_+}D)}\left(2(1-h_{r_+}D)r^{-1}\partial_{r}(r\cdot)-(h_+D)'(\cdot)\right).
\end{equation*}
Similarly, we can rearrange \eqref{eq:mainphysicalspaceeq} to write
\begin {equation*}
T\partial_{\rho}\phi=-\frac{1}{2}\partial_{\rho}(Dr^{-2}\partial_{\rho}\phi)+\frac{1}{2}(B_0+B_1\rho +B_2\rho^2)\phi-\frac{1}{2}\mathring{\slashed{\Delta}} \phi,
\end{equation*}
and hence,
\begin{align*}
T\phi|_{N_{\tau}}(\tau,\rho,\theta,\varphi)=T\phi|_{S}(\tau,r((\rho_c)_0),\theta,\varphi)- \frac{1}{2}\int_{\rho_0}^{\rho}\left[\partial_{\rho}(Dr^{-2}\partial_{\rho}\phi)-(B_0+B_1\rho +B_2\rho^2)\phi+\mathring{\slashed{\Delta}} \phi\right](\tau,\rho',\theta,\varphi)\,d\rho'.
\end{align*}
By using that $\Psi=\psi|_{\Sigma}$ and $\Psi'=T\psi|_{S}$, we find an explicit expression for $\mathcal{A}$, by using that $\mathcal{A}=T=\partial_{\tau}$ if we use the wave equation \eqref{eq:waveequation} to write $T$, which acts on $\psi$, as an operator acting on initial data $(\Psi,\Psi')$ in the following way:
\begin{equation*}
\mathcal{A}\begin{pmatrix}
\Psi\\
\Psi'
\end{pmatrix}
=
T\begin{pmatrix}
\Psi\\
\Psi'
\end{pmatrix}
= \begin{pmatrix}
T\psi|_{\Sigma_0}\\
T^2\psi|_{S_0}
\end{pmatrix}
= \begin{pmatrix}
\widetilde{\Psi}\\
\widetilde{\Psi}'
\end{pmatrix},
\end{equation*}
where
\begin{align*}
\widetilde{\Psi}|_{S}=&\:\Psi',\\
r\widetilde{\Psi}|_{N}=&\:r\Psi'(r((\rho_c)_0),\theta,\varphi)- \frac{1}{2}\int_{\rho_0}^{\rho}\left[\partial_{\rho}(Dr^{-2}\partial_{\rho}(r\Psi))-(B_0+B_1\rho +B_2\rho^2)r\Psi+\mathring{\slashed{\Delta}}( r\Psi)\right]|_{N}(\rho',\theta,\varphi)\,d\rho',\\
\widetilde{\Psi}'=&\:-P_1( \Psi')+\frac{1}{h_{r_+}(2-h_{r_+}D)}r^{-2}\partial_{r}(Dr^2\partial_{r}{\Psi})+\frac{1}{h_{r_+}(2-h_{r_+}D)}(r^{-2}\mathring{\slashed{\Delta}}-2\mathfrak{l}^{-2}) {\Psi}.
\end{align*}
So if we instead denote
\begin{equation*}
(\mathcal{A}-s)\begin{pmatrix}
\Psi\\
\Psi'
\end{pmatrix}
= \begin{pmatrix}
\widetilde{\Psi}\\
\widetilde{\Psi}'
\end{pmatrix},
\end{equation*}
we can express:
\begin{align*}
\widetilde{\Psi}|_{S}=&\:\Psi'-s\Psi|_S,\\
r\widetilde{\Psi}|_{N}=&\:-s r\Psi|_N+r\Psi'(r((\rho_c)_0),\theta,\varphi)\\
&- \frac{1}{2}\int_{\rho_0}^{\rho}\left[\partial_{\rho}(Dr^{-2}\partial_{\rho}(r\Psi))-(B_0+B_1\rho +B_2\rho^2)r\Psi+\mathring{\slashed{\Delta}}( r\Psi)\right](\rho',\theta,\varphi)\,d\rho',\\
\widetilde{\Psi}'=&\:-P_1( \Psi')+\frac{1}{h_{r_+}(2-h_{r_+}D)}r^{-2}\partial_{r}(Dr^2\partial_{r}{\Psi}|_S)+\frac{1}{h_{r_+}(2-h_{r_+}D)}(r^{-2}\mathring{\slashed{\Delta}}-2\mathfrak{l}^{-2}) {\Psi}|_S-s\Psi'.
\end{align*}

We can further rewrite the above expressions for the restriction $\mathcal{A}_{\ell}-s$ by making use of the operators $\hat{{L}}_{s,\ell,\kappa}$ (acting on the $\Psi_{\ell m}$ part of $\Psi_{\ell}$) :
\begin{align*}
\widetilde{\Psi}|_{S}=&\:\Psi'-s\Psi|_S,\\
r\widetilde{\Psi}|_{N}=&\:-s r\Psi|_N+r\Psi'(r((\rho_c)_0),\theta,\varphi)- \frac{1}{2}\int_{\rho_0}^{\rho}\left[r\hat{{L}}_{s,\ell,\kappa}(\Psi)|_N-2s\partial_{\rho}(r\Psi)|_N\right](\rho',\theta,\varphi)\,d\rho',\\
\widetilde{\Psi}'=&\:\frac{1}{h_{r_+}(2-h_{r_+}D)}r^{-2}\hat{{L}}_{s,\ell,\kappa}(\Psi)|_S+s(P_1+s)(\Psi)|_S-(P_1+s)(\Psi').
\end{align*}

Note in particular that it's possible to write the above relations on $S$ in matrix form:
\begin{equation*}
\begin{pmatrix}
\widetilde{\Psi}|_S\\
\widetilde{\Psi}'
\end{pmatrix}
=\begin{pmatrix}
0 & 1\\
1 & -(P_1+s)
\end{pmatrix}
\begin{pmatrix}
-\frac{1}{h_{r_+}(2-h_{r_+}D)}r^{-2}\hat{{L}}_{s,\ell,\kappa} & 0\\
0 & 1
\end{pmatrix}
\begin{pmatrix}
-1 & 0\\
-s & 1
\end{pmatrix}
\begin{pmatrix}
{\Psi}|_S\\
{\Psi}'
\end{pmatrix}
\end{equation*}
We can formally invert the above matrix to obtain
\begin{equation*}
\begin{pmatrix}
{\Psi}|_S\\
{\Psi}'
\end{pmatrix}
=\begin{pmatrix}
-1 & 1\\
-s& -0
\end{pmatrix}
\begin{pmatrix}
-{h_{r_+}(2-h_{r_+}D)}r^2\hat{{L}}_{s,\ell,\kappa}^{-1} & 0\\
0 & 1
\end{pmatrix}
\begin{pmatrix}
P_1+s & 1\\
1& 0
\end{pmatrix}
\begin{pmatrix}
\widetilde{\Psi}|_S\\
\widetilde{\Psi}'
\end{pmatrix}
\end{equation*}
We moreover have that
\begin{align*}
{\Psi}|_{N}=&\: -2(r\hat{{L}}_{s,\ell,\kappa})^{-1}( \partial_{\rho} (r\widetilde{\Psi}|_{N}))|_{N}.
\end{align*}
The above expressions constitute a formal definition for $(\mathcal{A}_{\ell}-s)^{-1}$, where we identify the functions on $S$ and $N$ with their trivial extensions to $\Sigma$ in order to act with $\hat{{L}}_{s,\ell,\kappa}^{-1}$.

We immediately see that $(\mathcal{A}_{\ell}-s)^{-1}$ is a well-defined bounded operator iff $\hat{{L}}_{s,\ell,\kappa}^{-1}$ exists as a bounded operator with $\mathcal{D}_{\sigma}(\hat{{L}}_{s,\ell,\kappa})\subset H_{\sigma,1,\rho_0}$. The identities $(\mathcal{A}_{\ell}-s)\circ (\mathcal{A}_{\ell}-s)^{-1}=\textnormal{id}_{\mathbf{H}_{\sigma,\rho_0}}$ and $(\mathcal{A}_{\ell}-s)^{-1}\circ (\mathcal{A}_{\ell}-s)=\textnormal{id}_{\mathcal{D}_{\sigma}(\mathcal{A}_{\ell})}$ then follow immediately.

In order to prove (ii) we observe that we can decompose $\Psi=\sum_{\ell=0}^{\infty}\pi_{\ell}(\Psi)$ and $\Psi'=\sum_{\ell=0}^{\infty}\pi_{\ell}(\Psi')$ and $(\pi_{\ell}(\Psi),(\pi_{\ell}(\Psi))\in \ker (\mathcal{A}_{\ell}-s)$.
\end{proof}

\section{Degenerate elliptic estimates}
\label{sec:degelliptic}
In this section we will derive two degenerate elliptic estimates, which are each valid for either small $|s|$ or large angular momentum $\ell$. The purpose of these estimates is to control $H^2$-norms of $\hat{\psi}$ \emph{away} from $\mathcal{H}^+$ and $\mathcal{C}^+$ or $\mathcal{I}^+$, in contrast with the Gevrey estimates of Section \ref{sec:gevrest}, which are restricted to suitably small neighbourhoods of $\mathcal{H}^+$ and $\mathcal{C}^+$ or $\mathcal{I}^+$ and involve higher-order Sobolev norms. The degenerate elliptic estimates and Gevrey estimates are \emph{coupled}. In order to make the coupling constant sufficiently small in the degenerate elliptic estimates below, we will make use of either the smallness of $|s|$ (Proposition \ref{prop:lowfreqellipticest}) or largeness of $\ell$ (Proposition \ref{prop:degellipticest}).

\begin{proposition}
\label{prop:lowfreqellipticest}
Let $\kappa_c,\kappa_+\geq 0$ and let $\rho_0>0$ and $\epsilon>0$ be arbitrarily small. Then there exists a constant $s_0>0$ and constants $C,C_{\epsilon}>0$ that are independent of $s$, $\kappa_+$ and $\kappa_c$, such that for all $\ell\in \N$ and $s\in \C$ such that $(1+\ell(\ell+1))^{-1}(|s|+M|s|^2)<s_0$:
\begin{equation}
\label{eq:lowfreqellipticest}
\begin{split}
\int_{r_+}^{r_c}&M^2r^{-2}|\partial_r(Dr^2\partial_r{\hat{\psi}})|^2+D|\partial_r\hphi|^2+Mr^{-3}\sqrt{D}|\hphi|^2+DM^2(\ell(\ell+1)+2r^2\mathfrak{l}^{-2})|\partial_{r}{\hat{\psi}}|^2\\
&+[\ell(\ell+1)r^{-2}+M^2\ell^2(\ell+1)^2r^{-4}]|\hphi|^2\,dr\\
\leq&\: C|s| \sum_{\star\in \{+,c\}} \int_{0}^{\rho_0} (M+|s|M^2+\epsilon)|\hphi|^2+M^2(M+|s|M^2)|\hphi_{(1)}|^2\,d\rho_{\star}+ C_{\epsilon}\int_{r_+}^{r_c}|\hat{L}_{s,\ell,\kappa}(\hat{\psi})|^2\,dr.
\end{split}
\end{equation}
\end{proposition}
\begin{proof}
We rewrite the equation $\hat{L}_{s,\ell,\kappa}(\hat{\psi})=\tilde{f}$, see \eqref{eq:fixedfreqpsis}, in terms of $\hat{\phi}=r\hat{\psi}$ instead of $\hat{\psi}$: obtain:
 \begin{equation*}
 2s(1-h_{r_+}D)\partial_r\hphi+\partial_r(D\partial_r\hphi)-r^{-2}(2Mr^{-1}-2e^2r^{-2})\hphi-r^{-2}\ell(\ell+1)\hphi-(h_{r_+}(2-h_{r_+}D)s^2+(h_+D)'s)\hphi=r^{-1}\tilde{f},
 \end{equation*}
We multiply both sides by $-\overline{\phi}$, take the real part and rearrange terms to obtain:

\begin{equation*}
\begin{split}
-\overline{\hphi}\partial_r(D\partial_r\hphi)+2r^{-4}(Mr-e^2)|\hphi|^2+\ell(\ell+1)r^{-2}|\hphi|^2=&\:2(1-h_{r_+}D)\re (s \partial_r\hphi \overline{\hphi})-\re(r^{-1}\tilde{f}\overline{\hphi})\\
&-\re(h_{r_+}(2-h_{r_+}D)s^2+(h_+D)'s)|\hphi|^2.
\end{split}
\end{equation*}
Integrating over $[r_+,r_c)$ with $r_c<\infty$ or $r_c=\infty$ and using that the corresponding boundary terms vanish, we obtain after applying Young's inequality on the right-hand side:
\begin{equation*}
\begin{split}
\int_{r_+}^{r_c}&D|\partial_r\hphi|^2+2r^{-4}(Mr-e^2)|\hphi|^2+\ell(\ell+1)r^{-2}|\hphi|^2\,dr\leq |s|M \int_{r_+}^{r_c} r^{-2}|\hphi|^2+M^{-2}r^2|\partial_r\hphi|^2\,dr\\
&+\epsilon \int_{r_+}^{r_c} r^{-2}|\hphi|^2\,dr+ C_{\epsilon}\int_{r_+}^{r_c} |\tilde{f}|^2\,dr,
\end{split}
\end{equation*}
where $Mr-e^2\geq 0$, since $r\geq r_+\geq M$ and $|e|\leq M$. Indeed we have that $D(M)\leq \frac{e^2}{M^2}-1\leq 0$ and $D(r)\geq 0$ for all $r\geq r_+$, so $r_+\geq M$.

From \eqref{eq:fixedfreqpsis} it follows that we can alternatively write:
\begin{equation}
\label{eq:maineqelliptic}
\begin{split}
r^{-1}\tilde{f}=&\: 2(1-h_{r_+}D)s\partial_{r}(r{\hat{\psi}})+r^{-1}\partial_{r}(Dr^2\partial_{r}{\hat{\psi}})+[-h_{r_+}(2-h_{r_+}D)s^2-(h_+D)'s]r{\hat{\psi}}\\
&-r^{-1}(\ell(\ell+1)+2r^2\mathfrak{l}^{-2}) {\hat{\psi}}.
\end{split}
\end{equation}
We moreover denote
\begin{equation*}
r^{-1} F:=-2(1-h_{r_+}D)s\partial_{r}(r{\hat{\psi}})+[h_{r_+}(2-h_{r_+}D)s^2+(h_+D)'s]r{\hat{\psi}}+r^{-1}\tilde{f},
\end{equation*}
then
\begin{equation}
\label{eq:ellipticeqv1}
r^{-1}\partial_{r}(Dr^2\partial_{r}{\hat{\psi}})-r^{-1}(\ell(\ell+1)+2r^2\mathfrak{l}^{-2}) {\hat{\psi}}=r^{-1}F.
\end{equation}
so after taking the square norm on both sides and apply the Leibniz rule several times, we obtain
\begin{equation*}
\begin{split}
r^{-2}|\partial_r(Dr^2\partial_r{\hat{\psi}})|^2&+(4r^2\mathfrak{l}^{-4}+\ell^2(\ell+1)^2r^{-2}+4\ell(\ell+1)\mathfrak{l}^{-2})|{\hat{\psi}}|^2-2\partial_r((\ell(\ell+1)r^{-2}+2\mathfrak{l}^{-2})\re(Dr^2\partial_{r}{\hat{\psi}} \overline{\hat{\psi}})\\
&+2D(\ell(\ell+1)+2r^2\mathfrak{l}^{-2})|\partial_{r}{\hat{\psi}}|^2-4\ell(\ell+1)Dr^{-1}\re(\partial_{r}{\hat{\psi}} \overline{\hat{\psi}})=r^{-2}|F|^2
\end{split}
\end{equation*}
We can estimate
\begin{equation*}
4\ell(\ell+1)Dr^{-1}|\partial_{r}{\hat{\psi}}||\overline{\hat{\psi}}|\leq 2D\ell(\ell+1)|\partial_r{\hat{\psi}}|^2+2\ell(\ell+1)Dr^{-2}|{\hat{\psi}}|^2,
\end{equation*}
so, using moreover that the total derivative term vanishes after integration, we can finally conclude that
\begin{equation*}
\begin{split}
\int_{r_+}^{r_c}&D|\partial_r\hphi|^2+2r^{-4}(Mr-e^2)|\hphi|^2+\ell(\ell+1)r^{-2}|\hphi|^2\,dr\\
&+M^2\int_{r_+}^{r_c}r^{-2}|\partial_r(Dr^2\partial_r{\hat{\psi}})|^2+D(\ell(\ell+1)+2r^2\mathfrak{l}^{-2})|\partial_{r}{\hat{\psi}}|^2\\
&+(4r^2\mathfrak{l}^{-4}+\ell^2(\ell+1)^2r^{-2}+4\ell(\ell+1)\mathfrak{l}^{-2})|{\hat{\psi}}|^2\,dr\\
\leq &\: C\int_{r_+}^{r_c}r^{-2}|F|^2+|\tilde{f}|^2\,dr+ M|s| \int_{r_+}^{r_c} r^{-2}|\hphi|^2+M^{-2}r^2|\partial_r\hphi|^2\,dr+\epsilon \int_{r_+}^{r_c} r^{-2}|\hphi|^2\,dr\\
\leq &\: C\int_{r_+}^{r_c}|\tilde{f}|^2\,dr+ \max\{|s|M,|s|^2M^2\} \int_{r_+}^{r_c} r^{-2}|\hphi|^2+M^{-2}r^2|\partial_r\hphi|^2\,dr+\epsilon \int_{r_+}^{r_c} r^{-2}|\hphi|^2\,dr.
\end{split}
\end{equation*}
For $M|s|(1+\ell(\ell+1))^{-1}$ suitably small compared to $M\rho_0$ and $\epsilon>0$ appropriately small, we can absorb the terms involving $\hphi$ and $\partial_r\hphi$ in the integrals restricted to $(R_+^0,R_c^0)$ on the right-hand side into left-hand side. We are then left with
\begin{equation*}
\begin{split}
\int_{r_+}^{r_c}&D|\partial_r\hphi|^2+2r^{-4}(Mr-e^2)|\hphi|^2+\ell(\ell+1)r^{-2}|\hphi|^2\,dr\\
&+M^2\int_{r_+}^{r_c}r^{-2}|\partial_r(Dr^2\partial_r{\hat{\psi}})|^2+D(\ell(\ell+1)+2r^2\mathfrak{l}^{-2})|\partial_{r}{\hat{\psi}}|^2\\
&+(4r^2\mathfrak{l}^{-4}+\ell^2(\ell+1)^2r^{-2}+4\ell(\ell+1)\mathfrak{l}^{-2})|{\hat{\psi}}|^2\,dr\\
\leq &\: C\int_{r_+}^{r_c}|\tilde{f}|^2\,dr+C\sum_{\star\in \{+,c\}} \int_{0}^{\rho_0} (M|s|+|s|^2M^2+\epsilon)|\hphi|^2+M^2(|s|M+|s|^2M^2)|\hphi_{(1)}|^2\,d\rho_{\star}.
\end{split}
\end{equation*}
\end{proof}

\begin{proposition}
\label{prop:degellipticest}
Fix $0<\rho_0=({\rho}_+)_0=({\rho}_c)_0<\frac{r_c-r_+}{r_c r_+}$. Let $K\in \N$, $K\geq 2$.  Let $\kappa_+,\kappa_c>0$ be suitably small compared to $\rho_0$ and $K$. Consider a smooth function $g:[r_+,r_c]\to \R_{\geq 0}$ which satisfies:
\begin{equation*}
g(r)=\left\{\begin{matrix}
1 & \textnormal{for $R_0^+\leq r \leq R^c_0$},\\
M^2\rho_c(r)\rho_+(r)& \textnormal{for $r\leq r((\frac{1}{2}\rho_+)_0)$ and $r\geq r(\frac{1}{2}(\rho_c)_0)$}.
\end{matrix}\right.
\end{equation*}
Furthermore, let $\chi: [r_+,r_c]\to \R_{\geq 0}$ be a cut-off function supported in $[r(\frac{1}{2K}({\rho}_+)_0),r(\frac{1}{2K}({\rho}_c)_0)]$ such that $\chi(r)=1$ for all $r\in [r(\frac{1}{K}({\rho}_+)_0),r(\frac{1}{K}({\rho}_c)_0)]$. 

Then there exists a constant $c_0>0$ depending only on $\chi$, such that for 
\begin{equation*}
k^2\leq c_0 M^3\rho_0^3\ell(\ell+1)
\end{equation*}
and
\begin{equation}
\label{eq:condlelliptic}
\ell(\ell+1){\gg}|s|^2(1+M^2|s|^2) K^2\rho_0^{-2}
\end{equation}
we have that for all $\kappa_+,\kappa_c$ suitably small compared to $\ell$:
\begin{equation}
\label{eq:ellipticestext}
\begin{split}
\int_{r(\frac{1}{2K}({\rho}_+)_0)}^{r(\frac{1}{2K}({\rho}_c)_0)}& g^k D^2(\ell(\ell+1)+2r^2\mathfrak{l}^{-2}) |\partial_r (\chi {\hat{\psi}})|^2+  r^{-2}g^k(\ell(\ell+1)+2r^2\mathfrak{l}^{-2})^2D \chi^2|{\hat{\psi}}|^2\,dr\\
\leq&\: C(\rho_c)_0^{-2}(K M^{-1}(\rho_c)_0^{-1})^{-k}\int_{\frac{1}{2K}(\rho_c)_0}^{\frac{1}{K}(\rho_c)_0}|\hat{\phi}|^2+ |\partial_{\rho }\hat{\phi} |^2 \,d\rho_c\\
&+ C(\rho_+)_0^{-2}(K M^{-1}(\rho_+)_0^{-1})^{-k}\int_{\frac{1}{2K}(\rho_+)_0}^{\frac{1}{K}(\rho_+)_0} |\hat{\phi}|^2+ |\partial_{\rho }\hat{\phi} |^2 \,d\rho_+\\
&+C \int_{r(\frac{1}{2K}(\rho_+)_0)}^{r(\frac{1}{2K}(\rho_c)_0)}g^{k}Dr^{-2}\chi^2|\hat{{L}}_{s,\ell,\kappa}({\hat{\psi}})|^2\,dr,
\end{split}
\end{equation}
where $C>0$ is a constant that is independent of $\ell$, $\kappa_+$ and $\kappa_c$.

In particular, for any constant $K_0>0$, and $\ell$ satisfying \eqref{eq:condlelliptic} with $K$ replaced by $K_0$, we can estimate
\begin{equation}
\label{eq:ellipticestextv2}
\begin{split}
\int_{R_0^+}^{R^c_0}& r^{-2}|\partial_{r}(Dr^2\partial_{r}{\hat{\psi}})|^2+D^2(\ell(\ell+1)+2r^2\mathfrak{l}^{-2}) |\partial_{r}{\hat{\psi}}|^2+  Dr^{-2}(\ell(\ell+1)+2r^2\mathfrak{l}^{-2})^2 |{\hat{\psi}}|^2\,dr\\
\leq&\: K_0^{-\ell}\int_{0}^{(\rho_c)_0}  |\hat{\phi}|^2+ |\partial_{\rho}\hat{\phi}|^2 \,d\rho_c+ K_0^{-\ell}\int_{0}^{(\rho_+)_0} |\hat{\phi}|^2+ |\partial_{\rho}\hat{\phi}|^2 \,d\rho_+\\
&+C \int_{R^+_0}^{R^c_0}Dr^{-2}|\hat{{L}}_{s,\ell,\kappa}({\hat{\psi}})|^2\,dr+C\int_{0}^{(\rho_c)_0} |r\hat{{L}}_{s,\ell,\kappa}({\hat{\psi}})|^2 \,d\rho_c+ C\int_{0}^{(\rho_+)_0}|r\hat{{L}}_{s,\ell,\kappa}({\hat{\psi}})|^2 \,d\rho_+,
\end{split}
\end{equation}
with $C>0$ a constant depending in particular on $K_0$, but independent of $\ell$.
\end{proposition}

\begin{remark}
It is important that the coupling constant $K_0^{-\ell}$ on the right-hand side of \eqref{eq:ellipticestextv2} depends \emph{exponentially} on $\ell$ and moreover $K_0$ can be chosen arbitrarily large if $\ell$ is taken appropriately large, because of the competition with an exponentially growing constant on the right-hand side of \eqref{eq:maingevreysum} in Corollary \ref{cor:maingevreysumbound}.
\end{remark}
\begin{proof}[Proof of Proposition \ref{prop:degellipticest}]
Let $\tilde{f}=\hat{{L}}_{s,\ell,\kappa}({\hat{\psi}})$. We denote
\begin{equation*}
r^{-1} F:=-2(1-h_{r_+}D)s\partial_{r}(r{\hat{\psi}})+[h_{r_+}(2-h_{r_+}D)s^2+(h_+D)'s]r{\hat{\psi}}+r^{-1}\tilde{f}.
\end{equation*}
By \eqref{eq:fixedfreqpsis}, we have that
\begin{equation}
\label{eq:ellipticeq}
r^{-1}\partial_{r}(Dr^2\partial_{r}{\hat{\psi}})-r^{-1}(\ell(\ell+1)+2r^2\mathfrak{l}^{-2}) {\hat{\psi}}=r^{-1}F.
\end{equation}
Then, by applying \eqref{eq:ellipticeq}, we can estimate, for $0<\eta<1$ arbitrarily small,
\begin{equation}
\label{eq:ellipticeqest}
r^{-2}|\partial_{r}(Dr^2\partial_{r}{\hat{\psi}})|^2\geq (1-\eta)r^{-2}(\ell(\ell+1)+2r^2\mathfrak{l}^{-2})^2 |{\hat{\psi}}|^2-C_{\eta}r^{-2}|F|^2,
\end{equation}
with $C_{\eta}>0$ an appropriately large numerical constant depending on $\eta>0$. 

Let $g: (r_+,r_c)\to \R_{>0}$ be a smooth weight function that we will specify below. Then we can take the square norm of both sides of \eqref{eq:ellipticeq}, multiply both sides by the weight function $Dg^k$ and apply \eqref{eq:ellipticeqest} to arrive at the following inequality:
\begin{equation*}
(2-\eta) Dg^k r^{-2}(\ell(\ell+1)+2r^2\mathfrak{l}^{-2})^2 |{\hat{\psi}}|^2-2\re(r^{-2}Dg^k\partial_{r}(Dr^2\partial_{r}{\hat{\psi}})(\ell(\ell+1)+2r^2\mathfrak{l}^{-2}) \overline{{\hat{\psi}}})\leq C_{\eta}r^{-2}Dg^k|F|^2,
\end{equation*}
where $k\in \N_0$.

By applying the Leibniz rule, we can further estimate
\begin{equation}
\label{eq:mainidentdegelliptic}
\begin{split}
-2\re(r^{-2}&Dg^k\partial_{r}(Dr^2\partial_{r}{\hat{\psi}})(\ell(\ell+1)+2r^2\mathfrak{l}^{-2}) \overline{{\hat{\psi}}})=-2\partial_{r}\left(\re\left(r^{-2}Dg^kDr^2\partial_{r}{\hat{\psi}}(\ell(\ell+1)+2r^2\mathfrak{l}^{-2}) \overline{{\hat{\psi}}}\right)\right)\\
&+2g^kD^2(\ell(\ell+1)+2r^2\mathfrak{l}^{-2}) |\partial_{r}{\hat{\psi}}|^2+2kg^{k-1} g'D^2\re\left(\partial_{r}{\hat{\psi}}\left(\ell(\ell+1)+2r^2\mathfrak{l}^{-2}\right) \overline{{\hat{\psi}}}\right)\\
&+[2D'D(\ell(\ell+1)+2r^2\mathfrak{l}^{-2})-4\ell(\ell+1) D^2 r^{-1}] g^k\re(\partial_{r}{\hat{\psi}}\overline{{\hat{\psi}}}).
\end{split}
\end{equation}

We now take the weight function $g$ to satisfy the following properties: $g(r)=1$ for $R_0^+\leq r \leq R^c_0$ and $g(r)=M^2\rho_c(r)\rho_+(r)$ for $r\leq r((\frac{1}{2}\rho_+)_0)$ and $r\geq r(\frac{1}{2}(\rho_c)_0)$, so that $g'(r)=0$ for $R_0^+\leq r \leq R_0^c$ and $g'(r)=-M^2r^{-2} (\rho_+-\rho_c)$  for $r\leq r((\frac{1}{2}\rho_+)_0)$ and $r\geq r(\frac{1}{2}(\rho_c)_0)$. 

Note that
\begin{equation*}
\begin{split}\
&\left|-2kg^{k-1}g' D^2\re\left(\partial_{r}\hat{\psi}(\ell(\ell+1)+2r^2\mathfrak{l}^{-2}) \overline{\hat{\psi}}\right)\right|\\
\leq&\: (2-2\eta)  g^k(\ell(\ell+1)+2r^2\mathfrak{l}^{-2})^2 Dr^{-2}|\hat{\psi}|^2+\frac{k^2}{(2-2\eta)\ell(\ell+1)} Dr^{-2}|g^{-1}\partial_{\rho}g |^2g^kD^2 \ell(\ell+1)|\partial_{r}\hat{\psi}|^2,
\end{split}
\end{equation*}
where $\partial_{\rho}=\partial_{\rho_+}$ in the region $r\leq R_0^+$ and $\partial_{\rho}=\partial_{\rho_c}$ in the region $r\geq R_0^c$.

By using standard properties of cut-off functions, we can infer that there exists a $c_0>0$ depending only on the choice of cut-off function appearing in the construction of $g$, such that 
\begin{equation*}
Dr^{-2}|g^{-1}\partial_{\rho}g|^2\leq c_0^{-1} \rho_0^{-3}.
\end{equation*}

Given $\epsilon'>0$ arbitrarily small, we can moreover estimate
\begin{equation*}
\begin{split}
&\left|\left[2D'D(\ell(\ell+1)+2r^2\mathfrak{l}^{-2})-4\ell(\ell+1) D^2 r^{-1}\right] g^k\re\left(\partial_{r}\hat{\psi}\overline{\hat{\psi}}\right)\right| \leq 2\epsilon' D^2g^{k} \ell(\ell+1) |\partial_{r}\hat{\psi}|^2\\
&+\epsilon' D^2g^{k} (\ell(\ell+1)+2r^2\mathfrak{l}^{-2})) |\partial_{r}\hat{\psi}|^2 + \frac{\epsilon'^{-1}}{\ell(\ell+1)} \cdot g^{k}(\ell(\ell+1)+2r^2\mathfrak{l}^{-2})^2 r^{-2} D^2|\hat{\psi}|^2\\
&+\frac{\epsilon'^{-1}}{\ell(\ell+1)} \cdot g^{k}(\ell(\ell+1)+2r^2\mathfrak{l}^{-2})^2 (rD')^2r^{-2} |\hat{\psi}|^2 .
\end{split}
\end{equation*}
Note that we can write $D'r= (Dr^{-2})'r^3+2D$, and using \eqref{eq:Devent}, we can therefore estimate
\begin{equation*}
\begin{split}
 (rD')^2\leq &\: 2r^2 |\partial_{\rho_+}(Dr^{-2})|^2+8D^2\\
 \leq &\: Cr^2\rho_+^2+C\kappa_+ \rho_+ r^2+C\kappa_+^2r^2
\end{split}
\end{equation*}
The above estimate holds also with $\rho_+$ and $\kappa_+$ replaced with $\rho_c$ and $\kappa_c$, so we can then estimate:
\begin{equation*}
(rD')^2\leq C  D+ C(\kappa_+^2+\kappa_c^2)r^2.
\end{equation*}

Now assume $k^2\leq c_0 M^3\rho_0^3 \ell(\ell+1)$ and choose $\eta,\epsilon'>0$ appropriately small and $\ell$ appropriately large. We then combine the above estimates to obtain:
\begin{equation}
\label{eq:keyestelliptic}
\begin{split}
g^kDr^{-2}&(\ell(\ell+1)+2r^2\mathfrak{l}^{-2})^2 |{\hat{\psi}}|^2+g^kD^2(\ell(\ell+1)+2r^2\mathfrak{l}^{-2}) |\partial_r{\hat{\psi}}|^2+\textnormal{total derivative}\\
\leq&\: C Dr^{-2}g^{k}|F|^2+C(\kappa_+^2+\kappa_c^2) \ell^{-2}\cdot g^{k}(\ell(\ell+1)+2r^2\mathfrak{l}^{-2})^2 |\hat{\psi}|^2.
\end{split}
\end{equation}

Rather than integrating \eqref{eq:keyestelliptic} directly from $r_+$ to $r_c$, we will first introduce a cut-off function $\chi: [r_+,r_c] \to \R$. Let $K\in \N$, $K\geq 2$. Take the cut-off function $\chi$ to satisfy:
\begin{align*}
\chi(r)&\:\geq 0\quad \textnormal{for all $r\in [r_+,r_c]$}, \\
\chi(r)&\:=0 \quad  \textnormal{if $r\leq r\left(\frac{1}{2}K^{-1}(\rho_+)_0\right)$ or $r\geq r\left(\frac{1}{2}K^{-1}(\rho_c)_0\right)$},\\
\chi(r)&\:=1 \quad  \textnormal{if  $r(K^{-1}(\rho_+)_0)\leq r\leq  r(K^{-1}(\rho_c)_0)$}.
\end{align*}

If we replace ${\hat{\psi}}$ in the above estimates by $\chi\cdot {\hat{\psi}}$, then we have to replace $r^{-1}F$ by
\begin{equation*}
\chi\cdot  r^{-1}F-\chi'' Dr {\hat{\psi}}-r^{-1}\partial_{r}(Dr^2)\chi' {\hat{\psi}}-2Dr \chi' \partial_{r}{\hat{\psi}},
\end{equation*}
and we can apply \eqref{eq:keyestelliptic} to obtain
\begin{equation}
\label{eq:keyestelliptic2}
\begin{split}
g^kDr^{-2}&(\ell(\ell+1)+2r^2\mathfrak{l}^{-2})^2\chi^2 |{\hat{\psi}}|^2+g^kD^2(\ell(\ell+1)+2r^2\mathfrak{l}^{-2}) |\partial_r(\chi{\hat{\psi}})|^2+\textnormal{total derivative}
\\
\leq&\: CDr^{-2}\chi^2 g^{k}|F|^2+ C D(D|\chi''|^2+r^{-2}|\chi'|^2) g^{k}|\hat{\phi}|^2+ CD r^{-2} |\chi'|^2   g^{k}|\partial_{\rho} \hat{\phi}|^2\\
&+C\chi^2(\kappa_+^2+\kappa_c^2) \ell^{-2}\cdot g^{k}(\ell(\ell+1)+2r^2\mathfrak{l}^{-2})^2 |\hat{\psi}|^2.
\end{split}
\end{equation}
We have that
\begin{equation*}
\chi^2 Dr^{-2}|F|^2\leq C Dr^2|s|^2 |\partial_{r}(\chi {\hat{\psi}})|^2 +C(|\chi'|^2r^2+1)(|s|^2+M^2|s|^4)D|{\hat{\psi}}|^2+CDr^{-2}\chi^2|\tilde{f}|^2.
\end{equation*}

If 
\begin{equation}
\label{eq:conditionl}
\ell(\ell+1){\gg} |s|^2(1+M^2|s|^2)\max\left \{(D^{-1}r^2)\left(\frac{1}{2}K^{-1}\left(\rho_+\right)_0\right),(D^{-1}r^2)\left(\frac{1}{2}K^{-1}(\rho_c)_0\right)\right\}, 
\end{equation}
then we can use that the cut-off $\chi$ allows us to restrict $r\in [r\left(\frac{1}{2}K^{-1}(\rho_+)_0\right),r\left(\frac{1}{2}K^{-1}(\rho_c)_0\right)]$ in order to absorb part of the $|F|^2$ term into the left-hand side of \eqref{eq:keyestelliptic2} to be left with
\begin{equation}
\label{eq:keyestelliptic3}
\begin{split}
g^{k} &r^{-2}(\ell(\ell+1)+2r^2\mathfrak{l}^{-2})^2 D\chi^2|{\hat{\psi}}|^2+ g^{k}D^2(\ell(\ell+1)+2r^2\mathfrak{l}^{-2}) |\partial_{r}(\chi {\hat{\psi}})|^2\\
&+\textnormal{total derivative}\\
\leq&\: CDr^{-2}\chi^2|\tilde{f}|^2+ CD (D|\chi''|^2+r^{-2}|\chi'|^2)|\hat{\phi}|^2+ C Dr^{-2}|\chi'|^2  |\partial_{\rho }\hat{\phi} |^2\\
&+C\chi^2(\kappa_+^2+\kappa_c^2)\ell^{-2}\cdot g^{k}(\ell(\ell+1)+2r^2\mathfrak{l}^{-2})^2 |\hat{\psi}|^2.
\end{split}
\end{equation}

Finally, we integrate \eqref{eq:keyestelliptic3} over $[r((\frac{1}{2}K^{-1}(\rho_+)_0)),r((\frac{1}{2}K^{-1}(\rho_c)_0))]$ and use that we can bound
\begin{align*}
Dr^{-2} |\chi'|^2+D^2|\chi''|^2\lesssim &\: r^{-4}+\kappa_+\rho_+^{-1}r^{-4}+\kappa_c\rho_c^{-1}r^{-4}+\kappa_+^2\rho_+^{-2}r^{-4}+\kappa_c^2\rho_c^{-2}r^{-4}
\end{align*}
to obtain
\begin{equation*}
\begin{split}
\int_{r(\frac{1}{2K}({\rho}_+)_0)}^{r(\frac{1}{2K}({\rho}_c)_0)}& g^k D^2(\ell(\ell+1)+2r^2\mathfrak{l}^{-2}) |\partial_r (\chi {\hat{\psi}})|^2+  r^{-2}g^k(\ell(\ell+1)+2r^2\mathfrak{l}^{-2})^2D \chi^2|{\hat{\psi}}|^2\,dr\\
\leq&\: C(\rho_c)_0^{-2}(K M^{-1}\rho_0^{-1})^{-k}\int_{\frac{1}{2K}(\rho_c)_0}^{\frac{1}{K}(\rho_c)_0}|\hat{\phi}|^2+ |\partial_{\rho }\hat{\phi} |^2 \,d\rho_c\\
&+ C(\rho_+)_0^{-2}(K M^{-1}\rho_0^{-1})^{-k}\int_{\frac{1}{2K}(\rho_+)_0}^{\frac{1}{K}(\rho_+)_0} |\hat{\phi}|^2+ |\partial_{\rho }\hat{\phi} |^2 \,d\rho_+\\
&+C \int_{r(\frac{1}{2K}(\rho_+)_0)}^{r(\frac{1}{2K}(\rho_c)_0)}g^{k}Dr^{-2}\chi^2|\tilde{f}|^2\,dr,
\end{split}
\end{equation*}
where $C$ is a constant that depends on $K$ but is independent of $\ell$ and $\kappa_+,\kappa_c$ are chosen suitably small, depending on $\rho_0$ and $K$. We conclude that \eqref{eq:ellipticestext} holds.

In order to prove \eqref{eq:ellipticestextv2}, we first observe that $K>0$ can be taken arbitrarily large (provided $\ell$ is taken suitably large according to \eqref{eq:conditionl}) so we simply make the integration interval on the left-hand side of \eqref{eq:ellipticestext} smaller and the integration interval on the right-hand side larger, and we use that $g\equiv 1$ on $[R^+_0,R^c_0]$. Finally, we can include the term $r^{-2}|\partial_{r}(Dr^2\partial_{r}{\hat{\psi}})|^2$ that appears on the left-hand side of \eqref{eq:ellipticestextv2} by applying additionally \eqref{eq:maineqelliptic}.

\end{proof}

\section{Main Gevrey estimates}
\label{sec:gevrest}
In this section, \textbf{we prove the key estimate of the paper}, namely a $L^2$-based Gevrey estimate near $\mathcal{H}^+$ and $\mathcal{C}^+$ or $\mathcal{I}^+$ under the assumption of suitably large angular momentum $\ell$.

\begin{theorem}
\label{thm:mainthmpaper}
Let ${\hat{\psi}}\in C^{\infty}([r_+,r_c])$. Fix $\sigma \in \R_{>0}$, $\ell\in \N_0$ and $(\rho_+)_0=(\rho_c)_0=\rho_0>0$. Let $\kappa_+=\kappa_c=\kappa>0$ or let $\kappa=0$ and assume additionally that ${\hat{\psi}} \in H_{\sigma,2,\rho_0}$. Let $s\in \Omega_{\sigma}\subset \C$, where
\begin{equation*}
\Omega_{\sigma}=\left\{ z\in \C\:|\: \re z<0,\: |z|<\sigma,\: 3(\im z)^2-5(\re z)^2>{\sigma}^2\right\}\cup \{ z\in \C\:|\: \re z \geq 0, |z|>0\}.
\end{equation*}
Then for all $\rho_0>0$ suitably small, there exist $\lambda_0\in \N$ suitably large such that for all $\lambda\geq \lambda_0$ and for $\kappa$ suitably small depending on $\ell$ and $\lambda$, we can find constants $C_{\ell,\lambda},C_{s,\rho_0}>0$, such that we can estimate

\begin{equation}
\label{eq:mainestpaper}
\begin{split}
\int_{R^+_0}^{R^c_0}& |{\hat{\psi}}|^2+|\partial_r{\hat{\psi}}|^2+|\partial_r^2{\hat{\psi}}|^2\,dr+\frac{1}{C_{\ell,\lambda}}\sum_{n=0}^{\infty} \frac{\sigma^{2n}}{(n+1)!^2n!^2} n^2(n+1)^2\int_{0}^{(\rho_+)_0} |\partial_{\rho_+}^n(r{\hat{\psi}})|^2+\rho^4|\partial_{\rho_+}^{n+1}(r{\hat{\psi}})|^2\, d{\rho_+}\\
&+\frac{1}{C_{\ell,\lambda}}\sum_{n=0}^{\infty} \frac{\sigma^{2n}}{(n+1)!^2n!^2} n^2(n+1)^2\int_{0}^{(\rho_c)_0} |\partial_{\rho_c}^n(r{\hat{\psi}})|^2+ \rho^4|\partial_{\rho_c}^{n+1}(r{\hat{\psi}})|^2\, d{\rho_c}\\
\leq &\: C_{ s,\rho_0}\int_{R_0^+}^{R^c_0} |{L}_{s,\ell+\lambda, \kappa} ({\hat{\psi}})|^2 \,dr +C_{ s,\rho_0} \sum_{n=0}^{\infty} \frac{\sigma^{2n}}{(n+1)!^2n!^2} \int_{0}^{(\rho_+)_0} |\partial_{\rho_+}^n(r {L}_{s,\ell+\lambda,\kappa} ({\hat{\psi}}))|^2\, d{\rho_+}\\
&+C_{s,\rho_0} \sum_{n=0}^{\infty} \frac{\sigma^{2n}}{(n+1)!^2n!^2} \int_{0}^{(\rho_c)_0} |\partial_{\rho_c}^n(r {L}_{s,\ell+\lambda,\kappa} ({\hat{\psi}}))|^2\, d{\rho_c}.
\end{split}
\end{equation}
\end{theorem}

For the convenience of the reader, we provide an outline of the main steps involved in the proof of  Theorem \ref{thm:mainthmpaper}.\\
\\
\emph{Outline of the proof of Theorem \ref{thm:mainthmpaper}:}
\begin{itemize}
\item[\textbf{Step 1}:] We restrict to a single region of the form $\{r\leq R_0^+ \}$ or $\{r\geq R_0^c\}$, with $\rho_0:=(\rho_+)_0=(\rho_c)_0$, where $0<\rho_0<\frac{r_c-r_+}{r_cr_+}$ will be chosen suitably small, and we derive estimates for the modified quantity $\hat{\Phi}_{(n)}$ in terms of $f_n$, which is defined in \eqref{def:fn} and appears on the left-hand side of \eqref{eq:fixedfreqpsihat}. We start by proving degenerate $H^2$-type estimates for all $n\geq \ell$ (\textbf{see Proposition \ref{eq:maingevprop}}). These can be viewed as the main (fixed-frequency) vector field multiplier estimates in the near-horizon regions.
\item[\textbf{Step 2}:] We consider the estimates from Proposition \ref{eq:maingevprop} and sum over $n$, starting from $n=\ell$ with appropriate $n$-dependent weights that allow us to absorb terms into (lower) $n-k$-order or (higher) $n+k$-order estimates with $k>0$ and we use the non-degeneracy of the wave operator \emph{away from the horizons} to express the sum over the boundary terms at $\rho=\rho_0$ in terms of $\hat{\Phi}(\rho_0)$ and $\hat{\Phi}_{(1)}(\rho_0)$ (\textbf{see Lemma \ref{lm:analytawayhor} and Proposition \ref{prop:gevreyestsum}}).
\item[\textbf{Step 3}:] We increase the summation from $n\geq \ell$ to $n\geq 0$ (\textbf{see Proposition \ref{prop:maingevreysum}}).
\item[\textbf{Step 4}:] We show that for $s\in \Omega_{\sigma}$ all but the top order terms on the left-hand side of Proposition \ref{prop:maingevreysum} are non-negative definite (\textbf{see Lemma \ref{lm:allowedvaluess}}) and show that we can moreover absorb the top order terms if we assume $\kappa>0$, making use of the \emph{enhanced red-shift effect} (\textbf{see Corollary \ref{cor:maingevreysumbound}}).
\item[\textbf{Step 5}:] We couple the estimate from Corollary \ref{cor:maingevreysumbound} to the estimates in Proposition \ref{prop:degellipticest} in order to get rid of the $\rho=\rho_0$ boundary terms (\textbf{see Proposition \ref{prop:maingevreysumv2}}).
\item[\textbf{Step 6}:] We finish the proof of Theorem \ref{thm:mainthmpaper} by converting the estimates for the modified higher-order quantities $\hat{\Phi}_{(n)}$ from Proposition \ref{prop:maingevreysumv2} into estimates for $\partial_{\rho}^n \hat{\phi}$  (\textbf{see Corollary \ref{cor:maingevreysum3}}).
\end{itemize}
\pagebreak
\begin{proposition}
\label{eq:maingevprop}
Let $s\neq 0$. For all $\epsilon>0$, there exist $\ell\in \N_0$ and $\rho_0$, $\gamma$ satisfying
\begin{equation*}
M^{-1}\tilde{\rho}_0^{-1}{\gg}\:\max\{1,\gamma\}\quad \textnormal{and}\quad \gamma{\gg}\:(M^{-1}|s|^{-1}+M |s|),
\end{equation*}
such that the following estimates hold when $\kappa_c+\kappa_c\lesssim (\ell+1)^{-1}$. Let $\ell\leq n\leq 2\ell$, then for all $\alpha,\beta,\nu>0$, $\mu\in [0,1]$ and $\hat{\Phi}\in C^{\infty}([0,\tilde{\rho}_{0}])$
\begin{equation*} 
\begin{split}
\int_{0}^{\tilde{\rho}_0}& (1-\mu \nu^{-1}(1+M\tilde{\rho})^4-\epsilon)(2\kappa\tilde{\rho}+\tilde{\rho}^2)^2e^{2\gamma M \sqrt{\ell(\ell+1)} \tilde{\rho} }|\hat{\Phi}_{(n+2)}|^2\, d\tilde{\rho}\\
&+\int_{0}^{\tilde{\rho}_0} \left(1-4 \beta^{-1}\left(\frac{\max\{-\re(s),0\}}{|s|}\right)^2-\nu \mu-\epsilon \right)|2s|^2e^{2\gamma M \sqrt{\ell(\ell+1)} \tilde{\rho} }|\hat{\Phi}_{(n+1)}|^2\, d\tilde{\rho}\\
&+\int_{0}^{\tilde{\rho}_0}  (4-\beta-\alpha(1-\mu)-\epsilon)(n+1)^2(\tilde{\rho}+\kappa)^2e^{2\gamma M \sqrt{\ell(\ell+1)} \tilde{\rho} }|\hat{\Phi}_{(n+1)}|^2 \, d\tilde{\rho}\\
&- \int_{0}^{\tilde{\rho}_0} \frac{|2s|^2\widetilde{\sigma}^2}{(n+1)^2(n+2)^2}\cdot \alpha^{-1}\widetilde{\sigma}^{-2}(1-\mu)(n+2)^2(1+M\tilde{\rho})^4(\tilde{\rho}^2+2\kappa\tilde{\rho})e^{2\gamma M \sqrt{\ell(\ell+1)} \tilde{\rho} }|\hat{\Phi}_{(n+1+1)}|^2 \, d\tilde{\rho}\\
&-\int_{0}^{\tilde{\rho}_0}\frac{[n(n+1)-\ell(\ell+1)]^2}{|2s|^2\widetilde{\sigma}^2} (1+\epsilon)\widetilde{\sigma}^2 |2s|^2e^{2\gamma M \sqrt{\ell(\ell+1)} \tilde{\rho}}|\hat{\Phi}_{(n-1+1)}|^2 \, d\tilde{\rho}\\
\leq &\: C_{\epsilon} (\ell+1)^{-3}  \sum_{k=1}^{n}\gamma^{\frac{1}{2} (n-k)} \frac{(n+1)!^2}{(k-1)!^2} e^{2\gamma M \sqrt{\ell(\ell+1)} \tilde{\rho}_0}|\hat{\Phi}_{(k)}|^2(\tilde{\rho}_0)\\
&+C_{\epsilon} (\ell+1)^{-1}\sum_{k=0}^{n}\gamma^{\frac{1}{2} (n-k)}   \frac{n!^2}{k!^2} e^{2\gamma M \sqrt{\ell(\ell+1)} \tilde{\rho}_0}|\hat{\Phi}_{(k)}|^2(\tilde{\rho}_0)+C_{\epsilon}\int_{0}^{\tilde{\rho}_0}e^{2\gamma M \sqrt{\ell(\ell+1)} \tilde{\rho}}|f_n|^2\, d\tilde{\rho},
\end{split}
\end{equation*}
with $C_{\epsilon}>0$ a constant that is independent of $\ell$.

Let $n>2\ell$, then
\begin{equation*}
\begin{split}
\int_{0}^{\tilde{\rho}_0}& (1-\mu \nu^{-1}(1+M\tilde{\rho})^4-\epsilon)(2\kappa\tilde{\rho}+\tilde{\rho}^2)^2e^{2\gamma M \sqrt{\ell(\ell+1)} \tilde{\rho} }|\hat{\Phi}_{(n+2)}|^2\, d\tilde{\rho}\\
&+\int_{0}^{\tilde{\rho}_0} \left(1-4 \beta^{-1}\left(\frac{\max\{-\re(s),0\}}{|s|}\right)^2-\nu \mu-\epsilon \right)|2s|^2e^{2\gamma M \sqrt{\ell(\ell+1)} \tilde{\rho} }|\hat{\Phi}_{(n+1)}|^2\, d\tilde{\rho}\\
&+\int_{0}^{\tilde{\rho}_0}  (4-\beta-\alpha(1-\mu)-\epsilon)(n+1)^2(\tilde{\rho}+\kappa)^2e^{2\gamma M \sqrt{\ell(\ell+1)} \tilde{\rho} }|\hat{\Phi}_{(n+1)}|^2 \, d\tilde{\rho}\\
&- \int_{0}^{\tilde{\rho}_0} \frac{|2s|^2\widetilde{\sigma}^2}{(n+1)^2(n+2)^2}\cdot \alpha^{-1}\widetilde{\sigma}^{-2}(1-\mu)(n+2)^2(1+M\tilde{\rho})^4(\tilde{\rho}^3+2\kappa\tilde{\rho}^2)e^{2\gamma M \sqrt{\ell(\ell+1)} \tilde{\rho} }|\hat{\Phi}_{(n+1+1)}|^2 \, d\tilde{\rho}\\
&-\int_{0}^{\tilde{\rho}_0}\frac{[n(n+1)-\ell(\ell+1)]^2}{|2s|^2\widetilde{\sigma}^2} (1+\epsilon)\widetilde{\sigma}^2 |2s|^2e^{2\gamma M \sqrt{\ell(\ell+1)} \tilde{\rho}}|\hat{\Phi}_{(n-1+1)}|^2 \, d\tilde{\rho}\\
&-C\int_{0}^{\tilde{\rho}_0}\widetilde{\sigma}^4 n^{-4} \frac{n^2(n+1)^2\cdot (n-1)^2 n^2}{|2s|^4\widetilde{\sigma}^4}  |2s|^2|e^{2\gamma M \sqrt{\ell(\ell+1)} \tilde{\rho}}\hat{\Phi}_{(n-2+1)}|^2\, d\tilde{\rho}\\
&-C (\ell+1)^{-2}\sum_{k=2\ell+1}^{n-1}\int_{0}^{\tilde{\rho}_0}\widetilde{\sigma}^{2(n-k+1)}|2s|^{2(n-k)}\frac{(n-2\ell) k!^2}{n!^2}\left[\frac{(n+1)!^2 n!^2}{k!^2(k-1)!^2}\frac{1}{|2s|^{2(n-k+1)} \widetilde{\sigma}^{2(n-k+1)}}\right]\\
& \cdot |2s|^2 e^{2\gamma M \sqrt{\ell(\ell+1)} \tilde{\rho}}|\hat{\Phi}_{(k-1+1)}|^2\, d\tilde{\rho}\\
&-C \gamma^{-2}(\ell+1)^{-2} \int_0^{\tilde{\rho}_0} \frac{(n+1)!^2}{(2\ell)!^2}e^{2\gamma M  \sqrt{\ell(\ell+1)} \tilde{\rho}}|2s|^2|\hat{\Phi}_{(2\ell+1)}|^2 \,d \tilde{\rho}\\
\leq &\: C_{\epsilon} (\ell+1)^{-3}  \sum_{k=1}^{2\ell}\gamma^{\frac{1}{2} (2\ell-k)} \frac{(n+1)!^2}{(k-1)!^2} e^{2\gamma M \sqrt{\ell(\ell+1)} \tilde{\rho}_0}|\hat{\Phi}_{(k)}|^2(\tilde{\rho}_0)\\
&+C_{\epsilon}\gamma^{-1} (\ell+1)^{-3} (n+2)^2 \sum_{k=0}^{2\ell}\gamma^{\frac{1}{2} (2\ell-k)}  \frac{n!^2}{k!^2} e^{2\gamma M \sqrt{\ell(\ell+1)} \tilde{\rho}_0}|\hat{\Phi}_{(k)}|^2(\tilde{\rho}_0)+C_{\epsilon}\int_{0}^{\tilde{\rho}_0}e^{2\gamma M \sqrt{\ell(\ell+1)} \tilde{\rho}} |f_n|^2\, d\tilde{\rho},
\end{split}
\end{equation*}
with $C,C_{\epsilon}>0$ constants that are independent of $\ell$.
\end{proposition}
\begin{remark}
Given an arbitrary choice of $\alpha,\beta,\mu,\nu$, the first three integrals on the left-hand sides of both estimates in Proposition \ref{eq:maingevprop} need not be non-negative definite. We will later fix these parameters (Lemma \ref{lm:allowedvaluess}), so that under suitable restrictions on $s\in \C$, we obtain non-negative definiteness.
\end{remark}

We will make frequent use of the lemma below that shows that we can absorb \emph{lower order} derivatives terms with respect to $\tilde{\rho}$ into \emph{higher order derivatives} by introducing suitable exponential weights in $\tilde{\rho}$ at the expense of introducing additional boundary terms on the right-hand side of the relevant estimates.
\begin{lemma}
\label{lm:expweightest}
Let $M,\gamma>0$ and consider a function $h\in C^1( [0,\tilde{\rho}_0])$.
\begin{equation}
\label{eq:expweightest}
\begin{split}
\int_{0}^{\tilde{\rho}_0}&e^{2\gamma M  \sqrt{\ell(\ell+1)} \tilde{\rho}} |h|^2 \,d\tilde{\rho}+\frac{1}{\gamma M\sqrt{\ell(\ell+1)} }|h|^2(0)\\
\leq &\: \frac{1}{M^2\gamma^2  \ell(\ell+1)   }\int_{0}^{\tilde{\rho}_0} e^{2\gamma M \sqrt{\ell(\ell+1)} \tilde{\rho}}|\partial_{\tilde{\rho}}h|^2\,d\tilde{\rho}+\frac{1}{ \gamma M \sqrt{\ell(\ell+1)} }e^{2\gamma M \sqrt{\ell(\ell+1)} \tilde{\rho}_0}|h|^2(\tilde{\rho}_0).
\end{split}
\end{equation}
\end{lemma}
\begin{proof}
We can expand
\begin{equation*}
\partial_{\tilde{\rho}}(e^{\gamma M \sqrt{\ell(\ell+1)} \tilde{\rho}} h)=e^{\gamma M \sqrt{\ell(\ell+1)} \tilde{\rho}}\partial_{\tilde{\rho}}h+ \gamma M \sqrt{\ell(\ell+1)} e^{\gamma M \sqrt{\ell(\ell+1)} \tilde{\rho}} h.
\end{equation*}
Hence,
\begin{equation*}
\begin{split}
e^{2\gamma M  \sqrt{\ell(\ell+1)} \tilde{\rho}}|\partial_{\tilde{\rho}}h|^2=&\:|\partial_{\tilde{\rho}}(e^{\gamma M \sqrt{\ell(\ell+1)} \tilde{\rho}} h)|^2+ \gamma^2M^2 \ell(\ell+1) e^{2\gamma M \sqrt{\ell(\ell+1)} \tilde{\rho}} |h|^2\\
&-\partial_{\tilde{\rho}}(\gamma M \sqrt{\ell(\ell+1)}e^{2\gamma M \sqrt{\ell(\ell+1)} \tilde{\rho}}|h|^2).
\end{split}
\end{equation*}
and after integrating, we obtain
\begin{equation*}
\begin{split}
\int_{0}^{\tilde{\rho}_0}& \gamma^2 M^2 \ell(\ell+1) e^{2\gamma M  \sqrt{\ell(\ell+1)} \tilde{\rho}} |h|^2 \,d\tilde{\rho}+\gamma M  \sqrt{\ell(\ell+1)} |h|^2(0)\\\
\leq &\: \int_{0}^{\tilde{\rho}_0} e^{2\gamma M \sqrt{\ell(\ell+1)} \tilde{\rho}}|\partial_{\tilde{\rho}}h|^2\,d\tilde{\rho}+\gamma M \sqrt{\ell(\ell+1)}e^{2\gamma M \sqrt{\ell(\ell+1)} \tilde{\rho}_0}|h|^2(\tilde{\rho}_0). \qedhere
\end{split}
\end{equation*}
\end{proof}

\begin{proof}[Proof of Proposition \ref{eq:maingevprop}]
We first of all rearrange \eqref{eq:fixedfreqpsihat} as follows:
\begin{equation}
\label{eq:fixedfreqpsihat2} 
\begin{split}
2(1+M\tilde{\rho})^2& s \hat{\Phi}_{(n+1)}+[2\kappa\tilde{\rho}+\tilde{\rho}^2] (\partial_{\tilde{\rho}}-s\hat{h})\hat{\Phi}_{(n+1)}+2(n+1)(\tilde{\rho}+\kappa) \hat{\Phi}_{(n+1)}\\
&+\sum_{k=1}^{n_{\ell}}  \frac{(n+1)!}{(k-1)!} M^{n-k} F_{n,k} \hat{\Phi}_{(k)}+\sum_{k=0}^{n_{\ell}} \frac{n! (n-k+1)}{k!}M^{n-k}\widetilde{F}_{n,k}\hat{\Phi}_{(k)}=G_n
\end{split}
\end{equation}
with $n_{\ell}=\min\{n,2\ell\}$ and
\begin{equation*}
\begin{split}
G_n:=&\:f_n- [n(n+1)-\ell(\ell+1)] \hat{\Phi}_{(n)}-(4n+2)M(1+M\tilde{\rho}) s \hat{\Phi}_{(n)}-2n^2M^2 s \hat{\Phi}_{(n-1)}\\
&-\tilde{\rho}^2F_{n,n+2}\hat{\Phi}_{(n+2)}-2(n+1)\tilde{\rho}F_{n,n+1} \hat{\Phi}_{(n+1)}
\end{split}
\end{equation*}
if $n\leq 2\ell$ and
\begin{equation*}
\begin{split}
G_n:=&\:f_n- [n(n+1)-\ell(\ell+1)] \hat{\Phi}_{(n)}-(4n+2)M(1+M\tilde{\rho}) s \hat{\Phi}_{(n)}-2n^2M^2 s \hat{\Phi}_{(n-1)}\\
&-\tilde{\rho}^2F_{n,n+2}\hat{\Phi}_{(n+2)}-2(n+1)\tilde{\rho}F_{n,n+1} \hat{\Phi}_{(n+1)}+\sum_{k=2\ell+1}^{n}  \frac{(n+1)!}{(k-1)!} M^{n-k} F_{n,k} \hat{\Phi}_{(k)}\\
&+\sum_{k=2\ell+1}^n \frac{n! (n-k+1)}{k!}M^{n-k}\widetilde{F}_{n,k}\hat{\Phi}_{(k)}
\end{split}
\end{equation*}
if $n>2\ell$. 

We take the square norm of both sides of the equation \eqref{eq:fixedfreqpsihat2} and then multipy by an exponential weight function to obtain the equation:
 \begin{equation*}
 \begin{split}
 e^{2\gamma M \sqrt{\ell(\ell+1)} \tilde{\rho} }&\Bigg|2(1+M\tilde{\rho})^2 s \hat{\Phi}_{(n+1)}+[2\kappa\tilde{\rho}+\tilde{\rho}^2] (\partial_{\tilde{\rho}}-s\hat{h})\hat{\Phi}_{(n+1)}\\
 &+2(n+1)(\tilde{\rho}+\kappa) \hat{\Phi}_{(n+1)}+\sum_{k=1}^{n_{\ell}}  \frac{(n+1)!}{(k-1)!} M^{n-k} F_{n,k} \hat{\Phi}_{(k)}+\sum_{k=0}^{n_{\ell}} \frac{n! (n-k+1)}{k!}M^{n-k}\widetilde{F}_{n,k}\hat{\Phi}_{(k)}\Bigg|^2\\
 =&\:e^{2\gamma M \sqrt{\ell(\ell+1)} \tilde{\rho} }|G_n|^2,
 \end{split}
 \end{equation*}
 where $\gamma>0$ is a dimensionless constant that will be chosen suitably large.
 
We can split up the left-hand side above as follows:
 \begin{equation*}
 \begin{split}
 e^{2\gamma M \sqrt{\ell(\ell+1)} \tilde{\rho} }&\Big|2(1+M\tilde{\rho})^2 s \hat{\Phi}_{(n+1)}+[2\kappa\tilde{\rho}+\tilde{\rho}^2] (\partial_{\tilde{\rho}}-s\hat{h})\hat{\Phi}_{(n+1)}\\
 &+2(n+1)(\tilde{\rho}+\kappa) \hat{\Phi}_{(n+1)}+\sum_{k=1}^{n_{\ell}}  \frac{(n+1)!}{(k-1)!} M^{n-k} F_{n,k} \hat{\Phi}_{(k)}+\sum_{k=0}^{n_{\ell}} \frac{n! (n-k+1)}{k!}M^{n-k}\widetilde{F}_{n,k}\hat{\Phi}_{(k)}\Big|^2=\sum_{i=0}^4 J_i,
 \end{split}
 \end{equation*}
 where
 \begin{equation*}
 \begin{split}
 J_0:=&\: |2s|^2(1+M\tilde{\rho})^4e^{2\gamma M \sqrt{\ell(\ell+1)} \tilde{\rho} }|\hat{\Phi}_{(n+1)}|^2+(2\kappa\tilde{\rho}+\tilde{\rho}^2)^2e^{2\gamma M  \sqrt{\ell(\ell+1)} \tilde{\rho} }|\hat{\Phi}_{(n+2)}|^2\\
 &+4(n+1)^2(\tilde{\rho}+\kappa)^2e^{2\gamma M \sqrt{\ell(\ell+1)} \tilde{\rho} }|\hat{\Phi}_{(n+1)}|^2
 \end{split}
 \end{equation*}
and
\begin{align*}
J_1:=&\:4 \re(2s)(n+1)(1+M\tilde{\rho})^2 (\tilde{\rho}+\kappa)e^{2\gamma M \sqrt{\ell(\ell+1)} \tilde{\rho} }|\hat{\Phi}_{(n+1)}|^2,\\
J_2:=&\:  2(n+1)(\tilde{\rho}+\kappa)[2\kappa\tilde{\rho}+\tilde{\rho}^2] e^{2\gamma M \sqrt{\ell(\ell+1)} \tilde{\rho} }\partial_{\tilde{\rho}}(|\hat{\Phi}_{(n+1)}|^2) \\
&-2(n+1)\re(2s)\hat{h}(\tilde{\rho}+\kappa)[2\kappa\tilde{\rho}+\tilde{\rho}^2]e^{2\gamma M \sqrt{\ell(\ell+1)} \tilde{\rho} }|\hat{\Phi}_{(n+1)}|^2 ,\\
J_3:=&\: 2(1+M\tilde{\rho})^2[2\kappa\tilde{\rho}+\tilde{\rho}^2]e^{2\gamma M \sqrt{\ell(\ell+1)} \tilde{\rho} }\re(2s\hat{\Phi}_{(n+1)} \overline{\hat{\Phi}_{(n+2)}}),\\
J_4:= &\: e^{2\gamma M \sqrt{\ell(\ell+1)} \tilde{\rho} }\left|\sum_{k=1}^{n_{\ell}}  \frac{(n+1)!}{(k-1)!}M^{n-k} F_{n,k} \hat{\Phi}_{(k)}+\sum_{k=0}^{n_{\ell}} \frac{n! (n-k+1)}{k!}M^{n-k}\widetilde{F}_{n,k}\hat{\Phi}_{(k)}\right|^2 \\
&+ 2 e^{2\gamma M \sqrt{\ell(\ell+1)} \tilde{\rho} }\re \Bigg[\left(\sum_{k=1}^{n_{\ell}}  \frac{(n+1)!}{(k-1)!} M^{n-k} F_{n,k}\overline{\hat{\Phi}_{(k)}}+\sum_{k=0}^{n_{\ell}} \frac{n! (n-k+1)}{k!}M^{n-k}\widetilde{F}_{n,k}\overline{\hat{\Phi}_{(k)}}\right)\cdot \\
&\cdot \left(2(1+M\tilde{\rho})^2 s \hat{\Phi}_{(n+1)}+[2\kappa\tilde{\rho}+\tilde{\rho}^2] \hat{\Phi}_{(n+2)}+2(n+1)(\tilde{\rho}+\kappa) \hat{\Phi}_{(n+1)} \right)\Bigg].
\end{align*}
\\
\emph{Step 1: Estimating $J_1$--$J_3$.} We will first estimate $J_1$. Let us first suppose $\re(s)\geq 0$. Then $J_1$ has a good sign. Now suppose $\re(s)<0$. Then we introduce the parameter $\beta>0$ and apply Young's inequality to estimate:
\begin{equation*}
|J_1|\leq \beta (n+1)^2(1+M\tilde{\rho})^4 (\tilde{\rho}+\kappa)^2e^{2\gamma M \sqrt{\ell(\ell+1)} \tilde{\rho} } |\hat{\Phi}_{(n+1)}|^2+ 4\beta^{-1}\left(\frac{|\re(2s)|}{|2s|}\right)^2 |2s|^2e^{2\gamma M \sqrt{\ell(\ell+1)} \tilde{\rho} }|\hat{\Phi}_{(n+1)}|^2.
\end{equation*}

Consider $J_2$. We integrate by parts (using that $\hat{\Phi}$ is smooth at $\rho=0$):
\begin{equation*}
\begin{split}
\int_{0}^{\tilde{\rho}_0}J_2\,d\rho=&\: 2(n+1)(\tilde{\rho} +\kappa)[2\kappa\tilde{\rho}+\tilde{\rho}^2] e^{2\gamma M \sqrt{\ell(\ell+1)} \tilde{\rho} }|\hat{\Phi}_{(n+1)}|^2\Big|_{\tilde{\rho}=\tilde{\rho}_0}\\
&-\int_{0}^{\tilde{\rho}_0}  (n+1)[4\kappa^2+12\kappa\tilde{\rho}+6\tilde{\rho}^2] e^{2\gamma M \sqrt{\ell(\ell+1)} \tilde{\rho} }|\hat{\Phi}_{(n+1)}|^2\\
&-4\gamma M  \sqrt{\ell(\ell+1)} (n+1)(\tilde{\rho}+\kappa)[2\kappa\tilde{\rho}+\tilde{\rho}^2] e^{2\gamma M \sqrt{\ell(\ell+1)} \tilde{\rho} }|\hat{\Phi}_{(n+1)}|^2\\
&-2(n+1)\re(2s)\hat{h}(\tilde{\rho}+\kappa)[2\kappa\tilde{\rho}+\tilde{\rho}^2]e^{2\gamma M \sqrt{\ell(\ell+1)} \tilde{\rho} }|\hat{\Phi}_{(n+1)}|^2\,d\tilde{\rho}.
\end{split}
\end{equation*}
Note in particular that we can absorb the second term in the integral on the RHS into $J_0$, provided $\gamma M{\ll} \tilde{\rho}_0^{-1}$.

Furthermore, using that $\ell<n$, we can estimate the third and fourth terms by:
\begin{equation*}
\int_{0}^{\tilde{\rho}_0} \epsilon (n+1)^2(\tilde{\rho}+\kappa)^2e^{2\gamma M \sqrt{\ell(\ell+1)} \tilde{\rho} }|\hat{\Phi}_{(n+1)}|^2\,d\tilde{\rho},
\end{equation*}
where $\epsilon>0$ here can be taken suitably small, using that $\gamma M{\ll} \tilde{\rho}_0^{-1}$ and $\kappa\lesssim (\ell+1)^{-1}$, where we will take $\ell$ arbitrarily large.

Now consider $J_3$. We introduce a parameter $\mu\in [0,1]$ and split:
\begin{equation*}
\begin{split}
|J_3|\leq&\: 2|2s|(1+M\tilde{\rho})^2[2\kappa\tilde{\rho}+\tilde{\rho}^2]e^{2\gamma M \sqrt{\ell(\ell+1)} \tilde{\rho} }|\hat{\Phi}_{(n+1)}|  |\hat{\Phi}_{(n+2)}|\\
=&\: (1-\mu)2|2s|(1+M\tilde{\rho})^2[2\kappa\tilde{\rho}+\tilde{\rho}^2]e^{2\gamma M \sqrt{\ell(\ell+1)} \tilde{\rho} }|\hat{\Phi}_{(n+1)}|  |\hat{\Phi}_{(n+2)}|\\
&+\mu2|2s|(1+M\tilde{\rho})^2[2\kappa\tilde{\rho}+\tilde{\rho}^2]e^{2\gamma M \sqrt{\ell(\ell+1)} \tilde{\rho} }|\hat{\Phi}_{(n+1)}|  |\hat{\Phi}_{(n+2)}|.
\end{split}
\end{equation*}
Subsequently, we introduce another parameter $\alpha>0$ in order to estimate the term with a factor $1-\mu$ via Young's inequality:
\begin{equation}
\begin{split}
\label{eq:alpha}
2&(1-\mu)|2s|(1+M\tilde{\rho})^2[2\kappa\tilde{\rho}+\tilde{\rho}^2]e^{2\gamma M \sqrt{\ell(\ell+1)} \tilde{\rho} }|\hat{\Phi}_{(n+1)}|  |\hat{\Phi}_{(n+2)}| \\
\leq &\:\alpha(1-\mu) (n+1)^2 (\tilde{\rho}^2+2\kappa\tilde{\rho})e^{2\gamma M \sqrt{\ell(\ell+1)} \tilde{\rho} }|\hat{\Phi}_{(n+1)}|^2\\
&+ \frac{|2s|^2\widetilde{\sigma}^2}{(n+1)^2(n+2)^2}\cdot \alpha^{-1}\widetilde{\sigma}^{-2}(1-\mu)(n+2)^2(1+M\tilde{\rho})^4(\tilde{\rho}^2+2\kappa\tilde{\rho})e^{2\gamma M \sqrt{\ell(\ell+1)} \tilde{\rho} }|\hat{\Phi}_{(n+1+1)}|^2.
\end{split}
\end{equation}
Furthermore, we introduce the parameter $\nu>0$ in order to further estimate the term with a factor $\mu$ via another application of Young's inequality:
\begin{equation*}
\begin{split}
2\mu|2s|&(1+M\tilde{\rho})^2[2\kappa\tilde{\rho}+\tilde{\rho}^2]e^{2\gamma M \sqrt{\ell(\ell+1)} \tilde{\rho} }|\hat{\Phi}_{(n+1)}|  |\hat{\Phi}_{(n+2)}|\\
\leq &\: \nu \mu |2s|^2e^{2\gamma M \sqrt{\ell(\ell+1)} \tilde{\rho} } |\hat{\Phi}_{(n+1)}|^2+\frac{\mu}{\nu}(1+M\tilde{\rho})^4(2\kappa\tilde{\rho}+\tilde{\rho}^2)^2e^{2\gamma M \sqrt{\ell(\ell+1)} \tilde{\rho} }|\hat{\Phi}_{(n+2)}|^2.
\end{split}
\end{equation*}

\emph{Step 2: Estimating $J_4$. } We consider $J_4$. We apply Young's inequality with suitable weights to obtain
\begin{equation*}
\begin{split}
J_4\geq& -\epsilon  \sum_{k=0}^{n_{\ell}} \gamma^{-\frac{1}{2}(n_{\ell}-k)}J_0- C\epsilon^{-1}\sum_{k=1}^{n_{\ell}}  \gamma^{\frac{1}{2} (n_{\ell}-k)} \frac{(n+1)!^2}{(k-1)!^2} M^{2(n-k)} |F_{n,k}|^2 e^{2\gamma M \sqrt{\ell(\ell+1)} \tilde{\rho} }|\hat{\Phi}_{(k)}|^2\\
&-C \epsilon^{-1}\sum_{k=0}^{n_{\ell}}  \gamma^{\frac{1}{2} (n_{\ell}-k)} \frac{n!^2}{k!^2}(n-k+2)^2 M^{2(n-k)} |\tilde{F}_{n,k}|^2 e^{2\gamma M \sqrt{\ell(\ell+1)} \tilde{\rho} }|\hat{\Phi}_{(k)}|^2.
\end{split}
\end{equation*}

Let $1\leq k\leq n_{\ell}$. We will ignore the $|F_{n,k}|^2$ and $|\tilde{F}_{n,k}|^2$  factors for now. Then we can apply Lemma \ref{lm:expweightest} to estimate:
\begin{equation}
\label{eq:estk1}
\begin{split}
\int_0^{\tilde{\rho}_0} \frac{(n+1)!^2}{(k-1)!^2} e^{2\gamma M \sqrt{\ell(\ell+1)} \tilde{\rho}}|\hat{\Phi}_{(k)}|^2 \,d \tilde{\rho}\leq&\: C\gamma^{-2}\int_0^{\tilde{\rho}_0} \frac{(n+1)!^2}{(k+1-1)!^2} e^{2\gamma M \sqrt{\ell(\ell+1)} \tilde{\rho}}(|\hat{\Phi}_{(k+1)}|^2 +\hat{h}^2|s|^2|\hat{\Phi}_{(k)}|^2)\,d \tilde{\rho}\\
&+ \frac{C}{\gamma (\ell(\ell+1))^{\frac{1}{2}}}\frac{(n+1)!^2}{(k-1)!^2}e^{2\gamma M \sqrt{\ell(\ell+1)} \tilde{\rho}_0}|\hat{\Phi}_{(k)}|^2(\tilde{\rho}_0),
\end{split}
\end{equation}
and then we can absorb the $\hat{h}^2|s|^2|\hat{\Phi}_{(k)}|^2$ term appearing in the integrand on the right-hand side into the left-hand side, for $\gamma$ suitably large.

Similarly, we have that for all $0\leq k\leq n_{\ell}$:
\begin{equation}
\label{eq:estk2}
\begin{split}
\int_0^{\tilde{\rho}_0}& \frac{n!^2}{k!^2} (n-k+2)^2 e^{2\gamma M \sqrt{\ell(\ell+1)} \tilde{\rho}}|\hat{\Phi}_{(k)}|^2 \,d \tilde{\rho}\\
\leq&\: C\gamma^{-2}\int_0^{\tilde{\rho}_0} \frac{n!^2}{(k+1)!^2}(n-(k+1)+2)^2 \left[\frac{(n-k+2)^2}{(n-(k+1)+2)^2}\right] e^{2\gamma M \sqrt{\ell(\ell+1)} \tilde{\rho}}|\hat{\Phi}_{(k+1)}|^2 \,d \tilde{\rho}\\
&+ \frac{C}{\gamma (\ell(\ell+1))^{\frac{1}{2}}}\frac{n!^2}{k!^2}(n-k+2)^2e^{2\gamma M \sqrt{\ell(\ell+1)} \tilde{\rho}_0}|\hat{\Phi}_{(k)}|^2(\tilde{\rho}_0).
\end{split}
\end{equation}

By repeatedly applying the estimate \eqref{eq:estk1} $n_{\ell}-k+1$ times, we obtain the following estimate for $\hat{\Phi}_{(k)}$ with $1\leq k\leq  n_{\ell}$:
\begin{equation*}
\begin{split}
\int_0^{\tilde{\rho}_0}&   \gamma^{\frac{1}{2} (n_{\ell}-k)} \frac{(n+1)!^2}{(k-1)!^2} e^{2\gamma M \sqrt{\ell(\ell+1)} \tilde{\rho}}|\hat{\Phi}_{(k)}|^2 \,d \tilde{\rho}\\
\leq&\: |s|^{-2} C^{n_{\ell}-k}\gamma^{-2}  \gamma^{-\frac{3}{2} (n_{\ell}-k)} \int_0^{\tilde{\rho}_0} |s|^2 \frac{(n+1)!^2}{n_{\ell}!^2} e^{2\gamma M \sqrt{\ell(\ell+1)} \tilde{\rho}}|\hat{\Phi}_{(n_{\ell}+1)}|^2 \,d \tilde{\rho}\\
&+ \frac{C}{ (\ell(\ell+1))^{\frac{1}{2}}} \gamma^{\frac{1}{2} (n_{\ell}-k)}\sum_{m=k}^{n_{\ell}} \gamma^{-1-2(m-k)}\frac{(n+1)!^2}{(m-1)!^2}e^{2\gamma M \sqrt{\ell(\ell+1)} \tilde{\rho}_0}|\hat{\Phi}_{(m)}|^2(\tilde{\rho}_0),
\end{split}
\end{equation*}
where $C>0$ is a constant independent of  $\gamma$, $\ell$. If we take $\gamma$ suitably large compared to $C$, then we can sum over $k$ to obtain:
\begin{equation*}
\begin{split}
\sum_{k=1}^ {n_{\ell}}\int_0^{\tilde{\rho}_0}&   \gamma^{\frac{1}{2} (n_{\ell}-k)} \frac{(n+1)!^2}{(k-1)!^2} e^{2\gamma M \sqrt{\ell(\ell+1)} \tilde{\rho}}|\hat{\Phi}_{(k)}|^2 \,d \tilde{\rho}\leq C |s|^{-2} \gamma^{-2} \int_0^{\tilde{\rho}_0} \frac{(n+1)!^2}{n_{\ell}!^2}e^{2\gamma M \sqrt{\ell(\ell+1)} \tilde{\rho}}|s|^2|\hat{\Phi}_{(n_{\ell}+1)}|^2 \,d \tilde{\rho}\\
&+ \frac{C }{\gamma (\ell+1)} \sum_{k=1}^{n_{\ell}}\gamma^{\frac{1}{2} (n_{\ell}-k)}  \frac{(n+1)!^2}{(k-1)!^2} e^{2\gamma M \sqrt{\ell(\ell+1)} \tilde{\rho}_0}|\hat{\Phi}_{(k)}|^2(\tilde{\rho}_0).
\end{split}
\end{equation*}

We can similarly apply \eqref{eq:estk2} $n_{\ell}-k+1$ times. Before we do that, we observe that for all $0\leq k\leq n_{\ell}$
\begin{equation*}
\begin{split}
\frac{(n-k+2)^2}{(n-(k+1)+2)^2}=&\: \left(1+\frac{1}{n-(k+1)+2}\right)^2\leq 4.
\end{split}
\end{equation*}
Hence,
\begin{equation*}
\begin{split}
\int_0^{\tilde{\rho}_0}&   \gamma^{\frac{1}{2} (n_{\ell}-k)} \frac{n!^2}{k!^2} (n-k+2)^2 e^{2\gamma M \sqrt{\ell(\ell+1)} \tilde{\rho}}|\hat{\Phi}_{(k)}|^2 \,d \tilde{\rho}\\
\leq &\:C^{n_{\ell}-k} |s|^{-2}\gamma^{-2}\gamma^{- \frac{3}{2} (n_{\ell}-k)} \int_0^{\tilde{\rho}_0} \frac{n!}{(n_{\ell} +1)!^2} (n-(n_{\ell}+1)+2)^2 e^{2\gamma M \sqrt{\ell(\ell+1)} \tilde{\rho}} |s|^2|\hat{\Phi}_{(n_{\ell}+1)}|^2 \,d \tilde{\rho}\\
&+ \frac{C}{ (\ell(\ell+1))^{\frac{1}{2}}} \gamma^{\frac{1}{2} (n_{\ell}-k)}\sum_{m=k}^{n_{\ell}} \gamma^{-1-2(m-k)}\frac{n!^2}{m!^2}(n-m+2)^2e^{2\gamma M \sqrt{\ell(\ell+1)} \tilde{\rho}_0}|\hat{\Phi}_{(m)}|^2(\tilde{\rho}_0).
\end{split}
\end{equation*}

Therefore, we can sum over $k$ to obtain:
\begin{equation*}
\begin{split}
\sum_{k=0}^ {n_{\ell}}&\int_0^{\tilde{\rho}_0}   \gamma^{\frac{1}{2} (n_{\ell}-k)} \frac{n!^2}{k!^2}(n-k+2)^2 e^{2\gamma M \sqrt{\ell(\ell+1)} \tilde{\rho}}|\hat{\Phi}_{(k)}|^2 \,d \tilde{\rho}\\
\leq&\: C \gamma^{-2} |s|^{-2} \frac{n!^2}{(n_{\ell} +1)!^2} (n-(n_{\ell}+1)+2)^2 \int_0^{\tilde{\rho}_0}e^{2\gamma M \sqrt{\ell(\ell+1)} \tilde{\rho}}|2s|^2|\hat{\Phi}_{(n_{\ell}+1)}|^2 \,d \tilde{\rho}\\
&+ C \gamma^{-1} (\ell+1)^{-1} (n+2)^2\sum_{k=0}^{n_{\ell}}\gamma^{\frac{1}{2} (n_{\ell}-k)}  \frac{n!^2}{k!^2} e^{2\gamma M \sqrt{\ell(\ell+1)} \tilde{\rho}_0}|\hat{\Phi}_{(k)}|^2(\tilde{\rho}_0).
\end{split}
\end{equation*}

We moreover have that $|F_{n,k}|^2+|\tilde{F}_{n,k}|^2\leq C(\kappa_+^2+\kappa_c^2)$, so in order to ensure that the integrals appearing in the estimates above can be absorbed into $J_0$, we will take $\kappa_+,\kappa_c\lesssim (\ell+1)^{-1}$. We can then conclude that for $\epsilon>0$ arbitrarily small and $|s|>0$, there exists a $\gamma{\gg}M^{-1}|s|^{-1}$ suitably large such that: for $n\leq 2\ell$ ($n_{\ell}=n$)
\begin{equation*}
\begin{split}
\int_0^{\tilde{\rho}_0} J_4 d\tilde{\rho}\geq&\: -\int_0^{\tilde{\rho}_0} \epsilon  J_0d\tilde{\rho}-C_{\epsilon} \gamma^{-1} (\ell+1)^{-3}  \sum_{k=1}^{n}\gamma^{\frac{1}{2} (n-k)}  \frac{(n+1)!^2}{(k-1)!^2} e^{2\gamma M \sqrt{\ell(\ell+1)} \tilde{\rho}_0}|\hat{\Phi}_{(k)}|^2(\tilde{\rho}_0)\\
&-C_{\epsilon} \gamma^{-1} (\ell+1)^{-1}\sum_{k=0}^{n}\gamma^{\frac{1}{2} (n-k)}  \frac{n!^2}{k!^2} e^{2\gamma M \sqrt{\ell(\ell+1)} \tilde{\rho}_0}|\hat{\Phi}_{(k)}|^2(\tilde{\rho}_0)
\end{split}
\end{equation*}
and for $n>2\ell$ ($n_{\ell}=2\ell$)
\begin{equation*}
\begin{split}
\int_0^{\tilde{\rho}_0}& J_4 d\tilde{\rho}\geq -\int_0^{\tilde{\rho}_0} \epsilon  J_0\,d\tilde{\rho}\\
&-C |s|^{-2}\gamma^{-2}(\ell+1)^{-2} \int_0^{\tilde{\rho}_0} \frac{(n+1)!^2}{(2\ell)!^2} e^{2\gamma M \sqrt{\ell(\ell+1)} \tilde{\rho}}|2s|^2|\hat{\Phi}_{(2\ell+1)}|^2 \,d \tilde{\rho}\\
&- C_{\epsilon} \gamma^{-1}(\ell+1)^{-3} \sum_{k=1}^{2\ell}\gamma^{\frac{1}{2} (2\ell-k)}  \frac{(n+1)!^2}{(k-1)!^2} e^{2\gamma M \sqrt{\ell(\ell+1)} \tilde{\rho}_0}|\hat{\Phi}_{(k)}|^2(\tilde{\rho}_0)\\
&-C_{\epsilon}\gamma^{-1} (\ell+1)^{-3} (n+2)^2 \sum_{k=0}^{2\ell}\gamma^{\frac{1}{2} (2\ell-k)}  \frac{n!^2}{k!^2} e^{2\gamma M \sqrt{\ell(\ell+1)} \tilde{\rho}_0}|\hat{\Phi}_{(k)}|^2(\tilde{\rho}_0).
\end{split}
\end{equation*}

\emph{Step 3: Estimating $G_n$.}
Estimate first of all
\begin{equation*}
\begin{split}
&\left|\sum_{k=2\ell+1}^{n-1} \left[ \frac{(n+1)!}{(k-1)!} M^{n-k} F_{n,k}+\frac{n!}{k!}(n-k+1)M^{n-k}\widetilde{ F}_{n,k}\right] \hat{\Phi}_{(k)}\right|^2\\
\leq&\: C (n-2\ell) \sum_{k=2\ell+1}^{n-1}  \Bigg(\frac{(n+1)!^2}{(k-1)!^2} M^{2(n-k)} \Big(|F_{n,k}|^2+ |\widetilde{F}_{n,k}|^2\Big)\Bigg)| \hat{\Phi}_{(k)}|^2.
\end{split}
\end{equation*}

We can therefore apply a suitably weighted Young's inequality in the $n>2\ell$ case: for $\epsilon>0$ arbitrary small, we have that there exists a constant $C_{\epsilon}>0$, such that
\begin{equation*}
\begin{split}
|G_n|^2\leq&\: (1+\epsilon) [n(n+1)-\ell(\ell+1)]^2 |\hat{\Phi}_{(n)}|^2+ C_{\epsilon}(n+1)^2n^2(|F_{n,n}|^2+|\widetilde{F}_{n,n}|^2)| \hat{\Phi}_{(n)}|^2\\
&+C_{\epsilon}(n+1)^2 |s|^2 |\hat{\Phi}_{(n)}|^2+C_{\epsilon}n^4 |s|^2|\hat{\Phi}_{(n-1)}|^2\\
&+C_{\epsilon}\tilde{\rho}^4|F_{n,n+2}|^2|\hat{\Phi}_{(n+2)}|^2+C_{\epsilon}(n+1)^2\tilde{\rho}^2|F_{n,n+1}|^2 |\hat{\Phi}_{(n+1)}|^2\\
&+C_{\epsilon}(n-2\ell) \sum_{k=2\ell+1}^{n-1}  \frac{(n+1)!^2}{(k-1)!^2} M^{2(n-k)}(|F_{n,k}|^2+|\widetilde{F}_{n,k}|^2)| \hat{\Phi}_{(k)}|^2\\
&+C_{\epsilon}|f_n|^2,
\end{split}
\end{equation*}
whereas in the $n\leq 2\ell$ terms we obtain the same inequality, but without the terms depending on $F_{n,k}$ and $\widetilde{F}_{n,k}$ with $2\ell+1\leq k\leq n$ on the right-hand side .

First of all, we rewrite
\begin{equation*}
(1+\epsilon) [n(n+1)-\ell(\ell+1)]^2 |\hat{\Phi}_{(n)}|^2= \frac{[n(n+1)-\ell(\ell+1)]^2}{|2s|^2\widetilde{\sigma}^2} (1+\epsilon)\widetilde{\sigma}^2 |2s|^2|\hat{\Phi}_{(n-1+1)}|^2.
\end{equation*}
We refer to the proof of Proposition \ref{prop:gevreyestsum} for the motivation to include the factor $\widetilde{\sigma}^2$ above. 

Similarly, we have that for $n>2\ell$
\begin{align*}
C_{\epsilon}(n+1)^2n^2|F_{n,n}|^2| \hat{\Phi}_{(n)}|^2\leq&\: C\widetilde{\sigma}^2(\kappa_+^2+\kappa_c^2)\frac{n^2(n+1)^2}{|2s|^2\widetilde{\sigma}^2}  |2s|^2|\hat{\Phi}_{(n-1+1)}|^2,\\
C_{\epsilon}(n+1)^2 |s|^2 |\hat{\Phi}_{(n)}|^2=&\: C_{\epsilon}n^{-2}\widetilde{\sigma}^2 |2s|^2 \frac{n^2(n+1)^2}{|2s|^2\widetilde{\sigma}^2}  |2s|^2|\hat{\Phi}_{(n-1+1)}|^2,\\
C_{\epsilon}n^4 |s|^2 |\hat{\Phi}_{(n-1)}|^2\leq&\: C_{\epsilon}n^{-4} \widetilde{\sigma}^4 |2s|^2\frac{n^2(n+1)^2\cdot (n-1)^2 n^2}{|2s|^2\widetilde{\sigma}^4}  |2s|^2|\hat{\Phi}_{(n-2+1)}|^2
\end{align*}
and, for $2\ell+1\leq k\leq n$,
\begin{equation*}
\begin{split}
M^{2(n-k)}&(|F_{n,k}|^2+|\widetilde{F}_{n,k}|^2)\frac{(n+1)!^2}{(k-1)!^2} (n-2\ell) |\Phi_{(k)}|^2\\
\leq&\: C M^{2(n-k)}(\kappa_+^2+\kappa_c^2) \widetilde{
\sigma}^{2(n-k+1)}|2s|^{2(n-k)}\frac{(n-2\ell)k!^2}{n!^2}\left[\frac{(n+1)!^2n!^2}{k!^2(k-1)!^2}\widetilde{
\sigma}^{-2(n-(k-1))}|2s|^{-2(n-(k-1))}\right]\\
\cdot &|2s|^2|\Phi_{(k-1+1)}|^2.
\end{split}
\end{equation*}
If $n\leq 2\ell$, we instead apply Lemma \ref{lm:expweightest} to estimate:
\begin{equation*}
\begin{split}
\int_0^{\tilde{\rho}_0} n^4 |s|^2 e^{2\gamma M \sqrt{\ell(\ell+1)} \tilde{\rho}}|\hat{\Phi}_{(n-1)}|^2\,d\tilde{\rho}\leq&\: \frac{C}{\gamma^2} n^2\int_0^{\tilde{\rho}_0}  |s|^2 e^{2\gamma M \sqrt{\ell(\ell+1)} \tilde{\rho}}|\hat{\Phi}_{(n)}|^2\,d\tilde{\rho}\\
&+ \frac{C}{(\ell(\ell+1))^{\frac{1}{2}} \gamma} n^4 |s|^2 e^{2\gamma M \sqrt{\ell(\ell+1)} \tilde{\rho}_0}|\hat{\Phi}_{(n-1)}|^2(\tilde{\rho_0}).
\end{split}
\end{equation*}
and
\begin{equation*}
\begin{split}
\int_0^{\tilde{\rho}_0}  n^2|s|^2 e^{2\gamma M \sqrt{\ell(\ell+1)} \tilde{\rho}}|\hat{\Phi}_{(n)}|^2\,d\tilde{\rho}\leq&\: \frac{C}{\gamma^2} \int_0^{\tilde{\rho}_0}  |s|^2 e^{2\gamma M \sqrt{\ell(\ell+1)} \tilde{\rho}}|\hat{\Phi}_{(n+1)}|^2\,d\tilde{\rho}\\
&+ \frac{C}{(\ell(\ell+1))^{\frac{1}{2}}\gamma}n^2  |s|^2 e^{2\gamma M \sqrt{\ell(\ell+1)} \tilde{\rho}_0}|\hat{\Phi}_{(n)}|^2(\tilde{\rho_0}).
\end{split}
\end{equation*}

\emph{Step 4: Putting everything together.} We combine the above estimates to obtain the following estimate: let $\epsilon>0$ be arbitrarily small, then there exist $\gamma$ suitably large with 
\begin{equation*}
M^{-1}\tilde{\rho}_0^{-1}{\gg}\gamma{\gg}(M^{-1}|s|^{-1}+M|s|),
\end{equation*}
and moreover $M^{-1}\tilde{\rho}_0^{-1}{\gg}1$, such that for $n\leq 2\ell$:
\begin{equation*} 
\begin{split}
\int_{0}^{\tilde{\rho}_0}& (1-\mu \nu^{-1}(1+M\tilde{\rho})^4-\epsilon)(2\kappa\tilde{\rho}+\tilde{\rho}^2)^2e^{2\gamma M \sqrt{\ell(\ell+1)} \tilde{\rho} }|\hat{\Phi}_{(n+2)}|^2\, d\tilde{\rho}\\
&+\int_{0}^{\tilde{\rho}_0} \left(1-4 \beta^{-1}\left(\frac{\max\{-\re(s),0\}}{|s|}\right)^2-\nu \mu-\epsilon \right)|2s|^2e^{2\gamma M \sqrt{\ell(\ell+1)} \tilde{\rho} }|\hat{\Phi}_{(n+1)}|^2\, d\tilde{\rho}\\
&+\int_{0}^{\tilde{\rho}_0}  (4-\beta-\alpha(1-\mu)-\epsilon)(n+1)^2(\tilde{\rho}+\kappa)^2e^{2\gamma M \sqrt{\ell(\ell+1)} \tilde{\rho} }|\hat{\Phi}_{(n+1)}|^2 \, d\tilde{\rho}\\
&- \int_{0}^{\tilde{\rho}_0} \frac{|2s|^2\widetilde{\sigma}^2}{(n+1)^2(n+2)^2}\cdot \alpha^{-1}\widetilde{\sigma}^{-2}(1-\mu)(n+2)^2(1+M\tilde{\rho})^4(\tilde{\rho}^2+2\kappa\tilde{\rho})e^{2\gamma M \sqrt{\ell(\ell+1)} \tilde{\rho} }|\hat{\Phi}_{(n+1+1)}|^2 \, d\tilde{\rho}\\
&-\int_{0}^{\tilde{\rho}_0}\frac{[n(n+1)-\ell(\ell+1)]^2}{|2s|^2\widetilde{\sigma}^2} (1+\epsilon)\widetilde{\sigma}^2 |2s|^2e^{2\gamma M \sqrt{\ell(\ell+1)} \tilde{\rho}}|\hat{\Phi}_{(n-1+1)}|^2 \, d\tilde{\rho}\\
\leq &\: C_{\epsilon} (\ell+1)^{-3}\gamma^{-1}  \sum_{k=1}^{n}\gamma^{\frac{1}{2} (n-k)} \frac{(n+1)!^2}{(k-1)!^2} e^{2\gamma M \sqrt{\ell(\ell+1)} \tilde{\rho}_0}|\hat{\Phi}_{(k)}|^2(\tilde{\rho}_0)\\
&+C_{\epsilon} \gamma^{-1} (\ell+1)^{-1}\sum_{k=0}^{n}\gamma^{\frac{1}{2} (n-k)}   \frac{n!^2}{k!^2} e^{2\gamma M \sqrt{\ell(\ell+1)} \tilde{\rho}_0}|\hat{\Phi}_{(k)}|^2(\tilde{\rho}_0)+C_{\epsilon}\int_{0}^{\tilde{\rho}_0}e^{2\gamma M \sqrt{\ell(\ell+1)} \tilde{\rho}}|f_n|^2\, d\tilde{\rho}.
\end{split}
\end{equation*}

Let $n>2\ell$, then
\begin{equation*}
\begin{split}
\int_{0}^{\tilde{\rho}_0}& (1-\mu \nu^{-1}(1+M\tilde{\rho})^4-\epsilon)(2\kappa\tilde{\rho}+\tilde{\rho}^2)^2e^{2\gamma M \sqrt{\ell(\ell+1)} \tilde{\rho} }|\hat{\Phi}_{(n+2)}|^2\, d\tilde{\rho}\\
&+\int_{0}^{\tilde{\rho}_0} \left(1-4 \beta^{-1}\left(\frac{\max\{-\re(s),0\}}{|s|}\right)^2-\nu \mu-\epsilon \right)|2s|^2e^{2\gamma M \sqrt{\ell(\ell+1)} \tilde{\rho} }|\hat{\Phi}_{(n+1)}|^2\, d\tilde{\rho}\\
&+\int_{0}^{\tilde{\rho}_0}  (4-\beta-\alpha(1-\mu)-\epsilon)(n+1)^2(\tilde{\rho}+\kappa)^2e^{2\gamma M \sqrt{\ell(\ell+1)} \tilde{\rho} }|\hat{\Phi}_{(n+1)}|^2 \, d\tilde{\rho}\\
&- \int_{0}^{\tilde{\rho}_0} \frac{|2s|^2\widetilde{\sigma}^2}{(n+1)^2(n+2)^2}\cdot \alpha^{-1}\widetilde{\sigma}^{-2}(1-\mu)(n+2)^2(1+M\tilde{\rho})^4(\tilde{\rho}^2+2\kappa\tilde{\rho})e^{2\gamma M \sqrt{\ell(\ell+1)} \tilde{\rho} }|\hat{\Phi}_{(n+1+1)}|^2 \, d\tilde{\rho}\\
&-\int_{0}^{\tilde{\rho}_0}\frac{[n(n+1)-\ell(\ell+1)]^2}{|2s|^2\widetilde{\sigma}^2} (1+\epsilon)\widetilde{\sigma}^2 |2s|^2e^{2\gamma M \sqrt{\ell(\ell+1)} \tilde{\rho}}|\hat{\Phi}_{(n-1+1)}|^2 \, d\tilde{\rho}\\
&-C\int_{0}^{\tilde{\rho}_0}\widetilde{\sigma}^4 n^{-4} \frac{n^2(n+1)^2\cdot (n-1)^2 n^2}{|2s|^4\widetilde{\sigma}^4}  |2s|^2|e^{2\gamma M \sqrt{\ell(\ell+1)} \tilde{\rho}}\hat{\Phi}_{(n-2+1)}|^2\, d\tilde{\rho}\\
&-C (\ell+1)^{-2}\sum_{k=2\ell+1}^{n-1 }\int_{0}^{\tilde{\rho}_0}\widetilde{\sigma}^{2(n-k+1)}|2s|^{2(n-k)}\frac{(n-2\ell) k!^2}{n!^2}\left[\frac{(n+1)!^2 n!^2}{k!^2(k-1)!^2}\frac{1}{|2s|^{2(n-k+1)} \widetilde{\sigma}^{2(n-k+1)}}\right]\\
&\times |2s|^2 e^{2\gamma M \sqrt{\ell(\ell+1)} \tilde{\rho}}|\Phi_{(k-1+1)}|^2\, d\tilde{\rho}\\
&-C \gamma^{-2}|s|^{-2}(\ell+1)^{-2} \int_0^{\tilde{\rho}_0} \frac{(n+1)!^2}{(2\ell)!^2}e^{2\gamma r_+ \sqrt{\ell(\ell+1)} \tilde{\rho}}|2s|^2|\hat{\Phi}_{(2\ell+1)}|^2 \,d \tilde{\rho}\\
\leq &\: C_{\epsilon} \gamma^{-1}(\ell+1)^{-3}  \sum_{k=1}^{2\ell}\gamma^{\frac{1}{2} (2\ell-k)} \frac{(n+1)!^2}{(k-1)!^2} e^{2\gamma M \sqrt{\ell(\ell+1)} \tilde{\rho}_0}|\hat{\Phi}_{(k)}|^2(\tilde{\rho}_0)\\
&+C_{\epsilon}\gamma^{-1} (\ell+1)^{-3} (n+2)^2 \sum_{k=0}^{2\ell}\gamma^{\frac{1}{2} (2\ell-k)}  \frac{n!^2}{k!^2} e^{2\gamma M \sqrt{\ell(\ell+1)} \tilde{\rho}_0}|\hat{\Phi}_{(k)}|^2(\tilde{\rho}_0)+C_{\epsilon}\int_{0}^{\tilde{\rho}_0}e^{2\gamma M \sqrt{\ell(\ell+1)} \tilde{\rho}} |f_n|^2\, d\tilde{\rho}.
\end{split}
\end{equation*}
\end{proof}

In order to estimate the boundary terms at $\tilde{\rho}=\tilde{\rho}_0$ that appear on the right-hand side of the inequalities in Proposition \ref{eq:maingevprop}, will will make use of the non-degeneracy of the wave equation at $\tilde{\rho}=\tilde{\rho}_0$ (away from the horizons) in the form of the following lemma:

\begin{lemma}
\label{lm:analytawayhor}
There exists a constant $A>0$ such that for all $2\leq n \leq 2\ell$, we can estimate
\begin{equation}
\label{eq:analytawayhor}
|\hat{\Phi}_{(n)}(\tilde{\rho}_0)|\leq A^{n} \tilde{\rho}_0^{-2(n-1)} (\ell+1)^{n}(|\hat{\Phi}|(\tilde{\rho}_0)+|\hat{\Phi}_{(1)}|(\tilde{\rho}_0))+\sum_{k=0}^{n-2} A^{n-2-k} \tilde{\rho}_0^{-2(n-2)-2+2k}(\ell+1)^{n-2-k} |f_{k}|(\tilde{\rho}_0).
\end{equation}
\end{lemma}
\begin{proof}
We will prove \eqref{eq:analytawayhor} by induction. The $n=2$ case follows immediately from \eqref{eq:fixedfreqpsihat} with $n=0$. Now suppose \eqref{eq:analytawayhor} holds for all $2\leq n\leq N$. We will show that then \eqref{eq:analytawayhor} must also hold for $n=N+1$. 

Observe first of all that for all $2\leq n \leq 2\ell$, we have that
\begin{align*}
|n(n+1)-\ell(\ell+1)|\leq&\: 3\ell(\ell+1).
\end{align*}

By  \eqref{eq:fixedfreqpsihat} with $n=N-1$, we have that there exists a constant $C>0$ independent of $N$ and $\ell$ such that
\begin{equation}
\label{eq:indstep}
\begin{split}
|\hat{\Phi}_{(N+1)}(\tilde{\rho}_0)|\leq&\: C M^{-1} N \tilde{\rho}_0^{-2} |\hat{\Phi}_{(N)}|(\tilde{\rho}_0)+ C\ell(\ell+1)\tilde{\rho}_0^{-2} |\hat{\Phi}_{(N-1)}|(\tilde{\rho}_0)\\
&+C(\kappa_++\kappa_c)\tilde{\rho}_0^{-2}\left[\sum_{k=1}^{N-1} \frac{N!}{(k-1)!} M^{n-k}|\hat{\Phi}_{(k)}|(\tilde{\rho}_0)+\sum_{k=0}^{N-1} \frac{(N-1)! (N-k)}{k!}M^{n-k}|\hat{\Phi}_{(k)}|(\tilde{\rho}_0)\right]\\
&+ C\tilde{\rho}_0^{-2} |f_{N-1}|.
\end{split}
\end{equation}
Note that by applying \eqref{eq:analytawayhor} with $n=N$ and $n=N-1$, we obtain
\begin{equation*}
\begin{split}
C M^{-1}N \tilde{\rho}_0^{-2}& |\hat{\Phi}_{(N)}|(\tilde{\rho}_0)+ C\ell(\ell+1)\tilde{\rho}_0^{-2} |\hat{\Phi}_{(N-1)}|(\tilde{\rho}_0)\\
\leq&\: (2C M^{-1} A^N+C A^{N-1}  )\tilde{\rho}_0^{-2N} (\ell+1)^{N+1}(|\hat{\Phi}|(\tilde{\rho}_0)+|\hat{\Phi}_{(1)}|(\tilde{\rho}_0))\\
&+(CM^{-1}A^{-1}+CA^{-2})\sum_{k=0}^{N-2} A^{N-1-k} \tilde{\rho}_0^{-2(N-1)-2+2k}(\ell+1)^{N-1-k} |f_{k}|(\tilde{\rho}_0).
\end{split}
\end{equation*}
Similarly, using Stirling's formula to obtain
\begin{equation*}
\frac{N!}{(k-1)!} \leq 2^{N-k+1} (\ell+1)^{N-k+1},
\end{equation*}
we can estimate
\begin{equation*}
\begin{split}
C&(\kappa_++\kappa_c)\tilde{\rho}_0^{-2}\sum_{k=1}^{N-1} \frac{N!}{(k-1)!} M^{n-k}|\hat{\Phi}_{(k)}|(\tilde{\rho}_0)\leq C(\kappa_++\kappa_c)\tilde{\rho}_0^{-2}\sum_{k=1}^{N-1} 2^{N-k+1} (\ell+1)^{N+1-k} M^{n-k}|\hat{\Phi}_{(k)}|(\tilde{\rho}_0)\\
\leq&\: C(\kappa_++\kappa_c)(\ell+1)^{N+1}A^{N+1} \tilde{\rho}_0^{-2N}\left[\sum_{k=1}^{N-1}\tilde{\rho}_0^{2(N-k)} \left(\frac{2}{A}\right)^{N-k+1}\right] (|\hat{\Phi}|(\tilde{\rho}_0)+|\hat{\Phi}_{(1)}|(\tilde{\rho}_0))\\
&+C(\kappa_++\kappa_c)(\ell+1)^{N+1}A^{N+1} \tilde{\rho}_0^{-2N}\sum_{k=1}^{N-1}\tilde{\rho}_0^{2(N-k)} \left(\frac{2}{A}\right)^{N-k+1} \sum_{m=0}^{k-2} A^{-2-m}\tilde{\rho}_0^{2+2m}(\ell+1)^{-2-m}|f_m|(\tilde{\rho}_0))\\
\leq &\:  \frac{1}{2} A^{N+1}(\ell+1)^{N+1}  \tilde{\rho}_0^{-2N}(|\hat{\Phi}|(\tilde{\rho}_0)+|\hat{\Phi}_{(1)}|(\tilde{\rho}_0))+\frac{1}{2}A^{N+1}\tilde{\rho}_0^{-2N} (\ell+1)^{N-1}\sum_{k=0}^{N-2} A^{-2-k} \tilde{\rho}_0^{2+2k}(\ell+1)^{-k} |f_{k}|(\tilde{\rho}_0),
\end{split}
\end{equation*}
for $\tilde{\rho}_0$ suitably small and $A$ chosen suitably large, where we used the convergence of the geometric series. The remaining term on the right-hand side of \eqref{eq:indstep} involving $\hat{\Phi}_{(k)}$ can be estimated similarly.

By taking the dimensionless constant $A\cdot M$ suitably large compared to the dimensionless constant $C>0$ that appears in the expressions above (and is independent of $\rho_0^{-1}$, $N$ and $\ell$), and by combining the above estimates, we obtain \eqref{eq:analytawayhor}  for $n=N+1$, thereby concluding the induction argument.
\end{proof}

\begin{proposition}
\label{prop:gevreyestsum}
For all $\epsilon>0$, there exist constants $B, C,\gamma>0$ depending on $|s|M$, such that for $\gamma{\gg}(|s|^{-1}M^{-1}+|s|M)$, $\rho_0^{-1} M^{-1}{\gg}\max\{\gamma,1\}$ and $\ell{\gg}1$, we can estimate for $N_{\infty}>2\ell$ arbitrarily large:
\begin{equation}
\label{eq:auxgevreyestsum}
\begin{split}
&\ell! (\ell+1)!\sum_{n=\ell}^{N_{\infty}} \frac{|2s|^{2n}\widetilde{\sigma}^{2n}}{(n+1)!^2n!^2} \int_{0}^{\tilde{\rho}_0} (1-\mu \nu^{-1}(1+M\tilde{\rho})^4-\epsilon)(2\kappa\tilde{\rho}+\tilde{\rho}^2)^2e^{2\gamma M \sqrt{\ell(\ell+1)} \tilde{\rho} }|\hat{\Phi}_{(n+2)}|^2\, d\tilde{\rho}\\
&+\ell! (\ell+1)!\sum_{n=\ell}^{N_{\infty}} \frac{|2s|^{2n}\widetilde{\sigma}^{2n}}{(n+1)!^2n!^2}\cdot  \\
 & \cdot \int_{0}^{\tilde{\rho}_0} \left(1-4 \beta^{-1}\left(\frac{\max\{-\re(s),0\}}{|s|}\right)^2-\nu \mu- (1+\epsilon)\widetilde{\sigma}^2 -\epsilon \right)|2s|^2e^{2\gamma M \sqrt{\ell(\ell+1)} \tilde{\rho} }|\hat{\Phi}_{(n+1)}|^2\, d\tilde{\rho}\\
&+\ell! (\ell+1)!\sum_{n=\ell}^{N_{\infty}} \frac{|2s|^{2n}\widetilde{\sigma}^{2n}}{(n+1)!^2n!^2}\int_{0}^{\tilde{\rho}_0}  (4-\beta-(\alpha+\alpha^{-1}\widetilde{\sigma}^{-2})(1-\mu)-\epsilon)(n+1)^2(\tilde{\rho}+\kappa)^2e^{2\gamma M \sqrt{\ell(\ell+1)} \tilde{\rho} }|\hat{\Phi}_{(n+1)}|^2 \, d\tilde{\rho}\\
&-\ell! (\ell+1)! \int_{0}^{\tilde{\rho}_0} \frac{|2s|^{2N_{\infty}}\widetilde{\sigma}^{2N_{\infty}}}{(N_{\infty}+1)!^2N_{\infty}!^2}\frac{|2s|^{2}}{(N_{\infty}+1)^2}\cdot 5\alpha^{-1}(1-\mu)(1+M\tilde{\rho})^4\tilde{\rho}^2e^{2\gamma M \sqrt{\ell(\ell+1)} \tilde{\rho} }|\hat{\Phi}_{(N_{\infty}+2)}|^2 \, d\tilde{\rho}\\
\leq &\: C B^{2\ell} \tilde{\rho}_0^{-8\ell} |2s|^{\ell} e^{2\gamma M \sqrt{\ell(\ell+1)} \tilde{\rho}_0} \left[|\hat{\Phi}|^2(\tilde{\rho}_0)+|\hat{\Phi}_{(1)}|^2(\tilde{\rho}_0)\right]\\
&+C B^{2\ell} \tilde{\rho}_0^{-8\ell} |2s|^{\ell} e^{2\gamma M \sqrt{\ell(\ell+1)} \tilde{\rho}_0}  \sum_{k=0}^{2\ell-2} A^{-4-4k} \tilde{\rho}_0^{4+4k}(\ell+1)^{-2k} |f_{k}|^2(\tilde{\rho}_0)\\
&+C\ell! (\ell+1)!\sum_{n=\ell}^{N_{\infty}} \frac{|2s|^{2n}\widetilde{\sigma}^{2n}}{(n+1)!^2n!^2}\int_{0}^{\tilde{\rho}_0}e^{2\gamma M \sqrt{\ell(\ell+1)} \tilde{\rho} }|f_n|^2\, d\tilde{\rho}.
\end{split}
\end{equation}
\end{proposition}
\begin{proof}
We will sum the inequalities appearing in Proposition \ref{eq:maingevprop} over $n$ with the following summation weights:
\begin{equation*}
\ell! (\ell+1)!\sum_{n=\ell}^{N_{\infty}} \frac{|2s|^{2n}\widetilde{\sigma}^{2n}}{(n+1)!^2n!^2} \Big[\cdot\Big],
\end{equation*}
where $N_{\infty}$ will be taken suitably large and the value of $\tilde{\sigma}$ will be determined later.

Before we carry out the summation, we will first apply Lemma \ref{lm:analytawayhor} to estimate further the boundary terms at $\tilde{\rho}=\tilde{\rho}_0$ that appear in the inequalities of Proposition \ref{eq:maingevprop}.

By Lemma \ref{lm:analytawayhor}, we have that there exists a constant $B>0$ such that for all $0\leq k\leq \min\{n,2\ell\}$, 
\begin{equation*}
\begin{split}
\frac{(n+1)!^2}{(k-1)!^2}& \gamma^{\frac{1}{2}(n_{\ell}-k)}e^{2\gamma M \sqrt{\ell(\ell+1)} \tilde{\rho}_0}|\hat{\Phi}_{(k)}|^2(\tilde{\rho}_0)\leq B^k\tilde{\rho}_0^{-4k} (n+1)!^2 \left(\frac{\ell+1}{k+1}\right)^{2(k+1)}  \gamma^{\frac{1}{2}(n_{\ell}-k)}e^{2\gamma M \sqrt{\ell(\ell+1)} \tilde{\rho}_0} \\
\cdot &\left[|\hat{\Phi}|^2(\tilde{\rho}_0)+|\hat{\Phi}_{(1)}|^2(\tilde{\rho}_0)+\sum_{m=0}^{k-2} A^{-4-4m} \tilde{\rho}_0^{4+4m}(\ell+1)^{-2m} |f_{m}|^2(\tilde{\rho}_0) \right].
\end{split}
\end{equation*}
Note that $\left(\frac{\ell+1}{k+1}\right)^{2(k+1)}$ attains its maximum when $\ell+1=e (k+1)$ so by redefining $B$ we can write
\begin{equation*}
\begin{split}
\frac{(n+1)!^2}{(k-1)!^2}& \gamma^{\frac{1}{2}(n_{\ell}-k)}e^{2\gamma M \sqrt{\ell(\ell+1)} \tilde{\rho}_0}|\hat{\Phi}_{(k)}|^2(\tilde{\rho}_0)\leq B^k\tilde{\rho}_0^{-4k} (n+1)!^2 \gamma^{\frac{1}{2}(n_{\ell}-k)} e^{2\gamma M \sqrt{\ell(\ell+1)} \tilde{\rho}_0}\cdot \\
\cdot &\left[|\hat{\Phi}|^2(\tilde{\rho}_0)+|\hat{\Phi}_{(1)}|^2(\tilde{\rho}_0)+\sum_{m=0}^{k-2} A^{-4-4m} \tilde{\rho}_0^{4+4m}(\ell+1)^{-2m} |f_{m}|^2(\tilde{\rho}_0) \right].
\end{split}
\end{equation*}
And hence,
\begin{equation*}
\begin{split}
 \ell! (\ell+1)!\sum_{n=\ell}^{N_{\infty}} \sum_{k=0}^{n_{\ell}}&\frac{|2s|^{2n}\widetilde{\sigma}^{2n}}{(n+1)!^2n!^2} \frac{(n+1)!^2}{(k-1)!^2}\gamma^{\frac{1}{2}(n_{\ell}-k)}e^{2\gamma M \sqrt{\ell(\ell+1)} \tilde{\rho}_0}|\hat{\Phi}_{(k)}|^2(\tilde{\rho}_0)\\
 \leq&\:  \ell! (\ell+1)!e^{2\gamma M \sqrt{\ell(\ell+1)} \tilde{\rho}_0} \sum_{n=\ell}^{N_{\infty}}\frac{|2s|^{2n}\widetilde{\sigma}^{2n}}{n!^2}\gamma^{\frac{1}{2}n_{\ell}} \sum_{k=0}^{n_{\ell}}B^{k} \tilde{\rho}_0^{-4k}\gamma^{-\frac{1}{2}k}\cdot   \\
&\cdot \left[|\hat{\Phi}|^2(\tilde{\rho}_0)+|\hat{\Phi}_{(1)}|^2(\tilde{\rho}_0)+\sum_{m=0}^{k-2} A^{-4-4m} \tilde{\rho}_0^{4+4m}(\ell+1)^{-2m} |f_{m}|^2(\tilde{\rho}_0) \right]\\
\leq&\: \ell! (\ell+1)!e^{2\gamma M \sqrt{\ell(\ell+1)} \tilde{\rho}_0} \sum_{n=\ell}^{N_{\infty}}\frac{|2s|^{2n}\widetilde{\sigma}^{2n}}{n!^2}B^{n_{\ell}} \tilde{\rho}_0^{-4n_{\ell}}\cdot   \\
&\cdot \sum_{k=0}^{n_{\ell}}B^{-k} \tilde{\rho}_0^{4k}\gamma^{\frac{1}{2}k} \left[|\hat{\Phi}|^2(\tilde{\rho}_0)+|\hat{\Phi}_{(1)}|^2(\tilde{\rho}_0)+\sum_{m=0}^{2\ell-2} A^{-4-4m} \tilde{\rho}_0^{4+4m}(\ell+1)^{-2m} |f_{m}|^2(\tilde{\rho}_0) \right]\\
 \leq&\:C \ell! (\ell+1)!e^{2\gamma M \sqrt{\ell(\ell+1)} \tilde{\rho}_0} \sum_{n=\ell}^{N_{\infty}}\frac{|2s|^{2n}\widetilde{\sigma}^{2n}}{n!^2}B^{n_{\ell}} \tilde{\rho}_0^{-4n_{\ell}}\cdot   \\
&\cdot \left[|\hat{\Phi}|^2(\tilde{\rho}_0)+|\hat{\Phi}_{(1)}|^2(\tilde{\rho}_0)+\sum_{m=0}^{2\ell-2} A^{-4-4m} \tilde{\rho}_0^{4+4m}(\ell+1)^{-2m} |f_{m}|^2(\tilde{\rho}_0) \right]\\
\leq &\: C B^{2\ell} |2s|^{4\ell}\tilde{\rho}_0^{-8\ell}e^{2\gamma M \sqrt{\ell(\ell+1)} \tilde{\rho}_0} \left[|\hat{\Phi}|^2(\tilde{\rho}_0)+|\hat{\Phi}_{(1)}|^2(\tilde{\rho}_0)\right]\\
&+C B^{2\ell} |2s|^{4\ell}\tilde{\rho}_0^{-8\ell} e^{2\gamma M \sqrt{\ell(\ell+1)} \tilde{\rho}_0}  \sum_{k=0}^{2\ell-2} A^{-4-4k} \tilde{\rho}_0^{4+4k}(\ell+1)^{-2k} |f_{k}|^2(\tilde{\rho}_0).
\end{split}
\end{equation*}
We can analogously estimate the remaining boundary terms.

We now take the sum over $n$ of the equations in Proposition \ref{eq:maingevprop} with $\ell$ suitably large. We will then see, in particular, the following double summation if $n>2\ell$:
\begin{equation*}
\begin{split}
(\ell+1)^{-2}&\sum_{n=2\ell+1}^{N_{\infty}}\sum_{k=2\ell+1}^{n-1}\int_{0}^{\tilde{\rho}_0}M^{2(n-k)}\widetilde{\sigma}^{2(n-k+1)}|2s|^{2(n-k)}\frac{(n-2\ell) k!^2}{n!^2}\left[\frac{|2s|^{2(k-1)} \widetilde{\sigma}^{2(k-1)}}{k!^2(k-1)!^2}\right]\cdot\\
& \cdot |2s|^2 e^{2\gamma M \sqrt{\ell(\ell+1)} \tilde{\rho}}|\hat{\Phi}_{(k-1+1)}|^2\, d\tilde{\rho}\\
=&(\ell+1)^{-2}\sum_{k=2\ell+1}^{N_{\infty}}\sum_{n=k+1}^{N_{\infty}}M^{2(n-k)}\widetilde{\sigma}^{2(n-k+1)}|2s|^{2(n-k)}\frac{(n-2\ell) k!^2}{n!^2}\int_{0}^{\tilde{\rho}_0}\left[\frac{|2s|^{2(k-1)} \widetilde{\sigma}^{2(k-1)}}{k!^2(k-1)!^2}\right]\cdot\\
& \cdot |2s|^2 e^{2\gamma M \sqrt{\ell(\ell+1)} \tilde{\rho}}|\hat{\Phi}_{(k-1+1)}|^2\, d\tilde{\rho}
\end{split}
\end{equation*}
We further estimate the factor:
\begin{equation*}
\begin{split}
\sum_{n=k+1}^{N_{\infty}}M^{2(n-k)}\widetilde{\sigma}^{2(n-k+1)}|2s|^{2(n-k)}\frac{(n-2\ell) k!^2}{n!^2}\leq &\:\widetilde{\sigma}^{2}\sum_{m=1}^{\infty}\widetilde{\sigma}^{2m}M^{2m}|2s|^{2m}\frac{(m+k-2\ell) k!^2}{(m+k)!^2}\\
\leq &\: C \widetilde{\sigma}^{2} \sum_{m=0}^{\infty}\frac{1}{(m+1)^2}\leq C \widetilde{\sigma}^{2},
\end{split}
\end{equation*}
for $\ell$ is sufficiently large compared to $M|s|$ and $\widetilde{\sigma}$.

Now, in view of the chosen summation, the terms with a negative sign on the left-hand sides of the inequalities in Proposition \ref{eq:maingevprop} at order $n$ can be directly absorbed into the terms with a positive sign in the inequalities in Proposition \ref{eq:maingevprop} at one order higher or lower, i.e.\ with $n$ replaced by either $n-1$ or $n+1$. Note that when $n=\ell$, only absorption into the estimates at order $n+1$ is necessary.

The absorption into the inequalities with $n$ replaced by $n+1$ is no longer possible at top order, i.e.\ when $n=N_{\infty}$. In that case, we moreover apply a slight variation of the estimates in the proof of Proposition \ref{eq:maingevprop}: we do not apply \eqref{eq:alpha}, but we estimate instead:
\begin{equation*}
\begin{split}
2&(1-\mu)|2s|(1+M\tilde{\rho})^2[2\kappa\tilde{\rho}+\tilde{\rho}^2]e^{2\gamma M \sqrt{\ell(\ell+1)} \tilde{\rho} }|\hat{\Phi}_{(N_{\infty}+1)}|  |\hat{\Phi}_{(N_{\infty}+2)}| \\
\leq &\:\alpha(1-\mu) (n+1)^2 (\tilde{\rho}^2+\kappa^2)e^{2\gamma M \sqrt{\ell(\ell+1)} \tilde{\rho} }|\hat{\Phi}_{(N_{\infty}+1)}|^2\\
&+ 5\frac{|2s|^2}{(N_{\infty}+1)^2}\cdot \alpha^{-1}(1-\mu)(1+M\tilde{\rho})^4 \tilde{\rho}^2e^{2\gamma M \sqrt{\ell(\ell+1)} \tilde{\rho} }|\hat{\Phi}_{(N_{\infty}+2)}|^2.
\end{split}
\end{equation*}
Let $\epsilon>0$ and $\ell$ appropriately large, then we can group the remaining terms in the summation of the equations in Proposition \ref{eq:maingevprop} to obtain \eqref{eq:auxgevreyestsum}.

\end{proof}
\begin{remark}
It is important to note that no terms involving $\hat{\Phi}_{(n)}$ with $n\leq \ell$ appear on the right-hand side of \eqref{prop:gevreyestsum} apart from the $n=0$ and $n=1$ boundary terms. Hence, the $n\geq \ell$ estimates are only coupled via the lowest order boundary terms at $\rho=\rho_0$. This makes it considerably easier to add estimates for the $n\leq \ell$ terms; see Proposition \ref{prop:maingevreysum} below.
\end{remark}

\begin{proposition}
\label{prop:maingevreysum}
For all $\epsilon>0$, there exist a constant $C>0$,  with $C$ independent of $\ell$, such that for $\gamma{\gg}(|s|^{-1}M^{-1}+|s|M)$, $\rho_0^{-1} M^{-1}{\gg}\max\{\gamma,1\}$, $\ell{\gg}1$ and $\kappa_c+\kappa_+\lesssim (\ell+1)^{-1}$, we can estimate for $N_{\infty}>2\ell$ arbitrarily large:

\begin{equation}
\label{eq:maingevreysum0}
\begin{split}
\sum_{n=0}^{\ell}& \frac{|2s|^{2n} \tilde{\sigma}^{2n}}{(\ell(\ell+1))^{n-1}}\int_{0}^{\tilde{\rho}_0}e^{2\gamma M \sqrt{\ell(\ell+1)} \tilde{\rho}} |\hat{\Phi}_{(n)}|^2 \,d\tilde{\rho}\\
&+\ell! (\ell+1)!\sum_{n=\ell}^{N_{\infty}} \frac{|2s|^{2n}\widetilde{\sigma}^{2n}}{(n+1)!^2n!^2} \int_{0}^{\tilde{\rho}_0} (1-\mu \nu^{-1}(1+M\tilde{\rho})^4-\epsilon)(2\kappa\tilde{\rho}+\tilde{\rho}^2)^2e^{2\gamma M \sqrt{\ell(\ell+1)} \tilde{\rho} }|\hat{\Phi}_{(n+2)}|^2\, d\tilde{\rho}\\
&+\ell! (\ell+1)!\sum_{n=\ell}^{N_{\infty}} \frac{|2s|^{2n}\widetilde{\sigma}^{2n}}{(n+1)!^2n!^2}\cdot \\
 & \cdot \int_{0}^{\tilde{\rho}_0} \left(1-4 \beta^{-1}\left(\frac{\max\{-\re(s),0\}}{|s|}\right)^2-\nu \mu- (1+\epsilon)\widetilde{\sigma}^2 -\epsilon \right)|2s|^2e^{2\gamma M \sqrt{\ell(\ell+1)} \tilde{\rho} }|\hat{\Phi}_{(n+1)}|^2\, d\tilde{\rho}\\
&+\ell! (\ell+1)!\sum_{n=\ell}^{N_{\infty}} \frac{|2s|^{2n}\widetilde{\sigma}^{2n}}{(n+1)!^2n!^2}\int_{0}^{\tilde{\rho}_0}  (4-\beta-(\alpha+\alpha^{-1}\widetilde{\sigma}^{-2})(1-\mu)-\epsilon)(n+1)^2(\tilde{\rho}+\kappa)^2e^{2\gamma M \sqrt{\ell(\ell+1)} \tilde{\rho} }|\hat{\Phi}_{(n+1)}|^2 \, d\tilde{\rho}\\
&-\ell! (\ell+1)! \int_{0}^{\tilde{\rho}_0} \frac{|2s|^{2N_{\infty}}\widetilde{\sigma}^{2N_{\infty}}}{(N_{\infty}+1)!^2N_{\infty}!^2}\frac{|2s|^{2}}{(N_{\infty}+2)^2}\cdot 5\alpha^{-1}(1-\mu)(1+M\tilde{\rho})^4\tilde{\rho}^2e^{2\gamma M \sqrt{\ell(\ell+1)} \tilde{\rho} }|\hat{\Phi}_{(N_{\infty}+2)}|^2 \, d\tilde{\rho}\\
\leq &\:C B^{2\ell} \tilde{\rho}_0^{-8\ell}|2s|^{4\ell} e^{2\gamma M \sqrt{\ell(\ell+1)} \tilde{\rho}_0} \left[|\hat{\Phi}|^2(\tilde{\rho}_0)+|\hat{\Phi}_{(1)}|^2(\tilde{\rho}_0)\right]\\
&+C B^{2\ell} \tilde{\rho}_0^{-8\ell} |2s|^{4\ell} e^{2\gamma M \sqrt{\ell(\ell+1)} \tilde{\rho}_0}  \sum_{k=0}^{2\ell-2}B^{-2n}\tilde{\rho}_0^{4n}|2s|^{-2n}(\ell+1)^{-2k} |f_{k}|^2(\tilde{\rho}_0)\\
&+C\ell! (\ell+1)!\sum_{n=\ell}^{N_{\infty}} \frac{|2s|^{2n}\widetilde{\sigma}^{2n}}{(n+1)!^2n!^2}\int_{0}^{\tilde{\rho}_0}e^{2\gamma M \sqrt{\ell(\ell+1)} \tilde{\rho} }|f_n|^2\, d\tilde{\rho}. 
\end{split}
\end{equation}

\end{proposition}
\begin{proof}
We apply Lemma \ref{lm:expweightest} $\ell+1-n$ times to estimate for all $0\leq n\leq \ell$
\begin{equation*}
\begin{split}
\int_{0}^{\tilde{\rho}_0}& e^{2\gamma M \sqrt{\ell(\ell+1)} \tilde{\rho}}\frac{ |2s|^{2n} \tilde{\sigma}^{2n}|\hat{\Phi}_{(n)}|^2}{(\ell(\ell+1))^{n-1}}\,d\tilde{\rho}\leq  C |2s|^{2n} \tilde{\sigma}^{2n} (M\gamma)^{-2(\ell+1-n)} (\ell(\ell+1))^{-\ell} \int_{0}^{\tilde{\rho}_0} e^{2\gamma M \sqrt{\ell(\ell+1)} \tilde{\rho}}|\hat{\Phi}_{(\ell+1)}|^2\,d\tilde{\rho}\\
&+ C(M\gamma)^{-1}(\ell(\ell+1))^{-\frac{1}{2}} |2s|^{2n} \tilde{\sigma}^{2n}\sum_{k=n}^{\ell} (M\gamma)^{-2(k-n)}   (\ell(\ell+1))^{-(k-1)}e^{2\gamma M \sqrt{\ell(\ell+1)} \tilde{\rho}_0} |\hat{\Phi}_{(k)}|^2(\tilde{\rho}_0).
\end{split}
\end{equation*}
After summing over $n$, and using that $\ell{\gg} M\gamma |2s|$, we obtain
\begin{equation*}
\begin{split}
\sum_{n=0}^{\ell}& |2s|^{2n} \tilde{\sigma}^{2n}\sum_{k=n}^{\ell} (M\gamma)^{-2(k-n)}   (\ell(\ell+1))^{-(k-1)}e^{2\gamma M \sqrt{\ell(\ell+1)} \tilde{\rho}_0} |\hat{\Phi}_{(k)}|^2(\tilde{\rho}_0)\\
\leq&\: \sum_{n=0}^{\ell}(M\gamma)^{-2n}(\ell(\ell+1))^{1-n} \left[\sum_{k=0}^{n}  |2s|^{2k}\tilde{\sigma}^{2k}(M\gamma)^{2k}\right]e^{2\gamma M \sqrt{\ell(\ell+1)} \tilde{\rho}_0} |\hat{\Phi}_{(n)}|^2(\tilde{\rho}_0)\\
\leq&\:\sum_{n=0}^{\ell}\frac{|2s|^{2n}\tilde{\sigma}^{2n}}{(\ell(\ell+1))^{n-1}}e^{2\gamma M \sqrt{\ell(\ell+1)} \tilde{\rho}_0} |\hat{\Phi}_{(n)}|^2(\tilde{\rho}_0)
\end{split}
\end{equation*}
and hence,
\begin{equation*}
\begin{split}
\sum_{n=0}^{\ell}& \frac{|2s|^{2n} \tilde{\sigma}^{2n}}{(\ell(\ell+1))^{n-1}}\int_{0}^{\tilde{\rho}_0}e^{2\gamma M \sqrt{\ell(\ell+1)} \tilde{\rho}} |\hat{\Phi}_{(n)}|^2 \,d\tilde{\rho}\\
\leq&\: C (M|2s| \tilde{\sigma} \gamma)^{-2}\frac{|s|^{2(\ell+1)} \tilde{\sigma}^{2(\ell+1)}}{(\ell(\ell+1))^{\ell}} \int_{0}^{\tilde{\rho}_0} e^{2\gamma M \sqrt{\ell(\ell+1)} \tilde{\rho}}|\hat{\Phi}_{(\ell+1)}|^2\,d\tilde{\rho}\\
&+ C(M\gamma)^{-1}(\ell(\ell+1))^{\frac{1}{2}}\sum_{n=0}^{\ell}\frac{|2s|^{2n}\tilde{\sigma}^{2n}}{(\ell(\ell+1))^{n}}e^{2\gamma M \sqrt{\ell(\ell+1)} \tilde{\rho}_0} |\hat{\Phi}_{(n)}|^2(\tilde{\rho}_0).
\end{split}
\end{equation*}

By Lemma \ref{lm:analytawayhor} we have that there exists a constant $B>0$ such that
\begin{equation*}
\begin{split}
(M\gamma)^{-1}&(\ell(\ell+1))^{\frac{1}{2}}\sum_{n=0}^{\ell}\frac{|2s|^{2n}\tilde{\sigma}^{2n}}{(\ell(\ell+1))^{n}}e^{2\gamma M \sqrt{\ell(\ell+1)} \tilde{\rho}_0} |\hat{\Phi}_{(n)}|^2(\tilde{\rho}_0)\\
\leq&\:  (M\gamma)^{-1} (\ell(\ell+1))^{\frac{1}{2}}\max_{0\leq n\leq \ell-2}\{B^{2n}\tilde{\rho}_0^{-4n}|2s|^{2n}\} (|\hat{\Phi}|^2(\tilde{\rho_0})+|\hat{\Phi}_{(1)}|^2(\tilde{\rho_0}))\\
&+  (M\gamma)^{-1} (\ell(\ell+1))^{\frac{1}{2}} B^{2\ell}\tilde{\rho}_0^{-4\ell}|2s|^{2\ell}\sum_{n=0}^{\ell-2}B^{-2n}\tilde{\rho}_0^{4n}|2s|^{-2n}(\ell(\ell+1))^{-2(n+2)} |f_n|^2(\tilde{\rho}_0).
\end{split}
\end{equation*}
By combining the above results with the estimates in Proposition \ref{prop:gevreyestsum}, we obtain \eqref{eq:maingevreysum0}.
\end{proof}

In order to ensure that the terms on the left-hand side of \eqref{eq:maingevreysum0} are positive-definite for $\kappa$ and $\alpha$ suitably small, we need to impose several \emph{compatibility conditions} on the parameters $\alpha,\beta,\mu,\nu,\widetilde{\sigma}$. 

For the sake of convenience, we introduce the notation $\varsigma:=\frac{|\re(s)|^2}{|s|^2}$. When $\re(s)<0$, the compatibility conditions are: if $\mu>0$, then
\begin{align}
\label{paramconst1}
0<\mu<&\:\nu,\\
\label{paramconst2}
4-\alpha(1-\mu)-\beta-\widetilde{\sigma}^{-2}\alpha^{-1}(1-\mu)>&\:0,\\
\label{paramconst3}
1-4\beta^{-1}\varsigma-\widetilde{\sigma}^2-\mu^2>&\:0.
\end{align}
If $\re(s)\geq 0$, we instead obtain:
\begin{align}
\label{paramconst1b}
\mu<&\:\nu,\\
\label{paramconst2b}
4-\alpha(1-\mu)-\beta-\widetilde{\sigma}^{-2}\alpha^{-1}(1-\mu)>&\:0,\\
\label{paramconst3b}
1-\mu^2>&\:0.
\end{align}
Note that if we take $\mu=0$, we can omit the parameter $\nu$ in the above expressions.

  \begin{figure}[H]
	\begin{center}
				\includegraphics[scale=0.5]{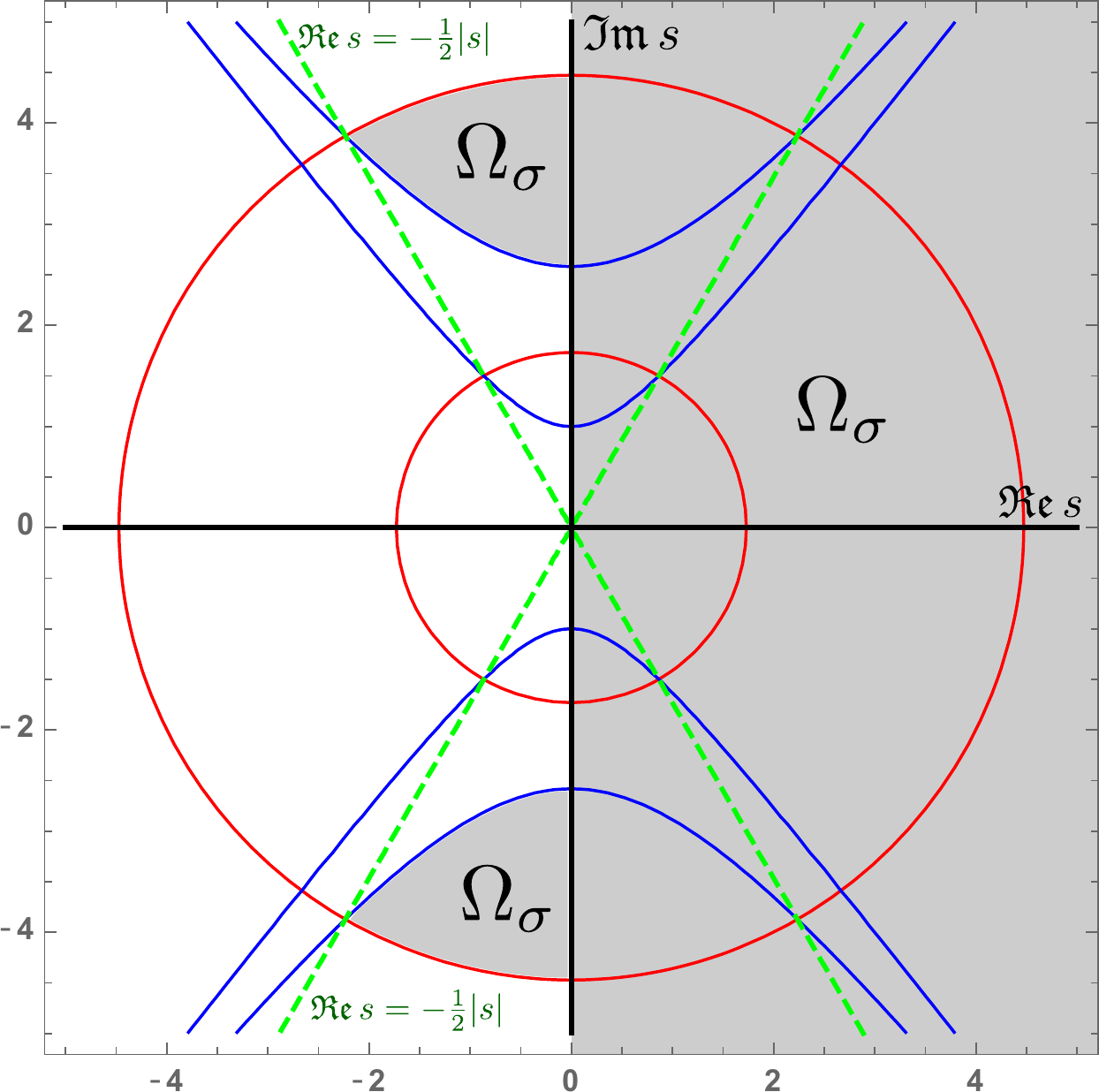}
\end{center}
\vspace{-0.2cm}
\caption{An example of a domain $\Omega_{\sigma}$, which is represented by the shaded region in the figure. The points satisfying $\re(s)=-\frac{1}{2}|s|$ are represented by the two dashed lines in the left half-plane in the picture.}
	\label{fig:res}
\end{figure}
\begin{lemma}
\label{lm:allowedvaluess}

\begin{enumerate}[\rm (i)]
\item
There exist $\alpha,\beta,\mu,\nu$ and $\widetilde{\sigma}^2$ satisfying \eqref{paramconst1}, \eqref{paramconst2} and \eqref{paramconst3} if and only if, when $\re(s)< 0$, we have that $|\re s|<\frac{1}{2}|s|$, or equivalently $|\textnormal{arg}(s)|<\frac{2}{3}\pi$.
\item
Fix the product $\sigma:=2\widetilde{\sigma} |s| \in \R$. Then, there exist $\alpha,\beta,\mu,\nu$ satisfying \eqref{paramconst1}, \eqref{paramconst2} and \eqref{paramconst3} if and only if, when $\re s< 0$, we have that:
\begin{align*}
|s|^2<&\:\sigma^2,\\
3|\im s|^2-5|\re s |^2>&\:\sigma^2.
\end{align*}
If we denote 
\begin{equation*}
\Omega_{\sigma}:=\left\{(x,y)\in \C\:|\: x<0,\: x^2+y^2<\sigma^2,\: 3y^2-5x^2>{\sigma}^2\right\}\cup\{(x,y)\in \C\:| x\geq 0,\, (x,y)\neq (0,0)\}\subset \C
\end{equation*}
then
\begin{equation*}
\bigcup_{\sigma \in \R_{>0}}\Omega_{\sigma}= \left\{(x,y)\in \C\:| x<0,\, x^2<\frac{1}{4}(x^2+y^2)\right\}\cup\{(x,y)\in \C\:| x\geq 0,\, (x,y)\neq (0,0)\}
\end{equation*}
\item Fix $s\in \{-\frac{1}{2}|s|\leq \re s\leq 0 \}$. There exist $\alpha,\beta,\mu,\nu$ and $\widetilde{\sigma}^2$ satisfying \eqref{paramconst1}, \eqref{paramconst2} and \eqref{paramconst3} if
\begin{equation*}
\frac{1}{4}<\tilde{\sigma}^2<\frac{3}{4}-2\left(\frac{|\re s|}{|s|}\right)^2.
\end{equation*}
\end{enumerate}
\end{lemma}
\begin{proof}
There exist $\alpha,\beta,\mu$ such that \eqref{paramconst2} and \eqref{paramconst3} hold if and only if
\begin{equation*}
\frac{1}{(4-\beta)(1-\mu)^{-1}\alpha-\alpha^2}<\widetilde{\sigma}^2<1-4\beta^{-1}\varsigma-\mu^2.
\end{equation*}
As a first step, we minimise the lower bound for $\widetilde{\sigma}^2$ in $\alpha$. It it easy to see that it is minimized if
\begin{equation*}
\alpha=\frac{1}{2}(4-\beta)(1-\mu)^{-1},
\end{equation*}
so with this choice of $\alpha$, we obtain
\begin{equation}
\label{eq:boundssigma}
\frac{4 (1-\mu)^2}{(4-\beta)^2}<\widetilde{\sigma}^2<1-4\beta^{-1}\varsigma-\mu^2.
\end{equation}

We now assume $\re(s)<0$.  We can find a $\widetilde{\sigma}^2$ satisfying \eqref{eq:boundssigma} if and only if 
\begin{equation*}
4\varsigma<2(2-b)\left[1-\mu^2-(1-\mu)^2b^{-2}\right],
\end{equation*}
where $2b:=4-\beta$.

Now, we will find a $\mu\in [0,1]$ such that
\begin{equation*}
1-\mu^2-(1-\mu)^2b^{-2}
\end{equation*}
is maximized.

Equivalently, we would like
\begin{equation*}
-2\mu+2(1-\mu)b^{-2}=0
\end{equation*}
which is the case when
\begin{equation*}
\mu=\frac{b^{-2}}{1+b^{-2}}=\frac{1}{1+b^2}.
\end{equation*}
Using the above choice of $\mu$, we are left with
\begin{equation*}
4\varsigma<2(2-b)\left[1-\frac{1}{(1+b^2)^2}-\frac{b^4}{b^2(1+b^2)^2}\right]=\frac{2(2-b)b^2}{1+b^2},
\end{equation*}
We have that
\begin{equation*}
\frac{d}{db}\left(\frac{2(2-b)b^2}{1+b^2}\right)=-\frac{2b(b-1)(b^2+b+4)}{(b^2+1)^2},
\end{equation*}
with $b\leq 2$. The above expression must attain a maximum when $b\in (0,2)$ at $b=1$ and we can conclude that
\begin{equation*}
4\varsigma<1,
\end{equation*}
or equivalently, we need to restrict
\begin{equation*}
|\re s|<\frac{1}{2}|s|,
\end{equation*}
after taking $\beta=2$, $\mu=\frac{1}{2}$ and $\alpha=2$.

Now assume $\re s\geq 0$. Then we can take $\mu=0$, $\beta=0$ and $\alpha=2$, and \eqref{paramconst2b} reduces to:
\begin{equation*}
\tilde{\sigma}^2>\frac{1}{4},
\end{equation*}
so no restriction on the allowed values of $s$ is required in this case. We have established $\rm (i)$.

We now turn to $\rm (ii)$. Fix $\sigma\in \R_{>0}$. Suppose $\re(s)<0$. From the above, we can moreover infer that with the choice $\alpha=2$, $\mu=\frac{1}{2}$ and $b=1$, we have that
\begin{equation}
\label{eq:conditionsigma}
\frac{1}{4}<\widetilde{\sigma}^2<\frac{3}{4}-2\varsigma.
\end{equation}
Hence, if we fix $\sigma^2=\widetilde{\sigma}^2|2s|^2$, we have that
\begin{align*}
|s|^2<&\:{\sigma}^2,\\
3|\im s|^2-5|\re s|^2>&\:{\sigma}^2.
\end{align*}
The statements in $\rm (ii)$ and $\rm (iii)$ now follow immediately.
\end{proof}

\begin{remark}
The choice of domains $\Omega_{\sigma}$ in Lemma \ref{lm:allowedvaluess} is sufficient for foliating $\{\re(z)>\frac{1}{2}|z|\}$, but it is not optimally chosen. One can for example show that all values of $s$ the subset $\{|z|>S\}\cap \{\re(z)>-C_0\}$, with arbitrary $C_0>0$ and $S>0$ suitably large can be arranged to satisfy \eqref{paramconst1}, \eqref{paramconst2} and \eqref{paramconst3} with a single choice of $\sigma$.
\end{remark}

\begin{corollary}
\label{cor:maingevreysumbound}
Fix $\sigma \in \R_{>0}$ and let $s\in \Omega_{\sigma}\subset \C$. Assume that $\kappa_c+\kappa_+\lesssim (\ell+1)^{-1}$. Then there exist $M\tilde{\rho}_0, \gamma, L>0$ suitably large depending on $s$, and constants $A, C>0$, depending on $s$, $M\tilde{\rho}_0$ and $\gamma$, such that for all $\ell\geq L$ and $N_{\infty}>2\ell$:
\begin{equation}
\label{eq:maingevreysum}
\begin{split}
\sum_{n=0}^{\ell}& \frac{|2s|^{2n} \tilde{\sigma}^{2n}}{(\ell(\ell+1))^{n-1}}\int_{0}^{\tilde{\rho}_0}e^{2\gamma M \sqrt{\ell(\ell+1)} \tilde{\rho}} |\hat{\Phi}_{(n)}|^2 \,d\tilde{\rho}\\
&+\ell! (\ell+1)!\sum_{n=\ell}^{N_{\infty}} \frac{|2s|^{2n}\widetilde{\sigma}^{2n}}{(n+1)!^2n!^2} \int_{0}^{\tilde{\rho}_0} (2\kappa\tilde{\rho}+\tilde{\rho}^2)^2e^{2\gamma M \sqrt{\ell(\ell+1)} \tilde{\rho} }|\hat{\Phi}_{(n+2)}|^2\, d\tilde{\rho}\\
&+\ell! (\ell+1)!\sum_{n=\ell}^{N_{\infty}} \frac{|2s|^{2n}\widetilde{\sigma}^{2n}}{(n+1)!^2n!^2} \int_{0}^{\tilde{\rho}_0} [|2s|^2+(n+1)^2(\tilde{\rho}+\kappa)^2]e^{2\gamma M \sqrt{\ell(\ell+1)} \tilde{\rho} }|\hat{\Phi}_{(n+1)}|^2\, d\tilde{\rho}\\
&-C\ell! (\ell+1)! \int_{0}^{\tilde{\rho}_0} \frac{|2s|^{2N_{\infty}}\widetilde{\sigma}^{2N_{\infty}}}{(N_{\infty}+1)!^2N_{\infty}!^2}\frac{|2s|^{2}}{(N_{\infty}+1)^2}\tilde{\rho}^2e^{2\gamma M \sqrt{\ell(\ell+1)} \tilde{\rho} }|\hat{\Phi}_{(N_{\infty}+2)}|^2 \, d\tilde{\rho}\\
\leq &\: A^{\ell}  e^{2\gamma M \sqrt{\ell(\ell+1)} \tilde{\rho}_0} \left[|\hat{\Phi}|^2(\tilde{\rho}_0)+|\hat{\Phi}_{(1)}|^2(\tilde{\rho}_0)\right]\\
&+A^{\ell} \sum_{n=0}^{\ell-1} \frac{|2s|^{2n} \tilde{\sigma}^{2n}}{(\ell(\ell+1))^{n-1}}\int_{0}^{\tilde{\rho}_0}e^{2\gamma M \sqrt{\ell(\ell+1)} \tilde{\rho}} |f_n|^2 \,d\tilde{\rho}+A^{\ell}\ell! (\ell+1)!\sum_{n=\ell}^{N_{\infty}} \frac{|2s|^{2n}\widetilde{\sigma}^{2n}}{(n+1)!^2n!^2}\int_{0}^{\tilde{\rho}_0}e^{2\gamma M \sqrt{\ell(\ell+1)} \tilde{\rho} }|f_n|^2\, d\tilde{\rho}.
\end{split}
\end{equation}
\end{corollary}
\begin{proof}
We apply Lemma \ref{lm:allowedvaluess} to conclude that we can choose $\alpha,\beta, \mu,\nu$ such that the constants in front of all but the final integral on the left-hand side of \eqref{eq:maingevreysum0} are positive, provided $s\in \Omega_{\sigma}$. Then we estimate the boundary terms involving $f_n$ on the right-hand side of \eqref{eq:maingevreysum0} in terms of integral norms by applying a standard Sobolev inequality.
\end{proof}

\begin{remark}
It is important to note that the constants appearing in the estimate \eqref{eq:maingevreysum} depend exponentially on $\ell$. The growth of the constants in $\ell$ turns out to be sufficiently slow so as to allow for a \emph{coupling} of \eqref{eq:maingevreysum} with the estimate \eqref{eq:ellipticestext} in Proposition \ref{prop:degellipticest} in order to convert \eqref{eq:maingevreysum} into a \emph{global} estimate; see Proposition \ref{prop:maingevreysumv2} below.
\end{remark}

\begin{proposition}
\label{prop:maingevreysumv2}
Let $\ell\in \N$ and assume $\kappa_+,\kappa_c>0$ and $\kappa_++\kappa_c\lesssim (\ell+1)^{-1}$. Fix $\sigma \in \R_{>0}$ and let $s\in \Omega_{\sigma}\subset \C$. If $\ell$ is suitably large, there exist constants $C_{\tilde{\rho}_0,s}>0$ and $A>0$ that are independent of $\kappa_+,\kappa_c$, $\ell$ and $N_{\kappa_+,\kappa_c}\in \N$, such that

\begin{equation}
\label{eq:maingevreysumv2}
\begin{split}
\int_{R^+_0}^{R^c_0}& (1+\ell^2(\ell+1)^2)|{\hat{\psi}}|^2+(1+\ell(\ell+1))|\partial_r{\hat{\psi}}|^2+|\partial_r^2{\hat{\psi}}|^2\,dr\\
&+A^{-\ell}\sum_{\star\in \{+,c\}} \sum_{n=0}^{\ell} \frac{M^n}{(\ell+1)^{2(n-1)}} \int_{0}^{\tilde{\rho}_0} e^{2\gamma M \sqrt{\ell(\ell+1)} \tilde{\rho}} |\hat{\Phi}_{(n)}|^2\, d\tilde{\rho}_{\star}\\
&+A^{-\ell}\ell! (\ell+1)!\sum_{\star\in \{+,c\}}\sum_{n=\ell}^{N_{\infty}} \frac{|2s|^{2n}\widetilde{\sigma}^{2n}}{(n+1)!^2n!^2} \int_{0}^{\tilde{\rho}_0} (|2s|^2+(n+1)^2\tilde{\rho}_{\star}^2)e^{2\gamma M \sqrt{\ell(\ell+1)} \tilde{\rho}} |\hat{\Phi}_{(n+1)}|^2\\
&+\tilde{\rho}_{\star}^4 e^{2\gamma M \sqrt{\ell(\ell+1)} \tilde{\rho}_0}|\hat{\Phi}_{(n+2)}|^2 \, d\tilde{\rho}_{\star}\\
\leq &\: C_{\tilde{\rho}_0,s}\int_{R^+_0}^{R^c_0} |\tilde{f}|^2\,dr+ C_{\tilde{\rho}_0,s} \sum_{\star\in \{+,c\}} \sum_{n=0}^{\ell}\frac{M^n}{(\ell+1)^{2n}} \int_{0}^{\tilde{\rho}_0} e^{2\gamma M \sqrt{\ell(\ell+1)} \tilde{\rho}} |f_n|^2\, d\tilde{\rho}_{\star} \\
&+C_{\tilde{\rho}_0,s} \ell! (\ell+1)!\sum_{\star\in \{+,c\}} \sum_{n=\ell}^{N_{\infty}} \frac{|2s|^{2n}\widetilde{\sigma}^{2n}}{(n+1)!^2n!^2}\int_{0}^{\tilde{\rho}_0}e^{2\gamma M \sqrt{\ell(\ell+1)} \tilde{\rho}_0} |f_n|^2\, d\tilde{\rho}_{\star},
\end{split}
\end{equation}
for all $N_{\infty}\geq N_{\kappa_+,\kappa_c}$.

Furthermore, if $\kappa_+=\kappa_c=0$ and we make the a priori assumption $\hat{\psi}\in H_{\sigma,2,\rho_0}$, then the above estimate holds with $N_{\infty}$ replaced by $\infty$.
\end{proposition}
\begin{proof}
We first use that for $\kappa>0$, we can absorb the non-positive definite integral on on the left-hand side of \eqref{eq:maingevreysum} into the second integral on the left-hand side, provided $N_{\infty}$ is taken suitably large, depending on the value of $\kappa$. This is a manifestation of the red-shift effect along the horizons.

Then we apply the fundamental theorem of calculus to estimate for both $\tilde{\rho}=\tilde{\rho}_c$ and $\tilde{\rho}=\tilde{\rho}_+$:
\begin{equation*}
 |\hat{\Phi}|^2(\tilde{\rho}_0)+ |\hat{\Phi}_{(1)}|^2(\tilde{\rho}_0)\lesssim  C(\tilde{\rho_0})[ |{\hat{\psi}}|^2(\tilde{\rho}_0)+|\partial_r{\hat{\psi}}|^2(\tilde{\rho}_0)]\lesssim C(\tilde{\rho_0}) \int_{r( (\tilde{\rho}_+)_0)}^{r( (\tilde{\rho}_c)_0)} |{\hat{\psi}}|^2+ |\partial_r{\hat{\psi}}|^2+ |\partial_r^2{\hat{\psi}}|^2\,dr.
\end{equation*}
In order to obtain \eqref{eq:maingevreysumv2}, we estimate the RHS above by applying \eqref{eq:ellipticestextv2} \textbf{and using that $K_0$ can be taken arbitrarily large compared to the constant $A$} on the RHS of \eqref{eq:maingevreysum} so that we can subsequently absorb the terms on the RHS of \eqref{eq:ellipticestextv2} involving $\hat{\phi}$ and $\partial_{\rho} \hat{\phi}$ into the LHS of \eqref{eq:maingevreysum}. Note that we require in particular $\ell(\ell+1)$ to grow as $|s|\downarrow 0$.

When $\kappa=0$, the non-positive definite integral on on the left-hand side  of \eqref{eq:maingevreysum} vanishes as $N_{\infty}\to \infty$ if we assume that $||\hat{\psi}||_{\sigma,2}<\infty$.
\end{proof}

\begin{corollary}
\label{cor:maingevreysum3}
Let $\ell\in \N$ and assume $\kappa_+,\kappa_c>0$ and $\kappa_++\kappa_c\lesssim (\ell+1)^{-1}$. Fix $\sigma \in \R_{>0}$ and let $s\in \Omega_{\sigma}\subset \C$. If $\ell$ is suitably large, there exist constants $C_{s,\rho_0},C_{\ell}>0$ that are independent of $\kappa_+,\kappa_c$, with $C_{s,\rho_0}$ moreover independent of $\ell$, and $N_{\kappa_+,\kappa_c}$ that does depend on $\kappa_+,\kappa_c$ , such that:
\begin{equation}
\label{eq:maingevreysum3}
\begin{split}
\int_{R^+_0}^{R^c_0}&(1+\ell^2(\ell+1)^2)|{\hat{\psi}}|^2+(1+\ell(\ell+1))|\partial_r{\hat{\psi}}|^2+|\partial_r^2{\hat{\psi}}|^2\,dr\\
&+C_{\ell}^{-1}\sum_{\star\in \{+,c\}}\sum_{n=0}^{N_{\infty}} \frac{|2s|^{2n}\widetilde{\sigma}^{2n}}{(n+1)!^2n!^2}  (n+1)^2n^2\int_{0}^{{\rho}_0}  |\partial_{\rho_{\star}}^n (r{\hat{\psi}})|^2+\rho_{\star}^4|\partial_{\rho_{\star}}^{n+1} (r{\hat{\psi}})|^2\, d{\rho}_{\star}\\
\leq &\: C_{s,{\rho}_0}\int_{R^+_0}^{R^c_0}|\tilde{f}|^2\,dr+ C_{\ell} \sum_{\star\in \{+,c\}} \sum_{n=0}^{N_{\infty}} \frac{|2s|^{2n}\widetilde{\sigma}^{2n}}{(n+1)!^2n!^2}\int_{0}^{{\rho}_0} |\partial_{\rho_{\star}}^n (r \tilde{f})|^2\, d{\rho}_{\star}.
\end{split}
\end{equation}
for all $N_{\infty}\geq N_{\kappa_+,\kappa_c}$.

Furthermore, if $\kappa_+=\kappa_c=0$ and we make the a priori assumption $\hat{\psi}\in H_{\sigma,2,\rho_0}$, then the above estimate holds with $N_{\infty}$ replaced by $\infty$.
\end{corollary}
\begin{proof}
We can immediately rewrite \eqref{eq:maingevreysumv2} to obtain
\begin{equation*}
\begin{split}
\int_{R^+_0}^{R^c_0}&(1+\ell^2(\ell+1)^2)|{\hat{\psi}}|^2+(1+\ell(\ell+1))|\partial_r{\hat{\psi}}|^2+|\partial_r^2{\hat{\psi}}|^2\,dr\\
&+A^{-\ell}\sum_{\star\in \{+,c\}}\sum_{n=0}^{N_{\infty}} \frac{|2s|^{2n}\widetilde{\sigma}^{2n}}{(n+1)!^2n!^2} (n+1)^2n^2\int_{0}^{\tilde{\rho}_0}  |\hat{\Phi}_{(n)}|^2+\tilde{\rho}^4|\hat{\Phi}_{(n+1)}|^2\, d\tilde{\rho}_{\star}\\
\leq &\: C_{\tilde{\rho}_0,s}\int_{R^+_0}^{R^c_0}|\tilde{f}|^2\,dr+ C_{\ell} \sum_{\star\in \{+,c\}} \sum_{n=0}^{N_{\infty}} \frac{|2s|^{2n}\widetilde{\sigma}^{2n}}{(n+1)!^2n!^2}\int_{0}^{\tilde{\rho}_0} |f_n|^2\, d\tilde{\rho}_{\star}.
\end{split}
\end{equation*}
Observe that
\begin{align*}
\hat{\Phi}_{(n)}=&\:( (1-M\rho)^2\partial_{\rho}-s \hat{h})^n(r{\hat{\psi}}),\\
f_n=&\:( (1-M\rho)^2\partial_{\rho}-s \hat{h})^n((1-M\rho)r^3 \tilde{f}),
\end{align*}
with $\rho=\rho_c$ or $\rho=\rho_+$ and $\hat{h}=\hat{h}_c$ or $\hat{h}=\hat{h}_+$, respectively.

Using analyticity of $\hat{h}$ as a function of ${\rho}$ in $[0,\rho_0]$ and the rapid growth in $n$ of the denominators in the summation factors, the estimate \eqref{eq:maingevreysum3} then follows in a straightforward manner.
\end{proof}

We will additionally need a slight variation of \eqref{eq:maingevreysum3} when investigating the convergence of resolvent operators with $\kappa>0$ as $\kappa\downarrow 0$; see 
\begin{proposition}
\label{prop:maingevreysum3v2}
Let $\ell\in \N$ and assume $\kappa_+,\kappa_c>0$ and $\kappa_++\kappa_c\lesssim (\ell+1)^{-1}$. Fix $\sigma \in \R_{>0}$ and let $s\in \Omega_{\sigma}\subset \C$. If $\ell$ is suitably large, there exist constants $C_{s,\rho_0},C_{\ell}>0$ that are independent of $\kappa_+,\kappa_c$ and $N_{\kappa_+,\kappa_c}$ that does depend on $\kappa_+,\kappa_c$ , such that

\begin{equation}
\label{eq:maingevreysum3v2}
\begin{split}
\int_{R^+_0}^{R^c_0}&(1+\ell^2(\ell+1)^2)|{\hat{\psi}}|^2+(1+\ell(\ell+1))|\partial_r{\hat{\psi}}|^2+|\partial_r^2{\hat{\psi}}|^2\,dr\\
&+C_{\ell}^{-1}\sum_{\star\in \{+,c\}}\sum_{n=0}^{N_{\infty}} \frac{|2s|^{2n}\widetilde{\sigma}^{2n}}{(n+1)!^2n!^2}  n^2\int_{0}^{{\rho}_0}  |\partial_{\rho_{\star}}^n (r{\hat{\psi}})|^2+\rho_{\star}^4|\partial_{\rho_{\star}}^{n+1} (r{\hat{\psi}})|^2\, d{\rho}_{\star}\\
\leq &\: C_{s,{\rho}_0}\int_{R^+_0}^{R^c_0}|\tilde{f}|^2\,dr+ C_{s,{\rho}_0} \sum_{\star\in \{+,c\}} \sum_{n=0}^{N_{\infty}} \frac{|2s|^{2n}\widetilde{\sigma}^{2n}}{(n+1)!^2n!^2}\int_{0}^{{\rho}_0} (n+1)^{-2}|\partial_{\rho_{\star}}^n (r \tilde{f})|^2\, d{\rho}_{\star}.
\end{split}
\end{equation}
for all $N_{\infty}\geq N_{\kappa_+,\kappa_c}$.
\end{proposition}
\begin{proof}
Observe that the difference between \eqref{eq:maingevreysum3} and \eqref{eq:maingevreysum3v2} amounts to an additional factor $(n+1)^{-2}$ in the summation over $n$. We arrive at by repeating the arguments in this section, employing instead the summation
\begin{equation*}
\ell! (\ell+1)!\sum_{n=\ell}^{N_{\infty}} \frac{|2s|^{2n}\widetilde{\sigma}^{2n}}{(n+1)!^2n!^2} (n+1)^{-2}\Big[\cdot\Big].
\end{equation*}
We omit the details of this procedure. We remark that the main difference is that will see in particular additional factors of $\frac{(n+1)^2}{n^2}$ appearing when trying to absorb terms into estimates of order $n+1$ or $n-1$, but we use that
\begin{equation*}
\frac{(n+1)^2}{n^2}=1+O(n^{-1})
\end{equation*}
so the procedure of absorbing terms poses no problem provided $\ell$ is taken suitably large.
\end{proof}

\section{Additional Gevrey estimates for low frequencies}
\label{sec:lowfreqgev}

In this section we obtain additional Gevrey estimates that are only valid for bounded $0\leq \ell\leq L$ and for correspondingly small values of $|s|$, cf. the estimates in Section \ref{sec:gevrest}, where $|s|>0$ is allowed to be arbitrary large, but we need to assume $\ell$ is suitably large.

\begin{theorem}
\label{thm:smallsthmpaper}
Let $L\in \N_0$, $\sigma \in {\R_{>0}}$ and $\rho_0>0$. Assume that $\kappa_+=\kappa_c=\kappa>0$ and $0\leq \ell\leq L$.  Let $s\in \Omega_{\sigma}\cap\{|s|<s_0\}$ for a given positive constant $s_0$. If $s_0$ is chosen suitably small and $\rho_0>0$ is a suitably small constant, both independent of $\kappa$, then we can find a constant $C_{s, L}>0$ that is also independent of $\kappa$, such that for all ${\hat{\psi}}\in C^{\infty}([r_+,r_c])$ we can estimate

\begin{equation}
\label{eq:smallsestpaper}
\begin{split}
\int_{R_0^+}^{R_0^c}& |{\hat{\psi}}|^2+|\partial_r{\hat{\psi}}|^2+|\partial_r^2{\hat{\psi}}|^2\,dr+\sum_{n=0}^{\infty} \frac{\sigma^{2n}}{(n+1)!^2n!^2} n^2(n+1)^2\int_{0}^{(\rho_+)_0} |\partial_{\rho_+}^n(r{\hat{\psi}})|^2+\rho_+^4|\partial_{\rho_+}^{n+1}(r{\hat{\psi}})|^2\, d{\rho_+}\\
&+\sum_{n=0}^{\infty} \frac{\sigma^{2n}}{(n+1)!^2n!^2} n^2(n+1)^2\int_{0}^{(\rho_c)_0} |\partial_{\rho_c}^n(r{\hat{\psi}})|^2+ \rho_c^4|\partial_{\rho_c}^{n+1}(r{\hat{\psi}})|^2\, d{\rho_c}\\
\leq &\: C_{ s, L}\int_{R_0^+}^{R_0^c} |{L}_{\kappa,s,\ell} ({\hat{\psi}})|^2 \,dr +C_{ s,L} \sum_{n=0}^{\infty} \frac{\sigma^{2n}}{(n+1)!^2n!^2} \int_{0}^{(\rho_+)_0} |\partial_{\rho_+}^n(r {L}_{s,\ell,\kappa} ({\hat{\psi}}))|^2\, d{\rho_+}\\
&+C_{ s,L} \sum_{n=0}^{\infty} \frac{\sigma^{2n}}{(n+1)!^2n!^2} \int_{0}^{(\rho_c)_0} |\partial_{\rho_c}^n(r {L}_{s,\ell,\kappa} ({\hat{\psi}}))|^2\, d{\rho_c}.
\end{split}
\end{equation} 
\end{theorem}

We will prove the theorem via several propositions which play a similar role to the propositions of Section \ref{sec:gevrest}.

\begin{proposition}
\label{prop:boundlgevrey}
Let $L\in \N_0$ and assume that $\ell\leq L$. Let $\sigma,\alpha,\beta,\mu,\nu\in \R$ and let $\epsilon>0$ be arbitrarily small. Then there exist integers $0<N_{L,\epsilon}<N_{\infty}=N_{\infty}(\kappa)$ and constants $\rho_0 =\rho_0(L)>0$ suitably small and $C_{\epsilon,L}>0$, such that
\begin{equation}
\label{eq:boundlgevrey}
\begin{split}
\sum_{n=0}^{N_{\infty}}&\frac{|2s|^{2n}\tilde{\sigma}^{2n}}{(n+1)!^2n!^2}\int_{0}^{\rho_0}(1-\mu \nu^{-1})\left((1-A_2)\rho^2+2\kappa \rho+A_3\rho^3+A_4 \rho^4\right)^2e^{2M\gamma\rho} |\hat{\phi}^{(n+2)}|^2\,d\rho\\
&+\sum_{n=0}^{N_{\infty}}\frac{|2s|^{2n}\tilde{\sigma}^{2n}}{(n+1)!^2n!^2}\int_{0}^{\rho_0}\left(1-4 \beta^{-1}\left(\frac{\max\{-\re(s),0\}}{|s|}\right)^2-(1+\epsilon)\sigma^{-2}-\epsilon-\nu \mu \right)|2s|^2e^{2M\gamma\rho}|\hat{\phi}^{(n+1)}|^2\,d\rho\\
&+\sum_{n=0}^{N_{\infty}}\frac{|2s|^{2n}\tilde{\sigma}^{2n}}{(n+1)!^2n!^2}\Bigg[\int_{0}^{\rho_0} (4-\beta-\alpha(1+A_2+\epsilon(1+\kappa))(1-\mu)-\alpha^{-1}(1+A_2+\epsilon(1+\kappa)) \sigma^{-2}(1-\mu))\cdot\\
&\cdot (n+1)^2\left((1-A_2)^2\rho^2+\kappa^2+2(1-A_2)\kappa \rho+O(\rho^3)+\kappa O(\rho^2)+\kappa^2 O(\rho) \right)e^{2M\gamma\rho}|\hat{\phi}^{(n+1)}|^2\,d\rho\\
&-\int_{0}^{\rho_0}(n+1) [6(1-A_2)^2\rho^2+4\kappa((1-A_2)^2+\kappa)+(1+|s|)O(\rho_0^3)+\kappa (1+|s|)O(\rho_0)]e^{2M\gamma\rho}|\hat{\phi}^{(n+1)}|^2\,d\rho\Bigg]\\
&- C\int_{0}^{\rho_0}\frac{|2s|^{2N_{\infty}}\widetilde{\sigma}^{2N_{\infty}}}{(N_{\infty}+1)!^2N_{\infty}!^2}\frac{|2s|^{2}}{(N_{\infty}+2)^2}\cdot \alpha^{-1}(1-\mu)(\rho^3+2\kappa\rho^2) |\hat{\phi}^{(N_{\infty}+2)}|^2 \, d{\rho}\\
&+\sum_{n=0}^{N_{\infty}}\frac{|2s|^{2n}\tilde{\sigma}^{2n}}{(n+1)!^2n!^2} (n+1)[2\rho_0^3(1-A_2)^2+4\kappa((1-A_2)^2+\kappa)\rho_0+O(\rho_0^4)+\kappa O(\rho_0^2)]e^{2M\gamma\rho}|\hat{\phi}^{(n+1)}|^2\Big|_{\rho=\rho_0}\\
\leq&\: C_{\epsilon,L}\sum_{n=0}^{N_{L,\epsilon}} |\hphi^{(n)}|^2(\rho_0)+C_{\epsilon,L}\sum_{n=0}^{N_{\infty}}\frac{|2s|^{2n}\tilde{\sigma}^{2n}}{(n+1)!^2n!^2}\int_{0}^{\rho_0} e^{2M\gamma\rho}|\partial_{\rho}^n(r\tilde{f})|^2\,d\rho.
\end{split}
\end{equation}
\end{proposition}
\begin{proof}
We consider the inhomogeneous version of \eqref{eq:maineqhatphin} and split
\begin{equation*}
\begin{split}
2s\hat{\phi}^{(n+1)}&+\left((1-A_2)\rho^2+2\kappa \rho+A_3\rho^3+A_4 \rho^4\right)\hat{\phi}^{(n+2)}+2(n+1)\left((1-A_2)\rho+\kappa+\frac{3}{2}A_3\rho^2+2A_4 \rho^3\right)\hat{\phi}^{(n+1)}\\
= &\:r{L}_{s,\ell,\kappa}({\hat{\psi}})-\left[n(n+1)\left(1-A_2+3A_3 \rho+6A_4 \rho^2\right)-(B_0+B_1\rho +B_2\rho^2)\right]\hat{\phi}^{(n)}+\ell(\ell+1)\hat{\phi}^{(n)}\\
 &-\left[(n+1)n(n-1)\left(A_3+4A_4\rho\right) -n(B_1+2B_2\rho)\right]\hat{\phi}^{(n-1)}\\
 &-\left[(n+1)n(n-1)(n-2)A_4-n(n-1)B_2\right]\hat{\phi}^{(n-2)}.
 \end{split}
\end{equation*}
Subsequently, we square both sides and multiply them by $e^{2M\gamma\rho}$, with $\gamma>0$ an $L$-dependent constant that will be chosen suitably large. Then we integrate over $[0,\rho_0]$  as in the proof of Proposition \ref{eq:maingevprop}, but with $\hat{\Phi}$ replaced by $\hat{\phi}$ (which satisfies a significantly simpler equation) and $\tilde{\rho}_0$ replaced by $\rho_0$. 

We first restrict to the case $n\geq N_{L,\epsilon}>2\ell$ suitably large depending on $\epsilon>0$ because in this case, we have that
\begin{equation*}
[n(n+1)-\ell(\ell+1)]^2=n^2(n+1)^2-2n(n+1)\ell(\ell+1)+\ell^2(\ell+1)^2<n^2(n+1)^2
\end{equation*}
so the $\ell^2(\ell+1)^2|\hphi^{(n)}|^2$ term forms no problem. We sum over $N_{L,\epsilon}\leq n\leq N_{\infty}$ with the same weights as in Proposition \ref{prop:maingevreysum}. We require here that $\rho_0{\ll}\frac{1}{\gamma M}$.

We can estimate the $n\leq N_{L,\epsilon}$ terms simply by repeatedly applying  Lemma \ref{lm:expweightest} as in the proof of Proposition \ref{prop:gevreyestsum}. The weights in the summation over $0\leq n\leq N_{L,\epsilon}$ do not matter as we are assuming $\ell\leq L$ and we do not keep track of the precise $L$ dependence of the constants.

\end{proof}

\begin{proposition}
\label{prop:boundlgev2}
Let $L\in \N_0$ and $\sigma \in {\R_{>0}}$. Let $s\in \Omega_{\sigma}$ with $|s|\leq s_0$. Assume that $\kappa_+,\kappa_c>0$. Then there exists $N_{\kappa_+,\kappa_c}>0$, $s_0=s_0(L)>0$, $(\rho_+)_0,(\rho_c)_0>0$ suitably small, and $C_{L,s}>0$, where $C_{L,s}>0$ remains bounded as $|s|\to0$, such that under the restriction $0<|s|\leq s_0$ and $0\leq \ell\leq L$:
\begin{equation}
\label{eq:boundlgevrey2}
\begin{split}
\int_{R_0^+}^{R_0^c}& |{\hat{\psi}}|^2+|\partial_r{\hat{\psi}}|^2+|\partial_r^2{\hat{\psi}}|^2\,dr\\
&+\sum_{\star\in \{+,c\}}\sum_{n=0}^{N_{\infty}} \frac{|2s|^{2n}\widetilde{\sigma}^{2n}}{(n+1)!^2n!^2} \int_{0}^{{\rho}_0} (|2s|^2+(n+1)^2{\rho}_{\star}^2+1)e^{2\gamma M \rho} |\hat{\phi}_{(n+1)}|^2\, d\rho_{\star}\\
\leq &\: C_{L,s}\int_{R_0^+}^{R_0^c} |\tilde{f}|^2\,dr+ C_{L,s}\sum_{\star\in \{+,c\}} \sum_{n=0}^{N_{\infty}} \frac{|2s|^{2n}\widetilde{\sigma}^{2n}}{(n+1)!^2n!^2}\int_{0}^{{\rho}_0}e^{2\gamma M \rho} |\partial_{\rho}^n(r\tilde{f})|^2\, d{\rho}_{\star},
\end{split}
\end{equation}
for all $N_{\infty}\geq N_{\kappa_+,\kappa_c}$.
\end{proposition}
\begin{proof}
First we apply the analogues of the estimates in Lemma \ref{lm:analytawayhor} with $\hat{\Phi}$ replaced by $\hat{\phi}$ to estimate the boundary terms on the right-hand side of \eqref{eq:boundlgevrey} by $\hphi(\rho_0)$ and $\hphi^{(1)}(\rho_0)$. We then repeat the proof of Corollary \ref{cor:maingevreysumbound}, using Proposition \ref{prop:boundlgevrey} instead of Proposition \ref{prop:maingevreysum}. Finally, we proceed as in the proof of Proposition \ref{prop:maingevreysumv2}, applying \eqref{eq:lowfreqellipticest} in place of \eqref{eq:ellipticestextv2}, taking $|s|$ to be suitably small (depending on $C_{L}$) so that the terms on the right-hand side of \eqref{eq:lowfreqellipticest} can be absorbed.
\end{proof}

\section{Quasinormal modes as eigenfunctions}
\label{sec:constrresolvent}
In this section, we apply the estimates of Theorem \ref{thm:mainthmpaper} to construct the resolvent operator $\hat{{L}}^{-1}_{s,\ell,0}$ on an appropriate Hilbert space and derive compactness properties.

\subsection{Definition of the resolvent operator}
We start by investigating the invertibility of $\hat{{L}}_{s,\ell,\kappa}$ with $\kappa>0$. In this case, the main required properties, which involve estimates that are \emph{not} uniform in $\kappa$, are adaptations of the estimates in \cite{warn15} and are derived in the setting of the present paper in Appendix \ref{sec:red-shift}.
\begin{theorem}
\label{thm:fredholmpositivekappa}
Fix $\ell,\lambda\in \N_0$ and $\kappa>0$. Let $s\in \C$ and $k\in \N_0$. We define $\mathcal{D}^{k}(\hat{{L}}_{s,\ell,\kappa})$ as the closure of $\{f\in \pi_{\ell}(C^{\infty}([r_+,r_c]\times \s^2;\C))\,|\, ||f||_{\widetilde{H}^{k}}+||\hat{{L}}_{s,\ell,\kappa}(f)||_{\widetilde{H}^{k}}<\infty\}$ with respect to the graph norm 
\begin{equation*}
||\cdot||_{\widetilde{H}^{k}}+||\hat{{L}}_{s,\ell,\kappa}(\cdot)||_{\widetilde{H}^{k}}.
\end{equation*}
Then, for all $k\in \N$, the map
\begin{equation*}
\hat{L}_{s,\ell, \kappa}: \mathcal{D}^{k}(\hat{{L}}_{s,\ell, \kappa})\to \pi_{\ell}(\widetilde{H}^{k})
\end{equation*}
is a densely defined, closed linear operator.

Moreover, there exists $k_{\kappa}\in \N$, independent of $\ell$ and $\lambda$, such that the following dichotomy holds.
\begin{enumerate}
\item Either
\begin{equation*}
\hat{{L}}_{s,\ell+\lambda, \kappa}^{-1}: \pi_{\ell+\lambda}(\widetilde{H}^{k_{\kappa}})\to \mathcal{D}^{k_{\kappa}}(\hat{{L}}_{s,\ell+\lambda, \kappa})\subset \pi_{\ell+\lambda}(\widetilde{H}_2^{k_{\kappa}+1})
\end{equation*}
is a well-defined bounded linear operator,
\item or the equation
\begin{equation*}
\hat{{L}}_{s,\ell+\lambda, \kappa}({\hat{\psi}})=0
\end{equation*}
admits non-trivial solutions ${\hat{\psi}}\in \mathcal{D}^{k_{\kappa}}(\hat{{L}}_{s,\ell+\lambda, \kappa})$.
\end{enumerate}
\end{theorem}
\begin{proof}
Since $\kappa>0$, we can apply the techniques developed in \cite{warn15}. Since \cite{warn15} deals with asymptotically anti de Sitter spacetimes, we include the main estimates in the setting of the present paper in Appendix \ref{sec:red-shift}. In particular, Proposition \ref{prop:poskappainvert} together with Rellich--Kondrachov implies that $\hat{{L}}_{s,\ell+\lambda, \kappa}^{-1}$ is well-defined for suitably large $\lambda$ (depending on $k$ and $\kappa$). We conclude the proof by applying the Analytic Fredholm Theorem as in Theorem 4.9 of \cite{warn15}.
\end{proof}

We now apply Corollary \ref{cor:maingevreysum3} to derive estimates for $\hat{{L}}_{s,\ell+\lambda, \kappa}^{-1}$ that are uniform in $\kappa$, after restricting to suitable subspaces.
\begin{proposition}
\label{prop:inversepositivekappa}
Fix $\sigma\in \R_{>0}$ and $\ell\in \N$. Then for all $s\in \Omega_{\sigma}$, there exist $\lambda_0\in \N$ and $\kappa_0>0$ independent of $\ell$ and $s_{\ell}>0$ dependent on $\ell$ such that for all $0<\kappa\leq \kappa_0$ and either $\lambda\geq \lambda_0$ or $|s|\leq s_{\ell}$ the inverse map
\begin{equation*}
\hat{{L}}_{s,\ell+\lambda, \kappa}^{-1}: H_{\sigma,\rho_0}\to \mathcal{D}_{\sigma}(\hat{{L}}_{s,\ell+\lambda, \kappa})\subseteq H_{\sigma,2,\rho_0}.
\end{equation*}
is a well-defined bounded linear operator and there exists a constant $C_{\ell,\lambda,s}>0$ such that the operator norm can be bounded:
\begin{equation*}
||\hat{{L}}_{s,\ell+\lambda, \kappa}^{-1}||_{H_{\sigma, \rho_0} \to H_{\sigma,2,\rho_0}}\leq C_{\ell,\lambda,s}.
\end{equation*}
\end{proposition}
\begin{proof}
Let $\kappa>0$. First suppose for our choice of $s$, we are in case (1) of Theorem \ref{thm:fredholmpositivekappa}. Suppose $f\in H_{\sigma,\rho_0}$ and consider $g=f\cdot Y_{\ell+\lambda,0}$. Then $g\in \pi_{\ell+\lambda}(\widetilde{H}^{k_{\kappa}})$ for all $\kappa>0$, so the restriction
\begin{equation*}
\hat{{L}}_{s,\ell+\lambda, \kappa}^{-1}: H_{\sigma,\rho_0}\to \mathcal{D}_{\sigma}(\hat{{L}}_{s,\ell+\lambda, \kappa})
\end{equation*}
is well-defined. Then we can apply Corollary \ref{cor:maingevreysum3} and take the limit $N_{\infty}\to \infty$ to conclude that there exists $\kappa_0>0$ and $\rho_0>0$ suitably small, such that for all $0\leq \kappa\leq \kappa_0$ and $\tilde{f}\in H_{\sigma,\rho_0}$
\begin{equation*}
||\hat{{L}}_{s,\ell+\lambda, \kappa}^{-1}(\tilde{f})||_{\sigma,2,\rho_0}\leq C_{\ell,\lambda,s} ||\tilde{f}||_{\sigma,\rho_0},
\end{equation*}
so $\mathcal{D}_{\sigma}(\hat{{L}}_{s,\ell+\lambda, \kappa})\subseteq H_{\sigma,2,\rho_0}$ and $||\hat{{L}}_{s,\ell+\lambda, \kappa}^{-1}||_{H_{\sigma,\rho_0} \to H_{\sigma,2,\rho_0}}\leq C_{\ell,\lambda,s}$. 

Now suppose we are in case (2) of Theorem \ref{thm:fredholmpositivekappa}. Let $\hat{{L}}_{s,\ell+\lambda, \kappa}{\hat{\psi}}=0$. Then we can apply Corollary \ref{cor:maingevreysum3} to conclude that ${\hat{\psi}}=0$, which is a contradiction.
\end{proof}

Finally, we use the uniform estimates from Proposition \ref{prop:inversepositivekappa} to construct  $\hat{{L}}_{s,\ell+\lambda, 0}^{-1}$ as a limit of a sequence of resolvent operators with strictly decreasing $\kappa>0$.

\begin{proposition}
\label{prop:resolventzerokappa}
Fix $\sigma\in \R_{>0}$ and let $s\in \Omega_{\sigma}$.  Let $\ell,\lambda\in \N_0$ and consider the linear operator
\begin{equation*}
\hat{{L}}_{s,\ell+\lambda, 0}: \mathcal{D}_{\sigma}(\hat{{L}}_{s,\ell+\lambda, 0})\to H_{\sigma,\rho_0}.
\end{equation*}

Let $\kappa_0, \lambda_0, s_{\ell}$ be the constants from Proposition \ref{prop:inversepositivekappa} and let $\{\kappa_n\}_{n\in \N}$ be a sequence of positive numbers such that $|\kappa_n|<\kappa_0$ and $\kappa_n\to 0$. Let $\lambda\geq \lambda_0$ or $|s|\leq s_{\ell}$. Then the sequence of linear operators 
\begin{equation*}
\left\{\hat{{L}}_{s, \ell+\lambda, \kappa_n}^{-1}: H_{\sigma,\rho_0} \to\mathcal{D}_{\sigma} (\hat{{L}}_{s, \ell+\lambda, \kappa_n}) \subset H_{\sigma, 2, \rho_0}\right\}_{n\in \N}
\end{equation*}
has a subsequence that converges in the Banach space of bounded linear operators $B(H_{\sigma,\rho_0} ,H_{\sigma,\rho_0})$ to
\begin{equation*}
\hat{{L}}_{s, \ell+\lambda, 0}^{-1}: H_{\sigma,\rho_0} \to \mathcal{D}_{\sigma} (\hat{{L}}_{s, \ell+\lambda,0}) \subseteq H_{\sigma,2, \rho_0},
\end{equation*}
which is the inverse of $\hat{{L}}_{s, \ell+\lambda, 0}$ and there exists a constant $C_{\ell,\lambda,s}>0$ such that
\begin{equation*}
||\hat{{L}}_{s, \ell+\lambda, 0}^{-1}||_{H_{\sigma,\rho_0} \to  H_{\sigma,2,\rho_0}}\leq C_{\ell,\lambda,s}.
\end{equation*}
\end{proposition}
\begin{proof}
We establish first the existence of a limit of $\hat{{L}}_{s, \ell+\lambda, \kappa_n}^{-1}$ with respect to strong operator convergence. Let $s\in \Omega_{\sigma}$. By openness of $\Omega_{\sigma}$, there exist $\sigma_2>\sigma_1>\sigma$ such that $s\in {\Omega}_{\sigma_1}\cap \Omega_{\sigma_2}$. Hence, for all $f\in H_{\sigma_2,\rho_0}$, $\hat{{L}}_{s, \ell+\lambda, \kappa_n}^{-1}(f)$ is a bounded sequence with respect to $||\cdot||_{\sigma_2,2,\rho_0}$. Since $H_{\sigma_2,\rho_0}\Subset H_{\sigma_1,\rho_0}$ by Lemma \ref{lm:cptembeddinggevrey} below, there exists a subsequence of $\{\hat{{L}}_{s, \ell+\lambda, \kappa_n}^{-1}(f)\}$ that converges with respect to $||\cdot||_{\sigma_1,2,\rho_0}$. We denote the limit by $\hat{\psi}$. We moreover have that
\begin{align*}
L_{s,\ell+\lambda, 0}(\hat{\psi})=&\:(L_{s, \ell+\lambda, 0}-L_{s, \ell+\lambda, \kappa_n})(\hat{\psi})+ L_{s, \ell+\lambda, \kappa_n}(\hat{\psi}-\hat{\psi}_n)+f,\\
||(L_{s, \ell+\lambda, 0}-L_{s, \ell+\lambda, \kappa_n})(\hat{\psi})||_{\sigma,\rho_0}\leq&\: \kappa_n \sum_{\star \in\{+,c\}}\left[||\hat{\psi}||_{\sigma,\rho_0}+||r^{-1} \partial_{\rho_{\star}}(r\hat{\psi})||_{\sigma,\rho_0}+||\rho r^{-1} \partial_{\rho_{\star}}^2(r\hat{\psi})||_{\sigma,\rho_0} \right],\\
||L_{s, \ell+\lambda, \kappa_n}(\hat{\psi}-\hat{\psi}_n)||_{\sigma,\rho_0}\leq&\:  \sum_{\star \in\{+,c\}}\Big[||\hat{\psi}-\hat{\psi}_n||_{\sigma,\rho_0}+||r^{-1} \partial_{\rho_{\star}}(r\hat{\psi}-r\hat{\psi}_n)||_{\sigma,\rho_0}\\
&+||\rho r^{-1} \partial_{\rho_{\star}}^2(r\hat{\psi}-r\hat{\psi}_n)||_{\sigma,\rho_0} \Big]
\end{align*}
where the norms on the right-hand side go to zero as $n\to \infty$ after passing to a subsequence if we take $\sigma_1>\sigma$, since $H_{\sigma_1,\rho_0}\Subset H_{\sigma,\rho_0}$ and we obtain
\begin{equation*}
L_{s, \ell+\lambda, 0}(\hat{\psi})=f.
\end{equation*}
This provides a definition of $L_{s, \ell+\lambda, 0}^{-1}: H_{\sigma_2,\rho_0} \to  H_{\sigma,2,\rho_0}$. By density of $H_{\sigma_2,\rho_0}$ in $H_{\sigma,\rho_0}$, we can extend $L_{s, \ell+\lambda, 0}^{-1}$ to $H_{\sigma,\rho_0}$. Furthermore, by Corollary \ref{cor:maingevreysum3} we have that
\begin{equation*}
||\hat{{L}}_{s, \ell+\lambda, 0}^{-1}||_{H_{\sigma,\rho_0} \to  H_{\sigma,2,\rho_0}}\leq C_{\ell,\lambda,s}.
\end{equation*}
Now we will show that in fact $\hat{{L}}_{s, \ell+\lambda, \kappa_n}^{-1}$ in the sense of \emph{uniform} operator convergence. We will show that $\hat{{L}}_{s, \ell+\lambda, \kappa_n}^{-1}$ form a Cauchy sequence in the Banach space $B(H_{\sigma,\rho_0}, H_{\sigma,\rho_0})$. We can express
\begin{equation*}
\hat{{L}}_{s, \ell+\lambda, \kappa_n}^{-1}-\hat{{L}}_{s, \ell+\lambda, \kappa_m}^{-1}=\hat{{L}}_{s, \ell+\lambda, \kappa_m}^{-1} \circ (\hat{{L}}_{s, \ell+\lambda, \kappa_m}-\hat{{L}}_{s, \ell+\lambda, \kappa_n})\circ \hat{{L}}_{s, \ell+\lambda, \kappa_n}^{-1}.
\end{equation*} 
Note that
\begin{equation*}
|| (\hat{{L}}_{s, \ell+\lambda, \kappa_m}-\hat{{L}}_{s, \ell+\lambda, \kappa_n})(f)||_{\sigma,\rho_0}\lesssim |\kappa_n-\kappa_m| \sum_{\star \in\{+,c\}}\left[||\hat{\psi}||_{\sigma,\rho_0}+||r^{-1} \partial_{\rho_{\star}}(r\hat{\psi})||_{\sigma,\rho_0}+||r^{-1} \rho \partial_{\rho_{\star}}^2(r\hat{\psi})||_{\sigma,\rho_0} \right]
\end{equation*}
with $\hat{\psi}=\hat{{L}}_{s, \ell+\lambda, \kappa_n}^{-1}(f)$. However, we cannot control $||r^{-1} \rho \partial_{\rho_{\star}}^2(r\hat{\psi})||_{\sigma,\rho_0}$ with $||\hat{\psi}||_{\sigma,2,\rho_0}$ (which we control by $||f||_{\sigma,\rho_0}$  using Corollary \ref{cor:maingevreysum3}). 

Instead consider a slightly weaker norm of $(\hat{{L}}_{s, \ell+\lambda, \kappa_m}-\hat{{L}}_{s, \ell+\lambda, \kappa_n})(f)$, where the terms in the sum defining the norms come with an extra factor $(n+1)^{-2}$. We will denote the relevant weaker norms by $||\cdot||_{\sigma,\rho_0,w}$ and $||\cdot||_{\sigma,2,\rho_0,w}$.  Then we have that
\begin{equation*}
||r^{-1} \rho \partial_{\rho_{\star}}^2(r\hat{\psi})||_{\sigma,\rho_0,w}+||r^{-1} \partial_{\rho_{\star}}(r\hat{\psi})||_{\sigma,\rho_0,w}\leq C ||\hat{\psi}||_{\sigma,2,\rho_0}\leq C||f||_{\sigma,\rho_0}.
\end{equation*}

From Proposition \ref{prop:maingevreysum3v2}, it moreover follows that
\begin{equation*}
||(\hat{{L}}_{s, \ell+\lambda, \kappa_m}^{-1}(g)||_{\sigma,2,\rho_0,w}\leq C ||g||_{\sigma,\rho_0,w},
\end{equation*}
so we obtain:
\begin{equation*}
||(\hat{{L}}_{s, \ell+\lambda, \kappa_n}^{-1}-\hat{{L}}_{s, \ell+\lambda, \kappa_m}^{-1})(f)||_{\sigma,2,\rho_0,w}\leq C |\kappa_n-\kappa_m| \cdot ||f||_{\sigma,\rho_0}.
\end{equation*}
Finally, we observe that:
\begin{equation*}
||(\hat{{L}}_{s, \ell+\lambda, \kappa_n}^{-1}-\hat{{L}}_{s, \ell+\lambda, \kappa_m}^{-1})(f)||_{\sigma,\rho_0}\leq C ||(\hat{{L}}_{s, \ell+\lambda, \kappa_n}^{-1}-\hat{{L}}_{s, \ell+\lambda, \kappa_m}^{-1})(f)||_{\sigma,2,\rho_0,w}\leq C |\kappa_n-\kappa_m| \cdot ||f||_{\sigma,\rho_0}.
\end{equation*}

The convergence of $\{\kappa_n\}$ ensures that indeed $\{\hat{{L}}_{s, \ell+\lambda, \kappa_n}^{-1}\}$ is Cauchy in the Banach space $B(H_{\sigma,\rho_0}, H_{\sigma,\rho_0})$ and therefore converges. By uniqueness of limits, the limit is $L_{s, \ell+\lambda, 0}^{-1}$.
\end{proof}

\subsection{Fredholm theory}
In this section we establish compactness of the operator ${L}_{0,\ell+\lambda, s}^{-1}$, which allows us to apply the Analytic Fredholm Theorem in order to obtain information about the boundedness properties of ${L}_{s,\ell, 0}^{-1}$ (with $\lambda=0$).
\begin{lemma}
\label{lm:cptembeddinggevrey}
Let $\sigma\in \R_{>0}$. The following embeddings hold
\begin{align}
\label{eq:cptgev1}
H_{\sigma,2,\rho_0}\subset H_{\sigma,1,\rho_0} \Subset&\: H_{\sigma,\rho_0},\\
\label{eq:cptgev2}
H_{\sigma', \rho_0}\Subset&\: H_{\sigma,\rho_0} \quad \textnormal{if $\sigma'>\sigma$.}
\end{align}
\end{lemma}
\begin{proof}
We will only prove \eqref{eq:cptgev1} as the proof of \eqref{eq:cptgev2} proceeds very similarly. Consider a sequence $\{f_n\}$ in $H_{\sigma,1,\rho_0}$, such that $||f_n||_{\sigma,1,\rho_0}=1$. We will show that there exists a subsequence that converges with respect to $||\cdot||_{\sigma,\rho_0}$. First of all, by $||f_n||_{\sigma,1,\rho_0}=1$, given any $N\in \N$, we can use standard Sobolev inequalities to estimate for all $0\leq k\leq N$
\begin{equation*}
||rf_n||_{C^{k+1}([0,(\rho_+)_0])}+||rf_n||_{C^{k+1}([0,(\rho_c)_0])}\leq C_N,
\end{equation*}
where the $C^k$-norm is taken with respect to the $\rho_+$ and $\rho_c$ coordinates. Hence, by Arzel\`a--Ascoli, there exists for all $N\in \N_0$ a subsequence $\{f_{n_j}\}$ satisfying the following Cauchy property: for all $\epsilon'>0$, there exists a $K>0$ such that for all $m>l>K$ and $0\leq k\leq N$:
\begin{equation*}
||rf_{n_m}-rf_{n_l}||_{C^{k}([0,(\rho_+)_0])}+||rf_{n_m}-rf_{n_l}||_{C^{k}([0,(\rho_c)_0])}<\epsilon'.
\end{equation*}
We moreover have that
\begin{equation*}
\sum_{\star\in\{+,c\}}\sum_{j=N}^{\infty} \frac{\sigma^{2j}}{j!^2(j+1)!^2} \int_0^{(\rho_{\star})_{0}} |\partial_{\rho_{\star}}^j (rf_{n_k}-rf_{n_l})|^2\,d\rho_{\star}\leq \frac{1}{N^2(N+1)^2}||f_{n_k}-f_{n_{\ell}}||^2_{\sigma,\rho_0,1}\leq \frac{2}{N^2(N+1)^2}.
\end{equation*}
Hence, for any $\epsilon>0$, there exists $N>0$ suitably large such that
\begin{equation*}
||f_{n_k}-f_{n_{l}}||_{G^2_{\sigma,0,\rho_0}}\leq C\epsilon'+ \frac{2}{N^2(N+1)^2}<\frac{\epsilon}{2}.
\end{equation*}
We moreover have that
\begin{equation*}
||f_{n_k}||_{H^1[R_0^+,R_0^c]}\leq 1,
\end{equation*}
so by Rellich--Kondrachov, there exists a further subsequence $\{f_{n_{k_m} }\}$ that is Cauchy with respect to the $L^2[R_0^+,R_0^c]$-norm. We can conclude from all the above that $\{f_{n_{k_m} }\}$ is a Cauchy sequence with respect to $||\cdot||_{\sigma,\rho_0}$ and must therefore converge.
\end{proof}

\begin{proposition}
\label{prop:fredholm}
Let $\sigma \in {\R_{>0}}$ and let $\kappa_0$ be the constant from Proposition \ref{prop:inversepositivekappa}. Then for all $0\leq \kappa\leq \kappa_0$
\begin{equation*}
 \hat{{L}}_{s,\ell,\kappa}^{-1}: H_{\sigma,\rho_0}\to H_{\sigma,2,\rho_0}
\end{equation*}
is holomorphic for all $s\in \Omega_{\sigma}\setminus \Lambda_{QNF}^{\sigma,\ell, \rho_0}$, with $\Lambda_{QNF}^{\sigma,\ell,\rho_0}\subset \Omega_{\sigma}$ a discrete set. Furthermore, if $s\in  \Lambda_{QNF}^{\sigma,\ell,\rho_0}$ then $\ker \hat{{L}}_{s,\ell,\kappa}=0$ is finite dimensional and $\hat{{L}}_{s,\ell,\kappa} {\hat{\psi}}=f$ admits a solution if and only if $f\in (\textnormal{coker}\,  \hat{{L}}_{s,\ell,\kappa})^{\perp}$ with $\textnormal{coker}\,  \hat{{L}}_{s,\ell,\kappa}< H_{\sigma,\rho_0}$ and $\dim \textnormal{coker}\,  \hat{{L}}_{s,\ell,\kappa}=\dim \ker \hat{{L}}_{s,\ell,\kappa}$. 
\end{proposition}
\begin{proof}
Let $s\in \Omega_{\sigma}$. By Proposition \ref {prop:resolventzerokappa} and Lemma \ref{lm:cptembeddinggevrey}
\begin{equation*}
B(s):=\lambda(\lambda +1) \hat{{L}}_{s,\ell+\lambda,\kappa}^{-1}=( \hat{{L}}_{s,\ell,\kappa} -\lambda(\lambda+1)id)^{-1}: H_{\sigma,\rho_0}\to H_{\sigma,\rho_0}
\end{equation*}
is a well-defined compact operator when $\lambda \geq \lambda_0$. Furthermore, one can easily verify that
\begin{equation*}
\begin{split}
\frac{B(s')-B(s)}{s-s'}(f)=&\:\lambda(\lambda+1) ( \hat{{L}}_{s,\ell+\lambda,\kappa})^{-1}\Big[r^2h_{r_+}(2-h_{r_+}D)P_1(r(\hat{{L}}_{s,\ell+\lambda,\kappa})^{-1}(f))\\
&-(s+s')r^2h_{r_+}(2-h_{r_+}D)(\hat{{L}}_{s,\ell+\lambda,\kappa})^{-1}(f)\Big],
\end{split}
\end{equation*}
with $P_1$ defined in Proposition \ref{prop:relationAL}. Since $( \hat{{L}}_{s,\ell+\lambda,\kappa})^{-1}(f)\in H_{\sigma,2,\rho_0}$, it follows that $\lim_{s'\to s} \frac{B(s')-B(s)}{s-s'}$ is a bounded linear operator for all $s\in \Omega_{\sigma}$, and hence $s\mapsto B(s)$ is analytic. 

We can relate the existence of $ \hat{{L}}_{s,\ell,\kappa}^{-1}: H_{\sigma,\rho_0}\to H_{\sigma,\rho_0}$ to the invertibility of $1+B(s)$. Indeed, one may easily verify that ${\hat{\psi}}\in H_{\sigma,\rho_0}$ satisfies
\begin{equation*}
\hat{{L}}_{s,\ell,\kappa}({\hat{\psi}})=f \quad \Longleftrightarrow \quad (1+B(s))({\hat{\psi}})=( \hat{{L}}_{s,\ell+\lambda,\kappa})^{-1}(f)
\end{equation*}

Since there exists $s\in \Omega_{\sigma}$ (with $\re(s)>0$) such that $(\hat{{L}}_{s,\ell,\kappa})^{-1}$ is well-defined by Theorem \ref{thm:semigroup} and Proposition \ref{prop:relationAL}, we can apply the Analytic Fredholm Theorem  (see for example Theorem 7.92 of \cite{ren93}) to conclude that $(id+B(s))^{-1}: H_{\sigma,\rho_0} \to H_{\sigma,\rho_0}$ is holomorphic for all $s\in \Omega_{\sigma}\setminus \Lambda_{QNF}^{\sigma,\ell,\rho_0}$, where $\Lambda_{QNF}^{\sigma,\ell, \rho_0}\subset \Omega_{\sigma}$ is a discrete set. Furthermore, if $s\in  \Lambda_{QNF}^{\sigma, \ell,\rho_0}$, then the space of solutions $\hat{{L}}_{s,\ell,\kappa}({\hat{\psi}})=0$ is finite dimensional.

By the above, we have that if $s\notin   \Lambda_{QNF}^{\sigma, \ell,\rho_0}$, then $\hat{\psi}=\hat{{L}}_{s,\ell,\kappa}^{-1}(f)\in H_{\sigma,\rho_0}$. Since we can moreover express $\hat{\psi}=\hat{{L}}_{s,\ell+\lambda,\kappa}^{-1}(f+\lambda(\lambda+1)\hat{\psi})$, we can take $\lambda$ suitably large and apply Proposition \ref{prop:inversepositivekappa} and Proposition \ref{prop:resolventzerokappa} to conclude that $\hat{\psi}\in H_{\sigma,2,\rho_0}$.
\end{proof}

\subsection{Convergence of quasinormal modes as $\kappa \downarrow 0$}

The proposition below is a variation of the proof of Theorem 7.92 of \cite{ren93}, utilising the uniform convergence of ${L}^{-1}_{\kappa,\ell+\lambda,s}$ as $\kappa\downarrow 0$ that is established in Proposition \ref{prop:resolventzerokappa}.
\begin{proposition}
\label{prop:convqnf} 
Let $\sigma\in \R_{>0}$ and denote 
\begin{equation*}
A_{\kappa,s}:=\lambda(\lambda+1)(\hat{L}_{s, \ell+\lambda, \kappa})^{-1}: H_{\sigma,\rho_0}\to H_{\sigma,\rho_0},
\end{equation*}
with $\kappa\geq 0$ and $s\in \Omega_{\sigma}$. Suppose that 
\begin{equation*}
\ker (1-A_{0,s_*})\neq \emptyset
\end{equation*}
for some $s_*\in \Omega_{\sigma}$.

Then there exists a sequence $\{\kappa_n\}$ in $\R_{>0}$ and $\{s_n\}$ in $\C$, such that $\kappa_n\to0$, $s_n\to s_*$ and
\begin{equation*}
\ker (1-A_{\kappa_n,s_n}) \neq \emptyset.
\end{equation*}
\end{proposition}
\begin{proof}
By Proposition \ref{prop:resolventzerokappa}, there exists a suitably small $\kappa_0>0$ and a neighbourhood $U_{s_*}\subset \Omega_{\sigma}$ of $s_*$, such that for all $s\in U_{s_*}$ and $0\leq \kappa\leq \kappa_0$,
\begin{equation*}
||A_{\kappa,s}-A_{0,s_*}||\leq \frac{1}{4}.
\end{equation*}
Since $A_{0,s_*}$ is compact, there exists an $N\in \N$ and a linear operator of rank $N$, $B: H_{\sigma,\rho_0}\to H_{\sigma,\rho_0}$, such that
\begin{equation*}
||B-A_{0,s_*}||\leq \frac{1}{4}.
\end{equation*}
We therefore have that for all $s\in U_{s_*}$, we can estimate
\begin{equation*}
||B-A_{\kappa,s}||\leq \frac{1}{2}.
\end{equation*}
Since $A_{\kappa,s}$ is analytic in $s$, we use the convergence of the appropriate Neumann series to conclude that the operator $(1-(A_{\kappa,s}-B))^{-1}$ exists and is also analytic in $s$. Furthermore, we have that for all $K \Subset U_{s_*}$
\begin{equation*}
\sup_{s\in K}||(1-(A_{\kappa,s}-B))^{-1}-(1-(A_{0,s}-B))^{-1}||\to 0
\end{equation*}
as $\kappa \downarrow 0$.

We define the linear operator
\begin{equation*}
F_{\kappa}(s):=B\circ (1-(A_{\kappa,s}-B))^{-1} : H_{\sigma,\rho_0}\to H_{\sigma,\rho_0}\,
\end{equation*}
so that we can write
\begin{equation*}
1-A_{\kappa,s}=(1-F_{\kappa}(s))\circ  (1-(A_{\kappa,s}-B)).
\end{equation*}
We have that $1-A_{\kappa,s}$ is invertible if and only if $1-F_{\kappa}(s)$ is invertible.

Note that $F_{\kappa}(s)$ is an operator of rank $N$ that is analytic in $s$ and for all $K \Subset U_{s_*}$
\begin{equation*}
\sup_{s\in K}||F_{\kappa}(s)-F_{0}(s)||\to 0
\end{equation*}
as $\kappa \downarrow 0$.

By the finite rank property, there exists a basis $\{e_i\}$ of $B( H_{\sigma,\rho_0})$ with $1\leq i\leq N$ and analytic functions $f_{\kappa,i}: U_{s_*} \to H_{\sigma,\rho_0}$, with $1\leq i \leq N$, such that for all $v \in  H_{\sigma,\rho_0}$,
\begin{equation*}
F_{\kappa}(s)v=\sum_{i=1}^N \la f_{\kappa,i},v\ra e_i.
\end{equation*}
Define $M_{ij; \kappa}(s):= \la f_{\kappa,i}, e_j\ra$, $1\leq i,j\leq N$ and denote with $M_{\kappa}(s)$ the corresponding $N\times N$ matrix. Then $M_{ij;\kappa}(s)$ is analytic in $s$ and for all $K \Subset U_{s_*}$, $\sup_{s\in K}|M_{ij;\kappa}(s)-M_{ij;0}(s)|\to 0$ as $\kappa \downarrow 0$.

We can conclude that $1-F_{\kappa}(s)$ and therefore $1-A_{\kappa,s}$ is invertible if and only if
\begin{equation*}
\det (1-M_{\kappa}(s))=0.
\end{equation*}
Since, $\det (1-M_{0}(s_*))=0$ by assumption, we can take $\{\kappa_n\}_{n\in \N}$ to be any sequence converging to $0$ with $|\kappa_n|< \kappa_0$ for all $n\in \N$ and use the above properties of $M_{ij;\kappa}(s)$ to apply Hurwitz's theorem, see for example Theorem 2.5 in \S 7.2 of \cite{conway78}, and conclude that there exists a corresponding sequence $\{s_n\}$ in $U_{s_*}$ such that $s_n\to s_*$ and
\begin{equation*}
\det (1-M_{\kappa_n}(s_n))=0. \qedhere
\end{equation*}
\end{proof}

We will use the above convergence property of the quasinormal frequencies associated to $\kappa>0$ to derive an improved regularity property of the $\kappa=0$ quasinormal modes.
\begin{proposition}
\label{prop:conveigenf}
Let $\sigma \in {\R_{>0}}$ and let $\hat{\psi}_{s_*;0,\ell}\in \ker \hat{L}_{s_*,\ell,0}$, with $s_* \in \Omega_{\sigma}$, $\kappa=0$ and angular frequency $\ell$. Then there exists a sequence
\begin{equation*}
\hat{\psi}_{s_n;\kappa_n,\ell} \in \ker \hat{L}_{s_n,\kappa_n,\ell}
\end{equation*}
with $\kappa_n\to 0$ and $s_n\to s_*$ such that
\begin{equation*}
||\hat{\psi}_{s_*;0,\ell}-\hat{\psi}_{s_{n}; \kappa_{n},\ell}||_{\sigma, \rho_0}\to 0
\end{equation*}
for all $\rho_0>0$ suitably small, as $n\to \infty$.
\end{proposition}
\begin{proof}
Define $A:=\lambda(\lambda+1)({L}_{s_*,\ell+\lambda,0})^{-1}$. By Proposition \ref{prop:convqnf}, for every sequence $\{\kappa_n\}$ in $\R_{>0}$, such that $\kappa_n\to0$ as $n\to \infty$, there exists a sequence $\{s_n\}$ in $\C$ such that $s_n\to s_*$ as $n\to \infty$. Furthermore, if we define $A_n:=\lambda(\lambda+1)({L}_{s_n,\kappa_n, \ell+\lambda})^{-1}$, then $\ker (1-A_n)\neq \emptyset$ and
\begin{equation*}
||A-A_n||\to 0
\end{equation*}
as $n\to \infty$. Hence, assuming without loss of generality that $||\hat{\psi}_{s_*;0,\ell}||_{\sigma}=1$, we can apply Lemma \ref{lm:conveigenvectors} to obtain a subsequence $\hat{\psi}_{s_{n_k};\kappa_{n_k},\ell}\in \ker (1-A_{n_k})$ such that 
\begin{equation*}
||\psi_{s_*;0,\ell}-\psi_{s_{n_k}; \kappa_{n_k},\ell}||_{\sigma}\to 0
\end{equation*}
as $k\to \infty$.
\end{proof}

\begin{proposition}
\label{prop:moreregqnm}
Let $\sigma\in \R_{>0}$ and let $\hat{\psi}_{s_*;0,\ell}\in \ker \hat{L}_{s_*,\ell,0}\subset  H_{\sigma,\rho_0}$ with $s_* \in \Omega_{\sigma}$. Then
\begin{equation*}
\hat{\psi}_{s_*;0,\ell}\in H_{\sigma',\rho_0}
\end{equation*}
for all
\begin{equation}
\label{eq:conditionsigmaprime}
\frac{1}{4}|s_*|^2<{\sigma'}^2<\frac{3}{4}|s_*|^2-2|\re(s_*)|^2.
\end{equation}
\end{proposition}
\begin{proof}
By Proposition \ref{prop:conveigenf}, there exists a sequence $\hat{\psi}_{s_{n}; \kappa_{n},\ell}\in  H_{\sigma,\rho_0}$, such that $||\hat{\psi}_{s_{n}; \kappa_{n},\ell}-\hat{\psi}_{s_*;0,\ell}||_{\sigma}\to 0$ as $n\to \infty$ and
\begin{equation*}
{L}_{\kappa_n,\ell+\lambda, s_n,}(\hat{\psi}_{s_{n}; \kappa_{n},\ell})=\lambda(\lambda+1)\hat{\psi}_{s_{n}; \kappa_{n},\ell}.
\end{equation*}
Let us assume, without loss of generality, that $\sigma'$ satisfies \eqref{eq:conditionsigmaprime} and $\sigma'>\sigma$. By Corollary \ref{cor:maingevreysum3}, we can estimate
\begin{equation}
\label{eq:poskappaqnmest}
\begin{split}
\int_{r_+}^{r_c}& (|\hat{\psi}_{s_{n}; \kappa_{n},\ell}|^2+|\partial_r\hat{\psi}_{s_{n}; \kappa_{n},\ell}|^2) r^2\,dr+\sum_{\star \in \{+,c\}}\sum_{m=0}^{N_{\infty}} \frac{|2s_n|^{2m}{\widetilde{\sigma'}}_n^{2m}}{(m+1)!^2n!^2} m^2(m+1)^2\int_{0}^{{\rho}_0} |(r \hat{\psi}_{s_{n}; \kappa_{n},\ell})_{(m)}|^2\, d{\rho}_{\star}\\
\leq &\: C_{\ell,\lambda,s_n}\int_{r_+}^{r_c} |\hat{\psi}_{s_{n}; \kappa_{n},\ell}|^2 r^2\,dr+ C_{\ell, \lambda,  s_n} \sum_{\star \in \{+,c\}} \sum_{m=0}^{N_{\infty}} \frac{|2s_n|^{2m}{\widetilde{\sigma'}}_n^{2m}}{(m+1)!^2m!^2} \int_{0}^{{\rho}_0} |(r \hat{\psi}_{s_{n}; \kappa_{n},\ell})_{(m)}|^2\, d{\rho}_{\star},
\end{split}
\end{equation}
for all $N_{\infty}>N_{\kappa_n}$, where $|s_n|\widetilde{\sigma'}_n:=\sigma'_n$, with
\begin{equation*}
\frac{1}{4}|s_n|^2<{\sigma'}_n^2<\frac{3}{4}|s_n|^2-2|\re(s_*)|^2.
\end{equation*}

Since $s_n\to s_*$, there exists $N_0=N_0(\sigma')\in \N$ such that for all $n\geq N_0$, we can fix $\sigma'_n=\sigma'$. Furthermore, as the constant $C_{\ell,\lambda,s_n}$ does not diverge as $s_n\to s_*$, we can replace it by a constant that is independent of $n$.

There exists $M_{\infty}$ suitably large depending on $\ell$ and $\lambda$, so that we can absorb all terms on the right-hand side of \eqref{eq:poskappaqnmest} with $m>M_{\infty}$ into the left-hand side, making use of the additional factor $m^2(m+1)^2$ that appears on the left-hand side, provided we take $N_{\infty}$ to be suitably large. We obtain:
\begin{equation*}
\begin{split}
\int_{r_+}^{r_c}& (|\hat{\psi}_{s_{n}; \kappa_{n},\ell}|^2+|\partial_r\hat{\psi}_{s_{n}; \kappa_{n},\ell}|^2) r^2\,dr+\sum_{\star \in \{+,c\}}\sum_{m=0}^{N_{\infty}} \frac{(2\sigma')^{2m}}{(m+1)!^2n!^2} m^2(m+1)^2\int_{0}^{{\rho}_0} |(r \hat{\psi}_{s_{n}; \kappa_{n},\ell})_{(m)}|^2\, d{\rho}_{\star}\\
\leq &\: C_{\ell,\lambda,s_*}\int_{r_+}^{r_c} |\hat{\psi}_{s_{n}; \kappa_{n},\ell}|^2 r^2\,dr+ C_{\ell, \lambda,s_*}\sum_{\star \in \{+,c\}} \sum_{m=0}^{M_{\infty}} \frac{(2\sigma')^{2m}}{(m+1)!^2m!^2} \int_{0}^{{\rho}_0} |(r \hat{\psi}_{s_{n}; \kappa_{n},\ell})_{(m)}|^2\, d{\rho}_{\star}.
\end{split}
\end{equation*}
By taking the limit $N_{\infty}\to \infty$ on the left-hand side, we therefore obtain that
\begin{equation*}
||\hat{\psi}_{s_{n}; \kappa_{n},\ell}||_{\sigma',\rho_0}\leq C ||\hat{\psi}_{s_{n}; \kappa_{n},\ell}||_{\sigma,\rho_0},
\end{equation*}
for some constant $C=C(\ell,\lambda, s_*, \sigma',\sigma)>0$. Since $\hat{\psi}_{s_{n}; \kappa_{n},\ell}$ is a convergent sequence in $H_{\sigma,\rho_0}$, we can conclude that $\hat{\psi}_{s_{n}; \kappa_{n},\ell}$ is a uniformly bounded sequence in $H_{\sigma',\rho_0}$
By considering the difference $\hat{\psi}_{s_{n}; \kappa_{n},\ell}-\psi_{s_{n'}; \kappa_{n'},\ell}$, with $n>n'$ and using that
\begin{equation*}
\hat{L}_{\kappa_n,\ell+\lambda, s_n,}(\hat{\psi}_{s_{n}; \kappa_{n},\ell}-\psi_{s_{n'}; \kappa_{n'},\ell})=\lambda(\lambda+1)(\hat{\psi}_{s_{n}; \kappa_{n},\ell}-\psi_{s_{n'}; \kappa_{n'},\ell})+ (\hat{L}_{\kappa_n,\ell+\lambda, s_n}-\hat{L}_{\kappa_n',\ell+\lambda, s_n'})\psi_{s_{n'}; \kappa_{n'},\ell}
\end{equation*}
and $\kappa_n\downarrow 0$, $s_n\to s_*$, it is straightforward to apply the above estimates to the difference $\hat{\psi}_{s_{n}; \kappa_{n},\ell}-\hat{\psi}_{s_{n'}; \kappa_{n'},\ell}$ in order to show that $\{\hat{\psi}_{s_{n}; \kappa_{n},\ell}\}$ is also a Cauchy sequence with respect to the $H_{\sigma',\rho_0}$ norm, so the corresponding (unique) limit must satisfy: $\hat{\psi}_{s_*;0,\ell}\in H_{\sigma',\rho_0}$.
\end{proof}

In the following proposition, we investigate the dependence of $\Lambda^{\sigma,\ell,\rho_0}_{QNF}$ on the choice of $\rho_0$.

\begin{proposition}
\label{prop:indprho0}
Consider $\rho_0$ and $\rho_0'$ suitably small so that Proposition \ref{prop:fredholm} can be applied to define $\Lambda^{\sigma,\ell,\rho_0}_{QNF}$ and $\Lambda^{\sigma,\ell,\rho_0'}_{QNF}$. Then
\begin{equation*}
\Lambda^{\sigma,\ell,\rho_0}_{QNF}=\Lambda^{\sigma,\ell,\rho_0'}_{QNF}.
\end{equation*}
\end{proposition}
\begin{proof}
Assume without loss of generality that $\rho_0'>\rho_0$.  Let $\hat{\psi}_{s_*; 0,\ell}\in  H_{\sigma,\rho_0}$. Then we use that, as in the proof of Proposition \ref{prop:moreregqnm}, there exists a sequence $\hat{\psi}_{s_{n}; \kappa_{n},\ell}\in  H_{\sigma,\rho_0}$, such that $||\hat{\psi}_{s_{n}; \kappa_{n},\ell}-\hat{\psi}_{s_*;0,\ell}||_{\sigma}\to 0$ as $n\to \infty$ and
\begin{equation*}
{L}_{\kappa_n,\ell+\lambda, s_n,}(\hat{\psi}_{s_{n}; \kappa_{n},\ell})=\lambda(\lambda+1)\hat{\psi}_{s_{n}; \kappa_{n},\ell}.
\end{equation*}
Similarly, we have that
\begin{equation*}
\begin{split}
\int_{r_+}^{r_c}& (|\hat{\psi}_{s_{n}; \kappa_{n},\ell}|^2+|\partial_r\hat{\psi}_{s_{n}; \kappa_{n},\ell}|^2) r^2\,dr+\sum_{\star \in \{+,c\}}\sum_{m=0}^{N_{\infty}} \frac{(2\sigma)^{2m}}{(m+1)!^2n!^2} m^2(m+1)^2\int_{0}^{{\rho}_0'} |(r \hat{\psi}_{s_{n}; \kappa_{n},\ell})_{(m)}|^2\, d{\rho}_{\star}\\
\leq &\: C_{\ell,\lambda,s_*}\int_{r_+}^{r_c} |\hat{\psi}_{s_{n}; \kappa_{n},\ell}|^2 r^2\,dr+ C_{\gamma, |s|,\ell} \sum_{\star \in \{+,c\}}\sum_{m=0}^{M_{\infty}} \frac{(2\sigma)^{2m}}{(m+1)!^2m!^2} \int_{0}^{{\rho}_0'} |(r \hat{\psi}_{s_{n}; \kappa_{n},\ell})_{(m)}|^2\, d{\rho}_{\star}.
\end{split}
\end{equation*}
In order to conclude that the term in the sum on the right-hand side is finite, we split the integral over $[0,\rho_0']$ into an integral over $[0,\rho_0]$ and $(\rho_0,\rho_0']$. The integral over $[0,\rho_0']$ is finite since $\hat{\psi}_{s_{n}; \kappa_{n},\ell}\in  H_{\sigma,\rho_0}$. In order to conclude that the remaining integral is finite, we use the equation ${L}_{\kappa_n,\ell, s_n,}(\hat{\psi}_{s_{n}; \kappa_{n},\ell})=0$ (and commute with $\partial_{\rho}^m$) to control all higher-order derivatives. Hence, $\hat{\psi}_{s_{n}; \kappa_{n},\ell}\in  H_{\sigma,\rho_0'}$. It follows straightforwardly that the limit $\hat{\psi}_{s_*; 0,\ell}$ must also be an element of $H_{\sigma,\rho_0'}$.
\end{proof}

By Proposition \ref{prop:indprho0}, we can unambiguously denote
\begin{equation*}
\Lambda^{\sigma,\ell}_{QNF}=\Lambda^{\sigma,\ell,\rho_0}_{QNF},
\end{equation*}
omitting $\rho_0$ in the superscript.
\begin{corollary}
\label{cor:mainresult}
Let $\sigma \in {\R_{>0}}$ and consider $\mathcal{A}: \mathbf{H}_{\sigma,\rho_0} \supseteq \mathcal{D}(\mathcal{A})\to \mathbf{H}_{\sigma,\rho_0}$. 
\begin{itemize}
\item[\rm (i)]Then 
\begin{align}
\label{eq:spect1}
\textnormal{Spect} (\mathcal{A}_{\ell})\cap \Omega_{\sigma}=&\:\textnormal{Spect}_{\rm point} (\mathcal{A}_{\ell})\cap \Omega_{\sigma}=\Lambda^{\sigma,\ell}_{QNF},\\
\label{eq:spect2}
\textnormal{Spect}_{\rm point} (\mathcal{A})\cap \Omega_{\sigma}=&\:\Lambda^{\sigma}_{QNF}:=\bigcup_{\ell \in \N_0} \Lambda^{\sigma,\ell}_{QNF},
\end{align}
with $\Lambda^{\sigma,\ell}_{QNF}\subset \Omega_{\sigma}$ the sets of isolated points from Proposition \ref{prop:fredholm}.
\item[\rm (ii)]
Define
\begin{equation*}
\Lambda_{QNF}:=\bigcup_{\sigma \in {\R_{>0}}}\Lambda^{\sigma}_{QNF}\subset \left\{|\textnormal{arg}(z)|<\frac{2}{3}\pi \right\},
\end{equation*}
then $\Lambda_{QNF}$ is a set of isolated points in $\left\{|\textnormal{arg}(z)|<\frac{2}{3}\pi \ \right\}$ (with possible accumulation only on the boundary of $\left\{|\textnormal{arg}(z)|<\frac{2}{3}\pi \ \right\}$ in $\C$).
\item[\rm (iii)]
\begin{equation*}
\textnormal{Spect}_{\rm point} (\mathcal{A})\cap \{\re z \geq 0\}=\emptyset.
\end{equation*}
\item[\rm (iv)]
For all $\ell\in \N_0$ there exists $\delta_{\ell}>0$ such that $\Lambda^{\sigma,\ell}_{QNF}\cap \{|z|<\delta_{\ell}\}=\emptyset$, i.e. the elements of $\Lambda^{\sigma,\ell}_{QNF}$ do not accumulate at the origin.
\end{itemize}
\end{corollary}
\begin{proof}
First, consider part (i) of the proposition. By Proposition \ref{prop:relationAL} combined with Proposition \ref{prop:fredholm}, we immediately obtain \eqref{eq:spect1}. We moreover obtain
\begin{equation*}
\bigcup_{\ell \in \N_0} \Lambda^{\sigma,\ell}_{QNF} \subseteq \textnormal{Spect}_{\rm point} (\mathcal{A})\cap \Omega_{\sigma},
\end{equation*}
since $\textnormal{Spect}_{\rm point} (\mathcal{A}_{\ell})\cap \Omega_{\sigma} \subset \textnormal{Spect}_{\rm point} (\mathcal{A}_{\ell})$. Suppose $s\in  \textnormal{Spect}_{\rm point} (\mathcal{A})$. Let $(\Psi,\Psi')\in \ker (\mathcal{A}-s)$ be initial data, then the corresponding solution $\psi$ is of the form $\psi(\tau,r,\theta,\varphi)=\hat{\psi}(r,\theta,\varphi)e^{s \tau}$ and we can decompose $\hat{\psi}=\sum_{\ell\in \N_0} \hat{\psi}_{\ell}$ with $\hat{L}_{s,\ell,0}(\hat{\psi}_{\ell})=0$. By Proposition \ref{prop:resolventzerokappa}, we have that there exists $L>0$ such that for all $\ell>L$,  $\hat{L}_{s,\ell,0}^{-1}$ is well-defined, so we must have that $\hat{\psi}=\sum_{\ell=0}^L\hat{\psi}_{\ell}$ and we can conclude that $s\in \bigcup_{\ell=0}^L  \Lambda^{\sigma,\ell}_{QNF}$ so 
\begin{equation*}
\bigcup_{\ell \in \N_0} \Lambda^{\sigma,\ell}_{QNF} \supseteq \textnormal{Spect}_{\rm point} (\mathcal{A})\cap \Omega_{\sigma}.
\end{equation*}

Consider part (ii). Let $s_0\in \Lambda_{QNF}^{\sigma}$ for some $\sigma \in {\R_{>0}}$. Since $\Lambda^{\sigma}_{QNF}$ is a set of isolated points by Proposition \ref{prop:fredholm}, there exists a neighbourhood $U_{s_0}$ of $s_0$ such that $U_{s_0}\cap \Omega_{\sigma}=\{s_0\}$. Suppose there exists $s_1\in U_{s_0}\setminus\{s_0\}$ and $\ell\in \N$, such that $s_1\in \Lambda_{QNF}^{\sigma',\ell}$ for some $\sigma'\in \R$. Let $\hat{\psi}_{s_1} \in \ker \hat{L}_{s_1,\ell,0}\subset H_{\sigma'}$. Since $s_1\in \Omega_{\sigma}$ we can apply Proposition \ref{prop:moreregqnm} to conclude that $\hat{\psi}_{s_1}\in H_{\sigma,\rho_0}$. But then $s_1\in \Lambda^{\sigma}_{QNF}$, which is a contradiction. So we can conclude that $\bigcup_{\sigma \in {\R_{>0}}}\Lambda^{\sigma}_{QNF}\cap U_{s_0}=\{s_0\}$.

Consider now part (iii). Let $\re(s)\geq 0$ and suppose $(\mathcal{A}-s)(\Psi,\Psi')=0$ with $(\Psi,\Psi')\in \mathbf{H}_{\sigma,\rho_0}$ for some $\sigma \in {\R_{>0}}$. Then there exists a corresponding solution $\psi(\tau,r,\theta,\varphi)=e^{s\tau}\Psi(r,\theta,\varphi)$ to \eqref{eq:waveequation}. However, by an application of the degenerate energy estimates (see for example \cite{aretakis1} for the relevant estimates in the extremal Reissner--Nordstr\"om setting), it follows that the (degenerate) energy with respect to $T$ must decrease in time, which contradicts the supposed exponential growth or non-decay in $\tau$ of the $T$-energy norm.

Finally, (iv) follows from Proposition \ref{prop:fredholm} combined with Proposition \ref{prop:relationAL}.
\end{proof}

\section{Relation to the scattering resonances}
\label{sec:tradqnms}

In this section, we show that the ``traditional'' notion of quasinormal frequencies as scattering resonances (with fixed angular frequency), defined as in Theorem \ref{thm:bach}, but in the setting of extremal Reissner--Nordstr\"om, can be interpreted as eigenvalues of $\mathcal{A}$. In fact, we determine the appropriate restriction to $\mathcal{A}$, which guarantees that \emph{all} eigenvalues correspond to scattering resonances in a suitable subset of the complex plane. Furthermore, we show that we can make sense of scattering resonances in extremal Reissner--Nordstr\"om without a restriction to fixed angular frequencies.

Let $\psi$ be a solution to \eqref{eq:waveequation}. Then we can express in $(t,r,\theta,\varphi)$ coordinates
\begin{equation}
\label{eq:asympflatwe}
D^{-1}r^{-1}\left[ -\partial_t^2(r\psi)+(D\partial_r)^2(r\psi)+Dr^{-2}\mathring{\slashed{\Delta}}(r\psi)-(r^{-1}DD'(r)+2D\mathfrak{l}^{-2})r\psi\right]=0.
\end{equation}
We introduce the following fixed-frequency operator with $\re s>0$:
\begin{equation*}
\hat{\mathcal{L}}_{s,\kappa}(\hat{\psi})=(D\partial_r)^2(r\hat{\psi})+Dr^{-2}\mathring{\slashed{\Delta}}(r\hat{\psi})-(s^2+r^{-1}DD'(r)+2D\mathfrak{l}^{-2})r\hat{\psi}.
\end{equation*}

Let $f\in L^2(\{t=0\})$, where $L^2(\{t=0\})$ is defined with respect to the natural volume form on $\{t=0\}$ with respect to the induced metric, and denote the trivial extension as follows: $\overline{f}: \mathcal{R}\to \R$, $\overline{f}(t,r,\theta,\varphi):=f(r,\theta,\varphi)$. Then
the map
\begin{align*}
Q_s:& L^2(\{t=0\})\to L^2(\Sigma\setminus (\mathcal{H}^+\cup\mathcal{C}^+)),\\
f&\mapsto \overline{f}e^{st}|_{\Sigma}
\end{align*}
is well-defined and invertible. Furthermore, if $\chi,\chi': \{t=0\}  \to \R$ are smooth cut-off functions that vanish near $\mathcal{H}^+$ and $\mathcal{C}^+$, then it follows immediately that
\begin{align*}
Q_s\circ \chi:  L^2(\{t=0\})&\to L^2(\Sigma),\\
\chi'\circ \mathcal{Q}_s^{-1}\circ \chi: H^2(\Sigma) &\to  H^2(\{t=0\})
\end{align*}
are bounded linear operators that are holomorphic in $s$.
\begin{lemma}
\label{lm:boundresposres}
Let $\kappa_c=\kappa_+=\kappa\geq 0$. Then we can express
\begin{equation*}
D r \hat{\mathcal{L}}_{s,\kappa}={Q}_s^{-1}\circ \hat{L}_{s,\ell,\kappa}\circ{Q}_s.
\end{equation*}
Let $\chi,\chi': \{t=0\}  \to \R$ be smooth cut-off functions that vanish near $\mathcal{H}^+$ and $\mathcal{C}^+$ or $\mathcal{I}^+$. Then we moreover have that for $\re s>0$:
\begin{equation*}
\chi'\circ {R}_{\kappa}(s)\circ \chi:=\chi'\circ \hat{\mathcal{L}}^{-1}_{s,\kappa}\circ \chi= \chi'\circ{Q}_s^{-1}\circ (D^{-1} r^{-1}\hat{L}^{-1}_{s,\kappa})\circ{Q}_s \circ \chi:  L^2(\{t=0\}) \to  H^2(\{t=0\})
\end{equation*}
is a bounded linear operator.
\end{lemma}
\begin{proof}
We first need to show that $(r^{-2}\hat{{L}})^{-1}_{s,\kappa}: L^2(\Sigma)\to H^1(\Sigma)$ is well-defined as a bounded linear operator when $\re s>0$. This is a standard result. By performing the steps in the proof of Proposition \ref{prop:mainred-shiftest} with $\re s>0$ and $\kappa_+,\kappa_c\geq 0$, it follows that
\begin{equation*}
||\hat{\psi}||_{H^1(\Sigma)}\leq C ||r^{-2}\hat{{L}}_{s,\kappa}\hat{\psi}||_{L^2(\Sigma)},
\end{equation*}
with $C>0$ a constant that is independent of $\kappa$. We then apply similar steps to those in Appendix \ref{sec:red-shift} to construct $\hat{\mathcal{L}}^{-1}_{s,\kappa}$. In fact, the above estimate moreover implies that
\begin{equation*}
||\chi (r^{-2}\hat{L})^{-1}_{s,\kappa}(\tilde{f})||_{H^2(\Sigma)}\leq C ||\tilde{f}||_{L^2(\Sigma)},
\end{equation*}
for a cut-off function $\chi$ that is supported away from $r=r_+$ and $r=r_c$ or $r=\infty$.
\end{proof}

\begin{proposition}
\label{prop:merocont}
Let $0\leq \kappa_c=\kappa_+=\kappa<\kappa_0$, with $\kappa_0>0$ suitably small. Then the linear operator
\begin{equation*}
\chi'\circ R_{\kappa}(s)\circ \chi:L^2(\{t=0\}) \to  H^2(\{t=0\})
\end{equation*}
can be meromorphically continued to $\{|\textnormal{arg}(z)|<\frac{2}{3}\pi \}\subset \C$ and the poles form a subset of $\Lambda_{QNF}$.
\end{proposition}
\begin{proof}
By Lemma \ref{lm:boundresposres}, we have to prove that 
\begin{equation*}
\chi' \circ \hat{L}^{-1}_{s,\kappa}\circ \chi: L^2(\{t=0\}) \to  H^2(\{t=0\})
\end{equation*}
can be meromorphically continued to $\{-\frac{\re s}{|s|}<\frac{1}{2}\}\subset \C$. Note that we can express for $f\in L^2(\{t=0\})$
\begin{equation}
\label{eq:sumlresolvent}
(\chi' \circ \hat{L}^{-1}_{s,\kappa}\circ \chi)(f)=\sum_{\ell=0}\sum_{|m|\leq \ell}(\chi' \circ \hat{L}^{-1}_{s,\kappa,\ell}\circ \chi)(f_{\ell m})Y_{\ell m}(\theta,\varphi)
\end{equation}
when $\re (s)>0$. 

Let $s\in \Omega$, with $\Omega \subset \{-\frac{\re z}{|z|}<\frac{1}{2}, z\neq 0\}$ such that $\Omega\cap \{\re z \leq 0\}$ is compact. Then, by Corollary \ref{cor:maingevreysum3} with $\rho_0>0$ suitably small depending on $\chi$ and $\chi'$, there exists $L=L(\Omega)>0$ suitably large such that for all $\ell\geq L+1$, $(\chi' \circ \hat{L}^{-1}_{s,\kappa,\ell}\circ \chi)$ is a holomorphic operator from $L^2[r_+,r_c]$ to $H^2[r_+,r_c]$ and there exists $C>0$ such that
\begin{equation*}
\left|\left|\sum_{\ell=L+1}^{\infty}\sum_{|m|\leq \ell}(\chi' \circ \hat{L}^{-1}_{s,\kappa,\ell}\circ \chi)(f_{\ell m})Y_{\ell m}(\theta,\varphi)\right|\right|_{H^2(\{t=0\})}\leq C ||f||_{L^2(\{t=0\}}.
\end{equation*}

 Furthermore, by Proposition \ref{prop:fredholm}, 
\begin{equation*}
f\mapsto \sum_{\ell=0}^L\sum_{|m|\leq \ell}(\chi' \circ \hat{L}^{-1}_{s,\kappa,\ell}\circ \chi)(f_{\ell m})Y_{\ell m}(\theta,\varphi)
\end{equation*}
is meromorphic on $\Omega$ as an operator from $L^2(\{t=0\})$ to $H^2(\{t=0\})$. We can conclude that \eqref{eq:sumlresolvent} is well-defined as a meromorphic operator from $L^2(\{t=0\})$ to $H^2(\{t=0\})$ when $\{-\frac{\re s}{|s|}<\frac{1}{2}\}\subset \C$.
\end{proof}

We would like to identify the poles of $R_{\kappa}(s)$ with the eigenvalues of $\mathcal{A}$, restricted to a suitable subspace of $\mathbf{H}_{\sigma,\rho_0}$. We introduce the following function space:
\begin{equation*}
X=\left\{\mathcal{S}(\tau)(\Psi,\Psi')\,|\, (\Psi,\Psi')\in (C_c^{\infty}(\Sigma)\times C^{\infty}(S))\cap \mathbf{H}_{\sigma,\rho_0},\, \tau\geq 0\right\}.
\end{equation*}
We denote with $\mathbf{H}_{\sigma,\rho_0}^{\rm res}$ the closure of $X$ under the norm on $\mathbf{H}_{\sigma,\rho_0}$. Then, $\mathbf{H}_{\sigma,\rho_0}^{\rm res}$ is the \emph{smallest, closed subspace of $\mathbf{H}_{\sigma,\rho_0}$ that is invariant under $\mathcal{S}(\tau)$ and contains $(C_c^{\infty}(\Sigma)\times C^{\infty}(S))\cap \mathbf{H}_{\sigma,\rho_0}$}. Furthermore,
\begin{equation*}
\mathcal{S}(\tau): \mathbf{H}_{\sigma,\rho_0}^{\rm res}\to \mathbf{H}_{\sigma,\rho_0}^{\rm res}
\end{equation*}
is well-defined, by construction, and hence,
\begin{equation*}
\mathcal{A}^{\rm res}: \mathcal{D}_{\sigma, \rho_0}^{ \rm res}(\mathcal{A}) \to \mathbf{H}_{\sigma,\rho_0}^{\rm res}
\end{equation*}
is a densely defined closed operator with $ \mathcal{D}_{\sigma,\rho_0}^{\rm res}(\mathcal{A}) \subseteq \mathbf{H}_{\sigma,\rho_0}^{\rm res}$ a dense subset.

We now denote with $\Lambda^{\sigma, \rm res}_{QNF}$ the subset of eigenvalues in $\Lambda^{\sigma}_{QNF}$ for $\mathcal{A}$ restricted to $\mathcal{D}_{\sigma, \rho_0}^{\rm res}(\mathcal{A})$. Moreover, we denote
\begin{equation*}
\Lambda_{QNF}^{ \rm res}=\bigcup_{\sigma \in {\R_{>0}}} \Lambda^{\sigma,  \rm res}_{QNF} \subset \Lambda_{QNF}.
\end{equation*}
\begin{proposition}
\label{prop:relwithtradition}
Let $\kappa_c=\kappa_+=\kappa\geq 0$ be suitably small. Then the poles of
\begin{equation*}
\chi'\circ R_{\kappa}(s)\circ \chi:L^2(\{t=0\}) \to  H^2(\{t=0\})
\end{equation*}
for any choice of cut-off functions $\chi,\chi'$ correspond precisely to elements of $\Lambda_{QNF}^{\rm res}$.
\end{proposition}
\begin{proof}
Fix $\ell \geq 0$, $\sigma>0$ and $\rho_0>0$ suitably small. In view of the fact that for $s$ in any compact neighbourhood in $\Omega_{\sigma}$, there exists an $L\in \N$ such that  $\ker (\mathcal{A}-s) \subseteq \mathbf{H}_{\sigma,\rho_0}\cap \sum_{\ell=0}^LV_{\ell}$, combined with Proposition \ref{prop:indprho0}, it is sufficient to show that $s$ is \emph{not} a pole of $\chi'\circ R_{\kappa,\ell}(s)\circ \chi$ for all cut-off functions $\chi,\chi'$ if and only if $s\notin \Lambda_{QNF}^{\sigma, \ell, \rm res}$ for all $\ell$. For the sake of convenience, we moreover restrict to functions in $V_{m\ell}$ for a fixed $m\in \Z$ with $|m|\leq \ell$ and we ignore their angular part.

We will use the following expressions, derived in the proof of Proposition \ref{prop:relationAL}:
\begin{align*}
(\mathcal{A}_{\ell}^{\rm res}-s)\begin{pmatrix}
{\Psi}\\
{\Psi}'
\end{pmatrix}=\begin{pmatrix}
\widetilde{\Psi}\\
\widetilde{\Psi}'
\end{pmatrix}
\end{align*}
for $(\Psi,\Psi'), (\widetilde{\Psi},\widetilde{\Psi}')\in \mathbf{H}^{\rm res}_{\sigma, \rho_0}\cap V_{m \ell}$, if and only if
\begin{align*}
\hat{L}_{s,\ell,\kappa}(\Psi|_{S})=&\:h_{r_+}(2-h_{r_+}D)r^2 \left[\widetilde{\Psi}'+(P_1+s)(\widetilde{\Psi}|_S)\right],\\
\Psi'=&\: s\Psi|_S+\widetilde{\Psi}|_S,\\
r\hat{L}_{s,\ell,\kappa}(\Psi|_{N})=&\:-2\partial_{\rho_c}(r \widetilde{\Psi}|_N),\\
r\hat{L}_{s,\ell,\kappa}(\Psi|_{\underline{N}})=&\:-2\partial_{\rho_+}(r \widetilde{\Psi}|_{\underline{N}}).
\end{align*}
Let $(\widetilde{\Psi},\widetilde{\Psi}')\in (C_c^{\infty}(\Sigma)\cap  \mathbf{H}_{\sigma, \rho_0} \cap V_{m \ell})\times (C^{\infty}(S) \cap V_{m \ell}) $ and suppose $s$ is not a pole of $\chi'\circ R_{\kappa,\ell}(s)\circ \chi$. 

We define the following auxilliary function:
\begin{align*}
rf=&\:r^3h_{r_+}(2-h_{r_+}D) \left[\widetilde{\Psi}'+(P_1+s)(\widetilde{\Psi}|_S)\right]\quad \textnormal{on $S$},\\
rf=&-2\partial_{\rho_c}(r \widetilde{\Psi}|_N)\quad \textnormal{on $N$,}\\
rf=&-2\partial_{\rho_+}(r \widetilde{\Psi}|_{\underline{N}})\quad \textnormal{on $\underline{N}$,}
\end{align*}
then $ f\in L^2_c(r_+,r_c)\cap H_{\sigma,\rho_0}$. By the meromorphicity of $\hat{L}_{s,\kappa,\ell}^{-1}$, there exists a neighbourhood $U_s$ of $s$ such that for all $s'\in U_s$ we can decompose
\begin{equation}
\label{eq:decompresolv}
\hat{L}_{s',\kappa,\ell}^{-1}=A(s')+ \sum_{k=1}^N (s-s')^{-k} B_k, 
\end{equation}
where $N\in \N_0$, $A(s), B_k: H_{\sigma,\rho_0}\to H_{\sigma,\rho_0}$ are holomorphic linear operators. We moreover have that $B_k$ are finite rank operators that are independent of $s'$ (note that $B_k=0$ for all $k$ if $s\notin \Lambda_{QNF}^{\sigma,\ell}$).

By Lemma \ref{lm:boundresposres}, we can moreover express
\begin{equation}
\label{eq:relationresolvents}
\chi' D^{-1}r^{-1}\hat{L}_{s',\kappa,\ell}^{-1}(f)= (Q_{s'}\circ \chi' \circ R_{\kappa}(s')\circ \chi \circ Q_{s'}^{-1}) (f)
\end{equation}
with $\chi'$ an arbitrary smooth cut-off function and $\chi$ a smooth cut-off function such that $\chi\equiv 1$ on $\supp f$. Since $s$ is not a pole of $\chi'\circ R(s')\circ \chi$ by supposition, we have that $\chi'\circ R(s')\circ \chi$ is uniformly bounded in $s'$ provided $U_{s}$ is suitably small, so we can multiply both sides of \eqref{eq:relationresolvents} by $(s-s')^k$ with $1\leq k\leq N$ and take the limit $s'\to s$ to conclude that $\chi' B_k(f)=0$ for all $k$. Since $\chi'$ was chosen arbitrarily, we must in fact have that $f\in \ker B_k$ for all $k$ and hence  $f\in \textnormal{Ran}(\hat{L}_{s,\kappa,\ell})=(\textnormal{coker}\, \hat{L}_{s,\kappa,\ell})^{\perp}$. 

Defining $\Psi'=s\Psi|_S+\widetilde{\Psi}|_S$, with $\Psi=A(s)(f)$, we can therefore conclude that $(\mathcal{A}_{\ell}^{\rm res}-s)(\Psi,\Psi')=(\widetilde{\Psi},\widetilde{\Psi}')\in (C_c^{\infty}(\Sigma)\cap  \mathbf{H}_{\sigma,\rho_0} \cap V_{m \ell})\times (C^{\infty}(S)\cap  \mathbf{H}_{\sigma,\rho_0} \cap V_{m \ell})$, with $(\Psi,\Psi')$ uniquely determined, and since $(\mathcal{A}_{\ell}^{\rm res}-s)$ commutes with $\mathcal{S}(\tau)$:
\begin{equation*}
(\mathcal{A}_{\ell}^{\rm res}-s)^{-1}:  \textnormal{Ran}(\mathcal{A}_{\ell}^{\rm res}-s)\supset X\cap V_{\ell}\to \mathcal{D}_{\sigma, \rho_0}^{ \rm res}(\mathcal{A}) 
\end{equation*}
is well-defined.

By the closedness of $\textnormal{Ran}(\hat{L}_{s,\kappa,\ell})=(\textnormal{coker}\, \hat{L}_{s,\kappa,\ell})^{\perp}$ it follows that for any $s\in \Omega_{\sigma}$, $\textnormal{Ran}(\mathcal{A}_{\ell}^{\rm res}-s)$ is closed. Since $\mathcal{A}_{\ell}^{\rm res}-s$ trivially commutes with $\mathcal{S}(\tau)$ we moreover have that $\textnormal{Ran}(\mathcal{A}_{\ell}^{\rm res}-s)$ is $\mathcal{S}(\tau)$-invariant. We therefore obtain the following identity:
\begin{equation}
\label{eq:rangeequality}
\textnormal{Ran}(\mathcal{A}_{\ell}^{\rm res}-s)=\mathbf{H}^{\rm res}_{\sigma,\rho_0}\cap V_{\ell},
\end{equation}
since $\mathbf{H}^{\rm res}_{\sigma,\rho_0}\cap V_{\ell}$ is the smallest closed, $\mathcal{S}(\tau)$-invariant subspace of $\mathbf{H}_{\sigma,\rho_0}\cap V_{\ell}$ containing $(C_c^{\infty}(\Sigma)\cap  \mathbf{H}_{\sigma,\rho_0} \cap V_{m \ell})\times (C^{\infty}(S)\cap  \mathbf{H}_{\sigma,\rho_0} \cap V_{m \ell})$. Since $X$ is dense in $\mathbf{H}^{\rm res}_{\sigma,\rho_0}$ and the inverse is well-defined on $X$, we must have that $\ker (\mathcal{A}_{\ell}^{\rm res}-s)=\{0\}$ and $\textnormal{Spect}(\mathcal{A}_{\ell}^{\rm res})\cap U_s=\emptyset$. We conclude that $s\notin \Lambda^{\rm res,\sigma,\ell}_{QNF}$. 

In order to conclude the proof, we need to show that if $s\notin \Lambda^{\rm res,\sigma,\ell}_{QNF}$, then $s$ is not a pole of $R_{\kappa,\ell}(s)$ for any choice of cut-off functions $\chi,\chi'$. By \eqref{eq:decompresolv} and \eqref{eq:relationresolvents}, we have that if $s\notin \Lambda^{\rm res,\sigma,\ell}_{QNF}$, then $\chi'\circ  R_{\kappa,\ell}(s)\circ \chi$ is a bounded operator on $H_{\sigma,\rho_0}$.  We can take $\rho_0$ to be suitably small, depending on the choice of $\chi,\chi'$ to conclude that $\chi'\circ  R_{\kappa,\ell}(s)\circ \chi$ is a bounded operator on $L^2(r_+,r_c)$, from which it follows that $s$ cannot be a pole.
\end{proof}
\appendix
\section{Basic lemmas}

\begin{lemma}
\label{lm:conveigenvectors}
Let $H$ be a complex Hilbert space and let $(A_n)_{n\in \N}$ be a sequence of compact operators on $H$ such that there exists a compact operator $A$ on $H$  with
\begin{equation*}
A_n\to A
\end{equation*}
with respect to the operator norm, as $n\to \infty$. Assume moreover that
\begin{equation*}
\ker (1-A_n)\neq \emptyset \quad \textnormal{and}\quad \dim \ker (1-A_n)<\infty\:\textnormal{for all $n\in \N$}.
\end{equation*}
Let $x\in \ker (1-A)$, with $||x||=1$. Then there exists a subsequence $x_{n_k}\in \ker (1-A_{n_k})$ such that
\begin{equation*}
||x-x_{n_k}||\to 0
\end{equation*}
as $k\to \infty$.
\end{lemma}
\begin{proof}
Denote $V=\ker (1-A)$ and let $n\in \N_0$. Consider the eigenspaces $V_n:= \ker (1-A_n)$. Let $N_n:=\dim V_n$ and $\{e_{n,m}\}$ an orthonormal basis for $V_n$, with $1\leq m\leq N_n$. 

Let $x\in V$. Then we can decompose
\begin{equation*}
x=\sum_{m=0}^{N_n}  \la x, e_{n,m}\ra e_{n,m}+y_n,
\end{equation*}
with $y_n\in V_n^{\perp}$. Let $x_n:=\sum_{m=0}^{N_n}  \la x, e_{n,m}\ra e_{n,m}$.

We can rearrange terms and use that $x\in V$ and $x_n\in V_n$ to obtain
\begin{equation*}
\begin{split}
y_n=&\:x-x_n\\
=&\: Ax-A_n x_n\\
=&\:Ay_n +(A-A_n) x_n
\end{split}
\end{equation*}
Note that
\begin{equation*}
1=||x||^2=||x_n||^2+||y_n||^2,
\end{equation*}
so
\begin{equation*}
||(A_n-A) x_n||\leq ||A_n-A||\to 0
\end{equation*}
as $n\to \infty$. Furthermore, $\{y_n\}$ is a bounded sequence, so by compactness of $A$, there exists a subsequence $\{y_{n_k}\}$ and $y\in H$ such that
\begin{equation*}
||A y_{n_k}-y||\to 0
\end{equation*}
as $k\to \infty$. Since $y_n=Ay_n +(A-A_n) x_n$ by the above, we can conclude that $||y_{n_k}-y||\to 0$ and $Ay=y$, so $y\in V$.

By compactness of $A_n$, we can moreover conclude that $\textnormal{Ran}\, (1-A_{n_k}^{\dag})$ is closed in $H$ and
\begin{equation*}
y_{n_k}\in V_{n_k}^{\perp}=\textnormal{Ran}\, (1-A_{n_k}^{\dag})
\end{equation*}
Hence, there exists $z_{n_k}\in H$ such that 
\begin{equation*}
y_{n_k}= (1-A_{n_k}^{\dag})z_{n_k}=(1-A^{\dag})z_{n_k}+(A-A_{n_k})^{\dag}z_{n_k}.
\end{equation*}
We have that $||1-A_{n_k}^{\dag}||\geq ||1-A^{\dag}||-||A^{\dag}-A_{n_k}^{\dag}||\to  ||1-A^{\dag}||$ as $k\to \infty$. Without loss of generality, $A\neq 1$, so $ ||1-A^{\dag}||>0$ and there exists a constant $C>0$ such that
\begin{equation*}
||z_{n,k}||\leq \frac{C}{  ||1-A^{\dag}||}||y_{n_k}||\leq \frac{C}{  ||1-A^{\dag}||}.
\end{equation*}
and therefore, $(A-A_{n_k})^{\dag}z_{n,k}\to 0$ as $k\to \infty$. We therefore have that
\begin{equation*}
y= \lim_{k\to\infty}(1-A^{\dag})z_{n_k},
\end{equation*}
so by closedness of $\textnormal{Ran}\, (1-A^{\dag})$, we must have that
\begin{equation*}
y\in \textnormal{Ran}\, (1-A^{\dag})=V^{\perp}.
\end{equation*}
Since we also showed that $y\in V$, we must have that $y=0$ and
\begin{equation*}
||x-x_{n_k}||\to 0
\end{equation*}
as $k\to \infty$.
\end{proof}

\begin{lemma}[Hardy inequalities]
\label{lm:hardy}
Let $q\in \R\setminus \{-1\}$. Let $f: [a,b] \to \R$ be a $C^1$ function with $a,b\geq 0$. Then
\begin{align}
\label{eq:hardy1}
\int_{a}^{b} x^qf^2(x)\,dx\leq &\:4(q+1)^{-2} \int_{a}^{b} x^{q+2}\left|\frac{df}{dx}\right|^2\,dx+ 2b^{q+1}f^2(b),\quad \textnormal{for $q>-1$},\\
\label{eq:hardy2}
\int_{a}^{b} x^qf^2(x)\,dx\leq &\:4(q+1)^{-2} \int_{a}^{b} x^{q+2}\left|\frac{df}{dx}\right|^2\,dx+ 2a^{q+1}f^2(a),\quad \textnormal{for $q<-1$}.
\end{align}
 \end{lemma}
\begin{proof}
See for example the proof of Lemma 2.2 in \cite{aag18}.
\end{proof}
\section{Red-shift estimates in Reissner--Nordstr\"om--de Sitter}
\label{sec:red-shift}
In this section, we include various estimates which are \emph{not} uniform in the surface gravities $\kappa_+$ and $\kappa_c$. These are small variations of the estimates derived in \cite{warn15} that fundamentally rely on the red-shift effect along both the event and cosmological horizon, and that we include for completeness.

Consider the operator
\begin{equation*}
\hat{{L}}_{s,\kappa}:  \mathcal{D}^k(\hat{{L}}_{s,\kappa})\to \widetilde{H}^k,
\end{equation*}
with $ \mathcal{D}^k(\hat{{L}}_{s,\kappa})$ the closure of $C^{\infty}([r_+,r_c]\times \s^2)$ under the graph norm $||f||_{\widetilde{H}^k}+||\hat{L}_{s,\kappa}(f)||_{\widetilde{H}^k}$. By construction, $( \mathcal{D}^k(\hat{{L}}_{s,\kappa}),\hat{{L}}_{s,\kappa})$ is a closed, densely defined operator.

\begin{proposition}
\label{prop:mainred-shiftest}
Let  $\kappa_+=\kappa_c=\kappa$, $s\in \C$ and $k\in \N_0$ such that $\re(s)>-\kappa(\frac{1}{2}+k)$. For suitably small $(\rho_0)_+>0$, $(\rho_0)_c>0$ and $\lambda>0$ suitably large, there exists a constant $C=C(s,\lambda,\kappa_+,\kappa_c)>0$ such that
\begin{equation*}
\begin{split}
||{\hat{\psi}}&||_{H^2([R_0^+,R_0^c]\times \s^2) }+||{\hat{\psi}}||_{H^{k+1}([r_+, R_0^+] \times \s^2)}+||{\hat{\psi}}||_{H^k([R_0^c,r_c]\times \s^2)}\\
\leq&\: C||(\hat{{L}}_{s,\kappa}-\lambda(\lambda+1))({\hat{\psi}})||_{L^2([(R_0^+,R_0^c]\times \s^2) }+C||(\hat{{L}}_{s,\kappa}-\lambda(\lambda+1))({\hat{\psi}})||_{H^k([r_+, R_0^+] \times \s^2)}\\
&+C||(\hat{{L}}_{s,\kappa}-\lambda(\lambda+1))({\hat{\psi}})||_{H^k([R_0^c,r_c]\times \s^2)}.
\end{split}
\end{equation*}
\end{proposition}
\begin{proof}
We multiply equation \eqref{eq:maineqhatphin} with inhomogeneity $\partial_{\rho}^n(r \tilde{f})$ by $\partial_{\rho}\hphi_{(n)}=(\overline{\hphi_{(n+1)}}-s\hat{h}\overline{\hphi}_{(n)})$ to obtain
\begin{equation*}
\begin{split}
\re \left( \partial_{\rho}\hphi_{(n)}\partial_{\rho}^n(r\tilde{f})\right)=&\:2(\re(s)+(n+1)\kappa+O(\rho))|\hphi_{(n+1)}|^2+(\kappa \rho +O(\rho^2))\partial_{\rho}(|\hphi_{(n+1)}|^2) +\re(O(1) \hphi_{(n)}  \hphi_{(n+1)} )\\
&\frac{1}{2}\ell(\ell+1)\partial_{\rho}(|\phi_{(n)}|^2)+\re(O(1) \hphi_{(n-1)}  \hphi_{(n+1)} )+\re(O(1) \hphi_{(n-2)}  \hphi_{(n+1)} ),
\end{split}
\end{equation*}
where the big O notation indicates the behaviour of factors as $\rho\to 0$ and we do not keep track of dependence on $n$ and $s$. We integrate the above equation on $[0,\rho_0]$ to obtain the following estimate for $\rho_0$ suitably small:
\begin{equation*}
\begin{split}
\int_{0}^{\rho_0}& (\re(2s)+(n+1)\kappa)|\hphi_{(n+1)}|^2\,d\rho+ \kappa \rho_0|\hphi_{(n+1)}|^2(\rho_0)+\ell(\ell+1)|\hphi_{(n)}|^2(0)\\\
\leq&\: C\int_{0}^{\rho_0} (|\hphi_{(n)}|+|\hphi_{(n-1)}|+|\hphi_{(n-2)}|)| \hphi_{(n+1)}|\,d\rho\\
&+C\ell(\ell+1)|\hphi_{(n)}|^2(\rho_0)+C\kappa \rho_0|\hphi_{(n)}|^2(\rho_0) +C\int_{0}^{\rho_0} |\partial_{\rho}^n(r\tilde{f})| | \hphi_{(n+1)}|\,d\rho.
\end{split}
\end{equation*}
We apply a Hardy inequality, see Lemma \ref{lm:hardy}, to estimate further:
\begin{equation*}
\begin{split}
\int_{0}^{\rho_0} (\re(2s)+(n+1)\kappa-\epsilon)|\hphi_{(n+1)}|^2\,d\rho\leq&\: C\sum_{k=0}^2 |\hat{\phi}_{(n-k)}|^2(\rho_0)+C\ell(\ell+1)|\hphi_{(n)}|^2(\rho_0)\\
&+C\epsilon^{-1}\int_{0}^{\rho_0} |\partial_{\rho}^n(r\tilde{f})|^2\,d\rho.
\end{split}
\end{equation*}
We can sum over $n$ to obtain for $\re{s}>-(n+\frac{1}{2})\kappa$:
\begin{equation*}
\begin{split}
\sum_{k=0}^n\sum_{k_1+k_2=k}\int_{0}^{\rho_0}(\ell(\ell+1))^{k_2}( |\hphi_{(k_1+1)}|^2+ |\hphi_{(k_1)}|^2)\,d\rho\leq&\: C\ell(\ell+1))^{n+1}|\hat{\phi}_{(0)}|^2(\rho_0)\\
&+C\epsilon^{-1}\sum_{k=0}^n\sum_{k_1+k_2=k}\int_{0}^{\rho_0}(\ell(\ell+1))^{k_2} |\partial_{\rho}^{k_1}(r\tilde{f})|^2\,d\rho,
\end{split}
\end{equation*}
with $C$ a constant independent of $\ell$. We can estimate the boundary term at $\rho=\rho_0$ using \eqref{eq:lowfreqellipticest} and by taking $\ell$ suitably large, we are left with
\begin{equation*}
\begin{split}
\int_{R_0^+}^{R_0^c}& |\partial_r^2{\hat{\psi}}|^2+|\partial_r{\hat{\psi}}|^2+(1+\ell^2(\ell+1)^2)|{\hat{\psi}}|^2\,dr+\sum_{\star\in\{+,c\}} \sum_{k=0}^{n+1}\sum_{k_1+k_2=k}\int_{0}^{(\rho_0)_{\star}}(\ell(\ell+1))^{k_2} |\hphi_{(k_1)}|^2\,d\rho_{\star}\\
\leq&\: C\int_{r_+}^{r_c}|\tilde{f}|^2\,dr+\sum_{k=0}^n\sum_{k_1+k_2=k}\int_{0}^{\rho_0}(\ell(\ell+1))^{k_2} |\partial_{\rho}^{k_1}(r\tilde{f})|^2\,d\rho.
\end{split}
\end{equation*}
\end{proof}

From Proposition \ref{prop:mainred-shiftest}, it follows immediately that:
\begin{corollary}
\label{cor:red-shift}
For suitably large $\lambda>0$, $\hat{{L}}_{s,\kappa}-\lambda(\lambda+1): \mathcal{D}^k(\hat{{L}}_{s,\kappa})\to \widetilde{H}^k$ is an injective linear operator with a closed range and $\mathcal{D}^k(\hat{{L}}_{s,\kappa})\subseteq \widetilde{H}^{k+1}_2$.
\end{corollary}
We now consider the adjoint operator of $\hat{{L}}_{s,\kappa}$.
\begin{lemma}
The linear operator 
\begin{equation*}
\hat{{L}}_{s,\kappa}^{\dag}:=\hat{{L}}_{-\overline{s},\kappa}: \mathcal{D}^k(\hat{{L}}^{\dag}_{s,\kappa})\to \widetilde{H}^k,
\end{equation*}
 with $\mathcal{D}^k(\hat{{L}}^{\dag}_{s,\kappa})$ the closure of 
\begin{equation*}
\{f\in C^{\infty}([r_+,r_c]\times \s^2)\:\big|\: \partial_r^kf(r_+)=\partial_r^kf(r_c)=0\:\textnormal{for all $k\in \N_0$}\} 
\end{equation*}
under the graph norm $||f||_{\widetilde{H}^k}+||{L}_{s,\kappa}(f)||_{\widetilde{H}^k}$, is the adjoint of $\hat{{L}}_{s,\kappa}$ with respect to $\widetilde{H}^k$.
\end{lemma}
\begin{proof}
See Lemma 4.5 of \cite{warn15}.
\end{proof}

\begin{proposition}
\label{prop:adjred-shiftest}
Let  $\kappa_+=\kappa_c=\kappa$ and $\re(s)>\frac{1}{2}\kappa $. For suitably small $(\rho_0)_+>0$, $(\rho_0)_c>0$ and $\lambda>0$ suitably large, there exists a constant $C=C(s,\lambda,\kappa)>0$ such that
\begin{equation}
\label{eq:adjest1}
||{\hat{\psi}}||_{\widetilde{H}^1_2}\leq C  ||(\hat{{L}}_{s,\kappa}^{\dag}-\lambda(\lambda+1))({\hat{\psi}})||_{\widetilde{H}^0}.
\end{equation}
Furthermore, for $k\in \N_0$ and $\re(s)<(\frac{1}{2}+k)\kappa$, there exist $(\rho_0)_+>0$, $(\rho_0)_c>0$ suitably small, $\lambda>0$ suitably large and a constant $C=C(s,\lambda,\kappa,k)>0$ such that
\begin{equation}
\label{eq:adjest2}
||{\hat{\psi}}||_{\widetilde{H}^{k+1}_2}\leq C  ||(\hat{{L}}_{s,\kappa}^{\dag}-\lambda(\lambda+1))({\hat{\psi}})||_{\widetilde{H}^k}.
\end{equation}
\end{proposition}
\begin{proof}
In order to prove \eqref{eq:adjest1}, we proceed as in the proof of Proposition \ref{prop:mainred-shiftest} for $n=0$, with $s$ replaced by $-\overline{s}$ but we multiply the equation with $\partial_{\rho}\hat{\phi}$ rather than $-\partial_{\rho}\hat{\phi}$.

We obtain \eqref{eq:adjest2} by repeating exactly the proof of Proposition \ref{prop:mainred-shiftest} with $s$ replaced by $-\overline{s}$.
\end{proof}

\begin{proposition}
\label{prop:poskappainvert}
Let $k\in \N_0$ and $\re(s)>-(\frac{1}{2}+k)\kappa$. Then, for $\lambda\in \N_0$ suitably large, the operator
\begin{equation*}
\hat{{L}}_{s,\kappa}-\lambda(\lambda+1):   \mathcal{D}^k(\hat{{L}}_{s,\kappa})\to \widetilde{H}^k
\end{equation*}
is invertible and
\begin{equation*}
\mathcal{D}^k(\hat{{L}}_{s,\kappa})\subset \widetilde{H}^{k+1}_2.
\end{equation*}
\end{proposition}
\begin{proof}
We first investigate the existence and uniqueness of solutions to 
\begin{equation}
\label{eq:bijection}
(\hat{{L}}_{s,\kappa}-\lambda(\lambda+1))({\hat{\psi}})=\tilde{f},
\end{equation}
 with $\tilde{f} \in C^{\infty}([r_+,r_c]\times \s^2)$. Uniqueness follows from Corollary \ref{cor:red-shift}. Furthermore, given any $k\in \N$, $\hat{{L}}_{s,\kappa}-\lambda(\lambda+1):  \mathcal{D}^k(\hat{{L}}_{s,\kappa})\to \widetilde{H}^k$ is a closed operator with closed range provided $\re(s)>-(\frac{1}{2}+k)\kappa$. Hence, by a standard functional analytic argument, see for example Lemma 4.4 of \cite{warn15}, 
 \begin{equation*}
 (\hat{{L}}_{s,\kappa}-\lambda(\lambda+1))(  \mathcal{D}^k(\hat{{L}}_{s,\kappa}))=\widetilde{H}^k\quad  \textnormal{if}\quad \hat{{L}}_{s,\kappa}^{\dag}-\lambda(\lambda+1):  \mathcal{D}^k(\hat{{L}}^{\dag}_{s,\kappa})\to \widetilde{H}^k \quad \textnormal{is injective}.
 \end{equation*}

By Proposition \ref {prop:adjred-shiftest}, injectivity of the adjoint holds for all $s\in \C$ provided we restrict to $k\geq 1$.

Since $\tilde{f} \in C^{\infty}([r_+,r_c]\times \s^2)$, we can take $k\geq 1$ and conclude that for $\re(s)>-(\frac{1}{2}+k)\kappa$, \eqref{eq:bijection} admits a solution in $ \mathcal{D}^k(\hat{{L}}_{s,\kappa})$. Hence, 
\begin{equation*}
(\hat{{L}}_{s,\kappa}-\lambda(\lambda+1))^{-1}:  C^{\infty}([r_+,r_c]\times \s^2)\to  \mathcal{D}^k(\hat{{L}}_{s,\kappa})
\end{equation*}
is well-defined if $\re(s)>-(\frac{1}{2}+k)\kappa$ with $k\geq 1$ and $\lambda$ is suitably large. By Corollary \ref{cor:red-shift} we in fact have that for all  $k\geq 0$ and $\re(s)>-(\frac{1}{2}+k)\kappa$ and $\lambda>0$ suitably large
\begin{equation*}
||(\hat{{L}}_{s,\kappa}-\lambda(\lambda+1))^{-1}(\tilde{f})||_{\widetilde{H}^{k+1}_2}\leq C ||\tilde{f}||_{\widetilde{H}^k},
\end{equation*}
so for all $k\geq 0$, $(\hat{{L}}_{s,\kappa}-\lambda(\lambda+1))^{-1}$ admits a unique extension as an operator from $\widetilde{H}^k$ to $\mathcal{D}^k(\hat{{L}}_{s,\kappa})$, provided $\re(s)>-(\frac{1}{2}+k)\kappa$, and moreover we have that
\begin{equation*}
\mathcal{D}^k(\hat{{L}}_{s,\kappa})\subset \widetilde{H}^{k+1}_2.
\end{equation*}
It immediately follows that the extended operator must be the inverse of $\hat{{L}}_{s,\kappa}-\lambda(\lambda+1)$.
\end{proof}

%\bibliographystyle{plain}
%\bibliography{../bibliography}

\end{document}